\documentclass[pra,aps,floatfix,amsmath,superscriptaddress,twocolumn,longbibliography,nofootinbib]{revtex4-2}
\usepackage{amssymb,enumerate}

\usepackage{amssymb,enumerate}
\usepackage{graphicx}
\usepackage{graphics}
\usepackage{amsmath}
\usepackage{amsthm,bbm}
\usepackage{color}
\usepackage{dsfont}
\usepackage{hyperref}

\usepackage[capitalise,nameinlink]{cleveref}

\usepackage{qcircuit}

\usepackage{graphicx}

\usepackage[dvipsnames]{xcolor}
\hypersetup{
    colorlinks,
    linkcolor={red!50!black},
    citecolor={blue!50!black},
    urlcolor={blue!80!black}
}

\usepackage{xfrac}

\usepackage{amsmath}
\usepackage{tikz-cd}
\usepackage{empheq}

\usepackage{enumitem}

\usepackage{cleveref}
\usepackage{nameref}

\usepackage{graphicx}

\usepackage{hyperref}
\usepackage[capitalise]{cleveref}
\usepackage{youngtab}

\usepackage{tikz}
\usepackage{mathtools}

\newcommand{\mE}{\mathcal{E}}

\newcommand{\mH}{\mathcal{H}}
\newcommand{\mI}{\mathcal{I}}

\newcommand{\mM}{\mathcal{M}}

\newcommand{\Cbb}{\mathbb{C}}

\newcommand{\Ibb}{\mathbb{I}}

\newcommand{\Nbb}{\mathbb{N}}

\newcommand{\Pbb}{\mathbb{P}}

\newcommand{\Rbb}{\mathbb{R}}
\newcommand{\Sbb}{\mathbb{S}}

\newcommand{\Zbb}{\mathbb{Z}}

\newcommand{\<}{\langle}
\renewcommand{\>}{\rangle}

\newtheorem{thm}{Theorem}

\usepackage{lipsum}
\usepackage{lmodern}
\usepackage{tcolorbox}

\usepackage{amsfonts}
\usepackage{graphicx,graphics,epsfig,times,bm,bbm,amssymb,amsmath,amsfonts,mathrsfs}
\usepackage[normalem]{ulem}
\usepackage{setspace}

\usepackage{subcaption}

\usepackage{dsfont}
\usepackage{braket}
\usepackage{upgreek}
\usepackage{tikz}
\usepackage{natbib}
\usepackage{chngcntr}

\newtheorem{theorem}{Theorem}

\newtheorem{definition}[theorem]{Definition}
\newtheorem{example}[theorem]{Example}

\newtheorem{lemma}[theorem]{Lemma}

\makeatletter 
\renewcommand\onecolumngrid{%
  \do@columngrid{one}{\@ne}%
  \def\set@footnotewidth{\onecolumngrid}%
  \def\footnoterule{\kern-6pt\hrule width 1.5in\kern6pt}%
}
\makeatother

\newcommand{\bes} {\begin{subequations}}
\newcommand{\ees} {\end{subequations}}
\newcommand{\bea} {\begin{eqnarray}}
\newcommand{\eea} {\end{eqnarray}}
\newcommand{\be} {\begin{equation}}
\newcommand{\ee} {\end{equation}}

\def\>{\rangle}
\def\<{\langle}

\newcommand{\abs}[1]{\lvert #1 \rvert}

\newcommand{\ignore}[1]{}

\crefformat{step}{#2Step #1#3}

\begin{document}

\title{Optimal Linear-Rate Conversion of Unknown Mixed Qubit States via SWAP Tests}

\author{Sujay Kazi}
\affiliation{Duke Quantum Center and Department of Electrical and Computer Engineering, Duke University, Durham, NC 27708, USA}

\author{Iman Marvian}
\affiliation{Duke Quantum Center and Department of Electrical and Computer Engineering, Duke University, Durham, NC 27708, USA}
\affiliation{Department of Physics, Duke University, Durham, NC 27708, USA}

\begin{abstract}
By consuming multiple copies of an unknown qubit state, one can modify its purity while preserving the direction of its Bloch vector. We determine the maximum linear rate at which qubit states of different purities can be interconverted, allowing a nonzero error, quantified, for instance, by the trace distance, provided that it vanishes in the limit of infinitely many copies. Interestingly, the optimal conversion rate is determined by the two eigenvalues of the complex right-logarithmic-derivative (RLD) Fisher information matrix associated with $\mathrm{SU}(2)$ rotations of the qubit state.  When the output qubits have higher purity, corresponding to concentration, the optimal rate is given by the ratio of the maximum eigenvalues of the input and output RLD matrices. In contrast, when the output qubits have lower purity, corresponding to dilution, the optimal rate is given by the ratio of their minimum eigenvalues. Remarkably, both concentration and dilution can be implemented using SWAP tests as the only nontrivial two-qubit measurement primitive, together with ancillary qubits initially prepared in maximally mixed states, without requiring any additional two-qubit gates. 
Our work thus provides a novel operational interpretation of the full complex RLD Fisher information matrix. Crucially, its antisymmetric, purely imaginary part encodes geometric information beyond the statistical distance between density operators and plays an essential role in determining the optimal state-conversion rates.

\end{abstract}

\maketitle

What is the information content of a qubit? Quantum information theorists have approached this question from several complementary perspectives. One operational answer comes from quantum communication theory: the Holevo bound \cite{holevo1973bounds, holevo2019} implies that transmitting a single qubit can communicate at most one classical bit, whereas, in the presence of shared entanglement, superdense coding allows the same transmitted qubit to communicate two classical bits. Another perspective is \emph{parameter estimation}: given many copies of
$\cos({\theta}/{2})\ket{0} + e^{i\phi}\sin({\theta}/{2})\ket{1}$, 
how well can one estimate $\theta$ or $\phi$? This question is quantified by the quantum Fisher information \cite{helstrom1969,holevo1973,holevo2019,braunsteincaves1994,barndorff2003}, which for a single qubit is an order-one quantity, up to convention-dependent normalization. 

Here, we focus on a third operational perspective, based on \emph{preparation cost}. Suppose one wishes to prepare a target state, but is given only copies of some imperfect reference state, for instance, copies of the target state corrupted by dephasing or depolarizing noise. How many copies of the reference state are required to prepare the desired target state within error $\varepsilon$?

The pioneering work of Cirac et al. \cite{Cirac1999} found the optimal procedure for transforming $N$ identical, uncorrelated, or i.i.d., copies of a single-qubit state 
\be\label{family}
\rho_\lambda=
\lambda\Psi+(1-\lambda) \frac{\mathbb{I}}{2}\ ,
\ee
into a single qubit that points in the same direction but is as close to pure state $\Psi=|\Psi\rangle\langle\Psi|$ as possible \cite{Cirac1999}. Namely, they showed that the optimal achievable fidelity is
\be\label{Eq:jdjd}
\langle \Psi|\mathcal{E}_{\text{opt}}(\rho_\lambda^{\otimes N})|\Psi\rangle = 1 - \frac{1-\lambda}{2\lambda^2}\frac{1}{N} + \mathcal{O}(\frac{1}{N^2})\ ,
\ee
where $\mathcal{E}_{\mathrm{opt}}$ is the quantum channel, i.e., completely-positive trace-preserving map, that maximizes the output fidelity and is independent of the unknown state $|\Psi\rangle$. Thus, exact preparation from noisy reference copies requires infinitely many copies. This observation raises the following fundamental question:\vspace{2mm}

\noindent\emph{Question 1:} Which information-theoretic quantity quantifies the asymptotic unreachability of pure states, characterizes the optimal achievable error in single-output distillation, and, in particular, reproduces the asymptotic behavior in \cref{Eq:jdjd}? \vspace{2mm} 

To the best of our knowledge, this question remains unanswered more than 25 years after the work of Cirac et al. \cite{Cirac1999}.

The preceding discussion concerns distillation, also known as purification, to a single output copy. More generally, one may consider \emph{linear-rate conversion}: Given $N$ copies of $\rho_{\lambda_{\rm in}}$, the goal is to produce $M_N$ copies of $\rho_{\lambda_{\rm out}}$ through a quantum channel independent of the unknown state $\Psi$, with an error that vanishes as $N\to\infty$. The achievable conversion rate is then $R=\lim_{N\to\infty}M_N/N$.  Depending on the purity levels, the conversion may correspond either to dilution, with $\lambda_{\rm in}>\lambda_{\rm out}$, or to concentration, with $\lambda_{\rm in}<\lambda_{\rm out}$. Dilution trades purity for a larger number of less pure systems, whereas concentration compresses the directional information into fewer systems of higher purity. These two regimes can be advantageous under different resource constraints, such as per-system transmission, storage, or measurement costs. For instance, with low-quality measurement devices, it may be useful to dilute the available qubits into a larger ensemble of lower-purity states, so that each noisy measurement corrupts less information. By contrast, when transmitting qubits incurs a per-qubit cost, one may prefer to concentrate the information into fewer, higher-purity qubits.

\vspace{2mm}
\noindent\emph{Question 2:} Given copies of the unknown noisy state
$\rho_{\lambda_{\rm in}}$, what is the maximal asymptotic rate at which they
can be converted into copies of $\rho_{\lambda_{\rm out}}$ with vanishing
error? Which protocol achieves this rate, and which information-theoretic
quantities determine the optimal rate?
\vspace{2mm}

Unlike the single-output distillation problem previously considered by Cirac
et al. \cite{Cirac1999}, determining the optimal rate of asymptotic state
conversion, and constructing protocols that achieve this rate, has remained an
open problem. This question is relevant to a broad range of areas in quantum
information science, including quantum reference frames and the resource theory of asymmetry \cite{BartlettRefFrames,gour2008,marvian2013,marvian2012}, quantum thermodynamics \cite{lostaglio2015quantum, lostaglio2015description, korzekwa2016extraction, gour2018quantum}, and quantum learning theory \cite{lloyd2014quantum, marvian2025efficient}.

A closely related variant arises when the state $\rho_{\lambda_{\rm in}}$ is
known, but the allowed quantum channels are required to respect ${\rm SU(2)}$
symmetry. More precisely, suppose that a channel $\mathcal{E}$ from $n$ input
qubits to $m$ output qubits satisfies the ${\rm SU(2)}$-covariance condition
\be
\mathcal{E}(U^{\otimes n}(\cdot){U^\dag}^{\otimes n})=U^{\otimes m} \mathcal{E}(\cdot){U^\dag}^{\otimes m}\ ,
\ee
for every single-qubit unitary $U$.
A standard symmetrization, or twirling, argument shows that, for the tasks defined in Questions 1 and 2, the optimal channel can be chosen to be covariant.

State conversion under covariant channels is a central topic in the resource theory of asymmetry \cite{BartlettRefFrames,gour2008,marvian2013,marvian2012}, which provides an operational framework for quantifying the asymmetry of quantum states and characterizing their interconversion under symmetry constraints. A function that is non-increasing under covariant channels is called an asymmetry measure (one also typically requires such a measure to vanish on symmetry-invariant states). In this framework, one can formulate the covariant counterpart of Question 2 as follows: Given copies of a known state $\rho_{\lambda_{\rm in}}$, what is the maximal asymptotic rate at which they can be converted into copies of $\rho_{\lambda_{\rm out}}$ with vanishing error using covariant channels? Which asymmetry measures determine the optimal conversion rate?
 For recent progress on asymptotic state conversion in this framework, see, e.g., Refs \cite{tajima2022,yamaguchi2026, yamaguchi2023beyond, marvian2022, marvian2020,kazi2025, yadavalli2025}.

\vspace{2mm}

In this Letter, we resolve both Questions 1 and 2. We determine the optimal asymptotic rates of linear-rate state conversion and construct explicit protocols that achieve them. Depending on whether $\lambda_{\rm in}>\lambda_{\rm out}$ (dilution) or $\lambda_{\rm in}<\lambda_{\rm out}$ ({concentration}), the optimal transformations are realized through suitable combinations of elementary and well-known quantum-information subroutines, including Schur sampling and optimal quantum cloner \cite{Werner1998}. Remarkably, we show that the entire protocol can be implemented using only a simple two-qubit measurement: the SWAP test, also known as the singlet--triplet measurement.

Most notably, we show that a single information-theoretic quantity governs both the optimal asymptotic conversion rate and the leading-order error in the single-output purification problem of Cirac et al. \cite{Cirac1999}. This quantity is the \emph{right-logarithmic-derivative (RLD) Fisher information matrix}, a lesser-known variant of quantum Fisher information whose imaginary, antisymmetric component captures the noncommutativity of the underlying quantum statistical model. Our results therefore provide a unified information-theoretic characterization of both single-output distillation and the asymptotic interconversion of noisy quantum states.\\

\noindent\emph{QFI metrics.}
For a smooth parametrized family of full-rank states $\rho(\vec{x})$, the RLD
Fisher information matrix is defined entrywise by
\be
\operatorname{RLD}_{\mu\nu}[\rho](\vec{x})
=
\operatorname{Tr}\left[
\bigl(\partial_\mu\rho(\vec{x})\bigr)
\rho(\vec{x})^{-1}
\bigl(\partial_\nu\rho(\vec{x})\bigr)
\right].
\ee
Petz showed that the RLD Fisher information is monotone under data processing
\cite{petz1996}: for any parameter-independent quantum channel $\mathcal{E}$,
\be\label{data}
\operatorname{RLD}[\mathcal{E}(\rho)](\vec{x})
\leq
\operatorname{RLD}[\rho](\vec{x}),
\ee
where $A\leq B$ means that $B-A$ is positive semidefinite. Furthermore, the
RLD Fisher information induces a Riemannian metric on the space of full-rank
density operators. The squared length associated with an infinitesimal real
displacement $d\vec{x}$ is
\be\label{Eq:length}
ds^2
=
\sum_{\mu,\nu}
{\rm{RLD}}_{\mu\nu}[\rho](\vec{x})\ dx_\mu dx_\nu\ .
\ee
This length provides a local measure of the statistical distance, or
distinguishability, between the density operators $\rho(\vec{x})$ and
$\rho(\vec{x}+d\vec{x})$. In the special case where the density operators
$\rho(\vec{x})$ commute in a neighborhood of $\vec{x}$, the RLD matrix reduces
to the classical Fisher information matrix associated with the probability
distribution defined by the eigenvalues of $\rho(\vec{x})$.

Classically, \v{C}encov's theorem~\cite{Chentsov1982}  states that, up to an overall normalization,
the Fisher information is the unique Riemannian metric that is monotone under
stochastic maps. In quantum theory, however, there is an
entire family of such metrics. Building on the earlier work of Morozova and
\v{C}encov \cite{morozova1991}, Petz \cite{petz1996} classified all Riemannian
metrics on the space of quantum states that are monotone under quantum
channels, or equivalently, satisfy a data-processing inequality of the form
given in \cref{data}. What is commonly called \emph{the} quantum Fisher information in quantum
metrology and sensing is another member of this family, namely, the
\emph{symmetric-logarithmic-derivative (SLD) Fisher information matrix}. This
matrix is real and symmetric and, in the single-parameter setting, quantifies
the ultimate precision allowed by quantum mechanics
\cite{helstrom1969,holevo1973,holevo2019,braunsteincaves1994,barndorff2003}. Beyond the SLD QFI, applications of
other QFI metrics have been comparatively limited.

For a general family of density operators, however, the RLD matrix contains a
purely imaginary antisymmetric part. For instance, consider the two-parameter
family of single-qubit states $\rho_\lambda(\theta,\phi)
=
\lambda\Psi(\theta,\phi)
+
(1-\lambda){\mathbb{I}}/{2}$, 
at a fixed purity level $0\leq\lambda<1$, where $\Psi(\theta,\phi)$ is the
pure-state projector to $ \ket{\Psi(\theta,\phi)} \coloneqq \cos({\theta}/{2})\ket{0} + e^{i\phi}\sin({\theta}/{2})\ket{1}$. The RLD matrix with respect to the parameters $\theta$ and $\phi$,
at $\theta=\pi/2$ and $\phi=0$, 
is
\small{\begin{equation}
\rm{RLD}(\lambda)
=
\frac{\lambda^2}{1-\lambda^2}
\begin{bmatrix}
1 & +\lambda i \\
-\lambda i & 1
\end{bmatrix}\ ,
\end{equation}} 
with the maximum and minimum eigenvalues
\begin{equation}
a(\lambda)=
\frac{\lambda^2}{1-\lambda}\ ,\qquad
\ 
b(\lambda)
=
\frac{\lambda^2}{1+\lambda}\ .
\end{equation}
At first glance, the presence of an imaginary antisymmetric part in the RLD
matrix, in contrast to the real symmetric SLD QFI matrix, may appear curious
or irrelevant. Indeed, this imaginary part does not contribute to the length
$ds^2$ in \cref{Eq:length}. 
Nevertheless, it captures information about the noncommutativity of the quantum statistical model and, as we show next, plays an essential role in determining the optimal state-conversion rates. Before turning to this point, it is worth noting that all QFI metrics are additive for tensor-product states $\rho(\vec{x})=\bigotimes_j\rho_j(\vec{x})$. In particular, ${\rm{RLD}}_{\mu\nu}[\rho](\vec{x})
=
\sum_j
{\rm{RLD}}_{\mu\nu}[\rho_j](\vec{x})$.

\noindent\emph{Single-copy distillation.}
We next show how the monotonicity of the RLD QFI determines the leading-order
error in the single-copy distillation problem solved by Cirac et al. more than
25 years ago \cite{Cirac1999}.  Suppose that $N$ copies of
$\rho_{\lambda_{\rm in}}$ are converted into a single copy of
$\rho_{\lambda_{\rm out}}$. Monotonicity and additivity of the RLD matrix imply
that, $N\times{\rm RLD}(\lambda_{\text{in}})
\ge {\rm RLD}(\lambda_{\text{out}})$, which in turn means the maximum eigenvalues satisfy
\be
N\times 
a(\lambda_{\rm in})
\geq
a(\lambda_{\rm out})
=
\frac{\lambda_{\rm out}^2}{1-\lambda_{\rm out}}\ .
\ee
Since
$\mathcal{E}_{\rm opt}
\bigl(\rho_{\lambda_{\rm in}}^{\otimes N}\bigr)=\rho_{\lambda_{\rm out}}
=
\lambda_{\rm out}\Psi
+
(1-\lambda_{\rm out}){\mathbb{I}}/{2}$,  
we immediately obtain the following lower bound on its infidelity with the
desired pure state:
\be
1-
\langle\Psi|
\mathcal{E}_{\rm opt}
\bigl(\rho_{\lambda_{\rm in}}^{\otimes N}\bigr)
|\Psi\rangle
=
\frac{1-\lambda_{\rm out}}{2}
\geq
\frac{\lambda_{\rm out}^2}
{2N\times a(\lambda_{\rm in})}\ .
\ee

In the large-$N$ regime, $\lambda_{\rm out}$ converges to $1$, so
the leading-order term on the right-hand side is
$({1-\lambda_{\rm in}})/(
{2N\lambda_{\rm in}^2})$, 
which agrees with the leading-order error term in \cref{Eq:jdjd}. 
Thus, we find that the larger eigenvalue of the ${\rm RLD}$ Fisher information matrix is the information-theoretic quantity governing the leading-order error in optimal single-output distillation, thereby answering Question 1.

Crucially, the imaginary part of the ${\rm RLD}$ matrix must be retained to recover this exact leading-order coefficient. If one instead considers only its real symmetric part, which is also monotone and additive, then
${\rm Re}[{\rm RLD}(\lambda)]
=\lambda^2(1-\lambda^2)^{-1}\mathbb{I}$,
and both eigenvalues are equal to
$\lambda^2(1-\lambda^2)^{-1}$.
In the large-$N$ regime, monotonicity of this quantity yields only the weaker leading-order bound
$({1-\lambda_{\rm in}^2})/({4N\lambda_{\rm in}^2})$,
which is smaller than the tight ${\rm RLD}$ coefficient by a factor of
$(1+\lambda_{\rm in})/2$.
Thus, although the imaginary antisymmetric part does not contribute to the associated statistical length $ds^2$, it is essential for recovering the exact leading-order distillation error for general $\lambda_{\rm in}$.

\noindent\emph{Linear-rate conversion.}
Next, we turn to the problem of linear-rate conversion in Question 2. 
Suppose that, for all sufficiently large $N$, there exists a quantum channel
$\mathcal{E}_N$, independent of $\vec{x}$, that converts
$\rho(\vec{x})^{\otimes N}$ into
$\sigma(\vec{x})^{\otimes\lfloor RN\rfloor}$.
Because exact conversion is often impossible and generally impractical, we follow the standard approach in Shannon information theory and quantum resource theories by allowing a nonzero error, quantified by the trace distance, provided that it vanishes asymptotically, i.e., $ \lim_{N\to\infty} \|\mE_N\left(\rho_{\lambda_{\text{in}}}^{\otimes N}\right) - \rho_{\lambda_{\text{out}}}^{\otimes\lfloor RN\rfloor}\|_1 = 0$, where as discussed above, the channel $\mathcal{E}_N$ may be chosen to be covariant.
By Helstrom's theorem \cite{}, this condition guarantees that the actual output state and the desired state become asymptotically indistinguishable as $N\rightarrow\infty$ (instead of trace distance, one could equivalently require the fidelity to converge to one).

\noindent\emph{Upper bound on the rate.}
We begin by establishing an upper bound on the interconversion rate. First,
instead of considering the conversion of
$\rho_{\lambda_{\text{in}}}^{\otimes N}$ into
$\rho_{\lambda_{\text{out}}}^{\otimes\lfloor RN\rfloor}$ with asymptotically
vanishing error, suppose that we impose the stronger requirement that the
transformation be realized with exactly zero error. The additivity of the QFI matrix for tensor-product states, together with its
monotonicity under quantum channels, then implies $N\times{\rm RLD}(\lambda_{\text{in}})
\ge
\lfloor RN\rfloor\times{\rm RLD}(\lambda_{\text{out}})$. 
Because the ${\rm RLD}$ matrices for this family share the same eigenvectors,
this matrix inequality is equivalent to
\be\label{bound}
\begin{aligned}
R \le \min\left\{
\frac{a(\lambda_{\mathrm{in}})}{a(\lambda_{\mathrm{out}})},
\frac{b(\lambda_{\mathrm{in}})}{b(\lambda_{\mathrm{out}})}
\right\}
&=
\left\|
{\rm RLD}(\lambda_{\rm out})
\left[{\rm RLD}(\lambda_{\rm in})\right]^{-1}
\right\|^{-1}
\\ &=:R^{\rm opt}(\lambda_{\mathrm{in}},\lambda_{\mathrm{out}})\ .
\end{aligned}
\ee
where $\|\cdot\|$ denotes the operator norm. For concentration, i.e., when
$\lambda_{\mathrm{out}}>\lambda_{\mathrm{in}}$, the right-hand side is equal to
the ratio of the maximum eigenvalues of the input and output RLD matrices,
whereas for dilution, it is equal to the ratio of their minimum eigenvalues.

However, this simple argument overlooks a fundamental feature of the conversion
problem: the transformation need not be exact. Allowing an error $\epsilon_N$
can drastically change the ${\rm RLD}$ Fisher information. In particular, even
when the state has full rank, so that the ${\rm RLD}$ Fisher information is
well defined and finite, an $\mathcal{O}(\epsilon_N)$ perturbation in trace distance can
produce an $\mathcal{O}(N^2\epsilon_N)$ change in the ${\rm RLD}$ Fisher information of
an $N$-copy state. Consequently, the ${\rm RLD}$ Fisher information per copy can change by
$\mathcal{O}(N\epsilon_N)$, which need not vanish in the asymptotic limit. Therefore,
monotonicity alone does not directly extend the exact-conversion bounds to
transformations with asymptotically vanishing error. Proving that the bounds
in \cref{bound} remain valid in the presence of vanishing error requires a new
technique, which is developed in the Supplemental Material (SM). \\

\noindent\emph{Optimal protocols.}
Interestingly, these bounds on the optimal conversion rate can be
achieved using a combination of relatively simple and well-known protocols. Both the concentration and dilution
protocols consist of three steps and share the same first step: unitary Schur
sampling, a standard subroutine in quantum information that was also used in
\cite{Cirac1999}. This procedure determines the total angular momentum of the
input qubits or, equivalently, projects them onto a Schur--Weyl sector labeled
by an irreducible representation of the permutation group $\mathbb{S}_N$ and
removes the corresponding multiplicity degrees of freedom. Recent work by Brahmachari et al.~\cite{Brahmachari2025} showed that unitary Schur sampling for qubits can be implemented using random two-qubit SWAP tests, or triplet--singlet measurements, which project each pair onto the singlet state $(|01\rangle-|10\rangle)/\sqrt{2}$ or its orthogonal subspace. Pairs found in the singlet state are discarded. After approximately $2N\ln(N/\varepsilon)$ randomly chosen SWAP tests, all singlet components are removed with probability at least $1-\varepsilon$.

The only other nontrivial subroutine is Werner's optimal cloner \cite{Werner1998}, required for dilution, which we also show can be implemented using SWAP tests (see below). Our concentration and dilution protocols can therefore be summarized as follows; see also Fig. \ref{fig:qubit-conversion-protocols-unified-presentation} and the SM for further details.
\vspace{2mm}

\noindent\textbf{(i) Unitary Schur sampling.}
Apply the unitary Schur transform to compress the $N$ input qubits into approximately $N\lambda_{\rm in}$ qubits supported on the totally symmetric subspace and carrying the directional, i.e., $\mathrm{SU}(2)$, information.

\noindent\textbf{(ii) Rescaling the number of qubits.}
This step differs between {concentration} and dilution.
\vspace{-2mm}
\begin{itemize}
\item \emph{{Concentration}} ($\lambda_{\rm out}>\lambda_{\rm in}$). 
Reduce the number of qubits by retaining the fraction
$\frac{\lambda_{\rm in}}{1-\lambda_{\rm in}}
\Big/
\frac{\lambda_{\rm out}}{1-\lambda_{\rm out}}$
and discarding the remaining qubits.
\vspace{-2mm}
\item \emph{Dilution} ($\lambda_{\rm out}<\lambda_{\rm in}$). 
Increase the number of qubits using the optimal cloner \cite{Werner1998} by the factor
$\frac{\lambda_{\text{in}}}{1+\lambda_{\text{in}}}\Big/
\frac{\lambda_{\text{out}}}{1+\lambda_{\text{out}}}$.
\end{itemize}
\vspace{-2mm}
\noindent\textbf{(iii) Inverse Schur sampling.}
Introduce ancillary qubits in singlet states, thereby increasing the total
number of qubits by a factor of approximately $\lambda_{\rm out}^{-1}$, and
apply a uniformly random permutation to all the qubits. This procedure realizes
inverse Schur sampling, and the resulting output state is close to an i.i.d.\
collection of qubits.
\vspace{2mm}

It is straightforward to verify that the total number of output qubits is
approximately $N\times R^{\rm opt}(\lambda_{\mathrm{in}},\lambda_{\mathrm{out}})$. The nontrivial part, proved in the
Supplementary Material, is that the error of these protocols vanishes as
$N$ grows.

It turns out that the first and third steps are both reversible. In particular, as $\lambda_{\rm out}\rightarrow\lambda_{\rm in}$, the intermediate step becomes trivial and the overall protocol recovers the original state. Thus, the only irreversible step is the intermediate one, involving either discarding or optimal cloning, which results in a loss of asymmetry (directional information). Such an irreversible step is necessary because there are two inequivalent asymmetry (information) measures, namely, the minimum and maximum eigenvalues of the RLD Fisher information matrix. If $\lambda_{\rm out}\neq\lambda_{\rm in}>0$, then at the optimal conversion rate, at least one of the two inequalities in \cref{bound} is strict. Consequently, either the minimum or the maximum eigenvalue of the total RLD Fisher information matrix is reduced by an amount proportional to $N$. The conversion process therefore cannot be fully reversible.

\vspace{2mm}

\begin{figure*}[t] 
    \centering
    \includegraphics[width=0.7\textwidth]{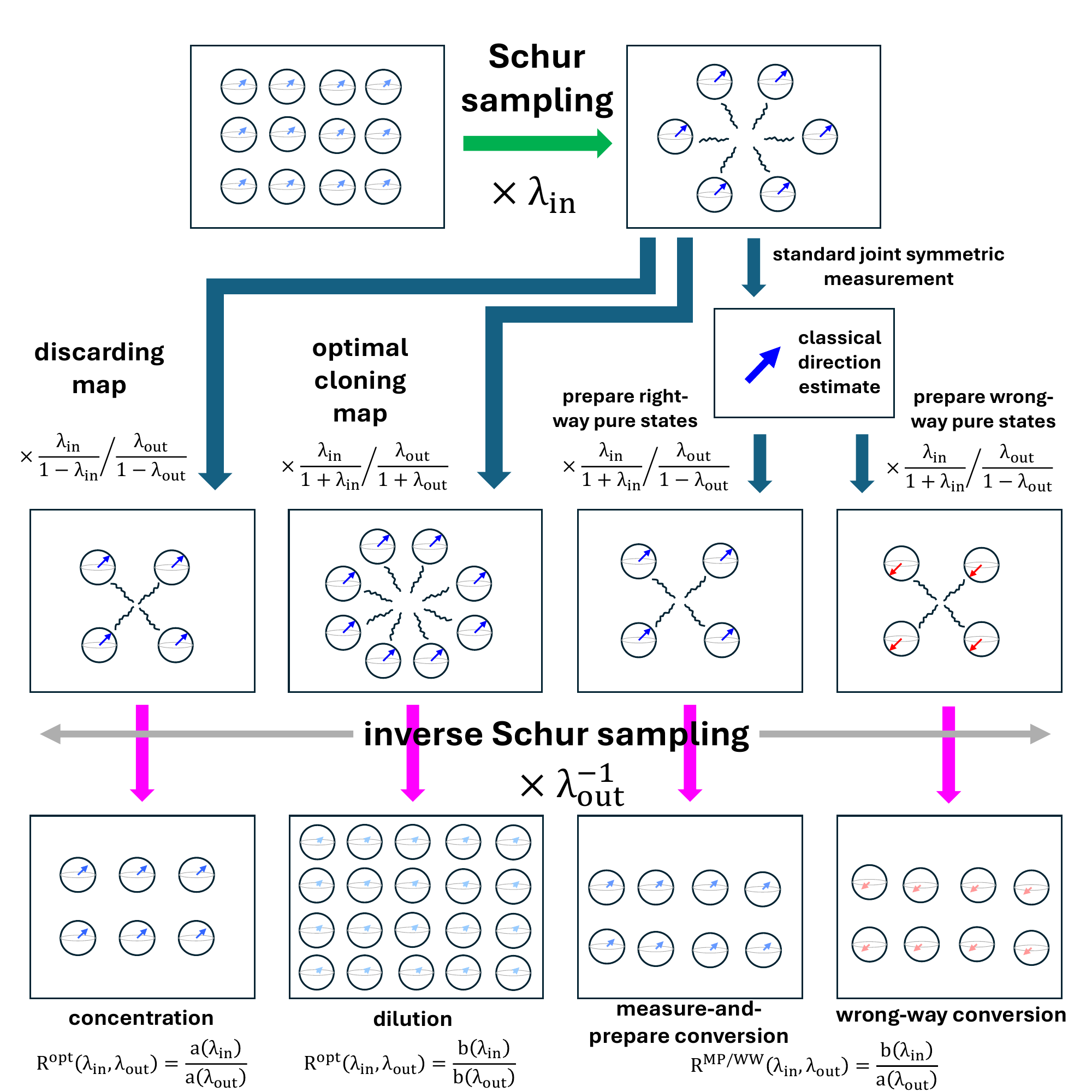}
    \caption{A unified presentation of the four linear-rate conversion protocols for qubits. At each step, the qubit count is multiplied approximately by some factor depending on $\lambda_{\text{in}}$ and/or $\lambda_{\text{out}}$, and multiplying these factors across the three steps yields the linear conversion rates shown at the bottom. Step \textbf{(i)}, unitary Schur sampling, is common across the protocols, and it can be performed solely using two-qubit SWAP tests \cite{Brahmachari2025}. Step \textbf{(ii)}, rescaling the number of qubits, varies based on the protocol. Step \textbf{(iii)}, inverse Schur sampling, is again common across the protocols, and it can be performed by introducing the appropriate number of singlet states and randomly permuting the qubits. Interestingly, we show that the optimal cloning map can also be performed solely using two-qubit SWAP tests. This allows us to conclude that both concentration and dilution have friendly implementations, requiring no two-qubit gate or measurement other than the two-qubit SWAP test.}
    \label{fig:qubit-conversion-protocols-unified-presentation}
\end{figure*}

\noindent\emph{Universality of the SWAP test.}
Remarkably, our qubit concentration and dilution protocols also admit simple
implementations that require only SWAP tests. As mentioned above, it was
recently shown that step (\textbf{i}) of the protocol, namely, unitary Schur
sampling, can be implemented using only SWAP tests
\cite{Brahmachari2025}. Step (\textbf{iii}) requires random permutations and
the preparation of singlets. The latter can again be achieved by performing
SWAP tests on pairs of qubits initially prepared in maximally mixed states and
postselecting on the singlet outcome.

For {concentration}, step (\textbf{ii}) consists simply of discarding qubits. For
dilution, on the other hand, one must implement the optimal cloning channel,
which, to the best of our knowledge, does not have an obvious simple
realization. Remarkably, however, we find that even the optimal cloning channel
can be implemented using SWAP tests alone. Fig. \ref{fig:qubit-conversion-protocols-unified-presentation} illustrates this
protocol, which is analyzed further in the Supplemental Material.
 
\begin{table}
    \centering
    \begin{tabular}{c|c}
        Conversion task & Maximum rate \\
        \hline
        {Concentration}
        & $a(\lambda_{\mathrm{in}})/a(\lambda_{\mathrm{out}})$ \\
        Dilution
        & $b(\lambda_{\mathrm{in}})/b(\lambda_{\mathrm{out}})$ \\
        Measure \& Prepare
        & $b(\lambda_{\mathrm{in}})/a(\lambda_{\mathrm{out}})$ \\
        Wrong-Way
        & $b(\lambda_{\mathrm{in}})/a(\lambda_{\mathrm{out}})$
    \end{tabular}
    \caption{Linear-rate conversion tasks considered in this work. Here, $a(\lambda)=\tfrac{\lambda^2}{1-\lambda}$ and $b(\lambda)=\tfrac{\lambda^2}{1+\lambda}$ are, respectively, the max and min eigenvalues of $\mathrm{RLD}(\lambda)$.
}  \label{tab:task-vs-max-rate}
\end{table}

\vspace{2mm}

\noindent\emph{Two additional conversion tasks.} Interestingly, the RLD Fisher information matrix also determines the optimal rates for two further tasks: \emph{measure-and-prepare conversion}, in which the input is measured and the output is prepared based on the measurement outcome, and \emph{wrong-way conversion}, in which the output Bloch vectors point opposite to the input direction. Furthermore, both tasks have the same maximum conversion rate, namely,
\be
R^{\rm MP/WW}(\lambda_{\text{in}},\lambda_{\text{out}})=\frac{b(\lambda_{\text{in}})}{a(\lambda_{\text{out}})} =R^{\rm opt}(0^+ , \lambda_{\text{out}}) \times R^{\rm opt}(\lambda_{\text{in}}, 0^+) \ ,
\ee
where $0^+$ denotes the limit of $\lambda>0$ going to zero, i.e., the maximally-mixed state limit.  Thus, this rate can be approached by first diluting the input qubits to an arbitrarily low but nonzero purity $\epsilon$, and then concentrating them to the target purity $\lambda_{\text{out}}$.

More practically, these rates are achieved by a simple modification of the protocol above. Namely, after Schur sampling in step (\textbf{i}), one performs the \emph{standard joint symmetric measurement} \cite{Massar1995} in step (\textbf{ii}), obtaining an estimate $|\Psi'\rangle$ of the unknown pure state $|\Psi\rangle$. One then prepares copies of $|\Psi'\rangle$ for measure-and-prepare conversion, or copies of the state orthogonal to $|\Psi'\rangle$ for wrong-way conversion. Step (\textbf{iii}) is again the inverse Schur transform. The results are summarized in Table~\ref{tab:task-vs-max-rate}.

For $0<\lambda_{\rm in},\lambda_{\rm out}<1$ with $\lambda_{\rm in}\neq\lambda_{\rm out}$, the strict inequality
$R^{\rm opt}>R^{\rm MP}$
shows that coherently processing the qubits provides an advantage when converting between purity levels while \emph{preserving} the direction. In contrast, the fact that $R^{\rm WW} = R^{\rm MP}$ shows that, when the direction is \emph{reversed}, a measure-and-prepare strategy is already optimal.

\vspace{2mm}

\noindent\emph{Conclusion.} Our work provides a novel operational interpretation for the full RLD Fisher information matrix of a multi-parameter family of quantum states. The full RLD matrix has been defined previously and used to derive Cram\'{e}r--Rao-type bounds for quantum parameter estimation \cite{Yuen1973,Belavkin1976}, but we are not aware of any  operational setting in which the full matrix has played a central role.

Furthermore, the complex nature of the RLD matrix is essential to our results. RLD Fisher information has also recently been used to establish lower bounds on the infidelity in a problem known as \emph{coherence distillation} \cite{marvian2020,yadavalli2025,kazi2025}. However, because that problem considers only one-parameter families of states, the RLD matrix reduces to a single number, which is necessarily real.

Other linear-rate conversion problems in the resource theory of asymmetry have been solved previously \cite{marvian2013, marvian2022,marvian2020,yamaguchi2026}, but they all restrict their attention to pure input states, and all except Ref.~\cite{marvian2022} also restrict their attention to pure output states. Furthermore, the techniques used in these prior works cannot be readily extended to the setting of mixed input states.  Most notably, Yamaguchi et al. recently computed the maximum conversion rate between families of pure states in the resource theory of asymmetry for an arbitrary compact Lie group \cite{yamaguchi2026}. Interestingly, their work emphasizes another complex metric, known as the \emph{quantum geometric tensor}, but this metric is restricted to pure states. \vspace{2mm}

\noindent\emph{Note added.} In the last stage of finalizing this manuscript, we became aware of independent and contemporaneous work by Yamaguchi and Tajima \cite{yamaguchi2026quantifying}, which determines the optimal asymptotic conversion rate between arbitrary finite-dimensional mixed states in the resource theory of asymmetry for compact Lie groups. Their general characterization is formulated in terms of a one-parameter family of quantum Fisher information matrices interpolating between the SLD and RLD Fisher information. Specialized to the $\mathrm{SU}(2)$ qubit family considered here, their result is consistent with the optimal conversion rates obtained in this work. Our work independently derives these rates for mixed qubits and, in addition, provides explicit optimal conversion protocols and their implementation using SWAP tests, determines the optimal rates under the measure-and-prepare restriction and for wrong-way conversion, and establishes the connection between the full RLD Fisher information matrix and the leading-order error in single-output qubit purification, thereby providing an information-theoretic explanation of the result of Ref.~\cite{Cirac1999}. \vspace{2mm}

\noindent\emph{Acknowledgments.} 
We acknowledge support from NSF Phy-2046195, NSF FET-2106448, and NSF QLCI grant OMA-2120757. SK is funded by the National Defense Science and Engineering Graduate (NDSEG) Fellowship. SK and IM would like to thank Shrigyan Brahmachari, Yash Chitgopekar, Govind Lal Sidhardh, Plato Deliyannis, Austin Hulse, David Jakab, M\'{a}rton Kar\'{a}csony, Nikolaos Koukoulekidis, Shiv Akshar Yadavalli, and Chu Zhao for many useful discussions.\\

\bibliography{REFERENCES}

\newpage

\onecolumngrid

\newpage

\maketitle
\vspace{-5in}
\begin{center}
\Large{Supplementary Material}
\end{center}

The supplementary material of this paper is organized as follows:
\begin{itemize}
    \item In Appendix \ref{sec:schur-sampling-commentary}, we define the multi-qubit states that result from applying unitary Schur sampling to a collection of i.i.d. qubit states and discarding the singlet states. We call these states ``Schur-transformed states'' for convenience. Several of their properties, as originally discussed by Cirac et al. in their work on qubit distillation \cite{Cirac1999}, will be essential for our analysis. A crucial mathematical property is that a Schur-transformed state is a classical mixture of so-called Dicke states, with the coefficients forming a geometric sequence. (The Dicke state $\ket{D^{(N)}_w}$ is the equal superposition of all $N$-qubit computational basis states with Hamming weight $w$.)
    \item In Appendix \ref{sec:four-important-channels}, we introduce four useful channels on the symmetric subspace and compute each of their actions on a single Dicke state. These channels are the discarding channel $\mE_{\text{discard}}$, the optimal cloning map $\mE_{\text{clone}}$ as devised by Werner \cite{Werner1998}, the optimal measure-and-prepare channel $\mE_{\text{MP}}$, and the optimal wrong-way measure-and-prepare channel $\mE_{\text{WW}}$.
    \item In Appendix \ref{sec:schur-transformed-state-conversion}, we show how applying each of the channels $\mE_{\text{discard}}$, $\mE_{\text{clone}}$, $\mE_{\text{MP}}$, $\mE_{\text{WW}}$ to a Schur-transformed state yields a state that approximates a new Schur-transformed state. These conversions between Schur-transformed states forms the backbone of our protocols to convert between i.i.d. states.
    \item In Appendix \ref{sec:unified-presentation}, we provide a unified presentation of our four qubit linear-rate conversion procedures. Each procedure begins with Schur sampling and ends with inverse Schur sampling. The intermediate step depends on the task of interest: concentration uses $\mE_{\text{discard}}$, dilution uses $\mE_{\text{clone}}$, measure-and-prepare conversion uses $\mE_{\text{MP}}$, and wrong-way conversion uses $\mE_{\text{WW}}$.
    \item In Appendix \ref{sec:heuristic-arguments}, we use tensor network diagrams to provide elegant heuristic arguments for why the discarding map achieves concentration and how the Werner optimal cloning map achieves dilution. The backbone of these arguments is the fact that applying the optimal cloning map to i.i.d. pure qubits yields approximately a Schur-transformed state.
    \item In Appendix \ref{sec:numerical-analysis}, we perform various numerical computations to corroborate our analytical results for the performance of our four linear-rate conversion protocols.
    \item In Appendix \ref{sec:monotonicity-complexified-rld}, we define RLD Fisher information and review the proof of its monotonicity under $2$-positive trace-preserving maps (which include all CPTP maps), a fact which was first shown by Petz \cite{petz1996}. We also emphasize the importance of $2$-positivity by providing two simple counterexamples for $1$-positive trace-preserving maps.
    \item In Appendix \ref{sec:single-shot-qubit-distillation}, we show how the monotonicity of RLD Fisher information provides an information-theoretic justification for the performance of optimal single-shot qubit distillation as derived by Cirac et al. \cite{Cirac1999}.
    \item In Appendix \ref{sec:qfi-metrics}, we discuss quantum Fisher information (QFI) metrics more generally and explain the distinguished role that RLD Fisher information plays within the family of QFI metrics.
    \item In Appendix \ref{sec:converse-bound-rld-sensitivity}, we prove the optimality of our achieved conversion rates by invoking the monotonicity of RLD Fisher information. The central technical challenge is to show that the allowance of nonzero but vanishing trace distance from our target i.i.d. qubit state does not allow us to excessively reduce the RLD Fisher information of the output state.
    \item In Appendix \ref{sec:friendly-implementations}, we show an elegant implementation of the optimal cloning map for qubits that does not appear to be common knowledge. We additionally leverage a work by Brahmachari et al. that shows how unitary Schur sampling can be realized by applying two-qubit SWAP tests to random pairs of qubits \cite{Brahmachari2025}. We conclude that our optimal concentration and dilution procedures can both be implemented solely using two-qubit SWAP tests and other simple operations.
    \item In Appendix \ref{sec:math-tidbits}, we explore some nice mathematical properties of the discarding map, the optimal cloning map, and the optimal measure-and-prepare channel. Although this appendix is not necessary to understand our work, it may be illuminating to the reader seeking to understand the behavior of these three channels more deeply.
\end{itemize}


\newcommand\appitemtwo[2]{
\newcommand\appitem[1]{\hyperref[{#1}]
{\textbf{\cref{#1}}} \textbf{\nameref*{#1}}
\dotfill \pageref{#1}\vspace{5pt}}
\newcommand\subappitem[1]{
\makeatletter
\newcommand{\appsec}[2]{%
  \section{#1}%
  \def\@currentlabelname{#1}%
  \def\@currentlabel{\thesection}%
  \label{#2}%
  \addcontentsline{toc}{section}{#1}%
}
\newcommand{\appsubsec}[2]{%
  \subsection{#1}%
  \def\@currentlabelname{#1}%
  \def\@currentlabel{\thesubsection}%
  \label{#2}%
}
\newcommand{\appsubsubsec}[2]{%
  \refstepcounter{subsubsection}%
  \subsubsection{#1}%
  \addtocounter{subsubsection}{-1}%
  \def\@currentlabelname{#1}%
  \def\@currentlabel{\thesubsubsection}%
  \label{#2}%
}
\makeatother


\newpage

\section*{Supplementary Material: Table of Contents}
\begin{itemize}[label={}]

\item \appitem{sec:schur-sampling-commentary}
\subitem \subappitem{subsec:typicality}

\item \appitem{sec:four-important-channels}
\subitem \subappitem{subsec:discarding-map-unified-presentation}
\subitem \subappitem{subsec:optimal-cloning-map-unified-presentation}
\subitem \subappitem{subsec:optimal-mp-channel-unified-presentation}
\subitem \subappitem{subsec:optimal-ww-channel-unified-presentation}

\item \appitem{sec:schur-transformed-state-conversion}
\subitem \subappitem{subsec:schur-transformed-state-concentration-unified-presentation}
\subitem \subappitem{subsec:schur-transformed-state-dilution-unified-presentation}
\subitem \subappitem{subsec:schur-transformed-state-mp-conversion-unified-presentation}
\subitem \subappitem{subsec:schur-transformed-state-ww-conversion-unified-presentation}
\subitem \subappitem{subsec:proofs-niche-lemmas-schur-transformed-state-conversion}

\item \appitem{sec:unified-presentation}
\subitem \subappitem{subsec:three-step-procedure-unified-presentation}
\subitem \subappitem{subsec:successful-conversion-theorems-unified-presentation}

\item \appitem{sec:heuristic-arguments}
\subitem \subappitem{subsec:tensor-network-symmetric-subspace}
\subitem \subappitem{subsec:heuristic-cloning-dilution}
\subitem \subappitem{subsec:heuristic-discarding-concentration}
\subitem \subappitem{subsec:heuristic-discarding-one-qubit}
\subitem \subappitem{subsec:heuristic-adding-one-qubit}

\item \appitem{sec:numerical-analysis}
\subitem \subappitem{subsec:numerical-analysis-schur-transformed-state-conversion-geometric-sequence}
\subitem \subappitem{subsec:numerical-analysis-schur-transformed-state-conversion-error-analysis}
\subitem \subappitem{subsec:numerical-analysis-iid-state-conversion-error-analysis}

\item \appitem{sec:monotonicity-complexified-rld}

\item \appitem{sec:single-shot-qubit-distillation}
\subitem \subappitem{subsec:right-way-distillation}
\subitem \subappitem{subsec:wrong-way-distillation}

\item \appitem{sec:qfi-metrics}
\subitem \subappitem{subsec:mcp-classification}
\subitem \subappitem{subsec:real-vs-complexified}
\subitem \subappitem{subsec:rld-in-context}
\subitem \subappitem{subsec:qfi-derived-su2-asymmetry-measures}
\subitem \subappitem{subsec:comparing-different-qfi-metrics}
\subitem \subappitem{subsec:resource-wastage-irreversibility}

\item \appitem{sec:converse-bound-rld-sensitivity}
\subitem \subappitem{subsec:conversion-rate-upper-bounds-intuitive}
\subitem \subappitem{subsec:output-state-restrictions}
\subitem \subappitem{subsec:rld-sensitivity}
\subitem \subappitem{subsec:proofs-niche-lemmas-converse-bound}

\item \appitem{sec:friendly-implementations}
\subitem \subappitem{subsec:optimal-cloning-implementation}
\subitem \subappitem{subsec:swap-test-universality}

\item \appitem{sec:math-tidbits}
\subitem \subappitem{subsec:optimal-cloning-map-composition}
\subitem \subappitem{subsec:optimal-cloning-map-polynomial}
\subitem \subappitem{subsec:optimal-mp-channel-polynomial}
\subitem \subappitem{subsec:extreme-dilution-extreme-concentration}
\subitem \subappitem{subsec:optimal-cloning-sjsm}
\subitem \subappitem{subsec:connection-discarding-optimal-cloning}

\end{itemize}

\appendix

\color{black}
\onecolumngrid

\appsec{Restricting to the Symmetric Subspace via Schur Sampling}
{sec:schur-sampling-commentary}

A major technique we use for qubit linear-rate conversion is to restrict our attention to the symmetric subspace on a number of qubits. In particular, instead of working directly with the i.i.d. states, we work with the multi-qubit states that result from applying Schur sampling to the i.i.d. states. In this section, we define these states, along with all other necessary concepts and notations.

\vspace{0.5\baselineskip}

We first define a few useful concepts. We begin by defining the computational basis with respect to a specific direction:

\begin{definition}[Directional qubit states]
\label{def:directional-states}
For an arbitrary unit vector $\hat{n} = \langle n_x,n_y,n_z\rangle$, we define $\ket{1}_{\hat{n}}$ and $\ket{0}_{\hat{n}}$ to be normalized eigenvectors of the operator $\hat{n}\cdot\vec{\sigma}$ with eigenvalues $+1$ and $-1$, respectively. We refer to these as the \textbf{computational basis states with respect to $\hat{n}$}.
\end{definition}

Under this convention, plugging in $\hat{n} = \hat{z}$ yields the usual computational basis states, but in the opposite order: $\ket{0}_{\hat{z}} = \ket{1}$ and $\ket{1}_{\hat{z}} = \ket{0}$. This convention may seem strange, but we use it to match the convention used in Ref. \cite{Cirac1999}.

\vspace{0.5\baselineskip}

Next, we define the notation we will use throughout the appendices for a single-qubit state:

\begin{definition}[direction and purity level]
\label{def:single-qubit-state}
The single-qubit state with \textbf{direction} $\hat{n}$ and \textbf{purity level} $\lambda$ is given by:
\begin{equation}
    \rho(\lambda,\hat{n}) = \frac{\Ibb + \lambda(\hat{n}\cdot\vec{\sigma})}{2} = \lambda\ket{1}\bra{1}_{\hat{n}} + (1-\lambda)\frac{\Ibb}{2} = c_1\ket{1}\bra{1}_{\hat{n}} + c_0\ket{0}\bra{0}_{\hat{n}},
\end{equation}
where for convenience we have defined $c_1 = \frac{1+\lambda}{2}$ and $c_0 = \frac{1-\lambda}{2}$.
\end{definition}

In the Bloch sphere picture, $\hat{n}$ is the direction of the Bloch vector, and $\lambda$ is the length of the Bloch vector. The expressions in Definition \ref{def:single-qubit-state} show that we can intuitively interpret $\rho(\lambda,\hat{n})$ as being the pure state $\ket{1}\bra{1}_{\hat{n}}$ with probability $\lambda$ and the maximally mixed state $\Ibb/2$ with probability $1-\lambda$, or alternatively as being the pure state $\ket{1}\bra{1}_{\hat{n}}$ with probability $c_1$ and the antipodal pure state $\ket{0}\bra{0}_{\hat{n}}$ with probability $c_0$.

\vspace{0.5\baselineskip}

We now define the symmetrized states of a fixed Hamming weight, also known as Dicke states:

\begin{definition}[Dicke states]
\label{def:Dicke-states}
The \textbf{Dicke state} with $N$ qubits, Hamming weight $w$, and direction $\hat{n}$ is the equal superposition of all computational basis states with respect to $\hat{n}$ that have $w$ ones and $(N-w)$ zeros:
\begin{equation}
    \ket{D^{(N)}_w}_{\hat{n}} = \binom{N}{w}^{-1/2}\sum_{b\in\{0,1\}^N, \, w(b)=w}\ket{b}_{\hat{n}}.
\end{equation}
\end{definition}

Finally, we define the states that result from applying Schur sampling to a collection of i.i.d. qubits:

\begin{definition}[Schur-transformed states]
\label{def:Schur-transformed-states}
The \textbf{Schur-transformed state} with $N_C$ qubits, purity level $\lambda$, and direction $\hat{n}$ is defined to be
\begin{equation}
    \rho_{\text{C}}(N_C,\lambda,\hat{n}) = \frac{c_1-c_0}{c_1^{N_C+1}-c_0^{N_C+1}}\sum_{w=0}^{N_C}c_1^wc_0^{N_C-w}\ket{D^{(N_C)}_w}\bra{D^{(N_C)}_w}_{\hat{n}}.
\end{equation}
In some cases, when the direction $\hat{n}$ is obvious and consistent throughout a derivation, we will suppress it for convenience.
\end{definition}

The formula in Definition \ref{def:Schur-transformed-states} was first presented in the work on single-shot qubit distillation by Cirac et al. \cite{Cirac1999}. They showed that, once this Schur-transformed state is obtained, the optimal way to return a single purified qubit is to return a single one of these qubits and throw away the rest \cite{Cirac1999}. However, in our work, we will need to make use of the full collection of qubits.

\vspace{0.5\baselineskip}

For clarity, if Schur sampling is performed on $\rho(\lambda,\hat{n})^{\otimes N}$, then the resulting state will be $\rho_C(N_C,\lambda,\hat{n})$ for some $0\le N_C\le N$ such that $N$ and $N_C$ have the same parity. The probability of each such outcome is given by
\begin{equation}
    p(N,N_C,\lambda) = d(N,N_C)\frac{c_1^{N_C+1} - c_0^{N_C+1}}{c_1-c_0}\left(c_1c_0\right)^{(N-N_C)/2}.
\end{equation}
For convenience, we have used
\begin{equation}
    d(N,N_C) = \binom{N}{\frac{N-N_C}{2}} - \binom{N}{\frac{N-N_C}{2}-1}
\end{equation}
to denote the dimension of the irreducible representation (irrep) of $S_N$ corresponding to the Young diagram with $\frac{N+N_C}{2}$ boxes in the first row and $\frac{N-N_C}{2}$ boxes in the second row.

\vspace{0.5\baselineskip}

The crucial features of the Schur-transformed states $\rho_C(N_C,\lambda,\hat{n})$ are as follows:
\begin{itemize}
    \item They live fully in the $N_C$-qubit symmetric subspace.
    \item They are purely a function of $N_C$ and $\lambda$. In particular, they retain no memory of the original number of i.i.d. qubits $N$.
    \item They are diagonal in the basis of Dicke states with respect to the Bloch vector $\hat{n}$.
    \item Their diagonal entries are in geometric sequence with first term $\frac{c_1-c_0}{c_1^{N_C+1}-c_0^{N_C+1}}c_1^{N_C} = \frac{2\lambda}{1+\lambda} + O\left((c_0/c_1)^{N_C}\right)$ and common ratio $\frac{c_0}{c_1} = \frac{1-\lambda}{1+\lambda}$ (where we actually start from $w=N_C$ and go down to $w=0$ so that the coefficients are decreasing).
    \item If you have the ensemble of states $\rho_C(N_C,\lambda,\hat{n})$ with corresponding probabilities $p(N,N_C,\lambda)$, then you can reconstruct the i.i.d. collection $\rho(\lambda,\hat{n})^{\otimes N}$. The justification for this is outlined in \cite{Cirac1999}. This explains why any information-processing task performed on the i.i.d. collection of states can instead be performed on the corresponding ensemble of Schur-transformed states without loss.
\end{itemize}

\appsubsec{Typicality of Schur Sampling Outcomes}
{subsec:typicality}

In contrast, we actually do not need to know that much about the probability distribution $p(N,N_C,\lambda)$. For our purposes, the only truly important fact is that, for large $N$, it is highly concentrated around $\lambda N$. More precisely, it approximates a normal distribution with mean $\sim\lambda N$ and variance $\sim(1-\lambda^2)N$ \cite{Keyl2001}, and the moments of this distribution are computed to further precision in Ref. \cite{kazi2025}.

\vspace{0.5\baselineskip}

This phenomenon, where the $N_C$ value is tightly concentrated around $\lambda N$, will be crucial for the analysis of our qubit linear-rate conversion protocols. We thus define a notion of typicality for a value sampled from the distribution $p(N,N_C,\lambda)$:

\begin{definition}
\label{def:schur-sampling-outcome-typicality}
Suppose that $N_C$ is sampled with probability $p(N,N_C,\lambda)$. We say that $N_C$ is \textbf{typical} if $\abs{N_C - \lambda N} \le \left[(1-\lambda^2)N\right]^{2/3}$ and \textbf{atypical} otherwise.
\end{definition}

The scaling of the allowed deviation as $N^{2/3}$ is not that crucial. We could have chosen some other scaling:
\begin{itemize}
    \item Of course, it must be $\omega(\sqrt{N})$ to ensure that the probability of an atypical $N_C$ goes to zero. In particular, standard manipulations of binomial coefficients demonstrate that deviations of $N_C$ from $\lambda_N$ demonstrate sub-Gaussian behavior, meaning that, for any $0 < \varepsilon < \frac{1}{2}$, the probability that $N_C$ deviates from $\lambda N$ by $\Theta(N^{1/2+\varepsilon})$ is at most $\exp(-cN^{2\varepsilon})$ for some constant $c$.
    \item Furthermore, it must also be $o(N)$, to ensure that a typical $N_C$ is still relatively close to $\lambda N$. We will see why this is necessary in much more detail in Appendix \ref{sec:schur-transformed-state-conversion}.
\end{itemize}
As a result, we can choose the allowed scaling in Definition \ref{def:schur-sampling-outcome-typicality} to be $N^p$ for any $\frac{1}{2} < p < 1$. In Appendix \ref{sec:schur-transformed-state-conversion}, we will see how the choice of $p$ affects the upper bound we can establish for the trace distance between the output state and the target state.

\newpage

\appsec{Four Important Channels on the Symmetric Subspace}
{sec:four-important-channels}

By leveraging the reversibility of Schur sampling for permutation-invariant states (as explained by Cirac et al. \cite{Cirac1999}, and as we briefly discuss in Appendix \ref{sec:schur-sampling-commentary}), we can reduce the task of converting between i.i.d. qubit collections at different purity levels to the task of converting between Schur-transformed states at different purity levels (see Definition \ref{def:Schur-transformed-states}). Since Schur-transformed states live on the symmetric subspace, it makes sense to consider channels whose input and output Hilbert spaces are both symmetric subspaces of some number of qubits.

\vspace{0.5\baselineskip}

In this appendix, we briefly discuss four useful channels on the symmetric subspace that are essential for our qubit linear-rate conversion procedures. They are as follows:
\begin{itemize}
    \item \textbf{Discarding map} $\mE_{\text{discard}}[N\to M]$: used for \textbf{concentration};
    \item \textbf{Optimal cloning map} $\mE_{\text{clone}}[N\to M]$: used for \textbf{dilution};
    \item \textbf{Optimal measure-and-prepare channel} $\mE_{\text{MP}}[N\to M]$: used for \textbf{measure-and-prepare conversion};
    \item \textbf{Optimal wrong-way measure-and-prepare channel} $\mE_{\text{WW}}[N\to M]$: used for \textbf{wrong-way conversion}.
\end{itemize}

\vspace{0.5\baselineskip}

In particular, in this appendix, we define each of these channels and compute their action on a single Dicke state (see Definition \ref{def:Dicke-states}). This will allow us to build up to the main result of the next appendix (Appendix \ref{sec:schur-transformed-state-conversion}), where we will apply each of these channels to a Schur-transformed state (which is a classical mixture of Dicke states, as shown in Definition \ref{def:Schur-transformed-states}).

\vspace{0.5\baselineskip}

Throughout this appendix, we will suppress the $\hat{n}$ notation used in Definitions \ref{def:directional-states}, \ref{def:Dicke-states}, and \ref{def:Schur-transformed-states}. Since all four of these channels are $\mathrm{SU}(2)$-covariant, the results we show here hold regardless of the choice of $\hat{n}$.

\appsubsec{Discarding Map}
{subsec:discarding-map-unified-presentation}

The discarding map $\mE_{\text{discard}}[N\to M]$ is just a partial trace over $(N-M)$ qubits. Since we apply this map to states on the symmetric subspace of $N$ qubits, it does not matter which qubits we discard, only how many. So if we write $[N]\coloneqq\{1,\cdots,N\}$, we can assume without loss of generality that we discard the last $(N-M)$ qubits:
\begin{equation}
    \mE_{\text{discard}}[N\to M] = \text{Tr}_{[N]\backslash[M]}.
\end{equation}
We now compute the action of the discarding map on a single Dicke state. The result is a classical mixture of Dicke states, as shown in the following lemma:

\begin{lemma}[applying the discarding map to a Dicke state]
When $N - M$ qubits are discarded from the Dicke state $\ket{D^{N}_w}\bra{D^{N}_w}$, the resulting state on $M$ qubits is as follows:
\begin{equation}
    \mE_{\text{discard}}[N\to M]\left(\ket{D^{(N)}_w}\bra{D^{(N)}_w}\right) = \binom{N}{M}^{-1}\sum_{\tilde{w}=0}^{M}\binom{w}{\tilde{w}}\binom{N - w}{M - \tilde{w}}\ket{D^{(M)}_{\tilde{w}}}\bra{D^{(M)}_{\tilde{w}}}.
\end{equation}
\label{lem:discarding-map-dicke-state-unified-presentation}
\end{lemma}

\begin{proof}
We first expand each Dicke state in the computational basis:
\begin{equation}
    \ket{D^{(N)}_w}\bra{D^{(N)}_w} = \binom{N}{w}^{-1}\sum_{b,b'\in\{0,1\}^{N}, \, w(b)=w(b')=w}\ket{b}\bra{b'}.
\end{equation}
We now take the reduced state of the first $M$ qubits and discard the rest. This yields
\begin{equation}
    \mE_{\text{discard}}[N\to M]\left(\ket{D^{(N)}_w}\bra{D^{(N)}_w}\right) = \binom{N}{w}^{-1}\sum_{b,b'\in\{0,1\}^{N}, \, w(b)=w(b')=w}\text{Tr}_{[N]\backslash[M]}\left[\ket{b}\bra{b'}\right].
\end{equation}
Each of these terms equals zero unless the last $(N-M)$ bits of $b$ and $b'$ are the same, in which case you get $\ket{c}\bra{c'}$, where $c$ and $c'$ are the first $M$ bits of $b$ and $b'$, respectively. Notice that $\ket{c}\bra{c'}$ can only be produced if $w(c) = w(c') = \tilde{w}$ for some $\tilde{w}\le w$, since $b$ and $b'$ must have the same Hamming weight $w$. In that case, there are $\binom{N-M}{w-\tilde{w}}$ ways to choose $b$ and $b'$ to produce the given $c$ and $c'$. Therefore,
\begin{align}
    \mE_{\text{discard}}[N\to M]\left(\ket{D^{(N)}_w}\bra{D^{(N)}_w}\right) &= \binom{N}{w}^{-1}\sum_{\tilde{w}=0}^{M}\binom{N - M}{w - \tilde{w}}\sum_{c,c'\in\{0,1\}^{M}, \, w(c)=w(c')=\tilde{w}}\ket{c}\bra{c'} \\
    &= \binom{N}{w}^{-1}\sum_{\tilde{w}=0}^{M}\binom{N-M}{w-\tilde{w}}\binom{M}{\tilde{w}}\ket{D^{(M)}_{\tilde{w}}}\bra{D^{(M)}_{\tilde{w}}} \\
    &= \binom{N}{M}^{-1}\sum_{\tilde{w}=0}^{M}\binom{w}{\tilde{w}}\binom{N-w}{M-\tilde{w}}\ket{D^{(M)}_{\tilde{w}}}\bra{D^{(M)}_{\tilde{w}}}.
\end{align}
\end{proof}

It is worth mentioning that the sequence of coefficients shown above precisely matches a \textbf{hypergeometric distribution}. Suppose you have an urn with $N$ balls, exactly $K$ of which are colored red, and the rest are colored black. Now you sample $n$ of them uniformly at random \emph{without replacement}. If $X$ denotes the number of red balls you sample, then $X$ follows a hypergeometric distribution, and we say that $X\sim\text{HG}(N,K,n)$. In particular:
\begin{equation}
    X\sim\text{HG}(N,K,n) \iff \Pbb[X=k] = \frac{\binom{K}{k}\binom{N-K}{n-k}}{\binom{N}{n}}.
\end{equation}
The discarding map acting on a Dicke state simulates precisely the random process described above. You have $N$ qubits, exactly $w$ of which are ones, and the rest are zeros. The discarding map $\mE_{\text{discard}}[N\to M]$ corresponds to sampling $M$ of these qubits uniformly at random \emph{without replacement}. Hence the resulting number of ones, which equals the new Hamming weight $\tilde{w}$, follows the hypergeometric distribution
\begin{equation}
    \tilde{w}\sim\text{HG}(N,w,M) \iff \text{coeff}[\text{Hamming weight }\tilde{w}] = \frac{\binom{w}{\tilde{w}}\binom{N-w}{M-\tilde{w}}}{\binom{N}{M}},
\end{equation}
which is exactly what we see in Lemma \ref{lem:discarding-map-dicke-state-unified-presentation}.

\appsubsec{Optimal Cloning Map}
{subsec:optimal-cloning-map-unified-presentation}

It is well known that it is impossible to clone quantum information with perfect fidelity. This naturally motivates the following question: how can one convert $N$ i.i.d. copies of a pure qudit state into $M >  N$ i.i.d. copies of that same state, with the highest possible fidelity? The channel that achieves this was determined by Werner \cite{Werner1998}, and it is now called the optimal cloning map $\mE_{\text{clone}}[N\to M]$.

\vspace{0.5\baselineskip}

The Werner optimal cloning map from $N$ qudits to $M$ qudits takes the form
\begin{equation}
    \mE_{\text{clone}}[N\to M](\rho_N) = \frac{\binom{N+d-1}{N}}{\binom{M+d-1}{M}}\Pi^M_{\text{sym}}\left(\rho_N\otimes\Ibb_d^{\otimes(M-N)}\right)\Pi^M_{\text{sym}},
\end{equation}
where $\rho_N$ is a state in the $N$-qubit symmetric subspace, $d$ is the dimension of the Hilbert space, and $\Pi^M_{\text{sym}}$ is the projector to the symmetric subspace on $M$ qudits \cite{Werner1998}. When we specialize to the qubit case $d=2$, we obtain
\begin{equation}
    \mE_{\text{clone}}[N\to M](\rho_N) = \frac{N+1}{M+1}\Pi^M_{\text{sym}}\left(\rho_N\otimes\Ibb_2^{\otimes(M-N)}\right)\Pi^M_{\text{sym}}.
\end{equation}
We now compute the action of the optimal cloning map on a single Dicke state:

\begin{lemma}[applying the optimal cloning map to a Dicke state]
When the optimal cloning map from $N$ qubits to $M$ qubits is applied to the Dicke state $\ket{D^{(N)}_w}\bra{D^{(N)}_w}$, the resulting state on $M$ qubits is as follows:
\begin{equation}
    \mE_{\text{clone}}[N\to M]\left(\ket{D^{(N)}_w}\bra{D^{(N)}_w}\right) = \binom{M+1}{N+1}^{-1}\sum_{\tilde{w}=0}^{M}\binom{\tilde{w}}{w}\binom{M-\tilde{w}}{N-w}\ket{D^{(M)}_{\tilde{w}}}\bra{D^{(M)}_{\tilde{w}}}.
\end{equation}
\label{lem:optimal-cloning-map-dicke-state-unified-presentation}
\end{lemma}

\begin{proof}
We begin by computing
\begin{equation}
    \bra{D^{(M)}_{\tilde{w}}}\left(\ket{D^{(N)}_w}\bra{D^{(N)}_w}\otimes\Ibb_d^{\otimes(M-N)}\right)\ket{D^{(M)}_{\tilde{w}'}}.
\end{equation}
This quantity is zero unless $\tilde{w} = \tilde{w}'$, because when the quantity in parentheses is written in the computational basis, it only ever has states where the bit string in the bra and the bit string in the ket have equal Hamming weight. Therefore, we can restrict our attention to $\tilde{w} = \tilde{w}'$ and proceed as follows:
\begin{align}
    & \bra{D^{(M)}_{\tilde{w}}}\left(\ket{D^{(N)}_w}\bra{D^{(N)}_w}\otimes\Ibb_d^{\otimes(M-N)}\right)\ket{D^{(M)}_{\tilde{w}}} \\
    = \,\, & \binom{M}{\tilde{w}}^{-1}\binom{N}{w}^{-1} \quad\quad \sum_{\mathclap{\substack{c,c'\in\{0,1\}^{M}, \\ w(c)=w(c')=\tilde{w}}}} \quad\quad\quad\quad\quad \sum_{\mathclap{\substack{b,b'\in\{0,1\}^{N}, \\ w(b)=w(b')=w}}} \quad\quad \bra{c}\left(\ket{b}\bra{b'}\otimes\Ibb_2^{\otimes(M-N)}\right)\ket{c'}.
\end{align}
The quantity in the innermost summation equals $1$ if and only if $c = bd$ and $c' = b'd$ for some bit string $d\in\{0,1\}^{M-N}$. In all other cases, that quantity equals $0$. Note that $b$ and $b'$ have Hamming weight $w$, while $c$ and $c'$ have Hamming weight $\tilde{w}$. Thus, there are $\binom{N}{w}$ ways to choose $b$, $\binom{N}{w}$ ways to choose $b'$, and $\binom{M-N}{\tilde{w}-w}$ ways to choose $d$. Therefore,
\begin{align}
    \bra{D^{(M)}_{\tilde{w}}}\left(\ket{D^{(N)}_w}\bra{D^{(N)}_w}\otimes\Ibb_d^{\otimes(M-N)}\right)\ket{D^{(M)}_{\tilde{w}}} &= \binom{M}{\tilde{w}}^{-1}\binom{N}{w}^{-1}\left[\binom{N}{w}\binom{N}{w}\binom{M-N}{\tilde{w}-w}\right] \\
    &= \binom{M}{\tilde{w}}^{-1}\binom{N}{w}\binom{M-N}{\tilde{w}-w}.
\end{align}
From this, we conclude that the optimal cloning map acts on a single Dicke state as follows:
\begin{align}
    \mE_{\text{clone}}[N\to M]\left(\ket{D^{(N)}_w}\bra{D^{(N)}_w}\right) &= \frac{N+1}{M+1}\Pi^{M}_{\text{sym}}\left(\ket{D^{(N)}_w}\bra{D^{(N)}_w}\otimes\Ibb_d^{\otimes(M-N)}\right)\Pi^{M}_{\text{sym}} \\
    &= \frac{N+1}{M+1}\sum_{\tilde{w}=0}^{M}\binom{M}{\tilde{w}}^{-1}\binom{N}{w}\binom{M-N}{\tilde{w}-w}\ket{D^{(M)}_{\tilde{w}}}\bra{D^{(M)}_{\tilde{w}}}.
\end{align}
Finally, we rearrange the binomial coefficients as follows:
\begin{equation}
    \frac{N+1}{M+1}\binom{M}{\tilde{w}}^{-1}\binom{N}{w}\binom{M-N}{\tilde{w}-w} = \binom{M+1}{N+1}^{-1}\binom{\tilde{w}}{w}\binom{M-\tilde{w}}{N-w}.
\end{equation}
This allows us to rewrite the above expression as
\begin{equation}
    \mE_{\text{clone}}[N\to M]\left(\ket{D^{(N)}_w}\bra{D^{(N)}_w}\right) = \binom{M+1}{N+1}^{-1}\sum_{\tilde{w}=0}^{M}\binom{\tilde{w}}{w}\binom{M-\tilde{w}}{N-w}\ket{D^{(M)}_{\tilde{w}}}\bra{D^{(M)}_{\tilde{w}}}.
\end{equation}
\end{proof}

It is worth mentioning that the sequence of coefficients shown above precisely matches a \textbf{negative hypergeometric distribution}. Suppose you have an urn with $N$ balls, exactly $K$ of which are colored red, and the rest are colored black. Now you sample them uniformly at random \emph{without replacement} until you get exactly $r$ black balls. If $X$ denotes the number of red balls you sample, then $X$ follows a negative hypergeometric distribution, and we say that $X\sim\text{NHG}(N,K,r)$. In particular:
\begin{equation}
    X\sim\text{NHG}(N,K,r) \iff \Pbb[X=k] = \frac{\binom{k+r-1}{k}\binom{N-r-k}{K-k}}{\binom{N}{K}}.
\end{equation}
We can see that the optimal cloning map yields precisely this distribution. In particular, the Hamming weight difference $\tilde{w}-w$ follows a negative hypergeometric distribution:
\begin{equation}
    \tilde{w}-w\sim\text{NHG}(M+1,M-N,w+1) \iff \text{coeff}[\text{Hamming weight }\tilde{w}] = \frac{\binom{\tilde{w}}{\tilde{w}-w}\binom{M-\tilde{w}}{(M-N)-(\tilde{w}-w)}}{\binom{M+1}{M-N}} = \frac{\binom{\tilde{w}}{w}\binom{M-\tilde{w}}{N-w}}{\binom{M+1}{N+1}},
\end{equation}
which is exactly what we see in Lemma \ref{lem:optimal-cloning-map-dicke-state-unified-presentation}.

\vspace{0.5\baselineskip}

For our work on qubit linear-rate conversion, we do not need any additional information about $\mE_{\text{clone}}$ beyond Lemma \ref{lem:optimal-cloning-map-dicke-state-unified-presentation}. Nonetheless, we encourage the curious reader to check out Appendix \ref{sec:math-tidbits} for some more interesting facts about $\mE_{\text{clone}}$ that can be derived using this result.

\appsubsec{Optimal Measure-and-Prepare Channel}
{subsec:optimal-mp-channel-unified-presentation}

Let us first define the measurement that we will perform on the symmetric subspace:

\begin{definition}[standard joint symmetric measurement]
\label{def:standard-joint-symmetric-measurement}
The \textbf{standard joint symmetric measurement (SJSM)} on the symmetric subspace of $N$ qubits is the positive operator-valued measure (POVM) given by
\begin{equation}
    dM = \frac{N+1}{4\pi}\ket{\Psi(\theta,\phi)}\bra{\Psi(\theta,\phi)}^{\otimes N}\,d\Omega,
\end{equation}
where $\ket{\Psi(\theta,\phi)}$ is the pure single-qubit state with polar angle $\theta$ and azimuthal angle $\phi$, and $d\Omega = \sin\theta\,d\theta\,d\phi$ is the area element on a unit sphere.
\end{definition}

Massar and Popescu showed that this measurement is the optimal measurement to estimate a pure qubit state given $N$ identical copies \cite{Massar1995}. They further showed that this measurement cannot be implemented solely using single-qubit measurements, even if an adaptive strategy is permitted (meaning that later measurements depend on earlier measurement outcomes) \cite{Massar1995}. Even though our problem is different, the optimality of the SJSM for estimating a pure qubit state makes it a very natural measurement to consider.

\vspace{0.5\baselineskip}

Recall that, for a POVM $\{M_j\}$ with corresponding preparations $\sigma_j$, the resulting measure-and-prepare channel takes the form
\begin{equation}
    \mE(\rho) = \sum_{j}\text{Tr}[\rho M_j]\sigma_j.
\end{equation}
The optimal measure-and-prepare channel $\mE_{\text{MP}}[N\to M]$ involves applying the SJSM, obtaining a unit vector $\hat{u} = \langle\sin\Theta\cos\Phi,\sin\Theta\sin\Phi,\cos\Theta\rangle$, and then preparing $M$ copies of the state $\ket{\Psi(\Theta,\Phi)} = \cos\frac{\Theta}{2}\ket{0} + e^{i\Phi}\sin\frac{\Theta}{2}\ket{1}$. Therefore, the action of our channel looks as follows:
\begin{equation}
    \mE_{\text{MP}}[N\to M](\rho_N) = \frac{N+1}{4\pi}\int_{\Sbb^2}\text{Tr}\left[\rho_N\ket{\Psi(\theta,\phi)}\bra{\Psi(\theta,\phi)}^{\otimes N}\right]\ket{\Psi(\theta,\phi)}\bra{\Psi(\theta,\phi)}^{\otimes M}\,d\Omega.
\end{equation}
We now compute the action of $\mE_{\text{MP}}$ on a single Dicke state. The result is as follows:

\begin{lemma}[applying the optimal measure-and-prepare channel to a Dicke state]
When the optimal measure-and-prepare channel from $N$ qubits to $M$ qubits is applied to the Dicke state $\ket{D^{(N)}_w}\bra{D^{(N)}_w}$, the result is as follows:
\begin{align}
    & \quad\,\, \mE_{\text{MP}}[N\to M]\left(\ket{D^{(N)}_w}\bra{D^{(N)}_w}\right) \\
    &= \binom{N+M+1}{N+1}^{-1}\sum_{\tilde{w}=0}^{M}\binom{w+\tilde{w}}{w}\binom{(N-w)+(M-\tilde{w})}{N-w}\ket{D^{(M)}_{\tilde{w}}}\bra{D^{(M)}_{\tilde{w}}}.
\end{align}
\label{lem:optimal-mp-channel-dicke-state-unified-presentation}
\end{lemma}

\begin{proof}
We first simplify the SJSM by noticing that the measurement outcome will be uniformly distributed over the azimuthal angle $\phi$. Therefore, for the sake of the POVM and the subsequent measurement, it makes sense to compute the following quantity, which is just the average of an identical pure state over all possible $\phi$ (but holding $\theta$ fixed):
\begin{equation}
    \sigma(N,\theta) \coloneqq \frac{1}{2\pi}\int_{0}^{2\pi}\ket{\Psi(\theta,\phi)}\bra{\Psi(\theta,\phi)}^{\otimes N}\,d\phi.
\end{equation}
As a result, for any input state $\rho_N$ that is diagonal in the basis of Dicke states, $\mE_{\text{MP}}$ can actually be simplified as follows:
\begin{equation}
    \mE_{\text{MP}}[N\to M](\rho_{N}) = \frac{N+1}{2}\int_{0}^{\pi}\text{Tr}\left[\sigma(N,\theta)\rho_{N}\right]\sigma(M,\theta)(\sin\theta\,d\theta).
\end{equation}

\vspace{0.5\baselineskip}

We first compute the density operator for an identical pure state:
\begin{align}
    \ket{\Psi(\theta,\phi)}^{\otimes N} &= \sum_{b\in\{0,1\}^N}\left(\cos\frac{\theta}{2}\right)^{N-w(b)}\left(\sin\frac{\theta}{2}\right)^{w(b)}e^{i\phi w(b)}\ket{b} \\
    \therefore \ket{\Psi(\theta,\phi)}\bra{\Psi(\theta,\phi)}^{\otimes N} &= \sum_{b,b'\in\{0,1\}^N}\left(\cos\frac{\theta}{2}\right)^{2N-w(b)-w(b')}\left(\sin\frac{\theta}{2}\right)^{w(b)+w(b')}e^{i\phi(w(b)-w(b'))}\ket{b}\bra{b'}.
\end{align}
We then integrate this quantity over $\phi$ to obtain $\sigma(N,\theta)$:
\begin{align}
    \sigma(N,\theta) &= \frac{1}{2\pi}\int_{0}^{2\pi}\ket{\Psi(\theta,\phi)}\bra{\Psi(\theta,\phi)}^{\otimes N}\,d\phi \\
    &= \frac{1}{2\pi}\int_{0}^{2\pi}\sum_{b,b'\in\{0,1\}^N}\left(\cos\frac{\theta}{2}\right)^{2N-w(b)-w(b')}\left(\sin\frac{\theta}{2}\right)^{w(b)+w(b')}e^{i\phi(w(b)-w(b'))}\ket{b}\bra{b'}\,d\phi \\
    &= \sum_{b,b'\in\{0,1\}^N}\left(\cos\frac{\theta}{2}\right)^{2N-w(b)-w(b')}\left(\sin\frac{\theta}{2}\right)^{w(b)+w(b')}\ket{b}\bra{b'}\left[\frac{1}{2\pi}\int_{0}^{2\pi}e^{i\phi(w(b)-w(b'))}\,d\phi\right] \\
    &= \sum_{b,b'\in\{0,1\}^N}\left(\cos\frac{\theta}{2}\right)^{2N-w(b)-w(b')}\left(\sin\frac{\theta}{2}\right)^{w(b)+w(b')}\ket{b}\bra{b'}\delta_{w(b),w(b')} \\
    &= \sum_{w=0}^{N}\sum_{w(b)=w(b')=w}\left(\cos\frac{\theta}{2}\right)^{2(N-w)}\left(\sin\frac{\theta}{2}\right)^{2w}\ket{b}\bra{b'} \\
    &= \sum_{w=0}^{N}\left(\cos^2\frac{\theta}{2}\right)^{N-w}\left(\sin^2\frac{\theta}{2}\right)^w\sum_{w(b)=w(b')=w}\ket{b}\bra{b'} \\
    &= \sum_{w=0}^{N}\binom{N}{w}\left(\cos^2\frac{\theta}{2}\right)^{N-w}\left(\sin^2\frac{\theta}{2}\right)^w\ket{D^{(N)}_w}\bra{D^{(N)}_w} \\
    &= 2^{-N}\sum_{w=0}^{N}\binom{N}{w}(1+\cos\theta)^{N-w}(1-\cos\theta)^w\ket{D^{(N)}_w}\bra{D^{(N)}_w}.
\end{align}

\vspace{0.5\baselineskip}

We are now ready to evaluate the action of $\mE_{\text{MP}}$ on a single Dicke state $\ket{D^{(N)}_w}\bra{D^{(N)}_w}$. First, notice that
\begin{equation}
    \text{Tr}\left[\sigma(N,\theta)\ket{D^{(N)}_w}\bra{D^{(N)}_w}\right] = 2^{-N}\binom{N}{w}(1+\cos\theta)^{N-w}(1-\cos\theta)^w.
\end{equation}
We now proceed as follows:
\begin{align}
    & \quad\,\, \mE_{\text{MP}}[N\to M]\left(\ket{D^{(N)}_w}\bra{D^{(N)}_w}\right) \\
    &= \frac{N+1}{2}\int_{0}^{\pi}\text{Tr}\left[\sigma(N,\theta)\ket{D^{(N)}_w}\bra{D^{(N)}_w}\right]\sigma(M,\theta)(\sin\theta\,d\theta) \\
    &= \frac{N+1}{2}\int_{0}^{\pi}\left[2^{-N}\binom{N}{w}(1+\cos\theta)^{N-w}(1-\cos\theta)^w\right] \\
    & \quad\quad \left[2^{-M}\sum_{\tilde{w}=0}^{M}\binom{M}{\tilde{w}}(1+\cos\theta)^{M-\tilde{w}}(1-\cos\theta)^{\tilde{w}}\ket{D^{(M)}_{\tilde{w}}}\bra{D^{(M)}_{\tilde{w}}}\right](\sin\theta\,d\theta) \\
    &= (N+1)2^{-(N+M+1)}\sum_{\tilde{w}=0}^{M}\binom{M}{\tilde{w}}\binom{N}{w}\ket{D^{(M)}_{\tilde{w}}}\bra{D^{(M)}_{\tilde{w}}}\\
    & \quad\quad \left[\int_{0}^{\pi}(1-\cos\theta)^{w+\tilde{w}}(1+\cos\theta)^{(N-w)+(M-\tilde{w})}\sin\theta\,d\theta\right].
\end{align}
Using the substitution $x = \frac{1-\cos\theta}{2}$ (which yields $dx = \frac{\sin\theta}{2}\,d\theta$), we can simplify the integral expression to obtain
\begin{align}
    & \quad\,\, \mE_{\text{MP}}[N\to M]\left(\ket{D^{(N)}_w}\bra{D^{(N)}_w}\right) \\
    &= (N+1)\sum_{\tilde{w}=0}^{M}\binom{M}{\tilde{w}}\binom{N}{w}\ket{D^{(M)}_{\tilde{w}}}\bra{D^{(M)}_{\tilde{w}}}\left[\int_{0}^{1}x^{w+\tilde{w}}(1-x)^{(N-w)+(M-\tilde{w})}\,dx\right].
\end{align}

\vspace{0.5\baselineskip}

We now invoke the integral formula
\begin{equation}
    \int_{0}^{1}x^a(1-x)^b\,dx = \frac{1}{a+b+1}\binom{a+b}{a}^{-1} = \frac{a!\,b!}{(a+b+1)!}.
\end{equation}
Applying the above integral formula and recombining the various binomial coefficients, we conclude that
\begin{align}
    & \quad\,\, \mE_{\text{MP}}[N\to M]\left(\ket{D^{(N)}_w}\bra{D^{(N)}_w}\right) \\
    &= \frac{N+1}{N+M+1}\sum_{\tilde{w}=0}^{M}\frac{\binom{M}{\tilde{w}}\binom{N}{w}}{\binom{N+M}{w+\tilde{w}}}\ket{D^{(M)}_{\tilde{w}}}\bra{D^{(M)}_{\tilde{w}}} \\
    &= \binom{N+M+1}{N+1}^{-1}\sum_{\tilde{w}=0}^{M}\binom{w+\tilde{w}}{w}\binom{(N-w)+(M-\tilde{w})}{N-w}\ket{D^{(M)}_{\tilde{w}}}\bra{D^{(M)}_{\tilde{w}}}.
\end{align}
\end{proof}

It is worth mentioning that the sequence of coefficients shown above precisely matches a \textbf{negative hypergeometric distribution}. Suppose you have an urn with $N$ balls, exactly $K$ of which are colored red, and the rest are colored black. Now you sample them uniformly at random \emph{without replacement} until you get exactly $r$ black balls. If $X$ denotes the number of red balls you sample, then $X$ follows a negative hypergeometric distribution, and we say that $X\sim\text{NHG}(N,K,r)$. In particular:
\begin{equation}
    X\sim\text{NHG}(N,K,r) \iff \Pbb[X=k] = \frac{\binom{k+r-1}{k}\binom{N-r-k}{K-k}}{\binom{N}{K}}.
\end{equation}
We can see that the optimal measure-and-prepare channel yields precisely this distribution for the output Hamming weight $\tilde{w}$:
\begin{equation}
    \tilde{w}\sim\text{NHG}(N+M+1,M,w+1) \iff \text{coeff}[\text{Hamming weight }\tilde{w}] = \frac{\binom{w+\tilde{w}}{\tilde{w}}\binom{(N-w)+(M-\tilde{w})}{M-\tilde{w}}}{\binom{N+M+1}{M}} = \frac{\binom{w+\tilde{w}}{w}\binom{(N-w)+(M-\tilde{w})}{N-w}}{\binom{N+M+1}{N+1}},
\end{equation}
which is exactly what we see in Lemma \ref{lem:optimal-mp-channel-dicke-state-unified-presentation}.

\vspace{0.5\baselineskip}

For our work on qubit linear-rate conversion, we do not need any additional information about $\mE_{\text{MP}}$ beyond Lemma \ref{lem:optimal-mp-channel-dicke-state-unified-presentation}. Nonetheless, we encourage the curious reader to check out Appendix \ref{sec:math-tidbits} for some more interesting facts about $\mE_{\text{MP}}$ that can be derived using this result.

\appsubsec{Optimal Wrong-Way Measure-and-Prepare Channel}
{subsec:optimal-ww-channel-unified-presentation}

Finally, we define the optimal wrong-way measure-and-prepare channel $\mE_{\text{WW}}$. The POVM for $\mE_{\text{WW}}$ is actually also the SJSM (see Definition \ref{def:standard-joint-symmetric-measurement}). The only difference with $\mE_{\text{MP}}$ is that, when preparing $M$ identical qubits based on the measurement outcome, we prepare them in the opposite direction. As a result, the action of this channel on a single Dicke state takes a form essentially identical to that of the optimal measure-and-prepare channel:

\begin{lemma}[applying the optimal wrong-way measure-and-prepare channel to a Dicke state]
When the optimal wrong-way measure-and-prepare channel from $N$ qubits to $M$ qubits is applied to the Dicke state $\ket{D^{(N)}_w}\bra{D^{(N)}_w}$, the result is as follows:
\begin{align}
    & \quad\,\, \mE_{\text{WW}}[N\to M]\left(\ket{D^{(N)}_w}\bra{D^{(N)}_w}\right) \\
    &= \binom{N+M+1}{N+1}^{-1}\sum_{\tilde{w}=0}^{M}\binom{w+(M-\tilde{w})}{w}\binom{(N-w)+\tilde{w}}{N-w}\ket{D^{(M)}_{\tilde{w}}}\bra{D^{(M)}_{\tilde{w}}}.
\end{align}
\label{lem:optimal-ww-channel-dicke-state-unified-presentation}
\end{lemma}

\begin{proof}
Preparing the qubits in the opposite direction is equivalent to replacing every Dicke state $\ket{D^{M}_{\tilde{w}}}\bra{D^{M}_{\tilde{w}}}$ with the Dicke state with the same Hamming weight but pointing in the opposite direction, which is simply $\ket{D^{M}_{M-\tilde{w}}}\bra{D^{M}_{M-\tilde{w}}}$. Therefore, we can simply use the result of Lemma \ref{lem:optimal-mp-channel-dicke-state-unified-presentation} and make this replacement to obtain
\begin{align}
    & \quad\,\, \mE_{\text{WW}}[N\to M]\left(\ket{D^{(N)}_w}\bra{D^{(N)}_w}\right) \\
    &= \binom{N+M+1}{N+1}^{-1}\sum_{\tilde{w}=0}^{M}\binom{w+\tilde{w}}{w}\binom{(N-w)+(M-\tilde{w})}{N-w}\ket{D^{M}_{M-\tilde{w}}}\bra{D^{M}_{M-\tilde{w}}} \\
    &= \binom{N+M+1}{N+1}^{-1}\sum_{\tilde{w}=0}^{M}\binom{w+(M-\tilde{w})}{w}\binom{(N-w)+\tilde{w}}{N-w}\ket{D^{(M)}_{\tilde{w}}}\bra{D^{(M)}_{\tilde{w}}}.
\end{align}
\end{proof}

\newpage

\appsec{Achieving Conversion Between Schur-Transformed States}
{sec:schur-transformed-state-conversion}

In Appendix \ref{sec:schur-sampling-commentary}, we defined the notion of a Schur-transformed state and explained how an appropriate ensemble of Schur-transformed states is equivalent to an i.i.d. state. Because of this, the problem of linear-rate conversion between i.i.d. qubit collections becomes essentially equivalent to a problem of converting between suitably chosen Schur-transformed states.

\vspace{0.5\baselineskip}

Intuitively, if we want to convert $N$ i.i.d. qubits at purity level $\lambda_{\text{in}}$ to $RN$ i.i.d. qubits at purity level $\lambda_{\text{out}}$, then we should convert a Schur-transformed state $\rho_C(N_C\approx\lambda_{\text{in}}N, \lambda_{\text{in}}, \hat{n})$ into a Schur-transformed state $\rho_C(\tilde{N}_C\approx\lambda_{\text{out}}RN, \lambda_{\text{out}}, \hat{n})$. Therefore, if we define $R_C\coloneqq\tilde{N}_C/N_C$ as the ``Schur-transformed conversion rate'', then we should expect to see
\begin{equation}
    R_C \coloneqq \frac{\tilde{N}_C}{N_C} \approx R\frac{\lambda_{\text{out}}}{\lambda_{\text{in}}} \approx \begin{cases}
        \frac{\lambda_{\text{in}}}{1-\lambda_{\text{in}}} \Big/ \frac{\lambda_{\text{out}}}{1-\lambda_{\text{out}}} & \text{concentration} \\
        \frac{\lambda_{\text{in}}}{1+\lambda_{\text{in}}} \Big/ \frac{\lambda_{\text{out}}}{1+\lambda_{\text{out}}} & \text{dilution} \\
        \frac{\lambda_{\text{in}}}{1+\lambda_{\text{in}}} \Big/ \frac{\lambda_{\text{out}}}{1-\lambda_{\text{out}}} & \text{MP/WW conversion.} \\
    \end{cases}
\end{equation}

\vspace{0.5\baselineskip}

In this appendix, we show how each of the four channels introduced in Appendix \ref{sec:four-important-channels} achieves one of the desired conversion tasks for Schur-transformed states, in the manner described above. In each case, the workflow is the same:
\begin{itemize}
    \item Cite the appropriate result from Appendix \ref{sec:four-important-channels} that tells us the action of the channel on a single Dicke state:
    \begin{itemize}
        \item Lemma \ref{lem:discarding-map-dicke-state-unified-presentation} for the discarding map;
        \item Lemma \ref{lem:optimal-cloning-map-dicke-state-unified-presentation} for the optimal cloning map;
        \item Lemma \ref{lem:optimal-mp-channel-dicke-state-unified-presentation} for the optimal measure-and-prepare channel.
        \item Lemma \ref{lem:optimal-ww-channel-dicke-state-unified-presentation} for the optimal wrong-way measure-and-prepare channel.
    \end{itemize}
    \item Use the formula for a Schur-transformed state (Definition \ref{def:Schur-transformed-states}) to compute the action of the channel on a Schur-transformed state $\rho_C(N_C,\lambda_{\text{in}},\hat{n})$.
    \item Make a sequence of approximations to show that the result has vanishing trace distance with a new Schur-transformed state $\rho_C(\tilde{N}_C,\lambda_{\text{out}},\hat{n})$. (This applies to concentration, dilution, and measure-and-prepare conversion. For wrong-way conversion, replace $\hat{n}$ with $-\hat{n}$.)
\end{itemize}

Later on, in Appendix \ref{sec:unified-presentation}, we will use the results of this appendix to tackle the original task, which is to convert between i.i.d. qubit collections. The first step will be Schur sampling, the last step will be inverse Schur sampling, and the conversion between Schur-transformed states (as shown in this appendix) will be the middle step.

\vspace{0.5\baselineskip}

Just as we did in Appendix \ref{sec:four-important-channels}, we will suppress the $\hat{n}$ notation used in Definitions \ref{def:directional-states}, \ref{def:Dicke-states}, and \ref{def:Schur-transformed-states}. Since all four channels we apply to a Schur-transformed state are $\mathrm{SU}(2)$-covariant, the results we show here hold regardless of the choice of $\hat{n}$.

\appsubsec{Concentration for Schur-Transformed States}
{subsec:schur-transformed-state-concentration-unified-presentation}

The discarding map achieves concentration for Schur-transformed states with vanishing trace distance in the limit of many qubits. To formalize this fact, we state and prove the following lemma:

\begin{lemma}[discarding map transforms Schur-transformed state into approximate Schur-transformed state]
Suppose $N_C = \lambda_{\text{in}}N + O(N^p)$ and $\tilde{N}_C = \lambda_{\text{out}}\tilde{N} + O(N^p)$ for some $0 < p < 1$, where $\tilde{N} = \lfloor RN\rfloor$ and $R = R^{\text{conc}}(\lambda_{\text{in}}\to\lambda_{\text{out}})$. Then
\begin{equation}
    d_{\text{Tr}}\left(\mE_{\text{discard}}[N_C\to\tilde{N}_C]\left(\rho_C(N_C,\lambda_{\text{in}},\hat{n})\right), \rho_C(\tilde{N}_C,\lambda_{\text{out}},\hat{n})\right) = O\left(N^{p+\varepsilon-1}\right)
\end{equation}
for any $\varepsilon > 0$.
\label{lem:discarding-map-schur-transformed-state}
\end{lemma}

\begin{proof}
We first apply Lemma \ref{lem:discarding-map-dicke-state-unified-presentation} with $N\mapsto N_C$ and $M\mapsto\tilde{N}_C$ to compute the action of the discarding map on the whole Schur-transformed state:
\begin{align}
    & \quad\,\, \mE_{\text{discard}}[N_C\to\tilde{N}_C]\left(\rho_C(N_C,\lambda_{\text{in}})\right) \\
    &= \frac{c_1-c_0}{c_1^{N_C+1}-c_0^{N_C+1}}\sum_{w=0}^{N_C}c_1^wc_0^{N_C-w}\mE_{\text{discard}}[N_C\to\tilde{N}_C]\left(\ket{D^{(N_C)}_w}\bra{D^{(N_C)}_w}\right) \\
    &= \sum_{\tilde{w}=0}^{\tilde{N}_C}A(\tilde{w})\ket{D^{(\tilde{N}_C)}_{\tilde{w}}}\bra{D^{(\tilde{N}_C)}_{\tilde{w}}},
\end{align}
where the coefficient $A(\tilde{w})$ takes the form
\begin{equation}
    A(\tilde{w}) = \frac{c_1-c_0}{c_1^{N_C+1}-c_0^{N_C+1}}\binom{N_C}{\tilde{N}_C}^{-1}\sum_{w=0}^{N_C}c_1^wc_0^{N_C-w}\binom{w}{\tilde{w}}\binom{N_C - w}{\tilde{N}_C - \tilde{w}},
\end{equation}
and where we have defined $c_1\coloneqq\frac{1+\lambda_{\text{in}}}{2}$ and $c_0\coloneqq\frac{1-\lambda_{\text{in}}}{2}$ for convenience.

\vspace{0.5\baselineskip}

The target state is the Schur-transformed state $\rho_C(\tilde{N}_C,\lambda_{\text{out}})$:
\begin{equation}
    \rho_C(\tilde{N}_C,\lambda_{\text{out}}) = \sum_{\tilde{w}=0}^{\tilde{N}_C}E(\tilde{w})\ket{D^{(\tilde{N}_C)}_{\tilde{w}}}\bra{D^{(\tilde{N}_C)}_{\tilde{w}}},
\end{equation}
where the coefficient $E(\tilde{w})$ takes the form
\begin{equation}
    E(\tilde{w}) = \frac{\tilde{c}_1-\tilde{c}_0}{\tilde{c}_1^{\tilde{N}_C+1}-\tilde{c}_0^{\tilde{N}_C+1}}\tilde{c}_1^{\tilde{w}}\tilde{c}_0^{\tilde{N}_C-\tilde{w}},
\end{equation}
and where we have defined $\tilde{c}_1\coloneqq\frac{1+\lambda_{\text{out}}}{2}$ and $\tilde{c}_0\coloneqq\frac{1-\lambda_{\text{out}}}{2}$ for convenience.

\vspace{0.5\baselineskip}

Since the output state and target state are both diagonal in the basis of Dicke states, we can write their trace distance as follows:
\begin{equation}
    d_{\text{Tr}}\left(\mE_{\text{discard}}[N_C\to\tilde{N}_C]\left(\rho_C(N_C,\lambda_{\text{in}})\right), \rho_C(\tilde{N}_C,\lambda_{\text{out}})\right) = \frac{1}{2}\sum_{\tilde{w}=0}^{\tilde{N}_C}\abs{A(\tilde{w}) - E(\tilde{w})}.
\end{equation}
To show that this trace distance vanishes as $N\to\infty$, we choose an arbitrarily small $\varepsilon_1 > 0$. As we will see later, we will especially take advantage of the fact that $\varepsilon_1 < \min\{p,1-p\}$. We then split the summation as follows:
\begin{align}
    \sum_{\tilde{w}=0}^{\tilde{N}_C}\abs{A(\tilde{w}) - E(\tilde{w})} &= \sum_{\tilde{w}=0}^{\tilde{N}_C-\tilde{N}_C^{\varepsilon_1}}\abs{A(\tilde{w}) - E(\tilde{w})} + \sum_{\tilde{w}=\tilde{N}_C-\tilde{N}_C^{\varepsilon_1}+1}^{\tilde{N}_C}\abs{A(\tilde{w}) - E(\tilde{w})}.
\end{align}
To upper bound the first summation, notice that
\begin{equation}
    E(\tilde{w}) = \frac{\tilde{c}_1-\tilde{c}_0}{\tilde{c}_1^{\tilde{N}_C+1}-\tilde{c}_0^{\tilde{N}_C+1}}\tilde{c}_1^{\tilde{w}}\tilde{c}_0^{\tilde{N}_C-\tilde{w}} \le \left(\frac{\tilde{c}_0}{\tilde{c}_1}\right)^{\tilde{N}_C-\tilde{w}}.
\end{equation}
Moreover, note that
\begin{align}
    A(\tilde{w}) &= \frac{c_1-c_0}{c_1^{N_C+1}-c_0^{N_C+1}}\binom{N_C}{\tilde{N}_C}^{-1}\sum_{w=0}^{N_C}c_1^wc_0^{N_C-w}\binom{w}{\tilde{w}}\binom{N_C - w}{\tilde{N}_C - \tilde{w}} \\
    &\le \binom{N_C}{\tilde{N}_C}^{-1}\sum_{w=0}^{N_C}\left(\frac{c_0}{c_1}\right)^{N_C-w}\binom{w}{\tilde{w}}\binom{N_C - w}{\tilde{N}_C - \tilde{w}} & (\text{upper bound factors involving $c_1$ and $c_0$}) \\
    &= \binom{N_C}{\tilde{N}_C}^{-1}\sum_{w=\tilde{w}}^{N_C-\tilde{N}_C+\tilde{w}}\left(\frac{c_0}{c_1}\right)^{N_C-w}\binom{w}{\tilde{w}}\binom{N_C - w}{\tilde{N}_C - \tilde{w}} \\
    &\le \sum_{w=\tilde{w}}^{N_C-\tilde{N}_C+\tilde{w}}\left(\frac{c_0}{c_1}\right)^{N_C-w} & (\text{because }\binom{n_1+n_2}{k_1+k_2}\ge\binom{n_1}{k_1}\binom{n_2}{k_2}) \\
    &\le \frac{c_1}{c_1-c_0}\left(\frac{c_0}{c_1}\right)^{\tilde{N}_C-\tilde{w}} & (\text{extend summation to }w=-\infty).
\end{align}
In other words, we upper bounded $E(\tilde{w})$ by a geometric sequence with common ratio $\tilde{c}_0/\tilde{c}_1$, and we upper bounded $A(\tilde{w})$ by a geometric sequence with common ratio $c_0/c_1$. Combining these two facts yields
\begin{align}
    & \quad\,\, \sum_{\tilde{w}=0}^{\tilde{N}_C-\tilde{N}_C^{\varepsilon_1}}\abs{A(\tilde{w}) - E(\tilde{w})} \\
    &\le \sum_{\tilde{w}=0}^{\tilde{N}_C-\tilde{N}_C^{\varepsilon_1}}A(\tilde{w}) + \sum_{\tilde{w}=0}^{\tilde{N}_C-\tilde{N}_C^{\varepsilon_1}}E(\tilde{w}) \\
    &= O\left(\left(\frac{c_0}{c_1}\right)^{\tilde{N}_C^{\varepsilon_1}}\right) + O\left(\left(\frac{\tilde{c}_0}{\tilde{c}_1}\right)^{\tilde{N}_C^{\varepsilon_1}}\right).
\end{align}
For any $\varepsilon_1 > 0$, this actually decays faster than any power law decay, i.e., it is $o(N^{-q})$ for all $q > 0$ in the $N\to\infty$ limit. Intuitively, this corresponds to the fact that both the output state and the target state only have nonnegligible probabilities in the highest Hamming weights (i.e., $\tilde{w}\approx\tilde{N}_C$), with all lower Hamming weights having negligible probabilities.

\vspace{0.5\baselineskip}

To upper bound the second summation, we will prove that
\begin{equation}
    \tilde{N}_C - \tilde{w} < \tilde{N}_C^{\varepsilon_1} \implies \abs{A(\tilde{w}) - E(\tilde{w})} = O\left(N^{p+\varepsilon_1-1}\right).
\end{equation}
To prove this, we first define $\delta\coloneqq N_C-w$ and $\tilde{\delta}\coloneqq\tilde{N}_C-\tilde{w}$ for convenience. We additionally define the ``Schur-transformed conversion rate''
\begin{equation}
    R_C \coloneqq R\frac{\lambda_{\text{out}}}{\lambda_{\text{in}}} = \frac{\lambda_{\text{in}}}{1-\lambda_{\text{in}}} \Bigg/ \frac{\lambda_{\text{out}}}{1-\lambda_{\text{out}}}.
\end{equation}
We now define the following sequence of quantities:
\begin{align}
    A_0(\tilde{\delta}) &= \frac{c_1-c_0}{c_1^{N_C+1}-c_0^{N_C+1}}\binom{N_C}{\tilde{N}_C}^{-1}\sum_{\delta=0}^{N_C}c_1^{N_C-\delta}c_0^{\delta}\binom{N_C-\delta}{\tilde{N}_C-\tilde{\delta}}\binom{\delta}{\tilde{\delta}} \\
    A_1(\tilde{\delta}) &= \frac{c_1-c_0}{c_1}\binom{N_C}{\tilde{N}_C}^{-1}\sum_{\delta=0}^{N_C}\left(\frac{c_0}{c_1}\right)^\delta\binom{N_C-\delta}{\tilde{N}_C-\tilde{\delta}}\binom{\delta}{\tilde{\delta}} \\
    A_2(\tilde{\delta}) &= \frac{c_1-c_0}{c_1}\binom{N_C}{\tilde{N}_C}^{-1}\sum_{\delta=\tilde{\delta}}^{N_C^{\varepsilon_1}-1}\left(\frac{c_0}{c_1}\right)^\delta\binom{N_C-\delta}{\tilde{N}_C-\tilde{\delta}}\binom{\delta}{\tilde{\delta}} \\
    A_3(\tilde{\delta}) &= \frac{c_1-c_0}{c_1}\sum_{\delta=\tilde{\delta}}^{N_C^{\varepsilon_1}-1}\left(\frac{c_0}{c_1}\right)^\delta R_C^{\tilde{\delta}}(1-R_C)^{\delta-\tilde{\delta}}\binom{\delta}{\tilde{\delta}} \\
    A_4(\tilde{\delta}) &= \frac{\tilde{c}_1-\tilde{c}_0}{\tilde{c}_1}\left(\frac{\tilde{c}_0}{\tilde{c}_1}\right)^{\tilde{\delta}} \\
    A_5(\tilde{\delta}) &= \frac{\tilde{c}_1-\tilde{c}_0}{\tilde{c}_1^{\tilde{N}_C+1}-\tilde{c}_0^{\tilde{N}_C+1}}\tilde{c}_1^{\tilde{N}_C-\tilde{\delta}}\tilde{c}_0^{\tilde{\delta}}.
\end{align}
Notice that $A_0(\tilde{\delta}) = A(\tilde{w})$ and $A_5(\tilde{\delta}) = E(\tilde{w})$. Therefore, we can set up a massive triangle inequality by bounding the successive differences:
\begin{equation}
    \abs{A(\tilde{w}) - E(\tilde{w})} \le \sum_{k=1}^{5}\abs{A_{k-1}(\tilde{\delta}) - A_k(\tilde{\delta})}.
\end{equation}
Each term corresponds to one approximation we will make. In particular, we will show that
\begin{align}
    \abs{A_0(\tilde{\delta}) - A_1(\tilde{\delta})} &\le \frac{(c_0/c_1)^{N_C+1}}{1-(c_0/c_1)^{N_C+1}} \\
    \abs{A_1(\tilde{\delta}) - A_2(\tilde{\delta})} &\le \left(\frac{c_0}{c_1}\right)^{N_C^{\varepsilon_1}} \\
    \abs{A_2(\tilde{\delta}) - A_3(\tilde{\delta})} &= O\left(N^{p+\varepsilon_1-1}\right) \\
    \abs{A_3(\tilde{\delta}) - A_4(\tilde{\delta})} &\le \left(\frac{c_0}{c_1}\right)^{N_C^{\varepsilon_1}} \\
    \abs{A_4(\tilde{\delta}) - A_5(\tilde{\delta})} &\le \frac{(\tilde{c}_0/\tilde{c}_1)^{\tilde{N}_C+1}}{1-(\tilde{c}_0/\tilde{c}_1)^{\tilde{N}_C+1}}.
\end{align}
Combining these approximations and applying the triangle inequality yields
\begin{equation}
    \abs{A(\tilde{w}) - E(\tilde{w})} = O\left(N^{p+\varepsilon_1-1}\right).
\end{equation}
Now all that remains is to prove each of the above approximations:
\begin{itemize}
    \item \textbf{First approximation:} Use the approximation
    \begin{equation}
        \frac{c_1-c_0}{c_1^{N_C+1}-c_0^{N_C+1}}c_1^{N_C} = \frac{c_1-c_0}{c_1}\left[1 + \frac{(c_0/c_1)^{N_C+1}}{1-(c_0/c_1)^{N_C+1}}\right].
    \end{equation}
    Therefore, $A_0(\tilde{\delta})$ is a very slight overestimate of $A_1(\tilde{\delta})$, as follows:
    \begin{align}
        0 \le A_0(\tilde{\delta}) - A_1(\tilde{\delta}) = \frac{(c_0/c_1)^{N_C+1}}{1-(c_0/c_1)^{N_C+1}}A_1(\tilde{\delta}) \le \frac{(c_0/c_1)^{N_C+1}}{1-(c_0/c_1)^{N_C+1}}.
    \end{align}
    For clarity, the last inequality comes from the fact that $A_1(\tilde{\delta})\le A_0(\tilde{\delta})\le 1$.
    
    \item \textbf{Second approximation:} The summand is zero for $\delta < \tilde{\delta}$, so we can freely remove these terms. Now additionally truncate the sum early at $\delta=N_C^{\varepsilon_1}-1$. Therefore, the error can be bounded via
    \begin{align}
        0 \le A_1(\tilde{\delta}) - A_2(\tilde{\delta}) &= \frac{c_1-c_0}{c_1}\binom{N_C}{\tilde{N}_C}^{-1}\sum_{\delta=N_C^{\varepsilon_1}}^{N_C}\left(\frac{c_0}{c_1}\right)^\delta\binom{N_C-\delta}{\tilde{N}_C-\tilde{\delta}}\binom{\delta}{\tilde{\delta}} \\
        &\stackrel{(1)}{\le} \frac{c_1-c_0}{c_1}\sum_{\delta=N_C^{\varepsilon_1}}^{N_C}\left(\frac{c_0}{c_1}\right)^\delta \\
        &\le \frac{c_1-c_0}{c_1}\sum_{\delta=N_C^{\varepsilon_1}}^{\infty}\left(\frac{c_0}{c_1}\right)^\delta \\
        &= \left(\frac{c_0}{c_1}\right)^{N_C^{\varepsilon_1}}.
    \end{align}
    For clarity, inequality (1) comes from the fact that $\binom{n_1+n_2}{k_1+k_2} \ge \binom{n_1}{k_1}\binom{n_2}{k_2}$.
    
    \item \textbf{Third approximation:} We first observe the following bounds for the ratio of two factorials:
    \begin{equation}
        n^k \ge \frac{n!}{(n-k)!} \ge (n-k)^k = n^k\left(1-\frac{k}{n}\right)^k.
    \end{equation}
    Furthermore, using the fact that $1-a\ge e^{-2a}$ when $0\le a\le 1/2$, we can say that
    \begin{equation}
        0 \le \frac{k}{n} \le \frac{1}{2} \implies \left(1-\frac{k}{n}\right)^k \ge e^{-2k^2/n}.
    \end{equation}
    Putting this all together yields
    \begin{equation}
        0 \le k \le \frac{n}{2} \implies n^ke^{-2k^2/n} \le \frac{n!}{(n-k)!} \le n^k \implies \frac{n!}{(n-k)!} = n^ke^{O(k^2/n)}.
    \end{equation}
    We can thus approximate the ratio of two binomial coefficients of interest as follows:
    \begin{align}
        \frac{\binom{N_C-\delta}{\tilde{N}_C-\tilde{\delta}}}{\binom{N_C}{\tilde{N}_C}} &= \frac{(N_C-\delta)!}{N_C!}\frac{\tilde{N}_C!}{(\tilde{N}_C-\tilde{\delta})!}\frac{(N_C-\tilde{N}_C)!}{((N_C-\tilde{N}_C)-(\delta-\tilde{\delta}))!} \\
        &= \left[N_C^{-\delta}e^{O(\delta^2/N_C)}\right]\left[\tilde{N}_C^{\tilde{\delta}}e^{O(\tilde{\delta}^2/\tilde{N}_C)}\right]\left[(N_C-\tilde{N}_C)^{\delta-\tilde{\delta}}e^{O\left((\delta-\tilde{\delta})^2/(N_C-\tilde{N}_C)\right)}\right] \\
        &\stackrel{(1)}{=} \left(\frac{\tilde{N}_C}{N_C}\right)^{\tilde{\delta}}\left(\frac{N_C-\tilde{N}_C}{N_C}\right)^{\delta-\tilde{\delta}}\exp\left[O\left(N^{2\varepsilon_1-1}\right)\right].
    \end{align}
    For clarity, equality (1) is ensured by the fact that $\tilde{\delta}\le\delta\le N_C^{\varepsilon_1}$. We additionally observe that the typicality of $N_C$ and $\tilde{N}_C$ permits the following approximations:
    \begin{align}
        \frac{\tilde{N}_C}{N_C} &= R_C\left[1 + O(N^{p-1})\right] = R_C\exp\left[O(N^{p-1})\right] \\
        \frac{N_C-\tilde{N}_C}{N_C} &= (1-R_C)\left[1 + O(N^{p-1})\right] = (1-R_C)\exp\left[O(N^{p-1})\right].
    \end{align}
    It might seem somewhat strange to write the above approximations using exponentials, but we will raise these two quantities to powers involving $\delta$ and $\tilde{\delta}$, so it will be easier to reason about them this way. We thus proceed as follows:
    \begin{align}
        \frac{\binom{N_C-\delta}{\tilde{N}_C-\tilde{\delta}}}{\binom{N_C}{\tilde{N}_C}} &= R_C^{\tilde{\delta}}\exp\left[O\left(N^{p-1}\right)\right]^{\tilde{\delta}}(1-R_C)^{\delta-\tilde{\delta}}\exp\left[O\left(N^{p-1}\right)\right]^{\delta-\tilde{\delta}}\exp\left[O\left(N^{2\varepsilon_1-1}\right)\right] \\
        &= R_C^{\tilde{\delta}}(1-R_C)^{\delta-\tilde{\delta}}\exp\left[O\left(N^{p+\varepsilon_1-1}\right)\right]\exp\left[O\left(N^{2\varepsilon_1-1}\right)\right] \\
        &\stackrel{(1)}{=} R_C^{\tilde{\delta}}(1-R_C)^{\delta-\tilde{\delta}}\exp\left[O\left(N^{p+\varepsilon_1-1}\right)\right].
    \end{align}
    For clarity, equality (1) is ensured by the fact that $\varepsilon_1$ is arbitrarily small (in particular, recall that we chose $0 < \varepsilon_1 < p$), which is why the two $\exp\left[O\left(N^{p+\varepsilon_1-1}\right)\right]$ factors dominate over the $\exp\left[O\left(N^{2\varepsilon_1-1}\right)\right]$ factor. We can now upper bound the total error resulting from this approximation:
    \begin{align}
        \abs{A_2(\tilde{\delta}) - A_3(\tilde{\delta})} &= \frac{c_1-c_0}{c_1}\sum_{\delta=\tilde{\delta}}^{N_C^{\varepsilon_1}-1}\left(\frac{c_0}{c_1}\right)^\delta\binom{\delta}{\tilde{\delta}}R_C^{\tilde{\delta}}(1-R_C)^{\delta-\tilde{\delta}}\Big\{\exp\left[O\left(N^{p+{\varepsilon_1}-1}\right)\right]-1\Big\} \\
        &\stackrel{(1)}{=} O\left(N^{p+\varepsilon_1-1}\right)\frac{c_1-c_0}{c_1}\sum_{\delta=\tilde{\delta}}^{N_C^{\varepsilon_1}-1}\left(\frac{c_0}{c_1}\right)^\delta\binom{\delta}{\tilde{\delta}}R_C^{\tilde{\delta}}(1-R_C)^{\delta-\tilde{\delta}} \\
        &\stackrel{(2)}{\le} O\left(N^{p+\varepsilon_1-1}\right)\frac{c_1-c_0}{c_1}\sum_{\delta=\tilde{\delta}}^{N_C^{\varepsilon_1}-1}\left(\frac{c_0}{c_1}\right)^\delta \\
        &\stackrel{(3)}{\le} O\left(N^{p+\varepsilon_1-1}\right)\left(\frac{c_0}{c_1}\right)^{\tilde{\delta}} \\
        &\le O\left(N^{p+\varepsilon_1-1}\right).
    \end{align}
    For clarity, equality (1) comes from the fact that $\varepsilon_1$ is arbitrarily small (in particular, recall that we chose $0 < \varepsilon_1 < 1-p$), so that $N^{p+\varepsilon_1-1} = o(1)$, along with the fact that $e^x-1 = O(x)$ for sufficiently small $x$ (for example, $\abs{e^x-1}\le 2x$ for $\abs{x}\le\frac{1}{2}$); inequality (2) comes from the fact that $\binom{n}{k}p^k(1-p)^{n-k} \le 1$ for all $0\le p\le 1$; and inequality (3) comes from the fact that the finite geometric series is upper bounded by the corresponding infinite geometric series.
    
    \item \textbf{Fourth approximation:} Extend the summation to $\delta=\infty$, instead of terminating it at $\delta=N_C^{\varepsilon_1}-1$. Using the identity
    \begin{equation}
        \sum_{n=0}^{\infty}x^n\binom{n+p}{n} = (1-x)^{-(p+1)},
    \label{eq:geom-series-binom-coeff-identity}
    \end{equation}
    we can conclude that the infinite sum yields
    \begin{align}
        & \quad\,\, \frac{c_1-c_0}{c_1}\sum_{\delta=\tilde{\delta}}^{\infty}\left(\frac{c_0}{c_1}\right)^\delta R_C^{\tilde{\delta}}(1-R_C)^{\delta-\tilde{\delta}}\binom{\delta}{\tilde{\delta}} \\
        &= \frac{c_1-c_0}{c_1}\left(\frac{c_0}{c_1}R_C\right)^{\tilde{\delta}}\sum_{n=0}^{\infty}\left[\frac{c_0}{c_1}(1-R_C)\right]^{n}\binom{\tilde{\delta}+n}{\tilde{\delta}} & (\text{define }n\coloneqq\delta-\tilde{\delta}) \\
        &= \frac{c_1-c_0}{c_1}\left(\frac{c_0}{c_1}R_C\right)^{\tilde{\delta}}\left[1-\frac{c_0}{c_1}(1-R_C)\right]^{-(\tilde{\delta}+1)} & (\text{by Eq. }\ref{eq:geom-series-binom-coeff-identity}) \\
        &= \frac{c_1-c_0}{c_1}\frac{c_1}{c_1-c_0(1-R_C)}\left[\frac{c_0R_C}{c_1-c_0(1 -R_C)}\right]^{\tilde{\delta}} \\
        &= \frac{c_1-c_0}{c_0R_C+(c_1-c_0)}\left[\frac{c_0R_C}{c_0R_C+(c_1-c_0)}\right]^{\tilde{\delta}}.
    \end{align}
    Plugging in the values for $c_1$, $c_0$, $R_C$ reveals that
    \begin{align}
        \frac{c_1-c_0}{c_0R_C+(c_1-c_0)} &= \frac{\lambda_{\text{in}}}{\frac{1-\lambda_{\text{in}}}{2}\frac{\lambda_{\text{in}}}{1-\lambda_{\text{in}}}\frac{1-\lambda_{\text{out}}}{\lambda_{\text{out}}} + \lambda_{\text{in}}} = \frac{1}{\frac{1-\lambda_{\text{out}}}{2\lambda_{\text{out}}} + 1} = \frac{2\lambda_{\text{out}}}{1+\lambda_{\text{out}}} = \frac{\tilde{c}_1-\tilde{c}_0}{\tilde{c}_1} \\
        \frac{c_0R_C}{c_0R_C+(c_1-c_0)} &= 1 - \frac{c_1-c_0}{c_0R_C+(c_1-c_0)} = \frac{\tilde{c}_0}{\tilde{c}_1}.
    \end{align}
    We thus obtain precisely $A_4(\tilde{\delta})$. The error in this approximation is the sum of the terms introduced when extending the summation to infinity:
    \begin{align}
        0 \le A_4(\tilde{\delta}) - A_3(\tilde{\delta}) &= \frac{c_1-c_0}{c_1}\sum_{\delta=N_C^{\varepsilon_1}}^{\infty}\left(\frac{c_0}{c_1}\right)^\delta R_C^{\tilde{\delta}}(1-R_C)^{\delta-\tilde{\delta}}\binom{\delta}{\tilde{\delta}} \\
        &\stackrel{(1)}{\le} \frac{c_1-c_0}{c_1}\sum_{\delta=N_C^{\varepsilon_1}}^{\infty}\left(\frac{c_0}{c_1}\right)^\delta \\
        &= \left(\frac{c_0}{c_1}\right)^{N_C^{\varepsilon_1}}.
    \end{align}
    For clarity, inequality (1) comes from the fact that $\binom{n}{k}p^k(1-p)^{n-k} \le 1$ for all $0\le p\le 1$.
    
    \item \textbf{Fifth approximation:} The idea here is actually identical to that of the first approximation, except we use $\tilde{c_0}$, $\tilde{c_1}$, $\tilde{N}_C$ instead, and we go in the other direction. In particular, use the approximation
    \begin{equation}
        \frac{\tilde{c}_1-\tilde{c}_0}{\tilde{c}_1^{\tilde{N}_C+1}-\tilde{c}_0^{\tilde{N}_C+1}}\tilde{c}_1^{\tilde{N}_C} = \frac{\tilde{c}_1-\tilde{c}_0}{\tilde{c}_1}\left[1 + \frac{(\tilde{c}_0/\tilde{c}_1)^{\tilde{N}_C+1}}{1-(\tilde{c}_0/\tilde{c}_1)^{\tilde{N}_C+1}}\right].
    \end{equation}
    Therefore, $A_4(\tilde{\delta})$ is a very slight underestimate of $A_5(\tilde{\delta})$, as follows:
    \begin{align}
        0 \le A_5(\tilde{\delta}) - A_4(\tilde{\delta}) = \frac{(\tilde{c}_0/\tilde{c}_1)^{\tilde{N}_C+1}}{1-(\tilde{c}_0/\tilde{c}_1)^{\tilde{N}_C+1}}A_4(\tilde{\delta}) \le \frac{(\tilde{c}_0/\tilde{c}_1)^{\tilde{N}_C+1}}{1-(\tilde{c}_0/\tilde{c}_1)^{\tilde{N}_C+1}}.
    \end{align}
    For clarity, the last inequality comes from the fact that $A_4(\tilde{\delta})\le A_5(\tilde{\delta})\le 1$.
\end{itemize}

Now that we have upper bounded each individual $\abs{A(\tilde{w}) - E(\tilde{w})}$, we take the summation over all $\tilde{w}\ge\tilde{N}_C-\tilde{N}_C^{\varepsilon_1}+1$ to obtain
\begin{equation}
    \sum_{\tilde{w}=\tilde{N}_C-\tilde{N}_C^{\varepsilon_1}+1}^{\tilde{N}_C}\abs{A(\tilde{w}) - E(\tilde{w})} = \sum_{\tilde{w}=\tilde{N}_C-\tilde{N}_C^{\varepsilon_1}+1}^{\tilde{N}_C}O\left(N^{p+\varepsilon_1-1}\right) = O\left(N^{p+2\varepsilon_1-1}\right).
\end{equation}
The very last step is to combine the summation over $\tilde{w}\le\tilde{N}_C-\tilde{N}_C^{\varepsilon_1}$ and the summation over $\tilde{w}\ge\tilde{N}_C-\tilde{N}_C^{\varepsilon_1}+1$:
\begin{equation}
    \sum_{\tilde{w}=0}^{\tilde{N}_C}\abs{A(\tilde{w}) - E(\tilde{w})} = O\left(N^{p+2\varepsilon_1-1}\right).
\end{equation}
Setting $\varepsilon = 2\varepsilon_1$ tells us that the trace distance is $O\left(N^{p+\varepsilon-1}\right)$, exactly as desired.
\end{proof}

\appsubsec{Dilution for Schur-Transformed States}
{subsec:schur-transformed-state-dilution-unified-presentation}

The optimal cloning map achieves dilution for Schur-transformed states with vanishing trace distance in the limit of many qubits. To formalize this fact, we state and prove the following lemma:

\begin{lemma}[optimal cloning map transforms Schur-transformed state into approximate Schur-transformed state]
Suppose $N_C = \lambda_{\text{in}}N + O(N^p)$ and $\tilde{N}_C = \lambda_{\text{out}}\tilde{N} + O(N^p)$ for some $0 < p < 1$, where $\tilde{N} = \lfloor RN\rfloor$ and $R = R^{\text{dilut}}(\lambda_{\text{in}}\to\lambda_{\text{out}})$. Then
\begin{equation}
    d_{\text{Tr}}\left(\mE_{\text{clone}}[N_C\to\tilde{N}_C]\left(\rho_C(N_C,\lambda_{\text{in}},\hat{n})\right), \rho_C(\tilde{N}_C,\lambda_{\text{out}},\hat{n})\right) = O\left(N^{p+\varepsilon-1}\right)
\end{equation}
for any $\varepsilon > 0$.
\label{lem:optimal-cloning-map-schur-transformed-state}
\end{lemma}

\begin{proof}
We first apply Lemma \ref{lem:optimal-cloning-map-dicke-state-unified-presentation} with $N\mapsto N_C$ and $M\mapsto\tilde{N}_C$ to compute the action of the optimal cloning map on the whole Schur-transformed state:
\begin{align}
    & \quad\,\, \mE_{\text{clone}}[N_C\to\tilde{N}_C]\left(\rho(N_C,\lambda)\right) \\
    &= \frac{c_1-c_0}{c_1^{N_C+1}-c_0^{N_C+1}}\sum_{w=0}^{N_C}c_1^wc_0^{N_C-w}\mE_{\text{clone}}[N_C\to\tilde{N}_C]\left(\ket{D^{(N_C)}_w}\bra{D^{(N_C)}_w}\right) \\
    &= \sum_{\tilde{w}=0}^{\tilde{N}_C}B(\tilde{w})\ket{D^{(\tilde{N}_C)}_{\tilde{w}}}\bra{D^{(\tilde{N}_C)}_{\tilde{w}}},
\end{align}
where the coefficient $B(\tilde{w})$ takes the form
\begin{equation}
    B(\tilde{w}) = \frac{c_1-c_0}{c_1^{N_C+1} - c_0^{N_C+1}}\binom{\tilde{N}_C+1}{N_C+1}^{-1}\sum_{w=0}^{N_C}c_1^wc_0^{N_C-w}\binom{\tilde{w}}{w}\binom{\tilde{N}_C-\tilde{w}}{N_C-w},
\end{equation}
and where we have defined $c_1\coloneqq\frac{1+\lambda_{\text{in}}}{2}$ and $c_0\coloneqq\frac{1-\lambda_{\text{in}}}{2}$ for convenience.

\vspace{0.5\baselineskip}

The target state is the Schur-transformed state $\rho_C(\tilde{N}_C,\lambda_{\text{out}})$:
\begin{equation}
    \rho_C(\tilde{N}_C,\lambda_{\text{out}}) = \sum_{\tilde{w}=0}^{\tilde{N}_C}E(\tilde{w})\ket{D^{(\tilde{N}_C)}_{\tilde{w}}}\bra{D^{(\tilde{N}_C)}_{\tilde{w}}},
\end{equation}
where the coefficient $E(\tilde{w})$ takes the form
\begin{equation}
    E(\tilde{w}) = \frac{\tilde{c}_1-\tilde{c}_0}{\tilde{c}_1^{\tilde{N}_C+1}-\tilde{c}_0^{\tilde{N}_C+1}}\tilde{c}_1^{\tilde{w}}\tilde{c}_0^{\tilde{N}_C-\tilde{w}},
\end{equation}
and where we have defined $\tilde{c}_1\coloneqq\frac{1+\lambda_{\text{out}}}{2}$ and $\tilde{c}_0\coloneqq\frac{1-\lambda_{\text{out}}}{2}$ for convenience.

\vspace{0.5\baselineskip}

Since the output state and target state are both diagonal in the basis of Dicke states, we can write their trace distance as follows:
\begin{equation}
    d_{\text{Tr}}\left(\mE_{\text{clone}}[N_C\to\tilde{N}_C]\left(\rho_C(N_C,\lambda_{\text{in}})\right), \rho_C(\tilde{N}_C,\lambda_{\text{out}})\right) = \frac{1}{2}\sum_{\tilde{w}=0}^{\tilde{N}_C}\abs{B(\tilde{w}) - E(\tilde{w})}.
\end{equation}
To show that this trace distance vanishes as $N\to\infty$, we choose an arbitrarily small $\varepsilon_1 > 0$. As we will see later, we will especially take advantage of the fact that $\varepsilon_1 < \min\{p,1-p\}$. We then split the summation as follows:
\begin{align}
    \sum_{\tilde{w}=0}^{\tilde{N}_C}\abs{B(\tilde{w}) - E(\tilde{w})} &= \sum_{\tilde{w}=0}^{\tilde{N}_C-\tilde{N}_C^\varepsilon}\abs{B(\tilde{w}) - E(\tilde{w})} + \sum_{\tilde{w}=\tilde{N}_C-\tilde{N}_C^{\varepsilon_1}+1}^{\tilde{N}_C}\abs{B(\tilde{w}) - E(\tilde{w})}.
\end{align}
To upper bound the first summation, we will upper bound both $E(\tilde{w})$ and $B(\tilde{w})$ for all but the highest $\tilde{w}$ values. For $E(\tilde{w})$, notice that
\begin{equation}
    E(\tilde{w}) = \frac{\tilde{c}_1-\tilde{c}_0}{\tilde{c}_1^{\tilde{N}_C+1}-\tilde{c}_0^{\tilde{N}_C+1}}\tilde{c}_1^{\tilde{w}}\tilde{c}_0^{\tilde{N}_C-\tilde{w}} \le \left(\frac{\tilde{c}_0}{\tilde{c}_1}\right)^{\tilde{N}_C-\tilde{w}}.
\end{equation}
For $B(\tilde{w})$, we state the following highly technical lemma:

\begin{lemma}
For sufficiently small $\varepsilon_1 > 0$, and for any $0 < \varepsilon_2 < \varepsilon_1$,
\begin{equation}
    \sum_{\tilde{w}=0}^{\tilde{N}_C-\tilde{N}_C^{\varepsilon_1}}B(\tilde{w}) = \exp\left[-\Omega\left(N_C^{\varepsilon_2}\right)\right].
\end{equation}
\label{lem:B-low-w-tilde-upper-bound}
\end{lemma}

To avoid cluttering the main proof, we defer the proof of Lemma \ref{lem:B-low-w-tilde-upper-bound} to Appendix \ref{sec:schur-transformed-state-conversion}\ref{subsec:proofs-niche-lemmas-schur-transformed-state-conversion}.

\vspace{0.5\baselineskip}

In other words, we upper bounded $E(\tilde{w})$ by a geometric sequence with common ratio $\tilde{c}_0/\tilde{c}_1$, and we upper bounded $B(\tilde{w})$ by a more complicated but still extremely small quantity. Combining these two facts yields
\begin{align}
    & \quad\,\, \sum_{\tilde{w}=0}^{\tilde{N}_C-\tilde{N}_C^{\varepsilon_1}}\abs{B(\tilde{w}) - E(\tilde{w})} \\
    &\le \sum_{\tilde{w}=0}^{\tilde{N}_C-\tilde{N}_C^{\varepsilon_1}}B(\tilde{w}) + \sum_{\tilde{w}=0}^{\tilde{N}_C-\tilde{N}_C^{\varepsilon_1}}E(\tilde{w}) \\
    &= \exp\left[-\Omega(N_C^{\varepsilon_2})\right] + O\left(\left(\frac{\tilde{c}_0}{\tilde{c}_1}\right)^{\tilde{N}_C^{\varepsilon_1}}\right).
\end{align}
For any $0 < \varepsilon_2 < \varepsilon_1$, this actually decays faster than any power law decay, i.e., it is $o(N^{-q})$ for all $q > 0$ in the $N\to\infty$ limit. Intuitively, this corresponds to the fact that both the output state and the target state only have nonnegligible probabilities in the highest Hamming weights (i.e., $\tilde{w}\approx\tilde{N}_C$), with all lower Hamming weights having negligible probabilities.

\vspace{0.5\baselineskip}

To upper bound the second summation, we will prove that
\begin{equation}
    \tilde{N}_C - \tilde{w} < \tilde{N}_C^{\varepsilon_1} \implies \abs{B(\tilde{w}) - E(\tilde{w})} = O\left(N^{p+\varepsilon_1-1}\right).
\end{equation}
Just as we did in the proof of Lemma \ref{lem:discarding-map-schur-transformed-state}, we first define $\delta\coloneqq N_C-w$ and $\tilde{\delta}\coloneqq\tilde{N}_C-\tilde{w}$ for convenience. We additionally define the ``Schur-transformed conversion rate''
\begin{equation}
    R_C \coloneqq R\frac{\lambda_{\text{out}}}{\lambda_{\text{in}}} = \frac{\lambda_{\text{in}}}{1+\lambda_{\text{in}}} \Bigg/ \frac{\lambda_{\text{out}}}{1+\lambda_{\text{out}}}.
\end{equation}
We now define the following sequence of quantities:
\begin{align}
    B_0(\tilde{\delta}) &= \frac{c_1-c_0}{c_1^{N_C+1} - c_0^{N_C+1}}\binom{\tilde{N}_C+1}{N_C+1}^{-1}\sum_{\delta=0}^{N_C}c_1^{N_C-\delta}c_0^\delta\binom{\tilde{N}_C-\tilde{\delta}}{N_C-\delta}\binom{\tilde{\delta}}{\delta} \\
    B_1(\tilde{\delta}) &= \frac{c_1-c_0}{c_1}\binom{\tilde{N}_C+1}{N_C+1}^{-1}\sum_{\delta=0}^{N_C}\left(\frac{c_0}{c_1}\right)^\delta\binom{\tilde{N}_C-\tilde{\delta}}{N_C-\delta}\binom{\tilde{\delta}}{\delta} \\
    B_2(\tilde{\delta}) &= \frac{c_1-c_0}{c_1}\binom{\tilde{N}_C+1}{N_C+1}^{-1}\sum_{\delta=0}^{\tilde{\delta}}\left(\frac{c_0}{c_1}\right)^\delta\binom{\tilde{N}_C-\tilde{\delta}}{N_C-\delta}\binom{\tilde{\delta}}{\delta} \\
    B_3(\tilde{\delta}) &= \frac{c_1-c_0}{c_1}\sum_{\delta=0}^{\tilde{\delta}}\left(\frac{c_0}{c_1}\right)^\delta\frac{(R_C-1)^{\tilde{\delta}-\delta}}{R_C^{\tilde{\delta}+1}}\binom{\tilde{\delta}}{\delta} \\
    B_4(\tilde{\delta}) &= \frac{\tilde{c}_1-\tilde{c}_0}{\tilde{c}_1}\left(\frac{\tilde{c}_0}{\tilde{c}_1}\right)^{\tilde{\delta}} \\
    B_5(\tilde{\delta}) &= \frac{\tilde{c}_1-\tilde{c}_0}{\tilde{c}_1^{\tilde{N}_C+1}-\tilde{c}_0^{\tilde{N}_C+1}}\tilde{c}_1^{\tilde{N}_C-\tilde{\delta}}\tilde{c}_0^{\tilde{\delta}}.
\end{align}
Notice that $B_0(\tilde{\delta}) = B(\tilde{w})$ and $B_5(\tilde{\delta}) = E(\tilde{w})$. Therefore, we can set up a massive triangle inequality by bounding the successive differences:
\begin{equation}
    \abs{B(\tilde{w}) - E(\tilde{w})} \le \sum_{k=1}^{5}\abs{B_{k-1}(\tilde{\delta}) - B_k(\tilde{\delta})}.
\end{equation}
Each term corresponds to one approximation we will make. In particular, we will show that
\begin{align}
    \abs{B_0(\tilde{\delta}) - B_1(\tilde{\delta})} &\le \frac{(c_0/c_1)^{N_C+1}}{1-(c_0/c_1)^{N_C+1}} \\
    \abs{B_1(\tilde{\delta}) - B_2(\tilde{\delta})} &= 0 \\
    \abs{B_2(\tilde{\delta}) - B_3(\tilde{\delta})} &= O\left(N^{p+\varepsilon_1-1}\right) \\
    \abs{B_3(\tilde{\delta}) - B_4(\tilde{\delta})} &= 0 \\
    \abs{B_4(\tilde{\delta}) - B_5(\tilde{\delta})} &\le \frac{(\tilde{c}_0/\tilde{c}_1)^{\tilde{N}_C+1}}{1-(\tilde{c}_0/\tilde{c}_1)^{\tilde{N}_C+1}}.
\end{align}
Combining these approximations and applying the triangle inequality yields
\begin{equation}
    \abs{B(\tilde{w}) - E(\tilde{w})} = O\left(N^{p+\varepsilon_1-1}\right).
\end{equation}
Now all that remains is to prove each of the above approximations.
\begin{itemize}
    \item \textbf{First approximation:} Identical to the first approximation in the proof of Lemma \ref{lem:discarding-map-schur-transformed-state}.
    
    \item \textbf{Second ``approximation'':} The last binomial coefficient is zero unless $\tilde{\delta}\ge\delta$, so we can freely truncate the summation to end at $\delta=\tilde{\delta}$, rather than $\delta=N_C$. This is actually not an approximation at all, but we mention it here because it involves a summation truncation, similarly to this step in the proof of Lemma \ref{lem:discarding-map-schur-transformed-state}.
    
    \item \textbf{Third approximation:} Just as we did at this step in the proof of Lemma \ref{lem:discarding-map-schur-transformed-state}, we first observe the following approximation for the ratio of two factorials:
    \begin{equation}
        0 \le k \le \frac{n}{2} \implies n^ke^{-2k^2/n} \le \frac{n!}{(n-k)!} \le n^k \implies \frac{n!}{(n-k)!} = n^ke^{O(k^2/n)}.
    \end{equation}
    We can thus approximate the ratio of two binomial coefficients of interest as follows:
    \begin{align}
        \frac{\binom{\tilde{N}_C-\tilde{\delta}}{N_C-\delta}}{\binom{\tilde{N}_C+1}{N_C+1}} &= \frac{(\tilde{N}_C-\tilde{\delta})!}{(\tilde{N}_C+1)!}\frac{(N_C+1)!}{(N_C-\delta)!}\frac{(\tilde{N}_C-N_C)!}{((\tilde{N}_C-N_C)-(\tilde{\delta}-\delta))!} \\
        &= \left[\tilde{N}_C^{-(\tilde{\delta}+1)}e^{O(\tilde{\delta}^2/\tilde{N}_C)}\right]\left[N_C^{\delta+1}e^{O(\delta^2/N_C)}\right]\left[(\tilde{N}_C-N_C)^{\tilde{\delta}-\delta}e^{O\left((\tilde{\delta}-\delta)^2/(\tilde{N}_C-N_C)\right)}\right] \\
        &\stackrel{(1)}{=} \left(\frac{N_C}{\tilde{N}_C}\right)^{\delta+1}\left(\frac{\tilde{N}_C-N_C}{\tilde{N}_C}\right)^{\tilde{\delta}-\delta}\exp\left[O\left(N^{2\varepsilon_1-1}\right)\right].
    \end{align}
    For clarity, equality (1) is ensured by the fact that $\delta$ and $\tilde{\delta}$ are both $O(N^{\varepsilon_1})$. We additionally observe that the typicality of $N_C$ and $\tilde{N}_C$ permits the following approximations:
    \begin{align}
        \frac{N_C}{\tilde{N}_C} &= \frac{1}{R_C}\left[1 + O(N^{p-1})\right] = \frac{1}{R_C}\exp\left[O(N^{p-1})\right] \\
        \frac{\tilde{N}_C-N_C}{\tilde{N}_C} &= \frac{R_C-1}{R_C}\left[1 + O(N^{p-1})\right] = \frac{R_C-1}{R_C}\exp\left[O(N^{p-1})\right].
    \end{align}
    We thus proceed as follows:
    \begin{align}
        \frac{\binom{\tilde{N}_C-\tilde{\delta}}{N_C-\delta}}{\binom{\tilde{N}_C+1}{N_C+1}} &= \left(\frac{1}{R_C}\right)^{\delta+1}\exp\left[O\left(N^{p-1}\right)\right]^{\delta+1}\left(\frac{R_C-1}{R_C}\right)^{\tilde{\delta}-\delta}\exp\left[O\left(N^{p-1}\right)\right]^{\tilde{\delta}-\delta}\exp\left[O\left(N^{2\varepsilon_1-1}\right)\right] \\
        &= \frac{(R_C-1)^{\tilde{\delta}-\delta}}{R_C^{\tilde{\delta}+1}}\exp\left[O\left(N^{p+\varepsilon_1-1}\right)\right]\exp\left[O\left(N^{2\varepsilon_1-1}\right)\right] \\
        &\stackrel{(1)}{=} \frac{(R_C-1)^{\tilde{\delta}-\delta}}{R_C^{\tilde{\delta}+1}}\exp\left[O\left(N^{p+\varepsilon_1-1}\right)\right].
    \end{align}
    For clarity, equality (1) is ensured by the fact that $\varepsilon_1$ is arbitrarily small (in particular, recall that we chose $0 < \varepsilon_1 < p$), which is why the two $\exp\left[O\left(N^{p+\varepsilon_1-1}\right)\right]$ factors dominate over the $\exp\left[O\left(N^{2\varepsilon_1-1}\right)\right]$ factor. We can now upper bound the total error resulting from this approximation:
    \begin{align}
        \abs{B_2(\tilde{\delta}) - B_3(\tilde{\delta})} &= \frac{c_1-c_0}{c_1}\sum_{\delta=0}^{\tilde{\delta}}\left(\frac{c_0}{c_1}\right)^\delta\frac{(R_C-1)^{\tilde{\delta}-\delta}}{R_C^{\tilde{\delta}+1}}\binom{\tilde{\delta}}{\delta}\Big\{\exp\left[O\left(N^{p+\varepsilon_1-1}\right)\right]-1\Big\} \\
        &\stackrel{(1)}{=} O\left(N^{p+\varepsilon_1-1}\right)\frac{c_1-c_0}{c_1}\sum_{\delta=0}^{\tilde{\delta}}\left(\frac{c_0}{c_1}\right)^\delta\frac{(R_C-1)^{\tilde{\delta}-\delta}}{R_C^{\tilde{\delta}+1}}\binom{\tilde{\delta}}{\delta} \\
        &= O\left(N^{p+\varepsilon_1-1}\right)\frac{c_1-c_0}{c_1}\frac{1}{R_C}\sum_{\delta=0}^{\tilde{\delta}}\left(\frac{c_0}{c_1}\right)^\delta\left(\frac{R_C-1}{R_C}\right)^{\tilde{\delta}-\delta}\left(\frac{1}{R_C}\right)^{\delta}\binom{\tilde{\delta}}{\delta} \\
        &\stackrel{(2)}{\le} O\left(N^{p+\varepsilon_1-1}\right)\frac{c_1-c_0}{c_1}\frac{1}{R_C}\sum_{\delta=0}^{\tilde{\delta}}\left(\frac{c_0}{c_1}\right)^\delta \\
        &\stackrel{(3)}{\le} O\left(N^{p+\varepsilon_1-1}\right)\frac{1}{R_C} \\
        &\le O\left(N^{p+\varepsilon_1-1}\right).
    \end{align}
    For clarity, equality (1) comes from the fact that $\varepsilon_1$ is arbitrarily small (in particular, recall that we chose $0 < \varepsilon_1 < 1-p$), so that $N^{p+\varepsilon_1-1} = o(1)$, along with the fact that $e^x-1 = O(x)$ for sufficiently small $x$ (for example, $\abs{e^x-1}\le 2x$ for $\abs{x}\le\frac{1}{2}$); inequality (2) comes from the fact that $\binom{n}{k}p^k(1-p)^{n-k} \le 1$ for all $0\le p\le 1$; and inequality (3) comes from the fact that the finite geometric series is upper bounded by the corresponding infinite geometric series.
    
    \item \textbf{Fourth ``approximation'':} Recall that, by the binomial theorem,
    \begin{equation}
        \sum_{k=0}^{n}x^k\binom{n}{k} = (1+x)^n.
    \label{eq:binom-thm}
    \end{equation}
    We can apply the binomial theorem as follows:
    \begin{align}
        B_3(\tilde{\delta}) &= \frac{c_1-c_0}{c_1}\sum_{\delta=0}^{\tilde{\delta}}\left(\frac{c_0}{c_1}\right)^\delta\frac{(R_C-1)^{\tilde{\delta}-\delta}}{R_C^{\tilde{\delta}+1}}\binom{\tilde{\delta}}{\delta} \\
        &= \frac{c_1-c_0}{c_1}\frac{(R_C-1)^{\tilde{\delta}}}{R_C^{\tilde{\delta}+1}}\sum_{\delta=0}^{\tilde{\delta}}\left[\frac{c_0}{c_1(R_C-1)}\right]^\delta\binom{\tilde{\delta}}{\delta} \\
        &= \frac{c_1-c_0}{c_1}\frac{(R_C-1)^{\tilde{\delta}}}{R_C^{\tilde{\delta}+1}}\left[1 + \frac{c_0}{c_1(R_C-1)}\right]^{\tilde{\delta}} & (\text{by Eq. }\ref{eq:binom-thm}) \\
        &= \frac{c_1-c_0}{c_1R_C}\left[\frac{c_1R_C-(c_1-c_0)}{c_1R_C}\right]^{\tilde{\delta}}.
    \end{align}
    Plugging in the values for $c_1$, $c_0$, $R_C$ reveals that
    \begin{align}
        \frac{c_1-c_0}{c_1R_C} = \frac{\lambda_{\text{in}}}{\frac{1+\lambda_{\text{in}}}{2}\frac{\lambda_{\text{in}}}{1+\lambda_{\text{in}}}\frac{1+\lambda_{\text{out}}}{\lambda_{\text{out}}}} = \frac{2\lambda_{\text{out}}}{1+\lambda_{\text{out}}} &= \frac{\tilde{c}_1-\tilde{c}_0}{\tilde{c}_1} \\
        \frac{c_1R_C-(c_1-c_0)}{c_1R_C} = 1 - \frac{c_1-c_0}{c_1R_C} &= \frac{\tilde{c}_0}{\tilde{c}_1}.
    \end{align}
    We thus obtain precisely $B_4(\tilde{\delta})$. This is actually not an approximation at all, but we mention it here because it involves an identity about binomial coefficients, similarly to this step in the proof of Lemma \ref{lem:discarding-map-schur-transformed-state}.
    
    \item \textbf{Fifth approximation:} Identical to the fifth approximation in the proof of Lemma \ref{lem:discarding-map-schur-transformed-state}.
\end{itemize}

Now that we have upper bounded each individual $\abs{B(\tilde{w}) - E(\tilde{w})}$, we take the summation over all $\tilde{w}\ge\tilde{N}_C-\tilde{N}_C^{\varepsilon_1}+1$ to obtain
\begin{equation}
    \sum_{\tilde{w}=\tilde{N}_C-\tilde{N}_C^{\varepsilon_1}+1}^{\tilde{N}_C}\abs{B(\tilde{w}) - E(\tilde{w})} = \sum_{\tilde{w}=\tilde{N}_C-\tilde{N}_C^{\varepsilon_1}+1}^{\tilde{N}_C}O\left(N^{p+\varepsilon_1-1}\right) = O\left(N^{p+2\varepsilon_1-1}\right).
\end{equation}
The very last step is to combine the summation over $\tilde{w}\le\tilde{N}_C-\tilde{N}_C^{\varepsilon_1}$ and the summation over $\tilde{w}\ge\tilde{N}_C-\tilde{N}_C^{\varepsilon_1}+1$:
\begin{equation}
    \sum_{\tilde{w}=0}^{\tilde{N}_C}\abs{B(\tilde{w}) - E(\tilde{w})} = O\left(N^{p+2\varepsilon_1-1}\right).
\end{equation}
Setting $\varepsilon = 2\varepsilon_1$ tells us that the trace distance is $O\left(N^{p+\varepsilon-1}\right)$, exactly as desired.
\end{proof}

\appsubsec{Measure-and-Prepare Conversion for Schur-Transformed States}
{subsec:schur-transformed-state-mp-conversion-unified-presentation}

The optimal measure-and-prepare channel achieves measure-and-prepare conversion for Schur-transformed states with vanishing trace distance in the limit of many qubits. To formalize this fact, we state and prove the following lemma:

\begin{lemma}[optimal measure-and-prepare channel transforms Schur-transformed state into approximate Schur-transformed state]
Suppose $N_C = \lambda_{\text{in}}N + O(N^p)$ and $\tilde{N}_C = \lambda_{\text{out}}\tilde{N} + O(N^p)$ for some $0 < p < 1$, where $\tilde{N} = \lfloor RN\rfloor$ and $R = R^{\text{MP}}(\lambda_{\text{in}}\to\lambda_{\text{out}})$. Then
\begin{equation}
    d_{\text{Tr}}\left(\mE_{\text{MP}}[N_C\to\tilde{N}_C]\left(\rho_C(N_C,\lambda_{\text{in}},\hat{n})\right), \rho_C(\tilde{N}_C,\lambda_{\text{out}},\hat{n})\right) = O\left(N^{p+\varepsilon-1}\right)
\end{equation}
for any $\varepsilon > 0$.
\label{lem:optimal-mp-channel-schur-transformed-state}
\end{lemma}

\begin{proof}
We first apply Lemma \ref{lem:optimal-mp-channel-dicke-state-unified-presentation} with $N\mapsto N_C$ and $M\mapsto\tilde{N}_C$ to compute the action of the optimal measure-and-prepare channel on the whole Schur-transformed state:
\begin{align}
    & \quad\,\, \mE_{\text{MP}}[N_C\to\tilde{N}_C]\left(\rho_C(N_C,\lambda_{\text{in}})\right) \\
    &= \frac{c_1-c_0}{c_1^{N_C+1}-c_0^{N_C+1}}\sum_{w=0}^{N_C}c_1^wc_0^{N_C-w}\mE_{\text{MP}}[N_C\to\tilde{N}_C]\left(\ket{D^{(N_C)}_w}\bra{D^{(N_C)}_w}\right) \\
    &= \sum_{\tilde{w}=0}^{\tilde{N}_C}C(\tilde{w})\ket{D^{(\tilde{N}_C)}_{\tilde{w}}}\bra{D^{(\tilde{N}_C)}_{\tilde{w}}},
\end{align}
where the coefficient $C(\tilde{w})$ takes the form
\begin{equation}
    C(\tilde{w}) = \frac{c_1-c_0}{c_1^{N_C+1} - c_0^{N_C+1}}\binom{N_C+\tilde{N}_C+1}{N_C+1}^{-1}\sum_{w=0}^{N_C}c_1^wc_0^{N_C-w}\binom{w+\tilde{w}}{w}\binom{(N_C-w)+(\tilde{N}_C-\tilde{w})}{N_C-w},
\end{equation}
and where we have defined $c_1\coloneqq\frac{1+\lambda_{\text{in}}}{2}$ and $c_0\coloneqq\frac{1-\lambda_{\text{in}}}{2}$ for convenience.

\vspace{0.5\baselineskip}

The target state is the Schur-transformed state $\rho_C(\tilde{N}_C,\lambda_{\text{out}})$:
\begin{equation}
    \rho_C(\tilde{N}_C,\lambda_{\text{out}}) = \sum_{\tilde{w}=0}^{\tilde{N}_C}E(\tilde{w})\ket{D^{(\tilde{N}_C)}_{\tilde{w}}}\bra{D^{(\tilde{N}_C)}_{\tilde{w}}},
\end{equation}
where the coefficient $E(\tilde{w})$ takes the form
\begin{equation}
    E(\tilde{w}) = \frac{\tilde{c}_1-\tilde{c}_0}{\tilde{c}_1^{\tilde{N}_C+1}-\tilde{c}_0^{\tilde{N}_C+1}}\tilde{c}_1^{\tilde{w}}\tilde{c}_0^{\tilde{N}_C-\tilde{w}},
\end{equation}
and where we have defined $\tilde{c}_1\coloneqq\frac{1+\lambda_{\text{out}}}{2}$ and $\tilde{c}_0\coloneqq\frac{1-\lambda_{\text{out}}}{2}$ for convenience.

\vspace{0.5\baselineskip}

Since the output state and target state are both diagonal in the basis of Dicke states, we can write their trace distance as follows:
\begin{equation}
    d_{\text{Tr}}\left(\mE_{\text{MP}}[N_C\to\tilde{N}_C]\left(\rho_C(N_C,\lambda_{\text{in}})\right), \rho_C(\tilde{N}_C,\lambda_{\text{out}})\right) = \frac{1}{2}\sum_{\tilde{w}=0}^{\tilde{N}_C}\abs{C(\tilde{w}) - E(\tilde{w})}.
\end{equation}
To show that this trace distance vanishes as $N\to\infty$, we choose an arbitrarily small $\varepsilon_1 > 0$. As we will see later, we will especially take advantage of the fact that $\varepsilon_1 < \min\{p,1-p\}$. We then split the summation as follows:
\begin{align}
    \sum_{\tilde{w}=0}^{\tilde{N}_C}\abs{C(\tilde{w}) - E(\tilde{w})} &= \sum_{\tilde{w}=0}^{\tilde{N}_C-\tilde{N}_C^{\varepsilon_1}}\abs{C(\tilde{w}) - E(\tilde{w})} + \sum_{\tilde{w}=\tilde{N}_C-\tilde{N}_C^{\varepsilon_1}+1}^{\tilde{N}_C}\abs{C(\tilde{w}) - E(\tilde{w})}.
\end{align}
To upper bound the first summation, we will upper bound both $E(\tilde{w})$ and $C(\tilde{w})$ for all but the highest $\tilde{w}$ values. For $E(\tilde{w})$, notice that
\begin{equation}
    E(\tilde{w}) = \frac{\tilde{c}_1-\tilde{c}_0}{\tilde{c}_1^{\tilde{N}_C+1}-\tilde{c}_0^{\tilde{N}_C+1}}\tilde{c}_1^{\tilde{w}}\tilde{c}_0^{\tilde{N}_C-\tilde{w}} \le \left(\frac{\tilde{c}_0}{\tilde{c}_1}\right)^{\tilde{N}_C-\tilde{w}}.
\end{equation}
For $C(\tilde{w})$, we state the following highly technical lemma:

\begin{lemma}
For sufficiently small $\varepsilon_1 > 0$, and for any $0 < \varepsilon_2 < \varepsilon_1$,
\begin{equation}
    \sum_{\tilde{w}=0}^{\tilde{N}_C-\tilde{N}_C^{\varepsilon_1}}C(\tilde{w}) = \exp\left[-\Omega\left(N_C^{\varepsilon_2}\right)\right].
\end{equation}
\label{lem:C-low-w-tilde-upper-bound}
\end{lemma}

To avoid cluttering the main proof, we defer the proof of Lemma \ref{lem:C-low-w-tilde-upper-bound} to Appendix \ref{sec:schur-transformed-state-conversion}\ref{subsec:proofs-niche-lemmas-schur-transformed-state-conversion}.

\vspace{0.5\baselineskip}

In other words, we upper bounded $E(\tilde{w})$ by a geometric sequence with common ratio $\tilde{c}_0/\tilde{c}_1$, and we upper bounded $C(\tilde{w})$ by a more complicated but still extremely small quantity. Combining these two facts yields
\begin{align}
    & \quad\,\, \sum_{\tilde{w}=0}^{\tilde{N}_C-\tilde{N}_C^{\varepsilon_1}}\abs{C(\tilde{w}) - E(\tilde{w})} \\
    &\le \sum_{\tilde{w}=0}^{\tilde{N}_C-\tilde{N}_C^{\varepsilon_1}}C(\tilde{w}) + \sum_{\tilde{w}=0}^{\tilde{N}_C-\tilde{N}_C^{\varepsilon_1}}E(\tilde{w}) \\
    &= \exp\left[-\Omega(N_C^{\varepsilon_2})\right] + O\left(\left(\frac{\tilde{c}_0}{\tilde{c}_1}\right)^{\tilde{N}_C^{\varepsilon_1}}\right).
\end{align}
For any $0 < \varepsilon_2 < \varepsilon_1$, this actually decays faster than any power law decay, i.e., it is $o(N^{-q})$ for all $q > 0$ in the $N\to\infty$ limit. Intuitively, this corresponds to the fact that both the output state and the target state only have nonnegligible probabilities in the highest Hamming weights (i.e., $\tilde{w}\approx\tilde{N}_C$), with all lower Hamming weights having negligible probabilities.

\vspace{0.5\baselineskip}

To upper bound the second summation, we will prove that
\begin{equation}
    \tilde{N}_C - \tilde{w} < \tilde{N}_C^{\varepsilon_1} \implies \abs{C(\tilde{w}) - E(\tilde{w})} = O\left(N^{p+\varepsilon_1-1}\right).
\end{equation}
As usual, we first define $\delta\coloneqq N_C-w$ and $\tilde{\delta}\coloneqq\tilde{N}_C-\tilde{w}$ for convenience. Also as usual, we additionally define the ``Schur-transformed conversion rate''
\begin{equation}
    R_C \coloneqq R\frac{\lambda_{\text{out}}}{\lambda_{\text{in}}} = \frac{\lambda_{\text{in}}}{1+\lambda_{\text{in}}} \Bigg/ \frac{\lambda_{\text{out}}}{1-\lambda_{\text{out}}}.
\end{equation}
We now define the following sequence of quantities:
\begin{align}
    C_0(\tilde{\delta}) &= \frac{c_1-c_0}{c_1^{N_C+1} - c_0^{N_C+1}}\binom{N_C+\tilde{N}_C+1}{N_C+1}^{-1}\sum_{\delta=0}^{N_C}c_1^{N_C-\delta}c_0^\delta\binom{(N_C-\delta)+(\tilde{N}_C-\tilde{\delta})}{N_C-\delta}\binom{\delta+\tilde{\delta}}{\delta} \\
    C_1(\tilde{\delta}) &= \frac{c_1-c_0}{c_1}\binom{N_C+\tilde{N}_C+1}{N_C+1}^{-1}\sum_{\delta=0}^{N_C}\left(\frac{c_0}{c_1}\right)^\delta\binom{(N_C-\delta)+(\tilde{N}_C-\tilde{\delta})}{N_C-\delta}\binom{\delta+\tilde{\delta}}{\delta} \\
    C_2(\tilde{\delta}) &= \frac{c_1-c_0}{c_1}\binom{N_C+\tilde{N}_C+1}{N_C+1}^{-1}\sum_{\delta=0}^{N_C^{\varepsilon_1}-1}\left(\frac{c_0}{c_1}\right)^\delta\binom{(N_C-\delta)+(\tilde{N}_C-\tilde{\delta})}{N_C-\delta}\binom{\delta+\tilde{\delta}}{\delta} \\
    C_3(\tilde{\delta}) &= \frac{c_1-c_0}{c_1}\sum_{\delta=0}^{N_C^{\varepsilon_1}-1}\left(\frac{c_0}{c_1}\right)^\delta\frac{R_C^{\tilde{\delta}}}{(R_C+1)^{\delta+\tilde{\delta}+1}}\binom{\delta+\tilde{\delta}}{\delta} \\
    C_4(\tilde{\delta}) &= \frac{\tilde{c}_1-\tilde{c}_0}{\tilde{c}_1}\left(\frac{\tilde{c}_0}{\tilde{c}_1}\right)^{\tilde{\delta}} \\
    C_5(\tilde{\delta}) &= \frac{\tilde{c}_1-\tilde{c}_0}{\tilde{c}_1^{\tilde{N}_C+1}-\tilde{c}_0^{\tilde{N}_C+1}}\tilde{c}_1^{\tilde{N}_C-\tilde{\delta}}\tilde{c}_0^{\tilde{\delta}}.
\end{align}
Notice that $C_0(\tilde{\delta}) = C(\tilde{w})$ and $C_5(\tilde{\delta}) = E(\tilde{w})$. Therefore, we can set up a massive triangle inequality by bounding the successive differences:
\begin{equation}
    \abs{C(\tilde{w}) - E(\tilde{w})} \le \sum_{k=1}^{5}\abs{C_{k-1}(\tilde{\delta}) - C_k(\tilde{\delta})}.
\end{equation}
Each term corresponds to one approximation we will make. In particular, we will show that
\begin{align}
    \abs{C_0(\tilde{\delta}) - C_1(\tilde{\delta})} &\le \frac{(c_0/c_1)^{N_C+1}}{1-(c_0/c_1)^{N_C+1}} \\
    \abs{C_1(\tilde{\delta}) - C_2(\tilde{\delta})} &\le \left(\frac{c_0}{c_1}\right)^{N_C^{\varepsilon_1}} \\
    \abs{C_2(\tilde{\delta}) - C_3(\tilde{\delta})} &= O\left(N^{p+\varepsilon_1-1}\right) \\
    \abs{C_3(\tilde{\delta}) - C_4(\tilde{\delta})} &\le \left(\frac{c_0}{c_1}\right)^{N_C^{\varepsilon_1}} \\
    \abs{C_4(\tilde{\delta}) - C_5(\tilde{\delta})} &\le \frac{(\tilde{c}_0/\tilde{c}_1)^{\tilde{N}_C+1}}{1-(\tilde{c}_0/\tilde{c}_1)^{\tilde{N}_C+1}}.
\end{align}
Combining these approximations and applying the triangle inequality yields
\begin{equation}
    \abs{C(\tilde{w}) - E(\tilde{w})} = O\left(N^{p+\varepsilon_1-1}\right).
\end{equation}
Now all that remains is to prove each of the above approximations.
\begin{itemize}
    \item \textbf{First approximation:} Identical to the first approximation in the proof of Lemma \ref{lem:discarding-map-schur-transformed-state}.

    \item \textbf{Second approximation:} Truncate the sum early at $\delta=N_C^{\varepsilon_1}-1$. Therefore, the error can be bounded via
    \begin{align}
        0 \le C_1(\tilde{\delta}) - C_2(\tilde{\delta}) &= \frac{c_1-c_0}{c_1}\binom{N_C+\tilde{N}_C+1}{N_C+1}^{-1}\sum_{\delta=N_C^{\varepsilon_1}}^{N_C}\left(\frac{c_0}{c_1}\right)^\delta\binom{(N_C-\delta)+(\tilde{N}_C-\tilde{\delta})}{N_C-\delta}\binom{\delta+\tilde{\delta}}{\delta} \\
        &\le \frac{c_1-c_0}{c_1}\binom{N_C+\tilde{N}_C}{N_C}^{-1}\sum_{\delta=N_C^{\varepsilon_1}}^{N_C}\left(\frac{c_0}{c_1}\right)^\delta\binom{(N_C-\delta)+(\tilde{N}_C-\tilde{\delta})}{N_C-\delta}\binom{\delta+\tilde{\delta}}{\delta} \\
        &\stackrel{(1)}{\le} \frac{c_1-c_0}{c_1}\sum_{\delta=N_C^{\varepsilon_1}}^{N_C}\left(\frac{c_0}{c_1}\right)^\delta \\
        &\le \frac{c_1-c_0}{c_1}\sum_{\delta=N_C^{\varepsilon_1}}^{\infty}\left(\frac{c_0}{c_1}\right)^\delta \\
        &= \left(\frac{c_0}{c_1}\right)^{N_C^{\varepsilon_1}}.
    \end{align}
    For clarity, inequality (1) comes from the fact that $\binom{n_1+n_2}{k_1+k_2} \ge \binom{n_1}{k_1}\binom{n_2}{k_2}$.

    \item \textbf{Third approximation:} Just as we did at this step in the proof of Lemma \ref{lem:discarding-map-schur-transformed-state}, we first observe the following approximation for the ratio of two factorials:
    \begin{equation}
        0 \le k \le \frac{n}{2} \implies n^ke^{-2k^2/n} \le \frac{n!}{(n-k)!} \le n^k \implies \frac{n!}{(n-k)!} = n^ke^{O(k^2/n)}.
    \end{equation}
    We can thus approximate the ratio of two binomial coefficients of interest as follows:
    \begin{align}
        \frac{\binom{(N_C-\delta)+(\tilde{N}_C-\tilde{\delta})}{N_C-\delta}}{\binom{N_C+\tilde{N}_C+1}{N_C+1}} &= \frac{(N_C+1)!}{(N_C-\delta)!}\frac{\tilde{N}_C!}{(\tilde{N}_C-\tilde{\delta})!}\frac{((N_C+\tilde{N}_C)-(\delta+\tilde{\delta}))!}{(N_C+\tilde{N}_C+1)!} \\
        &= \left[N_C^{\delta+1}e^{O(\delta^2/N_C)}\right]\left[\tilde{N}_C^{\tilde{\delta}}e^{O(\tilde{\delta}^2/\tilde{N}_C)}\right]\left[(N_C+\tilde{N}_C)^{-(\delta+\tilde{\delta}+1)}e^{O\left((\delta+\tilde{\delta})^2/(N_C+\tilde{N}_C)\right)}\right] \\
        &\stackrel{(1)}{=} \left(\frac{N_C}{N_C+\tilde{N}_C}\right)^{\delta+1}\left(\frac{\tilde{N}_C}{N_C+\tilde{N}_C}\right)^{\tilde{\delta}}\exp\left[O\left(N^{2\varepsilon_1-1}\right)\right].
    \end{align}
    For clarity, equality (1) is ensured by the fact that $\delta$ and $\tilde{\delta}$ are both $O(N^{\varepsilon_1})$. We additionally observe that the typicality of $N_C$ and $\tilde{N}_C$ permits the following approximations:
    \begin{align}
        \frac{N_C}{N_C+\tilde{N}_C} &= \frac{1}{R_C+1}\left[1 + O(N^{p-1})\right] = \frac{1}{R_C+1}\exp\left[O(N^{p-1})\right] \\
        \frac{\tilde{N}_C}{N_C+\tilde{N}_C} &= \frac{R_C}{R_C+1}\left[1 + O(N^{p-1})\right] = \frac{R_C}{R_C+1}\exp\left[O(N^{p-1})\right].
    \end{align}
    We thus proceed as follows:
    \begin{align}
        \frac{\binom{(N_C-\delta)+(\tilde{N}_C-\tilde{\delta})}{N_C-\delta}}{\binom{N_C+\tilde{N}_C+1}{N_C+1}} &= \left(\frac{1}{R_C+1}\right)^{\delta+1}\exp\left[O\left(N^{p-1}\right)\right]^{\delta+1}\left(\frac{R_C}{R_C+1}\right)^{\tilde{\delta}}\exp\left[O\left(N^{p-1}\right)\right]^{\tilde{\delta}}\exp\left[O\left(N^{2\varepsilon_1-1}\right)\right] \\
        &= \frac{R_C^{\tilde{\delta}}}{(R_C+1)^{\delta+\tilde{\delta}+1}}\exp\left[O\left(N^{p+\varepsilon_1-1}\right)\right]\exp\left[O\left(N^{2\varepsilon_1-1}\right)\right] \\ \\
        &\stackrel{(1)}{=} \frac{R_C^{\tilde{\delta}}}{(R_C+1)^{\delta+\tilde{\delta}+1}}\exp\left[O\left(N^{p+\varepsilon_1-1}\right)\right].
    \end{align}
    For clarity, equality (1) is ensured by the fact that $\varepsilon_1$ is arbitrarily small (in particular, recall that we chose $0 < \varepsilon_1 < p$), which is why the two $\exp\left[O\left(N^{p+\varepsilon_1-1}\right)\right]$ factors dominate over the $\exp\left[O\left(N^{2\varepsilon_1-1}\right)\right]$ factor. We can now upper bound the total error resulting from this approximation:
    \begin{align}
        \abs{C_2(\tilde{\delta}) - C_3(\tilde{\delta})} &= \frac{c_1-c_0}{c_1}\sum_{\delta=0}^{N_C^{\varepsilon_1}-1}\left(\frac{c_0}{c_1}\right)^\delta\frac{R_C^{\tilde{\delta}}}{(R_C+1)^{\delta+\tilde{\delta}+1}}\binom{\delta+\tilde{\delta}}{\delta}\Big\{\exp\left[O\left(N^{p+\varepsilon_1-1}\right)\right]-1\Big\} \\
        &\stackrel{(1)}{=} O\left(N^{p+\varepsilon_1-1}\right)\frac{c_1-c_0}{c_1}\sum_{\delta=0}^{N_C^{\varepsilon_1}-1}\left(\frac{c_0}{c_1}\right)^\delta\frac{R_C^{\tilde{\delta}}}{(R_C+1)^{\delta+\tilde{\delta}+1}}\binom{\delta+\tilde{\delta}}{\delta} \\
        &= O\left(N^{p+\varepsilon_1-1}\right)\frac{c_1-c_0}{c_1}\frac{1}{R_C+1}\sum_{\delta=0}^{N_C^{\varepsilon_1}-1}\left(\frac{c_0}{c_1}\right)^\delta\left(\frac{1}{R_C+1}\right)^\delta\left(\frac{R_C}{R_C+1}\right)^{\tilde{\delta}}\binom{\delta+\tilde{\delta}}{\delta} \\
        &\stackrel{(2)}{\le} O\left(N^{p+\varepsilon_1-1}\right)\frac{c_1-c_0}{c_1}\frac{1}{R_C+1}\sum_{\delta=0}^{N_C^{\varepsilon_1}-1}\left(\frac{c_0}{c_1}\right)^\delta \\
        &\stackrel{(3)}{\le} O\left(N^{p+\varepsilon_1-1}\right)\frac{1}{R_C+1} \\
        &\le O\left(N^{p+\varepsilon_1-1}\right).
    \end{align}
    For clarity, equality (1) comes from the fact that $\varepsilon_1$ is arbitrarily small (in particular, recall that we chose $0 < \varepsilon_1 < 1-p$), so that $N^{p+\varepsilon_1-1} = o(1)$, along with the fact that $e^x-1 = O(x)$ for sufficiently small $x$ (for example, $\abs{e^x-1}\le 2x$ for $\abs{x}\le\frac{1}{2}$); inequality (2) comes from the fact that $\binom{n}{k}p^k(1-p)^{n-k} \le 1$ for all $0\le p\le 1$; and inequality (3) comes from the fact that the finite geometric series is upper bounded by the corresponding infinite geometric series.
    
    \item \textbf{Fourth approximation:} Extend the summation to $\delta=\infty$, instead of terminating it at $\delta=N_C^{\varepsilon_1}-1$. Once again invoking the identity
    \begin{equation}
        \sum_{n=0}^{\infty}x^n\binom{n+p}{n} = (1-x)^{-(p+1)},
    \label{eq:geom-series-binom-coeff-identity-v2}
    \end{equation}
    we can conclude that the infinite sum yields
    \begin{align}
        & \quad\,\, \frac{c_1-c_0}{c_1}\sum_{\delta=0}^{\infty}\left(\frac{c_0}{c_1}\right)^\delta\frac{R_C^{\tilde{\delta}}}{(R_C+1)^{\delta+\tilde{\delta}+1}}\binom{\delta+\tilde{\delta}}{\delta} \\
        &= \frac{c_1-c_0}{c_1(R_C+1)}\left(\frac{R_C}{R_C+1}\right)^{\tilde{\delta}}\sum_{\delta=0}^{\infty}\left[\frac{c_0}{c_1(R_C+1)}\right]^\delta\binom{\delta+\tilde{\delta}}{\delta} \\
        &= \frac{c_1-c_0}{c_1(R_C+1)}\left(\frac{R_C}{R_C+1}\right)^{\tilde{\delta}}\left[1 - \frac{c_0}{c_1(R_C+1)}\right]^{-(\tilde{\delta}+1)} & (\text{by Eq. }\ref{eq:geom-series-binom-coeff-identity-v2}) \\
        &= \frac{c_1-c_0}{c_1(R_C+1)}\frac{c_1(R_C+1)}{c_1(R_C+1)-c_0}\left[\frac{c_1R_C}{c_1(R_C+1)-c_0}\right]^{\tilde{\delta}} \\
        &= \frac{c_1-c_0}{c_1R_C+(c_1-c_0)}\left[\frac{c_1R_C}{c_1R_C+(c_1-c_0)}\right]^{\tilde{\delta}}.
    \end{align}
    Plugging in the values for $c_1$, $c_0$, $R_C$ reveals that
    \begin{align}
        \frac{c_1-c_0}{c_1R_C+(c_1-c_0)} &= \frac{\lambda_{\text{in}}}{\frac{1+\lambda_{\text{in}}}{2}\frac{\lambda_{\text{in}}}{1+\lambda_{\text{in}}}\frac{1-\lambda_{\text{out}}}{\lambda_{\text{out}}} + \lambda_{\text{in}}} = \frac{1}{\frac{1-\lambda_{\text{out}}}{2\lambda_{\text{out}}} + 1} = \frac{2\lambda_{\text{out}}}{1+\lambda_{\text{out}}} = \frac{\tilde{c}_1-\tilde{c}_0}{\tilde{c}_1} \\
        \frac{c_1R_C}{c_1R_C+(c_1-c_0)} &= 1 - \frac{c_1-c_0}{c_1R_C+(c_1-c_0)} = \frac{\tilde{c}_0}{\tilde{c}_1}.
    \end{align}
    We thus obtain precisely $C_4(\tilde{\delta})$. The error in this approximation is the sum of the terms introduced when extending the summation to infinity:
    \begin{align}
        0 \le C_4(\tilde{\delta}) - C_3(\tilde{\delta}) &= \frac{c_1-c_0}{c_1}\sum_{\delta=N_C^{\varepsilon_1}}^{\infty}\left(\frac{c_0}{c_1}\right)^\delta\frac{R_C^{\tilde{\delta}}}{(R_C+1)^{\delta+\tilde{\delta}+1}}\binom{\delta+\tilde{\delta}}{\delta} \\
        &= \frac{c_1-c_0}{c_1}\frac{1}{R_C+1}\sum_{\delta=N_C^{\varepsilon_1}}^{\infty}\left(\frac{c_0}{c_1}\right)^\delta\left(\frac{1}{R_C+1}\right)^\delta\left(\frac{R_C}{R_C+1}\right)^{\tilde{\delta}}\binom{\delta+\tilde{\delta}}{\delta} \\
        &\stackrel{(1)}{\le} \frac{c_1-c_0}{c_1}\frac{1}{R_C+1}\sum_{\delta=N_C^{\varepsilon_1}}^{\infty}\left(\frac{c_0}{c_1}\right)^\delta \\
        &= \frac{1}{R_C+1}\left(\frac{c_0}{c_1}\right)^{N_C^{\varepsilon_1}} \\
        &\le \left(\frac{c_0}{c_1}\right)^{N_C^{\varepsilon_1}}.
    \end{align}
    For clarity, inequality (1) comes from the fact that $\binom{n}{k}p^k(1-p)^{n-k} \le 1$ for all $0\le p\le 1$.
    
    \item \textbf{Fifth approximation:} Identical to the fifth approximation in the proof of Lemma \ref{lem:discarding-map-schur-transformed-state}.
\end{itemize}

Now that we have upper bounded each individual $\abs{C(\tilde{w}) - E(\tilde{w})}$, we take the summation over all $\tilde{w}\ge\tilde{N}_C-\tilde{N}_C^{\varepsilon_1}+1$ to obtain
\begin{equation}
    \sum_{\tilde{w}=\tilde{N}_C-\tilde{N}_C^{\varepsilon_1}+1}^{\tilde{N}_C}\abs{C(\tilde{w}) - E(\tilde{w})} = \sum_{\tilde{w}=\tilde{N}_C-\tilde{N}_C^{\varepsilon_1}+1}^{\tilde{N}_C}O\left(N^{p+\varepsilon_1-1}\right) = O\left(N^{p+2\varepsilon_1-1}\right).
\end{equation}
The very last step is to combine the summation over $\tilde{w}\le\tilde{N}_C-\tilde{N}_C^{\varepsilon_1}$ and the summation over $\tilde{w}\ge\tilde{N}_C-\tilde{N}_C^{\varepsilon_1}+1$:
\begin{equation}
    \sum_{\tilde{w}=0}^{\tilde{N}_C}\abs{C(\tilde{w}) - E(\tilde{w})} = O\left(N^{p+2\varepsilon_1-1}\right).
\end{equation}
Setting $\varepsilon = 2\varepsilon_1$ tells us that the trace distance is $O\left(N^{p+\varepsilon-1}\right)$, exactly as desired.
\end{proof}

\appsubsec{Wrong-Way Conversion for Schur-Transformed States}
{subsec:schur-transformed-state-ww-conversion-unified-presentation}

The optimal wrong-way measure-and-prepare channel achieves wrong-way conversion for Schur-transformed states with vanishing trace distance in the limit of many qubits. To formalize this fact, we state and prove the following lemma:

\begin{lemma}[optimal wrong-way measure-and-prepare channel transforms Schur-transformed state into approximate Schur-transformed state]
Suppose $N_C = \lambda_{\text{in}}N + O(N^p)$ and $\tilde{N}_C = \lambda_{\text{out}}\tilde{N} + O(N^p)$ for some $0 < p < 1$, where $\tilde{N} = \lfloor RN\rfloor$ and $R = R^{\text{WW}}(\lambda_{\text{in}}\to\lambda_{\text{out}})$. Then
\begin{equation}
    d_{\text{Tr}}\left(\mE_{\text{WW}}[N_C\to\tilde{N}_C]\left(\rho_C(N_C,\lambda_{\text{in}},\hat{n})\right), \rho_C(\tilde{N}_C,\lambda_{\text{out}},-\hat{n})\right) = O\left(N^{p+\varepsilon-1}\right)
\end{equation}
for any $\varepsilon > 0$.
\label{lem:optimal-ww-channel-schur-transformed-state}
\end{lemma}

\begin{proof}
The proof is identical to the proof of Lemma \ref{lem:optimal-mp-channel-schur-transformed-state} for the optimal measure-and-prepare channel. Since the output of each individual Dicke state has the Hamming weights inverted (i.e., $\tilde{w}\mapsto\tilde{N}_C-\tilde{w}$), as described in Lemma \ref{lem:optimal-ww-channel-dicke-state-unified-presentation}, the same is true for the output of the entire Schur-transformed state.
\end{proof}

\appsubsec{Proofs of Niche Lemmas}
{subsec:proofs-niche-lemmas-schur-transformed-state-conversion}

To prove that our various channels of interest, when applied to a Schur-transformed state, yield approximately a new Schur-transformed state, we introduced two niche technical lemmas:
\begin{itemize}
    \item To prove Lemma \ref{lem:optimal-cloning-map-schur-transformed-state}, we had to introduce Lemma \ref{lem:B-low-w-tilde-upper-bound}, which upper bounds the sum of $B(\tilde{w})$ across all but the highest values of $\tilde{w}$.
    \item To prove Lemma \ref{lem:optimal-mp-channel-schur-transformed-state}, we had to introduce Lemma \ref{lem:C-low-w-tilde-upper-bound}, which upper bounds the sum of $C(\tilde{w})$ across all but the highest values of $\tilde{w}$.
\end{itemize}
We now present their proofs here, so as not to excessively clutter the main proofs in earlier subsections. Both proofs are essentially identical and follow these key steps:
\begin{itemize}
    \item Relate the coefficients of the output state to various negative hypergeometric distributions, as mentioned in the remarks at the ends of Appendices \ref{sec:four-important-channels}\ref{subsec:optimal-cloning-map-unified-presentation} and \ref{sec:four-important-channels}\ref{subsec:optimal-mp-channel-unified-presentation}.
    \item Use a standard relation between the negative hypergeometric distribution and the hypergeometric distribution.
    \item Use a well-known tail bound on the hypergeometric distribution.
\end{itemize}

We first present a useful relation between a negative hypergeometric distribution and a hypergeometric distribution:

\begin{lemma}[relation between negative hypergeometric distribution and hypergeometric distribution]
If $X\sim\text{NHG}(N,K,r)$ and $Y\sim\text{HG}(N,N-K,k+r)$, then
\begin{equation}
    \Pbb[X\le k] = \Pbb[Y\ge r].
\end{equation}
Equivalently, we can write
\begin{equation}
    \Pbb[X\ge k+1] = \Pbb[Y\le r-1].
\end{equation}
\label{lem:NHG-HG-relation}
\end{lemma}

This fact is actually immediate from the standard interpretations of these distributions. Suppose you have an urn with $N$ balls, such that $K$ are red and the remaining $N-K$ are black. Now suppose you draw exactly $k+r$ of them. Then $\Pbb[X\le k]$ is the probability that at most $k$ of them are red, and $\Pbb[Y\ge r]$ is the probability that at least $r$ of them are black. These two events are clearly equivalent.

\vspace{0.5\baselineskip}

We now present a useful tail bound on the hypergeometric distribution:

\begin{lemma}[tail bounds for hypergeometric distribution \cite{Chvatal1979}]
If $X\sim\text{HG}(N,K,n)$ and $p = K/N$, then for any $0 < t < p$,
\begin{equation}
    \Pbb[X\le(p-t)n] \le e^{-t^2n}, \quad \Pbb[X\ge(p+t)n] \le e^{-t^2n}.
\end{equation}
\label{lem:hypergeometric-tail-bounds}
\end{lemma}

We are now ready to prove Lemmas \ref{lem:B-low-w-tilde-upper-bound} and \ref{lem:C-low-w-tilde-upper-bound}:

\begin{proof}[Proof of Lemma \ref{lem:B-low-w-tilde-upper-bound}]
Consider a fixed Hamming weight $w\ge N_C-N_C^{\varepsilon_2}$. As usual, define $\delta\coloneqq N_C-w$ and $\tilde{\delta}\coloneqq\tilde{N}_C-\tilde{w}$, so that we have $\delta\le N_C^{\varepsilon_2}$ and $\tilde{\delta}\ge\tilde{N}_C^{\varepsilon_1}$. Then, based on the result of Lemma \ref{lem:optimal-cloning-map-dicke-state-unified-presentation},
\begin{equation}
    \tilde{\delta}-\delta \sim \text{NHG}(\tilde{N}_C+1,\tilde{N}_C-N_C,\delta+1).
\end{equation}
This is the same statement made in the remarks at the end of Appendix \ref{sec:four-important-channels}\ref{subsec:optimal-cloning-map-unified-presentation}, but with the Hamming weights inverted. Now define
\begin{equation}
    Y \sim \text{HG}(\tilde{N}_C+1,N_C+1,\tilde{N}_C^{\varepsilon_1}).
\end{equation}
Then we can apply Lemma \ref{lem:NHG-HG-relation} with
\begin{equation}
    N\mapsto\tilde{N}_C+1, \quad K\mapsto\tilde{N}_C-N_C, \quad r\mapsto\delta+1, \quad k\mapsto\tilde{N}_C^{\varepsilon_1}-\delta-1
\end{equation}
to conclude that
\begin{equation}
    \Pbb[\tilde{\delta}\ge\tilde{N}_C^{\varepsilon_1}] = \Pbb[\tilde{\delta}-\delta\ge \tilde{N}_C^{\varepsilon_1}-\delta] = \Pbb[Y\le\delta].
\end{equation}
Now we define the quantities
\begin{align}
    K' &\coloneqq N-K = N_C+1 \\
    p &\coloneqq \frac{K'}{N} = \frac{N_C+1}{\tilde{N}_C+1} = \frac{1}{R_C} + O(N^{p-1}) \\
    n &\coloneqq \tilde{N}_C^{\varepsilon_1} \\
    t &\coloneqq p-\frac{\delta}{n} = \frac{1}{R_C} + O(N^{\varepsilon_2-\varepsilon_1}).
\end{align}
In particular, we define $K'$ and $n$ such that $Y\sim\text{HG}(N,K',n)$, and we define $p$ and $t$ such that $p = K'/N$, and $0 < t < p$. We can thus apply Lemma \ref{lem:hypergeometric-tail-bounds} to conclude that
\begin{equation}
    \Pbb[Y\le\delta] = \Pbb[Y\le(p-t)n] \le e^{-t^2n} = \exp\left[-\Omega(\tilde{N}_C^{\varepsilon_1})\right].
\end{equation}
We conclude that, for any fixed single $\delta\le N_C^{\varepsilon_2}$, the total contribution across all $\tilde{\delta}\ge\tilde{N}_C^{\varepsilon_1}$ is $\exp\left[-\Omega(\tilde{N}_C^{\varepsilon_1})\right]$. More precisely,
\begin{equation}
    \sum_{\tilde{\delta}=\tilde{N}_C^{\varepsilon_1}}^{\tilde{N}_C}\binom{\tilde{N}_C+1}{N_C+1}^{-1}\binom{\tilde{N}_C-\tilde{\delta}}{N_C-\delta}\binom{\tilde{\delta}}{\delta} = \exp\left[-\Omega(\tilde{N}_C^{\varepsilon_1})\right].
\end{equation}
Rewriting this in terms of $w$ and $\tilde{w}$ yields
\begin{equation}
    \sum_{\tilde{w}=0}^{\tilde{N}_C-\tilde{N}_C^{\varepsilon_1}}\binom{\tilde{N}_C+1}{N_C+1}^{-1}\binom{\tilde{w}}{w}\binom{\tilde{N}_C-\tilde{w}}{N_C-w} = \exp\left[-\Omega(\tilde{N}_C^{\varepsilon_1})\right].
\end{equation}
Summing this across $w\ge N_C-N_C^{\varepsilon_2}+1$ yields
\begin{align}
    &\quad\,\, \sum_{\tilde{w}=0}^{\tilde{N}_C-\tilde{N}_C^{\varepsilon_1}}\frac{c_1-c_0}{c_1^{N_C+1}-c_0^{N_C+1}}\binom{\tilde{N}_C+1}{N_C+1}^{-1}\sum_{w=N_C-N_C^{\varepsilon_2}+1}^{N_C}c_1^wc_0^{N_C-w}\binom{\tilde{w}}{w}\binom{\tilde{N}_C-\tilde{w}}{N_C-w} \\
    &= \left[\sum_{w=N_C-N_C^{\varepsilon_2}+1}^{N_C}\frac{c_1-c_0}{c_1^{N_C+1}-c_0^{N_C+1}}c_1^wc_0^{N_C-w}\right]\exp\left[-\Omega(\tilde{N}_C^{\varepsilon_1})\right] \\
    &= \exp\left[-\Omega(\tilde{N}_C^{\varepsilon_1})\right].
\end{align}
Finally, we need to account for the terms where $w\le N_C-N_C^{\varepsilon_2}$. But fortunately this is easy, since now we can just use the geometric series decay to our advantage:
\begin{align}
    &\quad\,\, \sum_{\tilde{w}=0}^{\tilde{N}_C-\tilde{N}_C^{\varepsilon_1}}\frac{c_1-c_0}{c_1^{N_C+1}-c_0^{N_C+1}}\binom{\tilde{N}_C+1}{N_C+1}^{-1}\sum_{w=0}^{N_C-N_C^{\varepsilon_2}}c_1^wc_0^{N_C-w}\binom{\tilde{w}}{w}\binom{\tilde{N}_C-\tilde{w}}{N_C-w} \\
    &\stackrel{(1)}{\le} \sum_{\tilde{w}=0}^{\tilde{N}_C-\tilde{N}_C^{\varepsilon_1}}\binom{\tilde{N}_C+1}{N_C+1}^{-1}\sum_{w=0}^{N_C-N_C^{\varepsilon_2}}\left(\frac{c_0}{c_1}\right)^{N_C-w}\binom{\tilde{w}}{w}\binom{\tilde{N}_C-\tilde{w}}{N_C-w} \\
    &\stackrel{(2)}{\le} \sum_{\tilde{w}=0}^{\tilde{N}_C-\tilde{N}_C^{\varepsilon_1}}\sum_{w=0}^{N_C-N_C^{\varepsilon_2}}\left(\frac{c_0}{c_1}\right)^{N_C-w} \\
    &\stackrel{(3)}{\le} \sum_{\tilde{w}=0}^{\tilde{N}_C-\tilde{N}_C^{\varepsilon_1}}\frac{c_1}{c_1-c_0}\left(\frac{c_0}{c_1}\right)^{N_C^{\varepsilon_2}} \\
    &\le \tilde{N}_C\frac{c_1}{c_1-c_0}\left(\frac{c_0}{c_1}\right)^{N_C^{\varepsilon_2}} \\
    &= \exp\left[-\Omega(N_C^{\varepsilon_2})\right].
\end{align}
Let us clarify how we derive each numbered inequality:
\begin{itemize}
    \item Inequality (1) comes from the usual way in which we upper bound the factors involving $c_1$ and $c_0$.
    \item Inequality (2) comes from the fact that
    \begin{equation}
        \binom{n_1+n_2+1}{k_1+k_2+1} \ge \binom{n_1+n_2}{k_1+k_2} \ge \binom{n_1}{k_1}\binom{n_2}{k_2}.
    \end{equation}
    \item Inequality (3) comes from extending the geometric series to $w=-\infty$.
\end{itemize}

\vspace{0.5\baselineskip}

Finally, combining the upper bound for $w\ge N_C-N_C^{\varepsilon_2}+1$ and the upper bound for $w\le N_C-N_C^{\varepsilon_2}$ yields
\begin{align}
    & \quad\,\, \sum_{\tilde{w}=0}^{\tilde{N}_C-\tilde{N}_C^{\varepsilon_1}}B(\tilde{w}) \\
    &= \sum_{\tilde{w}=0}^{\tilde{N}_C-\tilde{N}_C^{\varepsilon_1}}\frac{c_1-c_0}{c_1^{N_C+1}-c_0^{N_C+1}}\binom{\tilde{N}_C+1}{N_C+1}^{-1}\sum_{w=N_C-N_C^{\varepsilon_2}+1}^{N_C}c_1^wc_0^{N_C-w}\binom{\tilde{w}}{w}\binom{\tilde{N}_C-\tilde{w}}{N_C-w} \\
    &= \sum_{\tilde{w}=0}^{\tilde{N}_C-\tilde{N}_C^{\varepsilon_1}}\sum_{w=N_C-N_C^{\varepsilon_2}+1}^{N_C}\left(\cdots\right) + \sum_{\tilde{w}=0}^{\tilde{N}_C-\tilde{N}_C^{\varepsilon_1}}\sum_{w=0}^{N_C-N_C^{\varepsilon_2}}\left(\cdots\right) \\
    &= \exp\left[-\Omega(\tilde{N}_C^{\varepsilon_1})\right] + \exp\left[-\Omega(N_C^{\varepsilon_2})\right] \\
    &= \exp\left[-\Omega(N_C^{\varepsilon_2})\right].
\end{align}
\end{proof}

\begin{proof}[Proof of Lemma \ref{lem:C-low-w-tilde-upper-bound}]
Consider a fixed Hamming weight $w\ge N_C-N_C^{\varepsilon_2}$. As usual, define $\delta\coloneqq N_C-w$ and $\tilde{\delta}\coloneqq\tilde{N}_C-\tilde{w}$, so that we have $\delta\le N_C^{\varepsilon_2}$ and $\tilde{\delta}\ge\tilde{N}_C^{\varepsilon_1}$. Then, based on the result of Lemma \ref{lem:optimal-mp-channel-dicke-state-unified-presentation},
\begin{equation}
    \tilde{\delta} \sim \text{NHG}(N_C+\tilde{N}_C+1,\tilde{N}_C,\delta+1).
\end{equation}
This is the same statement made in the remarks at the end of Appendix \ref{sec:four-important-channels}\ref{subsec:optimal-mp-channel-unified-presentation}, but with the Hamming weights inverted. Now define
\begin{equation}
    Y \sim \text{HG}(N_C+\tilde{N}_C+1,N_C+1,\tilde{N}_C^{\varepsilon_1}+\delta).
\end{equation}
Then we can apply Lemma \ref{lem:NHG-HG-relation} with
\begin{equation}
    N\mapsto N_C+\tilde{N}_C+1, \quad K\mapsto\tilde{N}_C, \quad r\mapsto\delta+1, \quad k\mapsto\tilde{N}_C^{\varepsilon_1}-1
\end{equation}
to conclude that
\begin{equation}
    \Pbb[\tilde{\delta}\ge\tilde{N}_C^{\varepsilon_1}] = \Pbb[Y\le\delta].
\end{equation}
Now we define the quantities
\begin{align}
    K' &\coloneqq N-K = N_C+1 \\
    p &\coloneqq \frac{K'}{N} = \frac{N_C+1}{N_C+\tilde{N}_C+1} = \frac{1}{R_C+1} + O(N^{p-1}) \\
    n &\coloneqq \tilde{N}_C^{\varepsilon_1}+\delta \\
    t &\coloneqq p-\frac{\delta}{n} = \frac{1}{R_C+1} + O(N^{\varepsilon_2-\varepsilon_1}).
\end{align}
In particular, we define $K'$ and $n$ such that $Y\sim\text{HG}(N,K',n)$, and we define $p$ and $t$ such that $p = K'/N$, and $0 < t < p$. We can thus apply Lemma \ref{lem:hypergeometric-tail-bounds} to conclude that
\begin{equation}
    \Pbb[Y\le\delta] = \Pbb[Y\le(p-t)n] \le e^{-t^2n} = \exp\left[-\Omega(\tilde{N}_C^{\varepsilon_1})\right].
\end{equation}
We conclude that, for any fixed single $\delta\le N_C^{\varepsilon_2}$, the total contribution across all $\tilde{\delta}\ge\tilde{N}_C^{\varepsilon_1}$ is $\exp\left[-\Omega(\tilde{N}_C^{\varepsilon_1})\right]$. More precisely,
\begin{equation}
    \sum_{\tilde{\delta}=\tilde{N}_C^{\varepsilon_1}}^{\tilde{N}_C}\binom{N_C+\tilde{N}_C+1}{N_C+1}^{-1}\binom{(N_C-\delta)+(\tilde{N}_C-\tilde{\delta})}{N_C-\delta}\binom{\delta+\tilde{\delta}}{\delta} = \exp\left[-\Omega(\tilde{N}_C^{\varepsilon_1})\right].
\end{equation}
Rewriting this in terms of $w$ and $\tilde{w}$ yields
\begin{equation}
    \sum_{\tilde{\delta}=0}^{\tilde{N}_C-\tilde{N}_C^{\varepsilon_1}}\binom{N_C+\tilde{N}_C+1}{N_C+1}^{-1}\binom{w+\tilde{w}}{w}\binom{(N_C-w)+(\tilde{N}_C-\tilde{w})}{N_C-w} = \exp\left[-\Omega(\tilde{N}_C^{\varepsilon_1})\right].
\end{equation}
Summing this across $w\ge N_C-N_C^{\varepsilon_2}+1$ yields
\begin{align}
    &\quad\,\, \sum_{\tilde{w}=0}^{\tilde{N}_C-\tilde{N}_C^{\varepsilon_1}}\frac{c_1-c_0}{c_1^{N_C+1}-c_0^{N_C+1}}\binom{N_C+\tilde{N}_C+1}{N_C+1}^{-1}\sum_{w=N_C-N_C^{\varepsilon_2}+1}^{N_C}c_1^wc_0^{N_C-w}\binom{w+\tilde{w}}{w}\binom{(N_C-w)+(\tilde{N}_C-\tilde{w})}{N_C-w} \\
    &= \left[\sum_{w=N_C-N_C^{\varepsilon_2}+1}^{N_C}\frac{c_1-c_0}{c_1^{N_C+1}-c_0^{N_C+1}}c_1^wc_0^{N_C-w}\right]\exp\left[-\Omega(\tilde{N}_C^{\varepsilon_1})\right] \\
    &= \exp\left[-\Omega(\tilde{N}_C^{\varepsilon_1})\right].
\end{align}
Finally, we need to account for the terms where $w\le N_C-N_C^{\varepsilon_2}$. But fortunately this is easy, since now we can just use the geometric series decay to our advantage:
\begin{align}
    &\quad\,\, \sum_{\tilde{w}=0}^{\tilde{N}_C-\tilde{N}_C^{\varepsilon_1}}\frac{c_1-c_0}{c_1^{N_C+1}-c_0^{N_C+1}}\binom{N_C+\tilde{N}_C+1}{N_C+1}^{-1}\sum_{w=0}^{N_C-N_C^{\varepsilon_2}}c_1^wc_0^{N_C-w}\binom{w+\tilde{w}}{w}\binom{(N_C-w)+(\tilde{N}_C-\tilde{w})}{N_C-w} \\
    &\stackrel{(1)}{\le} \sum_{\tilde{w}=0}^{\tilde{N}_C-\tilde{N}_C^{\varepsilon_1}}\binom{N_C+\tilde{N}_C+1}{N_C+1}^{-1}\sum_{w=0}^{N_C-N_C^{\varepsilon_2}}\left(\frac{c_0}{c_1}\right)^{N_C-w}\binom{w+\tilde{w}}{w}\binom{(N_C-w)+(\tilde{N}_C-\tilde{w})}{N_C-w} \\
    &\stackrel{(2)}{\le} \sum_{\tilde{w}=0}^{\tilde{N}_C-\tilde{N}_C^{\varepsilon_1}}\sum_{w=0}^{N_C-N_C^{\varepsilon_2}}\left(\frac{c_0}{c_1}\right)^{N_C-w} \\
    &\stackrel{(3)}{\le} \sum_{\tilde{w}=0}^{\tilde{N}_C-\tilde{N}_C^{\varepsilon_1}}\frac{c_1}{c_1-c_0}\left(\frac{c_0}{c_1}\right)^{N_C^{\varepsilon_2}} \\
    &\le \tilde{N}_C\frac{c_1}{c_1-c_0}\left(\frac{c_0}{c_1}\right)^{N_C^{\varepsilon_2}} \\
    &= \exp\left[-\Omega(N_C^{\varepsilon_2})\right].
\end{align}
Let us clarify how we derive each numbered inequality:
\begin{itemize}
    \item Inequality (1) comes from the usual way in which we upper bound the factors involving $c_1$ and $c_0$.
    \item Inequality (2) comes from the fact that
    \begin{equation}
        \binom{n_1+n_2+1}{k_1+k_2+1} \ge \binom{n_1+n_2}{k_1+k_2} \ge \binom{n_1}{k_1}\binom{n_2}{k_2}.
    \end{equation}
    \item Inequality (3) comes from extending the geometric series to $w=-\infty$.
\end{itemize}

\vspace{0.5\baselineskip}

Finally, combining the upper bound for $w\ge N_C-N_C^{\varepsilon_2}+1$ and the upper bound for $w\le N_C-N_C^{\varepsilon_2}$ yields
\begin{align}
    & \quad\,\, \sum_{\tilde{w}=0}^{\tilde{N}_C-\tilde{N}_C^{\varepsilon_1}}C(\tilde{w}) \\
    &= \sum_{\tilde{w}=0}^{\tilde{N}_C-\tilde{N}_C^{\varepsilon_1}}\frac{c_1-c_0}{c_1^{N_C+1}-c_0^{N_C+1}}\binom{N_C+\tilde{N}_C+1}{N_C+1}^{-1}\sum_{w=N_C-N_C^{\varepsilon_2}+1}^{N_C}c_1^wc_0^{N_C-w}\binom{w+\tilde{w}}{w}\binom{(N_C-w)+(\tilde{N}_C-\tilde{w})}{N_C-w} \\
    &= \sum_{\tilde{w}=0}^{\tilde{N}_C-\tilde{N}_C^{\varepsilon_1}}\sum_{w=N_C-N_C^{\varepsilon_2}+1}^{N_C}\left(\cdots\right) + \sum_{\tilde{w}=0}^{\tilde{N}_C-\tilde{N}_C^{\varepsilon_1}}\sum_{w=0}^{N_C-N_C^{\varepsilon_2}}\left(\cdots\right) \\
    &= \exp\left[-\Omega(\tilde{N}_C^{\varepsilon_1})\right] + \exp\left[-\Omega(N_C^{\varepsilon_2})\right] \\
    &= \exp\left[-\Omega(N_C^{\varepsilon_2})\right].
\end{align}
\end{proof}

\newpage

\appsec{A Unified Presentation of Qubit Linear-Rate Conversion Protocols}
{sec:unified-presentation}

In this appendix, we provide a unified presentation of all four of our qubit conversion protocols. We begin by outlining the basic three-step procedure for each protocol, where the first and last steps (Schur sampling and inverse Schur sampling, respectively) are identical for all four protocols, and only the middle step is different. We then prove that these protocols successfully achieve their respective linear-rate conversion tasks (these proofs rely crucially on calculations from Appendix \ref{sec:four-important-channels} and Appendix \ref{sec:schur-transformed-state-conversion}). When combined with the proof of the converse bounds in Appendix \ref{sec:converse-bound-rld-sensitivity}, we conclude that the linear conversion rates achieved in this appendix are the highest possible.

\vspace{0.5\baselineskip}

This appendix is broken down as follows:
\begin{itemize}
    \item In Appendix \ref{sec:unified-presentation}\ref{subsec:three-step-procedure-unified-presentation}, we state the three-step procedure that applies to all four conversion tasks. The first step is always Schur sampling, and the third step is always inverse Schur sampling. Only the second step differs among the conversion tasks; in particular, it is chosen from $\{\mE_{\text{discard}}, \mE_{\text{clone}}, \mE_{\text{MP}}, \mE_{\text{WW}}\}$, which we studied back in Appendices \ref{sec:four-important-channels} and \ref{sec:schur-transformed-state-conversion}.
    \item In Appendix \ref{sec:unified-presentation}\ref{subsec:successful-conversion-theorems-unified-presentation}, we prove that the three-step procedure successfully implements all four conversion tasks with vanishing trace distance in the $N\to\infty$ limit. The first crucial ingredient is the concentration of Schur sampling outcomes as $N_C\approx\lambda N$, as discussed in Appendix \ref{sec:schur-sampling-commentary}\ref{subsec:typicality}. The second necessary ingredient is the result of applying one of the channels $\mE_{\text{discard}}$, $\mE_{\text{clone}}$, $\mE_{\text{MP}}$, $\mE_{\text{WW}}$ to a Schur-transformed state, as discussed in Appendix \ref{sec:schur-transformed-state-conversion}.
\end{itemize}

\appsubsec{The Three-Step Procedure}
{subsec:three-step-procedure-unified-presentation}

We first outline the three-step procedures for our four linear-rate conversion tasks (concentration, dilution, measure-and-prepare conversion, and wrong-way conversion). The procedures work as follows:

\begin{enumerate}
    \item Perform Schur sampling on $\rho(\lambda_{\text{in}},\hat{n})^{\otimes N}$. The result is $\rho_C(N_C,\lambda_{\text{in}},\hat{n})$ with probability $p(N,N_C,\lambda_{\text{in}})$ for each $0\le N_C\le N$ such that $N$ and $N_C$ have the same parity.
    \item Randomly select $\tilde{N}_C$ with probability $p(\tilde{N},\tilde{N}_C,\lambda_{\text{out}})$. Now, depending on the task, perform the following channel from the symmetric subspace on $N_C$ qubits to the symmetric subspace on $\tilde{N}_C$ qubits:
    \begin{itemize}
        \item \textbf{Concentration:} discarding map (i.e., partial trace) $\mE_{\text{discard}}[N_C\to\tilde{N}_C]$.
        \item \textbf{Dilution:} Werner optimal cloning map $\mE_{\text{clone}}[N_C\to\tilde{N}_C]$.
        \item \textbf{Measure-and-prepare conversion:} optimal measure-and-prepare channel $\mE_{\text{MP}}[N_C\to\tilde{N}_C]$.
        \item \textbf{Wrong-way conversion:} optimal wrong-way measure-and-prepare channel $\mE_{\text{WW}}[N_C\to\tilde{N}_C]$.
    \end{itemize}
    \item Introduce $(\tilde{N} - \tilde{N}_C)/2$ singlet states, randomly permute the qubits, and forget the value of $\tilde{N}_C$ to produce the final output state, which has $\tilde{N}$ qubits. This has the effect of undoing Schur sampling.
\end{enumerate}

It is worth highlighting a few features of this protocol:
\begin{itemize}
    \item We define the protocol in such a way that it always produces a state with exactly $\tilde{N}$ qubits.
    \item We ensure that the probability of having $\tilde{N}_C$ qubits before the final inverse Schur sampling step exactly matches the distribution $p(\tilde{N},\tilde{N}_C,\lambda_{\text{out}})$ that would be obtained from applying Schur sampling to the target state $\rho(\lambda_{\text{out}},\hat{n})^{\otimes\tilde{N}}$. Although not strictly necessary, this choice makes our analysis slightly simpler when we upper bound the trace distance between the output of the above protocol and the target state.
\end{itemize}

Also, let us briefly mention one small caveat for the sake of completeness:
\begin{itemize}
    \item When carrying out the concentration protocol described above, it is possible to have $N_C < \tilde{N}_C$, in which case the discarding map is not defined. So if this happens, apply the optimal cloning map instead.
    \item Similarly, when carrying out the dilution protocol described above, it is possible to have $N_C > \tilde{N}_C$, in which case the optimal cloning map is not defined. So if this happens, apply the discarding map instead.
\end{itemize}
The probability that either of these events occurs becomes exponentially small as $N\to\infty$, due to the concentration of Schur sampling outcomes. In particular, as $N\to\infty$, these events will be restricted to atypical values of $N_C$ and $\tilde{N}_C$, and as we will see in the next subsection, we will not be concerned with the performance of the protocol in those regimes. Hence, we do not need to analyze these stray cases in detail.

\appsubsec{Successful Linear-Rate Conversion Theorems}
{subsec:successful-conversion-theorems-unified-presentation}

Now that we have defined our protocol, we can state the essential result that we want to prove, which states that the trace distance between the output of this protocol and the target i.i.d. qubit state vanishes in the $N\rightarrow\infty$ limit. We formalize this in the following theorems:

\begin{theorem}[concentration protocol achieves maximum rate]
\label{thm:concentration-protocol-works-unified-presentation}
Suppose that $0 < \lambda_{\text{in}} < \lambda_{\text{out}} < 1$. Let $R = R^{\text{conc}}(\lambda_{\text{in}}\to\lambda_{\text{out}})$ and $\tilde{N} = \lfloor RN\rfloor$. Also, let $\sigma_{\tilde{N}}(\hat{n})$ denote the output of the concentration protocol described above. Then
\begin{equation}
    \lim_{N\rightarrow\infty}d_{\text{Tr}}\left(\sigma_{\tilde{N}}(\hat{n}),\rho(\lambda_{\text{out}},\hat{n})^{\otimes\tilde{N}}\right) = 0.
\end{equation}
\end{theorem}

\begin{theorem}[dilution protocol achieves maximum rate]
\label{thm:dilution-protocol-works-unified-presentation}
Suppose that $0 < \lambda_{\text{out}} < \lambda_{\text{in}} \le 1$. Let $R = R^{\text{dilut}}(\lambda_{\text{in}}\to\lambda_{\text{out}})$ and $\tilde{N} = \lfloor RN\rfloor$. Also, let $\sigma_{\tilde{N}}(\hat{n})$ denote the output of the dilution protocol described above. Then
\begin{equation}
    \lim_{N\rightarrow\infty}d_{\text{Tr}}\left(\sigma_{\tilde{N}}(\hat{n}),\rho(\lambda_{\text{out}},\hat{n})^{\otimes\tilde{N}}\right) = 0.
\end{equation}
\end{theorem}

\begin{theorem}[measure-and-prepare conversion protocol achieves maximum rate]
\label{thm:mp-protocol-works-unified-presentation}
Suppose that $0 < \lambda_{\text{in}} \le 1$ and $0 < \lambda_{\text{out}} < 1$. Let $R = R^{\text{MP}}(\lambda_{\text{in}}\to\lambda_{\text{out}})$ and $\tilde{N} = \lfloor RN\rfloor$. Also, let $\sigma_{\tilde{N}}(\hat{n})$ denote the output of the measure-and-prepare conversion protocol described above. Then
\begin{equation}
    \lim_{N\rightarrow\infty}d_{\text{Tr}}\left(\sigma_{\tilde{N}}(\hat{n}),\rho(\lambda_{\text{out}},\hat{n})^{\otimes\tilde{N}}\right) = 0.
\end{equation}
\end{theorem}

\begin{theorem}[wrong-way conversion protocol achieves maximum rate]
\label{thm:ww-protocol-works-unified-presentation}
Suppose that $0 < \lambda_{\text{in}} \le 1$ and $0 < \lambda_{\text{out}} < 1$. Let $R = R^{\text{WW}}(\lambda_{\text{in}}\to\lambda_{\text{out}})$ and $\tilde{N} = \lfloor RN\rfloor$. Also, let $\sigma_{\tilde{N}}(\hat{n})$ denote the output of the wrong-way conversion protocol described above. Then
\begin{equation}
    \lim_{N\rightarrow\infty}d_{\text{Tr}}\left(\sigma_{\tilde{N}}(\hat{n}),\rho(\lambda_{\text{out}},-\hat{n})^{\otimes\tilde{N}}\right) = 0.
\end{equation}
\end{theorem}

We now prove the four theorems above. Fortunately, due to the many similarities among the four procedures, we will be able to prove all of them in one shot.

\vspace{0.5\baselineskip}

First, instead of comparing the final output of the protocol to the target i.i.d. state, we can compare the output of step $2$ of the protocol (before the inverse Schur sampling step) to the corresponding ensemble of Schur-transformed states that would result from the target state. The inverse Schur sampling step maps the step $2$ output to the final protocol output, and it maps the target ensemble of Schur-transformed states to the target i.i.d. state. Therefore, since the inverse Schur sampling step cannot increase the trace distance, we can look one step earlier and upper bound the trace distance there.

\vspace{0.5\baselineskip}

Second, since trace distance is jointly convex, we can upper bound the trace distance between the Schur-transformed ensembles by the average trace distance over all probabilistic outcomes.

\vspace{0.5\baselineskip}

Combining the two facts above yields
\begin{equation}
    d_{\text{Tr}}\left(\sigma_{\tilde{N}}(\hat{n}), \rho(\lambda_{\text{out}},\hat{n})^{\otimes\tilde{N}}\right) \le \sum_{N_C}\sum_{\tilde{N}_C}p(N,N_C,\lambda_{\text{in}})p(\tilde{N},\tilde{N}_C,\lambda_{\text{out}})d_{\text{Tr}}\left(\mE_{(\cdot)}[N_C\to\tilde{N}_C]\left(\rho_C(N_C,\lambda_{\text{in}},\hat{n})\right),\rho_C(\tilde{N}_C,\lambda_{\text{out}},\hat{n})\right),
\end{equation}
where $N_C$ and $\tilde{N}_C$ range over all valid values in the summation, and where the subscript $(\cdot)\in\{\text{discard},\text{clone},\text{MP},\text{WW}\}$ depends on which conversion task we are carrying out. (For wrong-way conversion, the last $\hat{n}$ in the above equation should be replaced with $-\hat{n}$. We will) It thus suffices to show that the right-hand-side of the above inequality vanishes in the $N\to\infty$ limit.

\vspace{0.5\baselineskip}

It is worthwhile to mention that, for wrong-way conversion, any $\hat{n}$ in the target i.i.d. state or Schur-transformed state should of course be replaced with $-\hat{n}$. For convenience, we will not repeat this observation for the rest of the proof.

\vspace{0.5\baselineskip}

Intuitively, there are three broad ways in which the result of applying $\mE_{(\cdot)}$ to the initial Schur-transformed state can differ from the target Schur-transformed state:
\begin{enumerate}
    \item The total angular momentum measurement in the initial Schur sampling step may produce a highly deviant value of $N_C$ (that is, a value far from $\lambda_{\text{in}}N$).
    \item The random choice of the number of qubits to keep in the discarding step may produce a highly deviant value of $\tilde{N}_C$ (that is, a value far from $\lambda_{\text{out}}\tilde{N}$).
    \item Even for typical $N_C$ and $\tilde{N}_C$, the state produced from $\rho_C(N_C,\lambda_{\text{in}},\hat{n})$ by applying the optimal cloning map from $N_C$ qubits to $\tilde{N}_C$ qubits merely approximates the desired Schur-transformed state $\rho_C(\tilde{N}_C,\lambda_{\text{out}},\hat{n})$.
\end{enumerate}

We can formalize this by upper bounding the above summation by the total of three different quantities:
\begin{align}
    & \quad\quad \sum_{N_C}\sum_{\tilde{N}_C}p(N,N_C,\lambda_{\text{in}})p(\tilde{N},\tilde{N}_C,\lambda_{\text{out}})d_{\text{Tr}}\left(\mE_{(\cdot)}[N_C\to\tilde{N}_C]\left(\rho_C(N_C,\lambda_{\text{in}},\hat{n})\right),\rho_C(\tilde{N}_C,\lambda_{\text{out}},\hat{n})\right) \\
    &\le \sum_{N_C\text{ atypical}}\sum_{\tilde{N}_C}(\cdots) + \sum_{N_C}\sum_{\tilde{N}_C\text{ atypical}}(\cdots) + \sum_{N_C\text{ typical}}\sum_{\tilde{N}_C\text{ typical}}(\cdots) \\
    &\le \Pbb\left[N_C\text{ atypical}\right] + \Pbb\left[\tilde{N}_C\text{ atypical}\right] + \max_{N_C,\tilde{N}_C\text{ both typical}}d_{\text{Tr}}\left(\mE_{(\cdot)}[N_C\to\tilde{N}_C]\left(\rho_C(N_C,\lambda_{\text{in}},\hat{n})\right),\rho_C(\tilde{N}_C,\lambda_{\text{out}},\hat{n})\right).
\end{align}
The three terms correspond precisely to the three sources of error outlined above. The first term is the probability that $N_C$ is atypical, the second term is the probability that $\tilde{N}_C$ is atypical, and the third term is the maximum possible trace distance if $N_C$ and $\tilde{N}_C$ are both typical.

\vspace{0.5\baselineskip}

We are already ready to upper bound the first two sources of error based on the typicality discussion in Appendix \ref{subsec:typicality}. This part is identical for all four procedures.

\vspace{0.5\baselineskip}

By Chebyshev's inequality, the probability that $N_C$ is atypical can be upper bounded as follows:
\begin{equation}
    \Pbb\left[\abs{N_C - \lambda_{\text{in}}N} > \left[(1-\lambda_{\text{in}}^2)N\right]^{2/3}\right] \lesssim \left[(1-\lambda_{\text{in}}^2)N\right]^{-1/3} = O\left(N^{-1/3}\right).
\end{equation}
Also by Chebyshev's inequality, the probability that $\tilde{N}_C$ is atypical can be upper bounded as follows:
\begin{equation}
    \Pbb\left[\abs{\tilde{N}_C - \lambda_{\text{out}}\tilde{N}} > \left[(1-\lambda_{\text{out}}^2)\tilde{N}\right]^{2/3}\right] \lesssim \left[(1-\lambda_{\text{out}}^2)\tilde{N}\right]^{-1/3} = O\left(N^{-1/3}\right).
\end{equation}
In particular, both of these probabilities go to zero as $N\rightarrow\infty$. This error is actually exponentially small due to the concentration of the distribution of Schur sampling outcomes \cite{Keyl2001}, so these two bounds are actually very crude. However, we do not worry about making these bounds tighter, because the third and final source of error will be the largest anyway.

\vspace{0.5\baselineskip}

Upper bounding the final source of error requires us to show that, if we apply the suitable channel to a Schur-transformed state with the input qubit count $N_C$ and the output qubit count $\tilde{N}_C$ both being typical, then the resulting state approximates a Schur-transformed state, with the output purity level depending on the input purity level and the ratio by which the qubit count changes.

\vspace{0.5\baselineskip}

But in fact, we already showed this in Appendix \ref{sec:schur-transformed-state-conversion}! In particular:
\begin{itemize}
    \item Refer to Lemma \ref{lem:discarding-map-schur-transformed-state} to see how the discarding map achieves concentration for Schur-transformed states.
    \item Refer to Lemma \ref{lem:optimal-cloning-map-schur-transformed-state} to see how the optimal cloning map achieves dilution for Schur-transformed states.
    \item Refer to Lemma \ref{lem:optimal-mp-channel-schur-transformed-state} to see how the optimal measure-and-prepare channel achieves measure-and-prepare conversion for Schur-transformed states.
    \item Refer to Lemma \ref{lem:optimal-ww-channel-schur-transformed-state} to see how the optimal wrong-way measure-and-prepare channel achieves wrong-way conversion for Schur-transformed states.
\end{itemize}
These four lemmas immediately show that the third and final source of error tends to zero in the $N\to\infty$ limit, since if both $N_C$ and $\tilde{N}_C$ are typical, we can plug $p=\frac{2}{3}$ into the lemmas to determine that the third and final source of error is $O\left(N^{-\frac{1}{3}+\varepsilon}\right)$ for arbitrarily small $\varepsilon > 0$. This concludes the proof that the linear conversion rates that we claim to be optimal are indeed achievable.

\vspace{0.5\baselineskip}

In fact, using these lemmas, we can draw two other interesting conclusions:
\begin{itemize}
    \item The allowed deviation of $N_C$ in Definition \ref{def:schur-sampling-outcome-typicality} can be chosen to scale as $N^p$ for any $\frac{1}{2} < p < 1$. On one end, we only need $p > \frac{1}{2}$ to ensure that the first two sources of error (corresponding to the probability of atypical $N_C$ or atypical $\tilde{N}_C$ vanish as $N\to\infty$. On the other end, we only need $p < 1$ to be able to use the above lemmas to ensure that the third source of error vanishes as $N\to\infty$.
    \item Since the first and second sources of error are actually exponentially small, if we chose $p = \frac{1}{2} + \varepsilon'$ for very small $\varepsilon' > 0$, then the above lemmas tell us that the overall trace distance scales as $O\left(N^{-\frac{1}{2}+\varepsilon''}\right)$, where $\varepsilon'' \coloneqq \varepsilon + \varepsilon' > 0$ can be made arbitrarily small. In Appendix \ref{sec:numerical-analysis}, we will show numerical simulations that support this analytical error analysis.
\end{itemize}

\newpage

\appsec{Heuristic Arguments for the Concentration and Dilution Protocols}
{sec:heuristic-arguments}

In Appendix \ref{sec:unified-presentation}, we showed three-step procedures to achieve maximum-rate qubit concentration and dilution. The crucial facts that make these procedures work were presented earlier in Lemmas \ref{lem:discarding-map-schur-transformed-state} and \ref{lem:optimal-cloning-map-schur-transformed-state}. In particular, we proved that applying the discarding map (optimal cloning map) to a Schur-transformed state to decrease (increase) the number of qubits by a certain proportion yields approximately another Schur-transformed state with a higher (lower) purity level that depends on the initial purity level and the relative decrease (increase) in the number of qubits.

\vspace{0.5\baselineskip}

However, we suspect that Lemmas \ref{lem:discarding-map-schur-transformed-state} and \ref{lem:optimal-cloning-map-schur-transformed-state} may still seem a bit ``magical''. In this appendix, we aim to make them much more intuitive by providing two different heuristic arguments that show intuitively how the discarding and optimal cloning maps achieve their respective effects.

\vspace{0.5\baselineskip}

The first heuristic argument illustrates the discarding and cloning maps on a Schur-transformed state using tensor networks, and it covers the following subsections:
\begin{itemize}
    \item In Appendix \ref{sec:heuristic-arguments}\ref{subsec:tensor-network-symmetric-subspace}, we introduce a tensor network representation for projecting a multi-qubit state to the symmetric subspace and apply this to an i.i.d. qubit state to produce a representation of a Schur-transformed state. We then show diagrammatic representations of the discarding map and optimal cloning map. With these diagrams, we will already be able to show the basic idea behind why discarding yields concentration while optimal cloning yields dilution, although we will not immediately be able to see why the numbers work out the way they do.
    \item In Appendix \ref{sec:heuristic-arguments}\ref{subsec:heuristic-cloning-dilution}, we show how the optimal cloning map achieves dilution. We first explain the special case where the initial qubits are pure, which actually provides a novel interpretation of a Schur-transformed state as (approximately) the result of applying the optimal cloning map to a collection of i.i.d. pure qubits. We then use this novel interpretation to extend the argument to the general setting where the initial qubits are mixed.
    \item In Appendix \ref{sec:heuristic-arguments}\ref{subsec:heuristic-discarding-concentration}, we show how the discarding map achieves concentration. The novel interpretation of a Schur-transformed state will also serve as the backbone of this argument.
\end{itemize}
The second heuristic argument studies what happens when we start with a Schur-transformed state and add or remove just one qubit (rather than altering the number of qubits by a constant factor), and it covers the following subsections:
\begin{itemize}
    \item In Appendix \ref{sec:heuristic-arguments}\ref{subsec:heuristic-discarding-one-qubit}, we show that discarding just one qubit from a Schur-transformed state approximately maintains the ``geometric sequence'' property of a Schur-transformed state, while slightly nudging the purity level upward by the right amount to approximately conserve $\text{RLD}_{\text{max}}$.
    \item In Appendix \ref{sec:heuristic-arguments}\ref{subsec:heuristic-adding-one-qubit}, we show that adding just one qubit to a Schur-transformed state via the optimal cloning map approximately maintains the ``geometric sequence'' property of a Schur-transformed state, while slightly nudging the purity level downward by the right amount to approximately conserve $\text{RLD}_{\text{min}}$.
\end{itemize}

\appsubsec{Tensor Network Representations for the Symmetric Subspace}
{subsec:tensor-network-symmetric-subspace}

The projector to the symmetric subspace on $M$ qubits is
\begin{equation}
    \Pi_{\text{sym}}^{M} = \frac{1}{M!}\sum_{\sigma\in S_{M}}P_\sigma,
\end{equation}
where $P_\sigma$ is the operator that permutes the qubits according to $\sigma\in S_M$. Therefore, the projection of a general $M$-qubit state $\rho_M$ to the symmetric subspace takes the form
\begin{equation}
    \Pi_{\text{sym}}^{M}\rho_M\Pi_{\text{sym}}^{M} = \left(\frac{1}{M!}\sum_{\sigma\in S_{M}}P_\sigma\right)\rho_M\left(\frac{1}{M!}\sum_{\tau\in S_{M}}P_\tau\right).
\end{equation}
To represent such a projection as a tensor network, we can take the $M$-qubit state $\rho_M$ and sandwich it between two permutations chosen independently and uniformly at random from the symmetric group $S_M$. (Technically, the projection we want equals the \textit{average} over all such diagrams.) This is shown in Figure \ref{fig:SUB-symmetric-subspace-state-general}, where the \textbf{\textcolor{cyan}{cyan}} box represents some multi-qubit state, and the \textbf{\textcolor{LimeGreen}{green}} boxes represent random permutations (which just correspond to scrambling the wires).

\vspace{0.5\baselineskip}

The discarding map is simply a partial trace, which is very easy to represent with tensor networks. For each qubit that is being discarded, simply take the output end of the wire and wrap it back around to connect to the input end. For a state in the symmetric subspace, it does not matter which qubits we discard, only how many, so we will always discard qubits from the right side. This is shown in Figure \ref{fig:SUB-discarding-map-general}, where the \textbf{\textcolor{gray}{gray}} wires implement the partial trace.

\vspace{0.5\baselineskip}

The optimal cloning map is also straightforward. The optimal cloning map from $N$ qubits to $M\ge N$ qubits takes the form
\begin{equation}
    \mE_{\text{clone}}[N\rightarrow M](\rho_N) = \frac{N+1}{M+1}\Pi^M_{\text{sym}}\left(\rho_N\otimes\Ibb_2^{\otimes(M-N)}\right)\Pi^M_{\text{sym}},
\end{equation}
where $\rho_N$ is a state in the $N$-qubit symmetric subspace \cite{Werner1998}. Therefore, if we ignore the constant factor, we can understand the optimal cloning map $\mE_{\text{clone}}[N\rightarrow M]$ as introducing $(M-N)$ identity states (which are simply blank wires) and sandwiching the result between two additional permutations on $M$ qubits, chosen independently and uniformly at random. The resulting tensor network is shown in Figure \ref{fig:SUB-optimal_cloning_map_general}. However, notice that the inner random permutations that act on the original $N$ qubits can be absorbed into the outer random permutations that act on all $M$ qubits. We can thus simplify the tensor network, as shown in Figure \ref{fig:SUB-optimal_cloning_map_general_permutations_absorbed}. (Due to the constant factor discrepancy, this should NOT be mistaken for an actual implementation of the optimal cloning map! However, in Appendix \ref{sec:friendly-implementations}\ref{subsec:optimal-cloning-implementation}, we show an elegant implementation of the optimal cloning map that was inspired at least somewhat by this tensor network representation.)

\begin{figure}
    \centering
    \begin{subfigure}{0.48\textwidth}
        \centering
        \includegraphics[width=\textwidth]{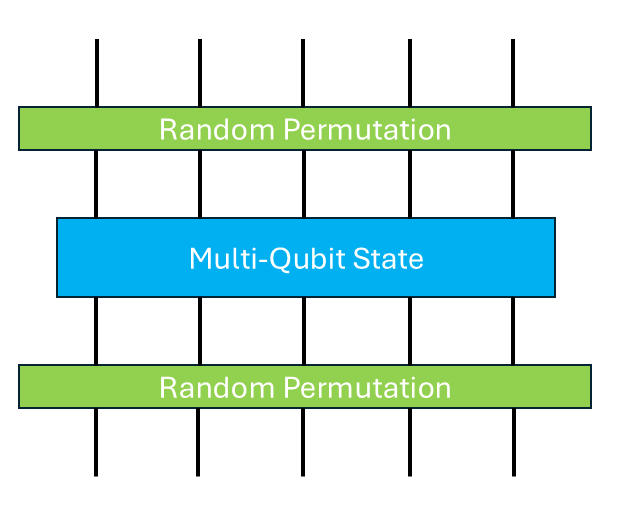}
        \caption{The projector to the symmetric subspace is the average of all qubit permutation operators. As a result, any density operator in the symmetric subspace can be understood (up to proportionality) as the result of taking some multi-qubit density operator and then sandwiching it between two qubit permutations that are chosen independently and uniformly at random. More precisely, the state in the symmetric subspace is described (up to proportionality) by the average over all such diagrams.}
        \label{fig:SUB-symmetric-subspace-state-general}
    \end{subfigure}
    \hfill
    \begin{subfigure}{0.48\textwidth}
        \centering
        \includegraphics[width=\textwidth]{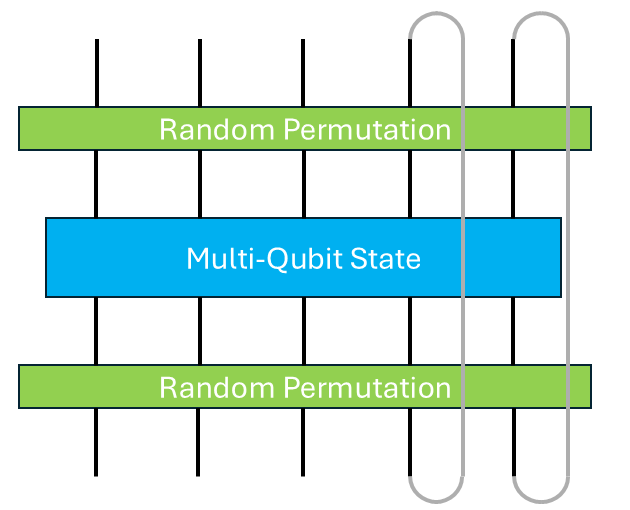}
        \caption{Applying the discarding map to a state in the symmetric subspace corresponds to taking the output end of each wire corresponding to a discarded qubit and wrapping it back around to connect to the input end. Of course, since we are working in the symmetric subspace, it only matters how many qubits we discard, not which ones.}
        \label{fig:SUB-discarding-map-general}
    \end{subfigure}
    \\
    \begin{subfigure}{0.48\textwidth}
        \centering
        \includegraphics[width=\textwidth]{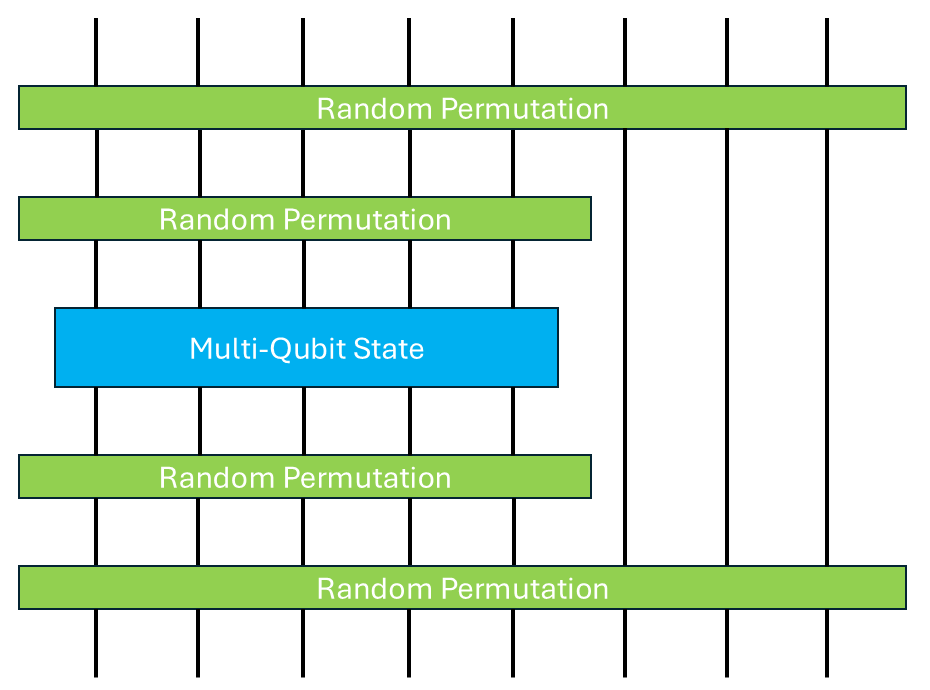}
        \caption{Applying the optimal cloning map corresponds to (up to proportionality) adding a bunch of identity wires and then projecting the whole collection to the symmetric subspace. We can once again represent the outer projection as (up to proportionality) the average over all possible sandwiches of permutations chosen independently and uniformly at random.}
        \label{fig:SUB-optimal_cloning_map_general}
    \end{subfigure}
    \hfill
    \begin{subfigure}{0.48\textwidth}
        \centering
        \includegraphics[width=\textwidth]{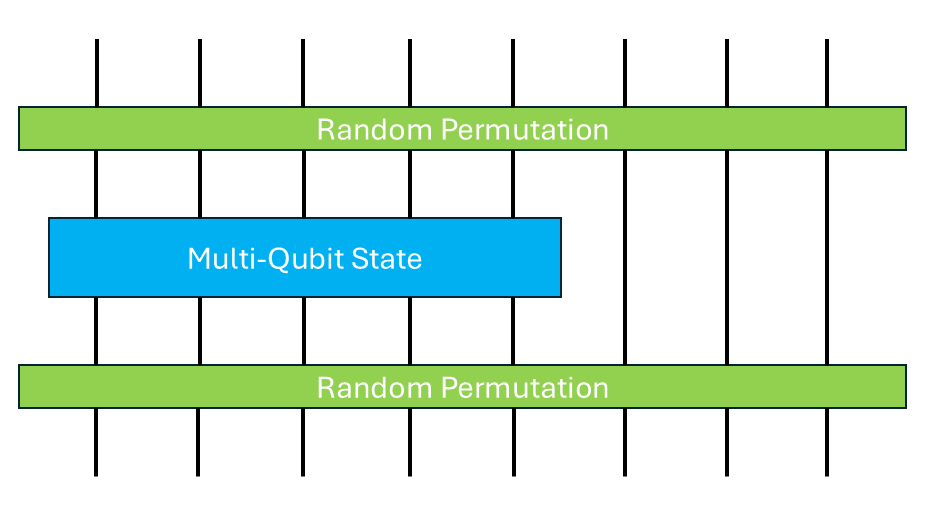}
        \caption{For the optimal cloning map, the inner random permutations can be absorbed into the outer random permutations, producing a simplified representation.}
        \label{fig:SUB-optimal_cloning_map_general_permutations_absorbed}
    \end{subfigure}
    \caption{The tensor network representation of a state on the symmetric subspace (Figure \ref{fig:SUB-symmetric-subspace-state-general}), followed by an application of either the discarding map (Figure \ref{fig:SUB-discarding-map-general}) or the optimal cloning map (Figures \ref{fig:SUB-optimal_cloning_map_general} and \ref{fig:SUB-optimal_cloning_map_general_permutations_absorbed}).}
    \label{fig:symmetric-subspace-state-discarding-map-optimal-cloning-map-general}
\end{figure}

\vspace{0.5\baselineskip}

We now turn to the special case of the Schur-transformed state. The state $\rho_C(N_C,\lambda,\hat{n})$ is proportional to the projection of $\rho(\lambda,\hat{n})^{\otimes N_C}$ to the $N_C$-qubit symmetric subspace:
\begin{equation}
    p(N_C,N_C,\lambda)\rho_C(N_C,\lambda,\hat{n}) = \Pi_{\text{sym}}(N_C)\rho(\lambda,\hat{n})^{\otimes N_C}\Pi_{\text{sym}}(N_C).
\end{equation}
Therefore, it has a tensor network representation as in Figure \ref{fig:SUB-symmetric-subspace-state-general}, except now the generic multi-qubit state is replaced with an i.i.d. collection $\rho(\lambda,\hat{n})^{\otimes N_C}$.

\vspace{0.5\baselineskip}

In fact, since $\rho(\lambda,\hat{n})^{\otimes N_C}$ is permutation-invariant, it commutes with $\Pi_{\text{sym}}(N_C)$, so we can actually pass one factor of $\Pi_{\text{sym}}(N_C)$ to the other side and absorb it into the other factor (since a projector squared equals itself):
\begin{equation}
    p(N_C,N_C,\lambda)\rho_C(N_C,\lambda,\hat{n}) = \Pi_{\text{sym}}(N_C)\rho(\lambda,\hat{n})^{\otimes N_C}.
\end{equation}
This means that, to represent a Schur-transformed state, we actually only need the random permutation on one side. However, in this appendix, we will maintain visual consistency by keeping the random permutations on both sides.

\vspace{0.5\baselineskip}

At this point, we can already explain at a high level why the discarding and optimal cloning maps have the desired effects in our concentration and dilution protocols, respectively:
\begin{itemize}
    \item First consider discarding, shown in Figure \ref{fig:discarding-map-Schur-transformed-state}. Because of the wraparound caused by the partial trace, each open-ended wire may now thread multiple copies of $\rho(\lambda,\hat{n})$. This means that the relative difference between the coefficients of $\ket{1}\bra{1}_{\hat{n}}$ and $\ket{0}\bra{0}_{\hat{n}}$ grows larger, which can be understood as concentration.
    \item Now consider optimal cloning, shown in Figure \ref{fig:optimal-cloning-map-Schur-transformed-state}. Because of the identity wires, each open-ended wire may now thread either a single copy of $\rho(\lambda,\hat{n})$ or none at all. This means that the relative difference between the coefficients of $\ket{1}\bra{1}_{\hat{n}}$ and $\ket{0}\bra{0}_{\hat{n}}$ grows smaller, which can be understood as dilution.
\end{itemize}

\begin{figure}
    \includegraphics[scale=0.5]{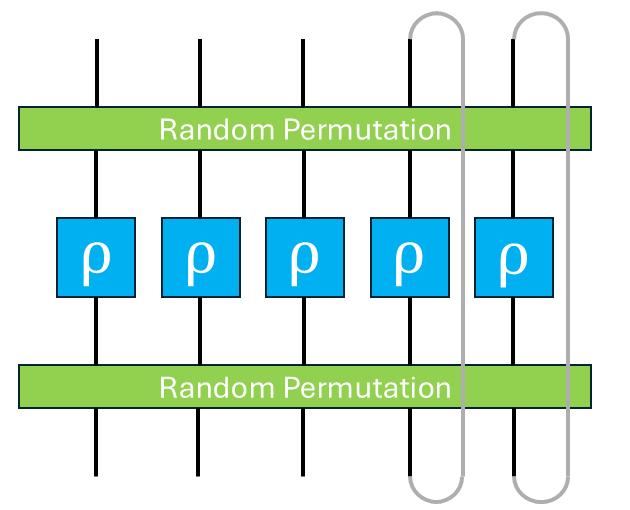}
    \includegraphics[scale=0.5]{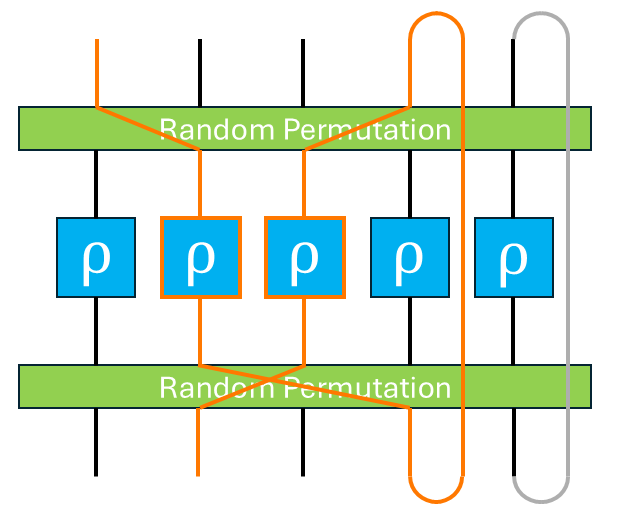}
    \caption{(LEFT) The discarding map applied to a Schur-transformed state. Instead of using a generic multi-qubit state, as we did in Figure \ref{fig:symmetric-subspace-state-discarding-map-optimal-cloning-map-general}, we now use an i.i.d. collection, which results in a Schur-transformed state (up to proportionality) after projecting to the symmetric subspace. (RIGHT) The essential intuition for why discarding qubits from the Schur-transformed state achieves concentration. Before the discarding map, each wire threads exactly one copy of $\rho$. But now, each wire can thread multiple copies of $\rho$ due to the wraparound caused by the partial trace. Since each copy of $\rho$ multiplies the coefficient of $\ket{1}\bra{1}_{\hat{n}}$ by $c_1 = \frac{1+\lambda}{2}$ and the coefficient of $\ket{0}\bra{0}_{\hat{n}}$ by $c_0 = \frac{1-\lambda}{2}$, each additional copy of $\rho$ increases the relative separation between the coefficients based on Hamming weight, which can be understood as concentration.}
    \label{fig:discarding-map-Schur-transformed-state}
\end{figure}

\begin{figure}
    \includegraphics[scale=0.35]{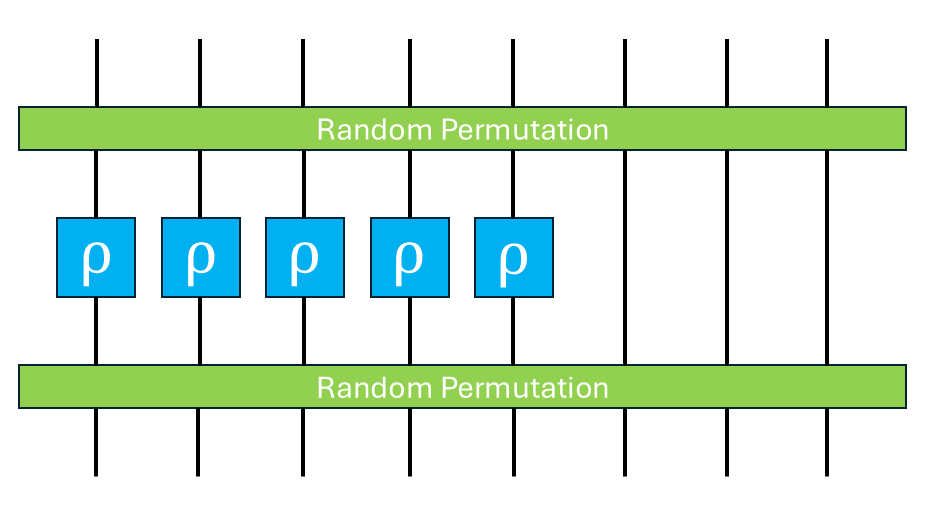}
    \includegraphics[scale=0.35]{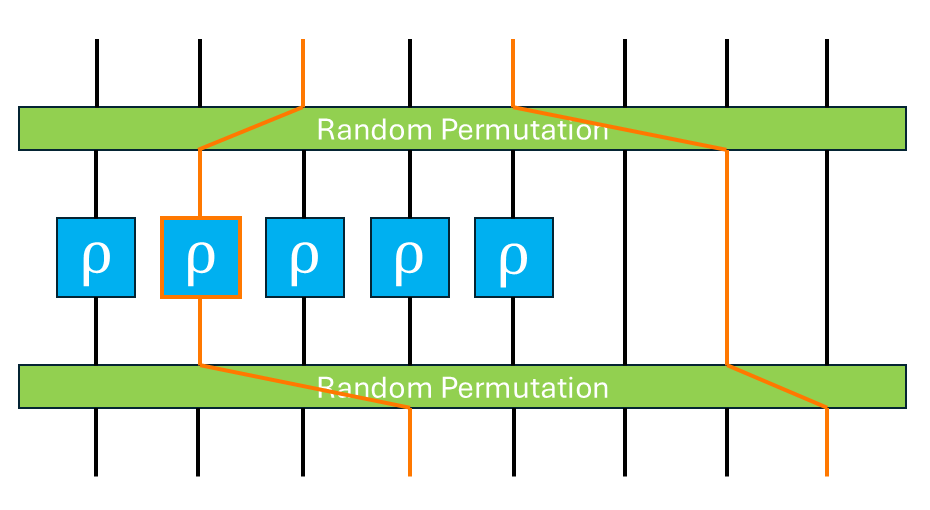}
    \caption{(LEFT) The optimal cloning map applied to a Schur-transformed state. Instead of using a generic multi-qubit state, as we did in Figure \ref{fig:symmetric-subspace-state-discarding-map-optimal-cloning-map-general}, we now use an i.i.d. collection, which results in a Schur-transformed state (up to proportionality) after projecting to the symmetric subspace. (RIGHT) The essential intuition for why applying the optimal cloning map to the Schur-transformed state achieves dilution. Before the optimal cloning map, each wire threads exactly one copy of $\rho$. But now, each wire can thread either one copy of $\rho$ or none at all. Hence the relative separation between the coefficients based on Hamming weight is now smaller, which can be understood as dilution.}
    \label{fig:optimal-cloning-map-Schur-transformed-state}
\end{figure}

\vspace{0.5\baselineskip}

Figures \ref{fig:discarding-map-Schur-transformed-state} and \ref{fig:optimal-cloning-map-Schur-transformed-state} already show roughly why discarding and optimal cloning help us achieve concentration and dilution, respectively. However, they do not yet make it clear how a specific input purity level and a specific relative change in the number of qubits yields a specific output purity level. In the next two subsections, we tackle this problem, first for dilution, and then for concentration.

\appsubsec{How Does the Optimal Cloning Map Realize Dilution?}
{subsec:heuristic-cloning-dilution}

It turns out that the easiest case to understand is the case where we dilute from pure states $(\lambda_{\text{in}}=1)$ to mixed states $(\lambda_{\text{out}} = \lambda < 1)$. This case will form the foundation for the more general dilution setting (where we may also have $\lambda_{\text{in}} < 1$) and for the concentration setting as well.

\vspace{0.5\baselineskip}

Suppose we have $M$ i.i.d. pure qubit states. In this case, the i.i.d. state and the Schur-transformed state are the same, since Schur sampling will yield $j=M/2$ with probability $1$. Now suppose that we apply the optimal cloning map to increase the number of qubits to $\tilde{M} = rM$ for some $r > 1$. The expression of the optimal cloning map shows that, up to a constant factor, this channel can be understood as introducing $(r-1)M$ maximally mixed qubits and putting the result in the symmetric subspace on $\tilde{M}$ qubits by scrambling the wires at the beginning and at the end. The result will take the form shown in Figure \ref{fig:optimal-cloning-map-initial-state-pure-counts}.

\begin{figure}
    \includegraphics[scale=0.5]{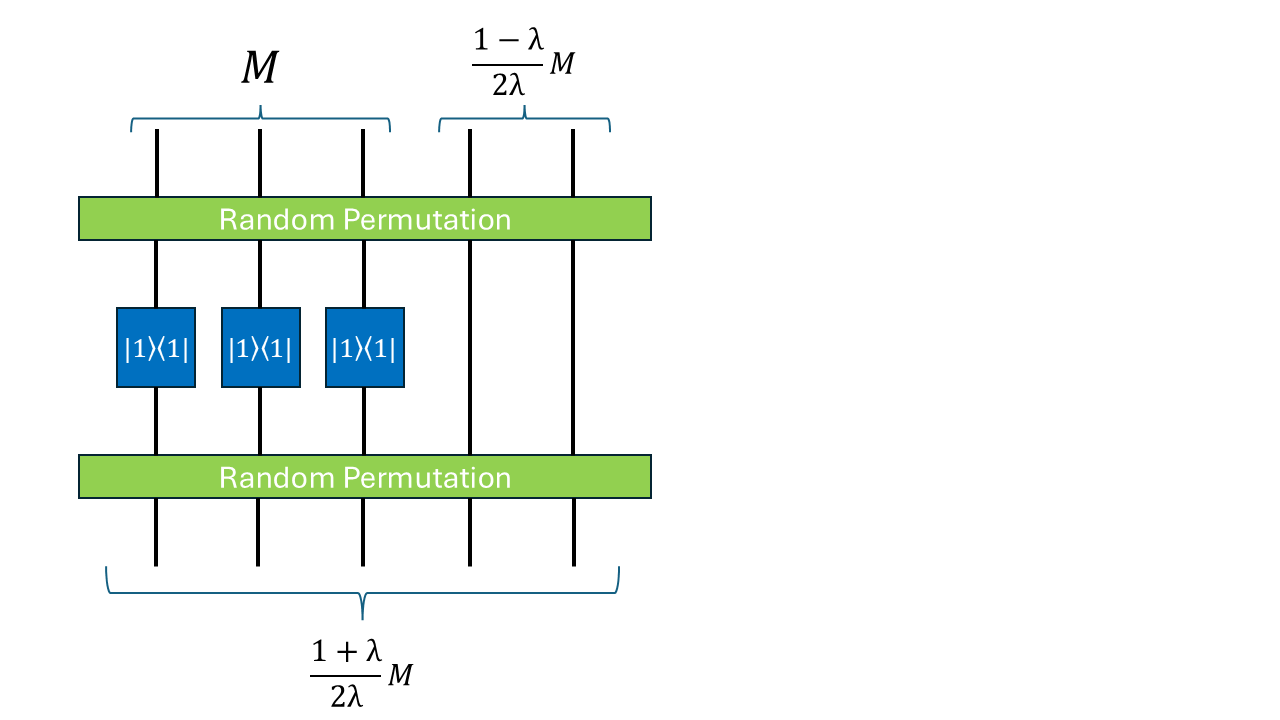}
    \caption{A tensor network representation of the cloning step of our qubit dilution procedure, in the special case where the initial qubits are pure. We begin with $M$ i.i.d. pure qubits, represented by the \textbf{\textcolor{MidnightBlue}{dark blue}} boxes. We then apply the optimal cloning map to increase the number of qubits to $rM$, where $r = \frac{1+\lambda}{2\lambda}$ has been strategically chosen based on the desired purity level $\lambda$. In the tensor network picture, this corresponds to introducing $\frac{1-\lambda}{2\lambda}M$ identity wires, and then sandwiching the result between two independent random permutations, represented by the usual \textbf{\textcolor{LimeGreen}{green}} boxes.}
    \label{fig:optimal-cloning-map-initial-state-pure-counts}
\end{figure}

\vspace{0.5\baselineskip}

Now we want to see how this operator acts on a computational basis state in the eigenbasis defined by $\hat{n}$. Notice that the random permutations, the $\ket{1}\bra{1}_{\hat{n}}$ projector wires, and the identity wires all preserve the Hamming weight. Therefore, the result will be either zero (if a $\ket{0}_{\hat{n}}$ state ever collides with a $\ket{1}\bra{1}_{\hat{n}}$ projector) or a computational basis state with the same Hamming weight (if that never happens).

\vspace{0.5\baselineskip}

Therefore, we only really need to compute the probability that a $\ket{0}_{\hat{n}}$ state never collides with a $\ket{1}\bra{1}_{\hat{n}}$ projector. This occurs if and only if all the wires with $\ket{0}_{\hat{n}}$ get mapped to the identity wires in the initial permutation. In general, this probability will be the ratio of two binomial coefficients. However, if there are a lot of zeros in the bit string, then this probability will be exponentially small anyway. As a result, we can choose to focus on the Hamming weights $\tilde{w}$ close to $\tilde{M}$, meaning that the number of zeros, given by $\tilde{\delta} \coloneqq \tilde{M} - \tilde{w}$, is relatively small.

\vspace{0.5\baselineskip}

In this case, the events where each $\ket{0}_{\hat{n}}$ is successfully mapped to one of the identity wires are roughly independent. Furthermore, each successful mapping occurs with probability $\frac{r-1}{r}$, since there are $(r-1)M$ identity wires out of $rM$ total positions. Therefore, the relative contribution of bit strings with Hamming weight $\tilde{w} = \tilde{M} - \tilde{\delta}$ is approximately $\left(\frac{r-1}{r}\right)^{\tilde{\delta}}$. As a result, the Dicke state coefficients in the operator are approximately in geometric sequence, and as we can see from Definition \ref{def:Schur-transformed-states}, this means that the operator approximates a Schur-transformed state! More specifically, we know that the common ratio between successive Dicke state coefficients is $\frac{c_0}{c_1} = \frac{1-\lambda}{1+\lambda}$. Solving for $\lambda$ in terms of $r$ and vice versa yields
\begin{equation}
    \frac{1-\lambda}{1+\lambda} = \frac{r-1}{r} \iff \boxed{\lambda = \frac{1}{2r-1}} \iff \boxed{r = \frac{1+\lambda}{2\lambda}}.
\end{equation}
We conclude that taking $M$ i.i.d. pure qubit states and applying the optimal cloning map to increase the number of qubits to $\frac{1+\lambda}{2\lambda}M$ yields approximately a Schur-transformed state with purity level $\lambda$. This is the essential fact that will provide the foundation for the more general dilution protocol, and even the concentration protocol.

\begin{figure}
    \includegraphics[scale=0.5]{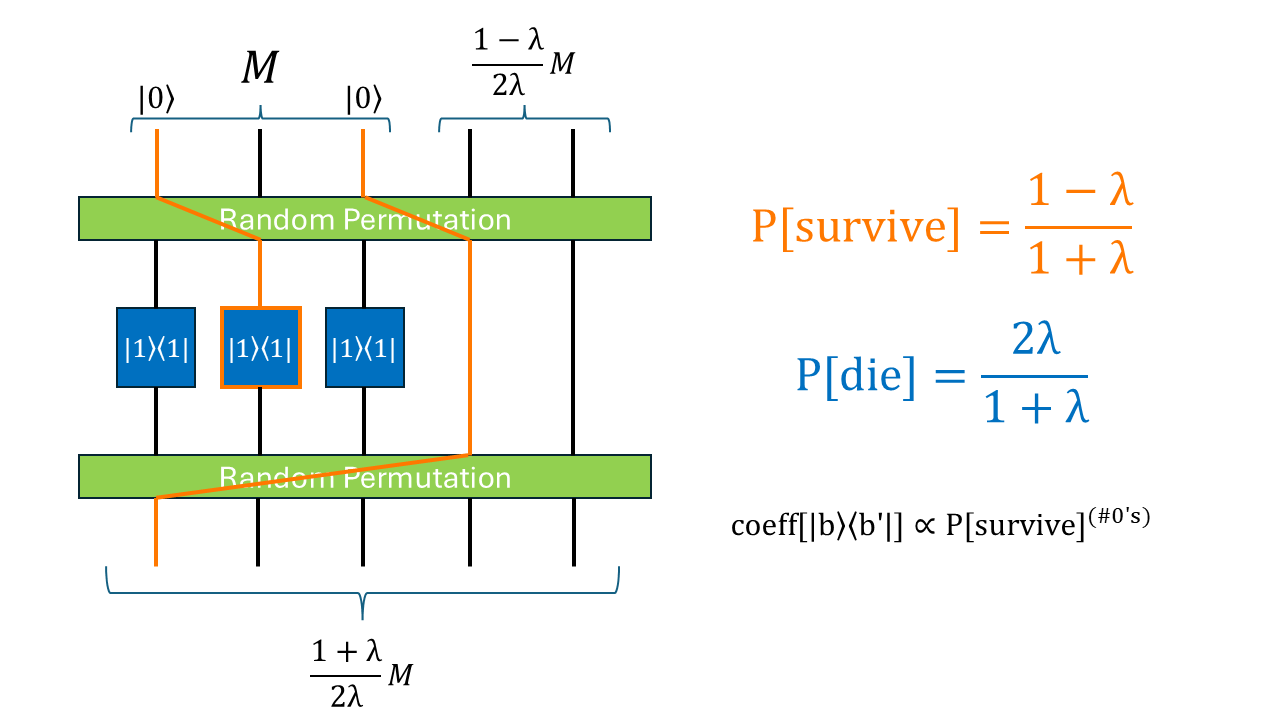}
    \caption{An explanation of why applying the optimal cloning map to $M$ i.i.d. pure qubits to increase the qubit count to $\frac{1+\lambda}{2\lambda}M$ produces approximately a Schur-transformed state with purity level $\lambda$. Consider a computational basis state $\ket{b}$ passing through this tensor network. Each $\ket{1}_{\hat{n}}$ will pass through unobstructed, but each $\ket{0}_{\hat{n}}$ will either go through an identity wire and survive or collide with a $\ket{1}\bra{1}_{\hat{n}}$ projector and be annihilated. Out of $\frac{1+\lambda}{2\lambda}M$ possible wires, there are $\frac{1-\lambda}{2\lambda}M$ identity wires, so each $\ket{0}_{\hat{n}}$ has probability $\Pbb[\text{survive}] = \frac{1-\lambda}{1+\lambda}$ of surviving. If \emph{every} $\ket{0}_{\hat{n}}$ is mapped to an identity wire, then the computational basis state $\ket{b}$ survives and becomes a new computational basis state $\ket{b'}$ with the same Hamming weight, but if even one $\ket{0}_{\hat{n}}$ is mapped to a $\ket{1}\bra{1}_{\hat{n}}$ projector, then the computational basis state is annihilated. If the number of zeros in the computational basis state is sufficiently small, then the different survival events are roughly independent, which is why the overall survival probability is approximately $\Pbb[\text{survive}]$ raised to the power of the number of zeros. (If the number of zeros is too large, then this independence fails, but those coefficients are negligible.) Hence the coefficients of the Dicke states decay approximately as a geometric sequence with common ratio $\frac{1-\lambda}{1+\lambda} = \frac{c_0}{c_1}$, which is exactly the defining property of a Schur-transformed state with purity level $\lambda$.}
    \label{fig:optimal-cloning-map-initial-state-pure-counts-probs}
\end{figure}

\vspace{0.5\baselineskip}

We are now ready to present the more general heuristic argument for our dilution protocol, which is shown in Figure \ref{fig:further-cloning-counts}. We begin with a Schur-transformed state $\rho_C(N_C,\lambda,\hat{n})$. The first step is to think of this state as itself the result of the optimal cloning map applied to a smaller number of pure qubits! The advantage of this way of thinking is that it allows us to treat every wire as either a $\ket{1}\bra{1}_{\hat{n}}$ projector (corresponding to a pure qubit) or an identity matrix (corresponding to a maximally mixed qubit), which simplifies the math. More precisely, we saw that multiplying the number of qubits by $r = \frac{1+\lambda}{2\lambda}$ reduces the purity level from $1$ to $\lambda$. Therefore, we can approximate our Schur-transformed state as the result of starting with $M = \frac{2\lambda_{\text{in}}}{1+\lambda_{\text{in}}}N_C$ i.i.d. pure qubit states and then applying the optimal cloning map to increase the number of qubits to $N_C$.

\begin{figure}
    \includegraphics[scale=0.5]{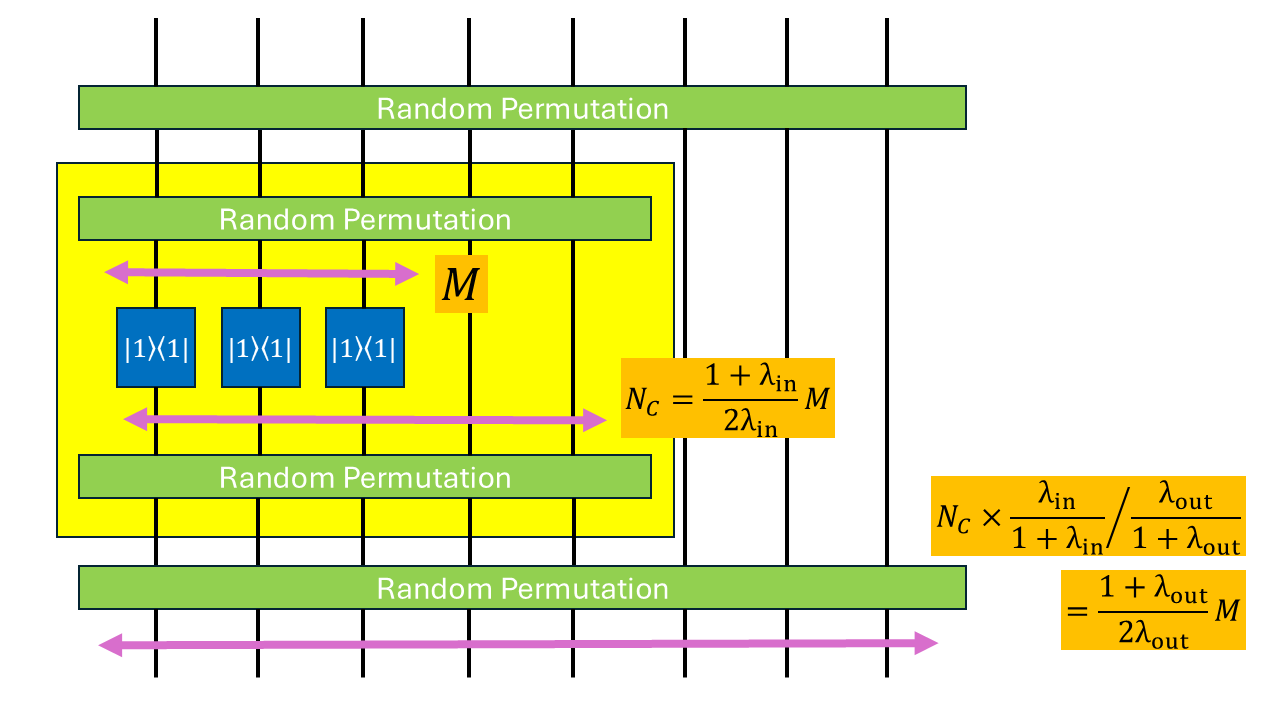}
    \includegraphics[scale=0.5]{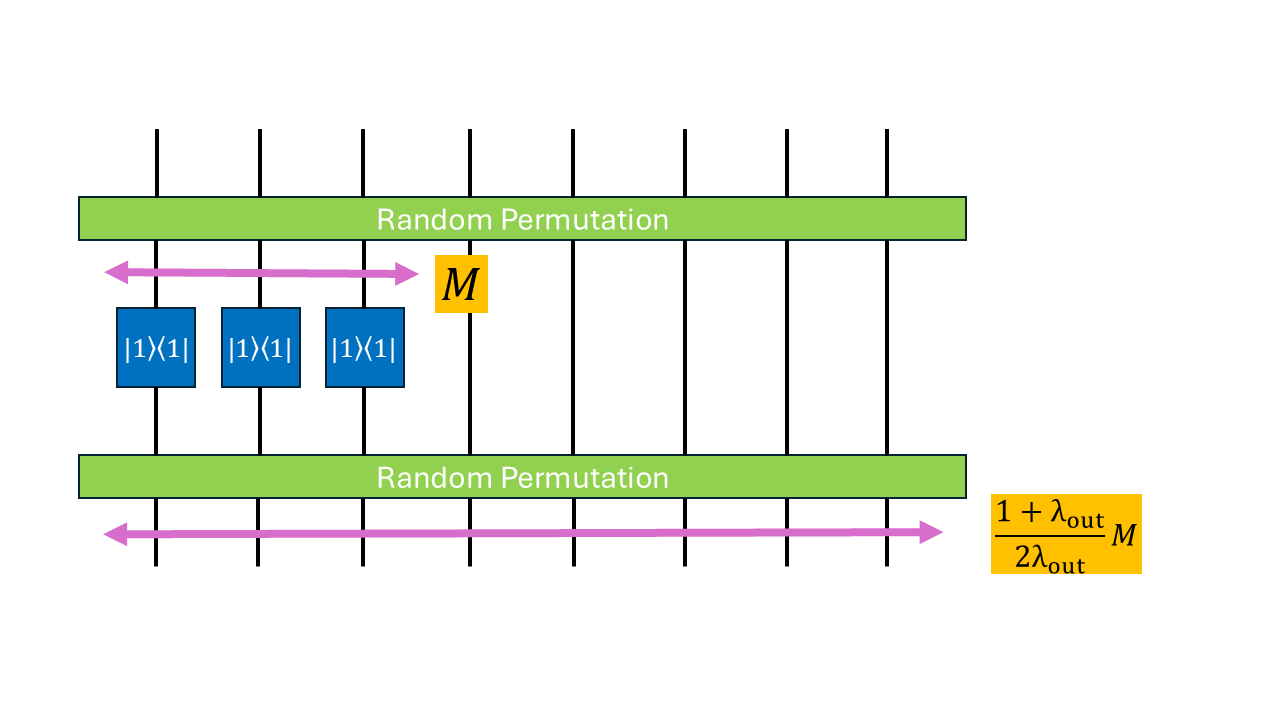}
    \caption{A tensor network representation of the cloning step of our qubit dilution procedure, in the general case where the initial qubits are mixed. We begin by treating the Schur-transformed state $\rho_C(N_C,\lambda,\hat{n})$ as the result of applying the optimal cloning map to $M = \frac{2\lambda_{\text{in}}}{1+\lambda_{\text{in}}}N_C$ i.i.d. pure qubits. This was previously shown in Figure \ref{fig:optimal-cloning-map-initial-state-pure-counts}, and the result is shown here in the \textbf{\textcolor{YellowOrange}{yellow}} box. We now apply the optimal cloning map again to increase the number of qubits to $\tilde{N}_C = R_CN_C = \frac{\lambda_{\text{in}}}{1+\lambda_{\text{in}}}\frac{1+\lambda_{\text{out}}}{\lambda_{\text{out}}}N_C = \frac{1+\lambda_{\text{out}}}{2\lambda_{\text{out}}}M$. We can actually see diagrammatically that the result is the same as though we had just applied the optimal cloning map once to go directly from $M$ qubits to $\tilde{N}_C$ qubits, since the random permutations of the wires in the first optimal cloning map can be absorbed into the random permutations of the wires in the second optimal cloning map. Since the number of qubits has been multiplied by the factor $\tilde{r} = \frac{1+\lambda_{\text{out}}}{2\lambda_{\text{out}}}$ over the total process, we conclude that the resulting state approximates the Schur-transformed state $\rho_C(\tilde{N}_C,\lambda_{\text{out}},\hat{n})$.}
    \label{fig:further-cloning-counts}
\end{figure}

\vspace{0.5\baselineskip}

As discussed in Lemma \ref{lem:optimal-cloning-map-schur-transformed-state}, we now apply the optimal cloning map to multiply the number of qubits by $R_C = \frac{\lambda_{\text{in}}}{1+\lambda_{\text{in}}}\Big/\frac{\lambda_{\text{out}}}{1+\lambda_{\text{out}}}$. However, these two uses of the optimal cloning map can be merged into just one, where we go directly from $M$ qubits to $\tilde{N}_C = R_CN_C$ qubits. We can see this diagrammatically in Figure \ref{fig:further-cloning-counts}, since the random permutations of the wires in the first optimal cloning map can be absorbed into the random permutations of the wires in the second optimal cloning map. (This holds even for qudits! The optimal cloning map behaves naturally with respect to composition. More precisely, $\mE_{\text{clone}}[M_1\to M_3] = \mE_{\text{clone}}[M_2\to M_3]\circ\mE_{\text{clone}}[M_1\to M_2]$ for any $M_1\le M_2\le M_3$.)

\vspace{0.5\baselineskip}

The net factor by which the number of qubits has increased, which we call $\tilde{r}$ for convenience, is given by
\begin{equation}
    \tilde{r} = \frac{\tilde{N}_C}{M} = \frac{R_C}{M/N_C} = \frac{\frac{\lambda_{\text{in}}}{1+\lambda_{\text{in}}}\Big/\frac{\lambda_{\text{out}}}{1+\lambda_{\text{out}}}}{\frac{2\lambda_{\text{in}}}{1+\lambda_{\text{in}}}} \\
    = \frac{1+\lambda_{\text{out}}}{2\lambda_{\text{out}}}.
\end{equation}
But notice that this is exactly the same expression as $r = \frac{1+\lambda}{2\lambda}$, except now with $\lambda_{\text{out}}$ in place of $\lambda$. Therefore, since we started with $M$ i.i.d. pure qubits and applied an optimal cloning map to increase the number of qubits by a factor of $\tilde{r} = \frac{1+\lambda_{\text{out}}}{2\lambda_{\text{out}}}$, the final result should approximate a Schur-transformed state with purity level $\lambda_{\text{out}}$, exactly as desired.

\appsubsec{How Does the Discarding Map Realize Concentration?}
{subsec:heuristic-discarding-concentration}

We now turn to the problem of qubit concentration. In particular, we want to show heuristically that discarding a constant fraction of the qubits from a Schur-transformed state produces approximately another Schur-transformed state while increasing the purity level to a new value, which will be a function of the initial purity level and the fraction of qubits we discard.

\vspace{0.5\baselineskip}

Our discarding procedure is shown diagrammatically in Figure \ref{fig:discarding-map-Schur-transformed-state-counts}. In particular, the fraction of qubits that we keep is $R_C = \frac{\lambda_{\text{in}}}{1-\lambda_{\text{in}}}\Big/\frac{\lambda_{\text{out}}}{1-\lambda_{\text{out}}}$. Some intuition for why the discarding map achieves concentration was already offered earlier by Figure \ref{fig:discarding-map-Schur-transformed-state}. To make this idea more precise, we once again invoke the key result that we showed in Appendix \ref{sec:heuristic-arguments}\ref{subsec:heuristic-cloning-dilution}, namely, that the Schur-transformed state can be approximated by applying the optimal cloning map to i.i.d. pure qubits. Just as in the general dilution setting, this allows us to treat every wire as either a $\ket{1}\bra{1}_{\hat{n}}$ projector or an identity matrix.

\begin{figure}
    \includegraphics[scale=0.5]{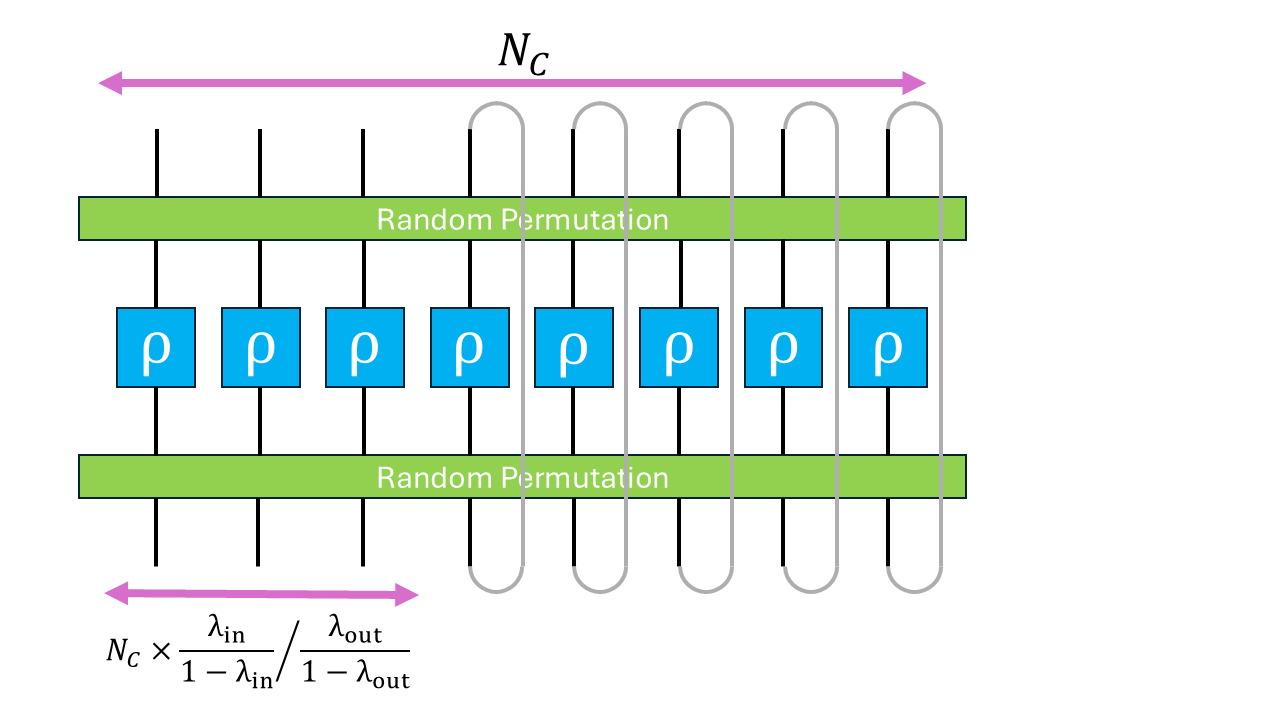}
    \caption{A tensor network representation of the discarding step of our qubit concentration procedure. We begin with the Schur-transformed state on $N_C$ qubits, which is (up to a constant factor) a tensor product state sandwiched by independent random permutations. We then keep a cleverly chosen fraction $\frac{\lambda_{\text{in}}}{1-\lambda_{\text{in}}}\Big/\frac{\lambda_{\text{out}}}{1-\lambda_{\text{out}}}$ of the qubits and discard (trace out) the rest.}
    \label{fig:discarding-map-Schur-transformed-state-counts}
\end{figure}

\vspace{0.5\baselineskip}

In this case, we begin with the Schur-transformed state $\rho_C(N_C,\lambda_{\text{in}},\hat{n})$. As we showed in Appendix \ref{sec:heuristic-arguments}\ref{subsec:heuristic-cloning-dilution}, this can be approximated by taking $M = \frac{2\lambda_{\text{in}}}{1+\lambda_{\text{in}}}N_C$ i.i.d. pure qubit states $\ket{1}\bra{1}_{\hat{n}}$ and applying the optimal cloning map to increase the number of qubits to $N_C$. We now apply the discarding map to this setup; as we state in Lemma \ref{lem:discarding-map-schur-transformed-state}, we keep a fraction $R_C = \frac{\lambda_{\text{in}}}{1-\lambda_{\text{in}}}\Big/\frac{\lambda_{\text{out}}}{1-\lambda_{\text{out}}}$ of the qubits and discard the rest. The final number of qubits is thus $\tilde{N}_C = R_CN_C$. The resulting tensor network is shown in Figure \ref{fig:discarding-map-post-cloning-state-counts}.

\begin{figure}
    \includegraphics[scale=0.5]{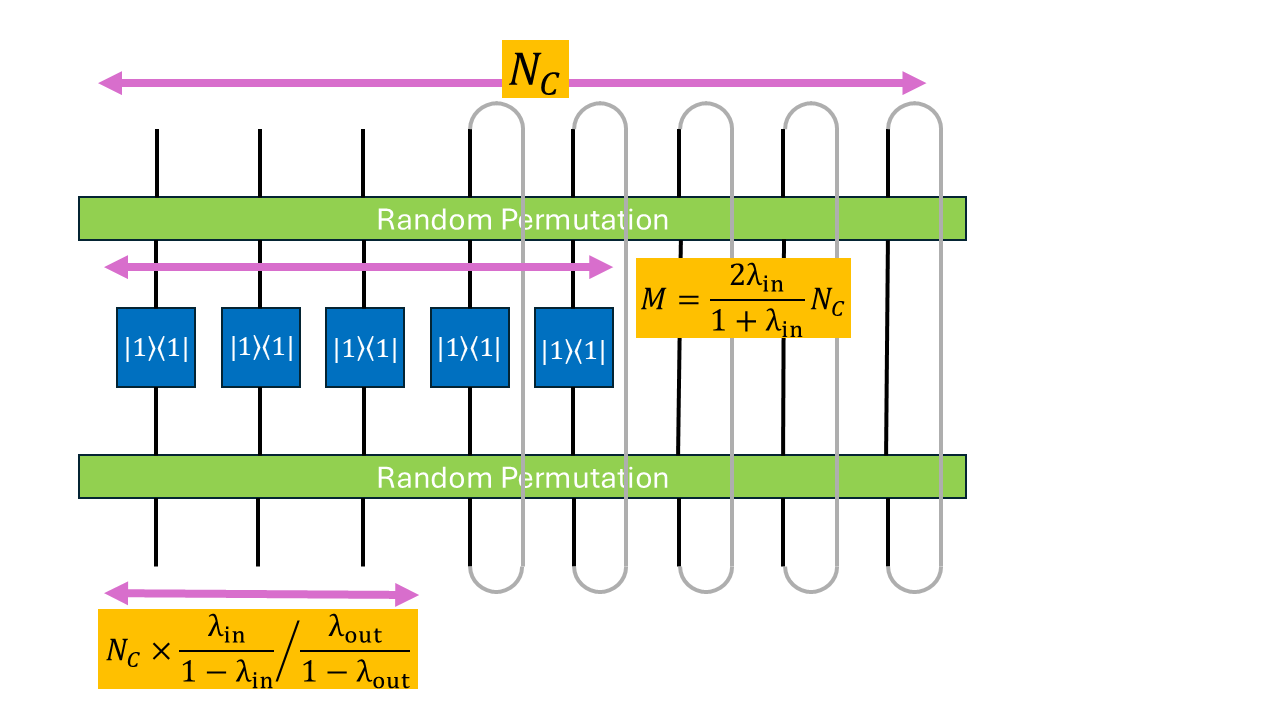}
    \caption{A tensor network representation of the discarding step of our qubit concentration procedure. The only change from Figure \ref{fig:discarding-map-Schur-transformed-state-counts} is that we now represent the initial Schur-transformed state as itself the result of applying the optimal cloning map to a collection of i.i.d. pure qubit states.}
    \label{fig:discarding-map-post-cloning-state-counts}
\end{figure}

\vspace{0.5\baselineskip}Just as for the dilution protocol, we want to see how the operator shown in Figure \ref{fig:discarding-map-post-cloning-state-counts} acts on a computational basis state in the eigenbasis defined by $\hat{n}$. For the same reason as before, the result will be either zero (if a $\ket{0}_{\hat{n}}$ state ever collides with a $\ket{1}\bra{1}_{\hat{n}}$ projector) or a computational basis state with the same Hamming weight (if that never happens). However, unlike for the dilution protocol, each $\ket{0}_{\hat{n}}$ might have to pass through multiple positions before potentially making it out, since it might get mapped by the random permutation to a discarded position, whose wire will go back around to the input side due to the partial trace.

\vspace{0.5\baselineskip}

Once again, we only really need to compute the probability that a $\ket{0}_{\hat{n}}$ state never collides with a $\ket{1}\bra{1}_{\hat{n}}$ projector. In general, this probability may be somewhat complicated. However, if there are a lot of zeros in the bit string, then this probability will be exponentially small anyway. As a result, we can once again choose to focus on the Hamming weights $\tilde{w}$ close to $\tilde{N}_C$, meaning that the number of zeros, given by $\tilde{\delta} \coloneqq \tilde{N}_C - \tilde{w}$, is relatively small.

\vspace{0.5\baselineskip}

In this setting, each wire with a $\ket{0}_{\hat{n}}$ can be treated roughly independently. Furthermore, each such wire can be composed of multiple roughly independent ``trials'', where on each trial, it hits either a $\ket{1}\bra{1}_{\hat{n}}$ projector or an identity, and then subsequently gets permuted to either a non-discarded position (in which case, it terminates at the output end) or a discarded position (in which case, it wraps around to the input end, and a new trial commences).

\vspace{0.5\baselineskip}

In other words, each $\ket{0}_{\hat{n}}$ has to play a ``game'' with potentially multiple rounds, where each round proceeds as follows. First, it survives with probability $\Pbb[\text{survive}]$ and dies with probability $\Pbb[\text{die}] = 1 - \Pbb[\text{survive}]$. Afterward, if it survives, it escapes with probability $\Pbb[\text{escape}]$, but with probability $\Pbb[\text{reset}] = 1 - \Pbb[\text{escape}]$, the game resets and it starts a new round. The ``survive vs. die'' and ``escape vs. reset'' events within each round are independent, and different rounds are independent as well. If the $\ket{0}_{\hat{n}}$ escapes without dying, it ``wins'', but if you it ever collides with a $\ket{1}\bra{1}_{\hat{n}}$ projector and dies, the game ends and it ``loses''. Refer to Figure \ref{fig:markov-chain-wire-trajectory-win-lose} for an example of a winning instance and an example of a losing instance.

\vspace{0.5\baselineskip}

Given this somewhat complicated setup, what is the overall probability of winning for each $\ket{0}_{\hat{n}}$? The answer turns out to be
\begin{equation}
    \Pbb[\text{win}] = \frac{\Pbb[\text{survive}]\Pbb[\text{escape}]}{\Pbb[\text{die}] + \Pbb[\text{survive}]\Pbb[\text{escape}]}.
\end{equation}
There are several ways to determine this, but the most popular is to study the game as a Markov chain. At the start of a round, you can die to automatically enter a losing state, or you can survive to enter a ``middle'' state. If you reach the middle state, you can either escape to enter a winning state, or you can reset back to the starting state, in which case your probability of winning is once again $\Pbb[\text{win}]$. This Markov chain is shown in Figure \ref{fig:markov-chain-reasoning}. Therefore, you can set up the equation
\begin{equation}
    \Pbb[\text{win}] = \Pbb[\text{survive}]\left(\Pbb[\text{escape}] + \Pbb[\text{reset}]\Pbb[\text{win}]\right)
\end{equation}
and solve it for $\Pbb[\text{win}]$ to obtain the result shown above.

\begin{figure}
    \includegraphics[scale=0.35]{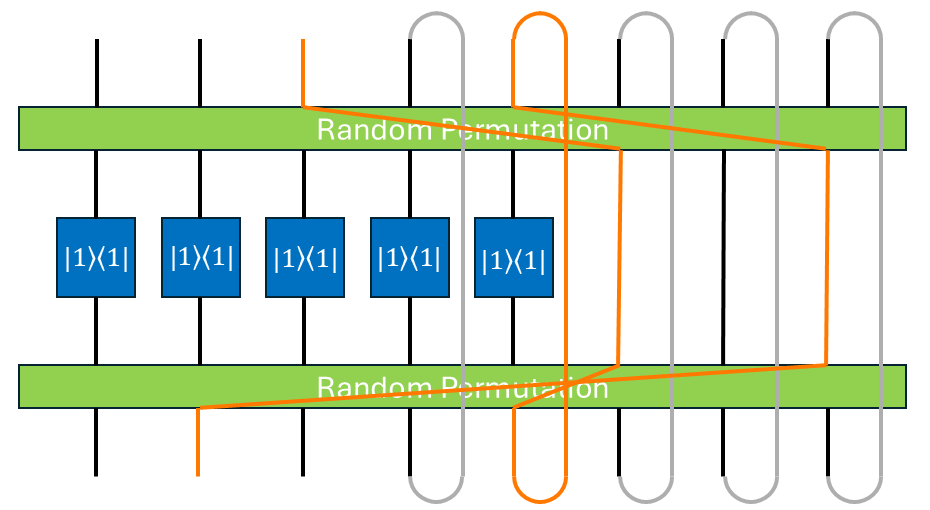}
    \includegraphics[scale=0.35]{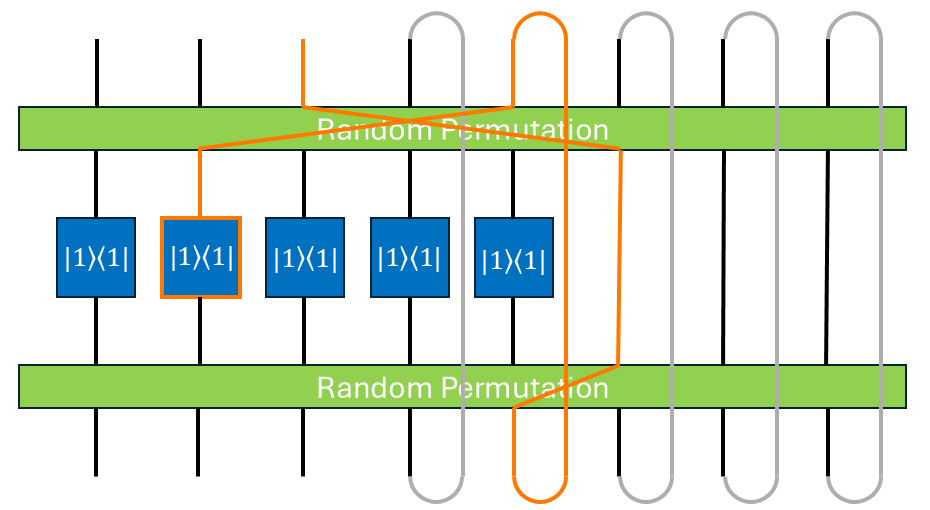}
    \caption{(LEFT) An example of a trajectory where the $\ket{0}_{\hat{n}}$ state ``wins'' by successfully avoiding a $\ket{1}\bra{1}_{\hat{n}}$ projector. (RIGHT) An example of a trajectory where the $\ket{0}_{\hat{n}}$ state ``loses'' by colliding into a $\ket{1}\bra{1}_{\hat{n}}$ projector and being annihilated.}
    \label{fig:markov-chain-wire-trajectory-win-lose}
\end{figure}

\vspace{0.5\baselineskip}

The only remaining task is to compute the probabilities $\Pbb[\text{survive}]$, $\Pbb[\text{die}]$, $\Pbb[\text{escape}]$, $\Pbb[\text{reset}]$ that go into the formula above. First, the ``survive vs. die'' event depends solely on the relative counts of identity wires and $\ket{1}\bra{1}_{\hat{n}}$ projectors. We thus obtain
\begin{equation}
    \Pbb[\text{die}] = \frac{M}{N_C} = \frac{2\lambda_{\text{in}}}{1+\lambda_{\text{in}}}, \quad \Pbb[\text{survive}] = 1 - \Pbb[\text{die}] = \frac{1-\lambda_{\text{in}}}{1+\lambda_{\text{in}}}.
\end{equation}
Second, the ``escape vs. reset'' event depends solely on the relative counts of non-discarded positions (wires with a free end at the bottom) and discarded positions (wires that get wrapped back around to the top). We thus obtain
\begin{equation}
    \Pbb[\text{escape}] = R_C = \frac{\lambda_{\text{in}}}{1-\lambda_{\text{in}}}\Big/\frac{\lambda_{\text{out}}}{1-\lambda_{\text{out}}}, \quad \Pbb[\text{reset}] = 1 - \Pbb[\text{escape}].
\end{equation}
Plugging these values into the above formula yields
\begin{align}
    \Pbb[\text{win}] &= \frac{\Pbb[\text{survive}]\Pbb[\text{escape}]}{\Pbb[\text{die}] + \Pbb[\text{survive}]\Pbb[\text{escape}]} \\
    &= \frac{\left(\frac{1-\lambda_{\text{in}}}{1+\lambda_{\text{in}}}\right)\left(\frac{\lambda_{\text{in}}}{1-\lambda_{\text{in}}}\Big/\frac{\lambda_{\text{out}}}{1-\lambda_{\text{out}}}\right)}{\frac{2\lambda_{\text{in}}}{1+\lambda_{\text{in}}} + \left(\frac{1-\lambda_{\text{in}}}{1+\lambda_{\text{in}}}\right)\left(\frac{\lambda_{\text{in}}}{1-\lambda_{\text{in}}}\Big/\frac{\lambda_{\text{out}}}{1-\lambda_{\text{out}}}\right)} \\
    &= \frac{\left(\frac{\lambda_{\text{in}}}{1+\lambda_{\text{in}}}\right)\left(\frac{1-\lambda_{\text{out}}}{\lambda_{\text{out}}}\right)}{\left(\frac{\lambda_{\text{in}}}{1+\lambda_{\text{in}}}\right)\left(2 + \frac{1-\lambda_{\text{out}}}{\lambda_{\text{out}}}\right)} \\
    &= \frac{\frac{1-\lambda_{\text{out}}}{\lambda_{\text{out}}}}{2 + \frac{1-\lambda_{\text{out}}}{\lambda_{\text{out}}}} \\
    &= \frac{1-\lambda_{\text{out}}}{1+\lambda_{\text{out}}}.
\end{align}
We get one such factor for each $\ket{0}_{\hat{n}}$ in the original computational basis state. As a result, the relative contribution of bit strings with Hamming weight $\tilde{w} = \tilde{N}_C - \tilde{\delta}$ is approximately $\left(\frac{1-\tilde{\lambda}}{1+\tilde{\lambda}}\right)^{\tilde{\delta}}$. As a result, the Dicke state coefficients in the operator are approximately in geometric sequence, and as we discussed in Appendix \ref{sec:schur-sampling-commentary}, this means that the operator approximates a Schur-transformed state! More specifically, we know that, for a Schur-transformed state with purity level $\lambda$, the common ratio between successive Dicke state coefficients is $\frac{1-\lambda}{1+\lambda}$. Sure enough, as shown above, the common ratio is $\frac{1-\lambda_{\text{out}}}{1+\lambda_{\text{out}}}$, which means that the purity level is none other than $\lambda_{\text{out}}$. We conclude that the resulting tensor network approximates a Schur-transformed state with purity level $\lambda_{\text{out}}$, exactly as desired.

\begin{figure}
    \includegraphics[scale=0.25]{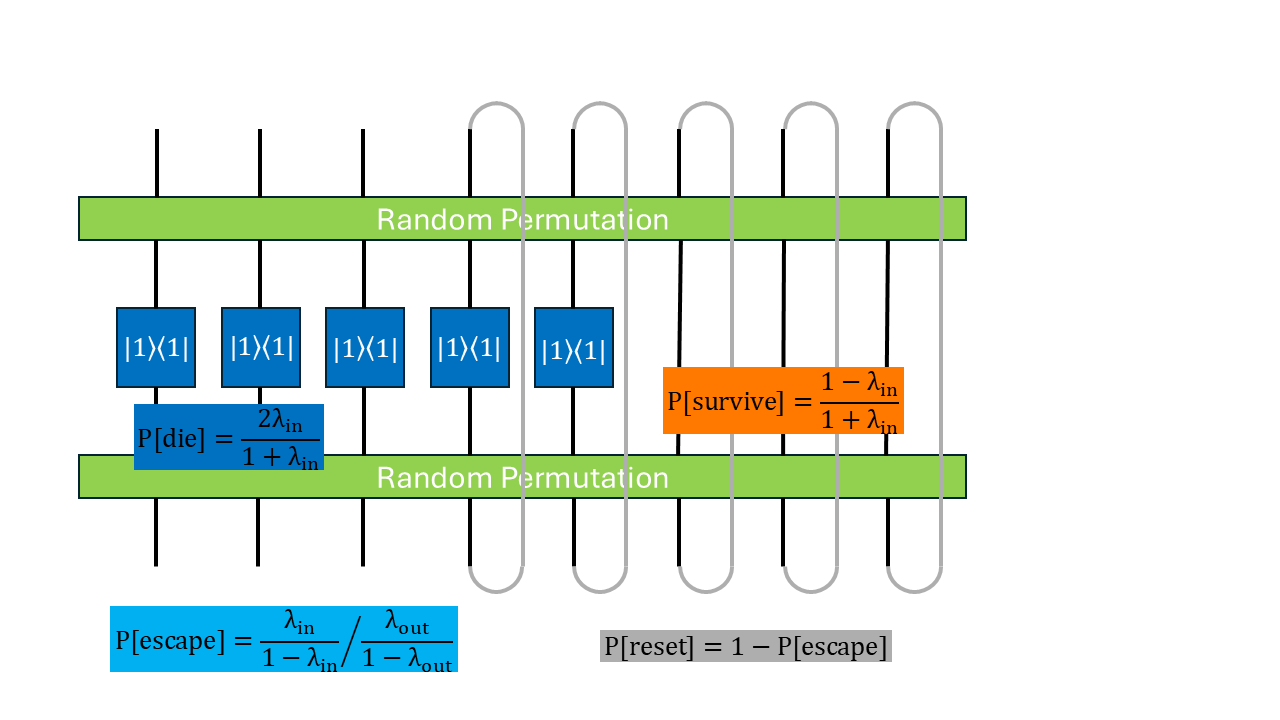}
    \includegraphics[scale=0.25]{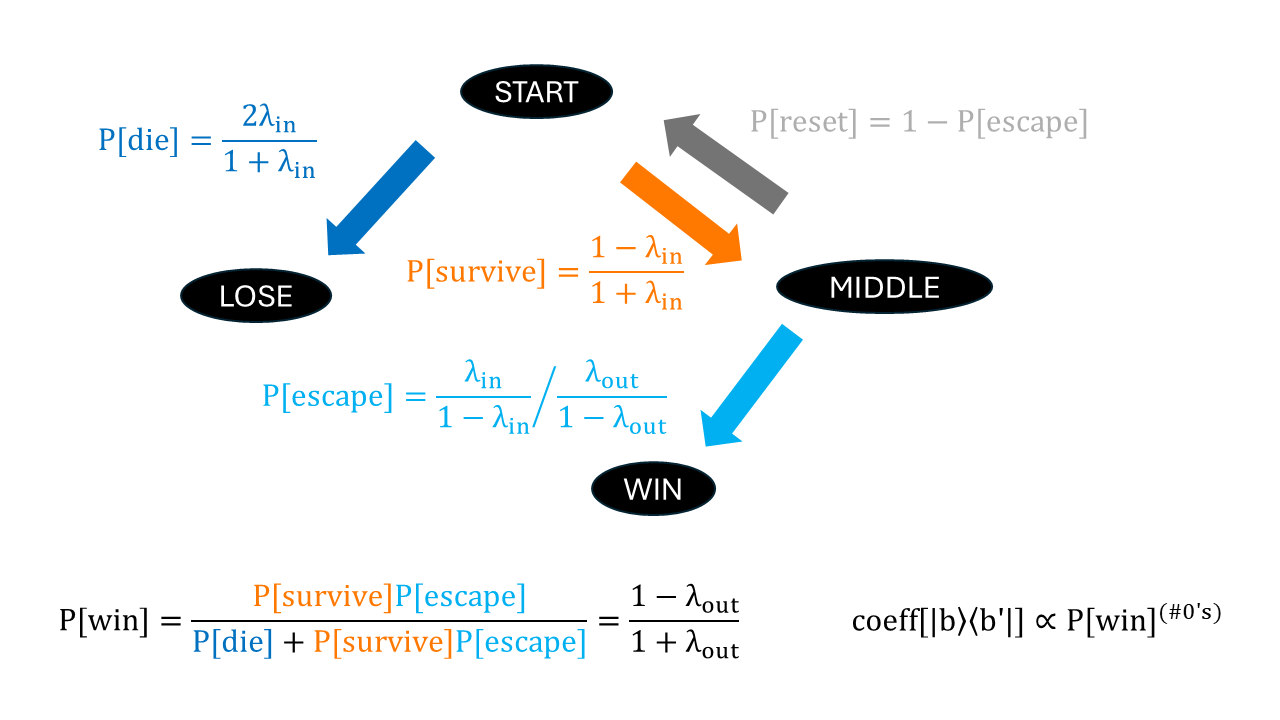}
    \caption{(LEFT) Each $\ket{0}_{\hat{n}}$ on the input end of one of the wires can be understood as playing a game. Each round proceeds as follows. At the beginning, the $\ket{0}_{\hat{n}}$ must be mapped by the first permutation to one of the identity wires, which occurs with probability $\Pbb[\text{survive}] = \frac{1-\lambda_{\text{in}}}{1+\lambda_{\text{in}}}$. Afterward, if the second permutation maps the $\ket{0}_{\hat{n}}$ to one of the wires that does not wrap around, which occurs with probability $\Pbb[\text{escape}] = \frac{\lambda_{\text{in}}}{1-\lambda_{\text{in}}}\Big/\frac{\lambda_{\text{out}}}{1-\lambda_{\text{out}}}$, it ``wins'' the game. However, if the second permutation maps the $\ket{0}_{\hat{n}}$ to one of the wires that does wrap around, then the game resets, and a new round begins. (RIGHT) The path that each $\ket{0}_{\hat{n}}$ must traverse can be approximately simplified into a Markov chain. A standard calculation shows that the overall probability of winning this game is $\Pbb[\text{win}] = \frac{1-\lambda_{\text{out}}}{1+\lambda_{\text{out}}}$. Based on our derivation from Figure \ref{fig:optimal-cloning-map-initial-state-pure-counts-probs}, this is exactly what we expect from a Schur-transformed state with the target purity level $\lambda_{\text{out}}$.}
    \label{fig:markov-chain-reasoning}
\end{figure}

\vspace{0.5\baselineskip}

Technically, there is one other issue that we have not yet addressed, which is the possibility of closed loops, which will produce constant factors. If a closed loop passes through even one non-discarded position, it will produce a constant factor of $\text{Tr}[\ket{1}\bra{1}_{\hat{n}}] = 1$. However, if a closed loop only ever passes through discarded positions, it will produce a constant factor of $\text{Tr}[\Ibb] = 2$.

\vspace{0.5\baselineskip}

Fortunately, as long as there are only a small number of $\ket{0}_{\hat{n}}$ states in the computational basis state (and we know that these are the only states that contribute substantially), the number of rounds that each $\ket{0}_{\hat{n}}$ plays in its ``game'' as described previously does not have a substantial effect on the number of discarded positions that are free to partake in a closed loop.

\vspace{0.5\baselineskip}

The probability that a closed loop only passes through discarded positions decays exponentially with the length of the loop, since there is each position is mapped to a discarded position with probability $\Pbb[\text{reset}] = 1 - \Pbb[\text{escape}] = 1 - R_C$. Therefore, the expected number of loops that only pass through discarded positions is only $\Theta(1)$, even with $\Theta(N)$ total qubits. Hence, the expected number of such loops will not be substantially affected by the the number of rounds that each $\ket{0}_{\hat{n}}$ plays in its ``game'', since we assume that there are very few zeros, each playing very few rounds, or else that computational basis state will be heavily suppressed.

\vspace{0.5\baselineskip}

Nonetheless, the annoyance of technical details such as these is why the act of turning these heuristic arguments directly into formal proofs would be somewhat messy and unappealing. As a result, we found it a lot cleaner to keep the formal proofs of Lemmas \ref{lem:discarding-map-schur-transformed-state} and \ref{lem:optimal-cloning-map-schur-transformed-state} as they were originally derived, rather than recasting them in the language of these heuristic arguments using tensor network diagrams.

\appsubsec{Discarding One Qubit at a Time}
{subsec:heuristic-discarding-one-qubit}

As one additional heuristic, we will consider what happens if we add or discard just a single qubit. Of course, the discarding (optimal cloning) step of our qubit concentration (dilution) procedure can be implemented by just repeatedly discarding (adding) a single qubit. We will see how the ``geometric series'' property that defines a Schur-transformed state is approximately maintained, and how the purity level gets nudged slightly upward (for discarding) or downward (for optimal cloning) by just the right amount. In this subsection, we study the act of discarding a single qubit from a Schur-transformed state, and in Appendix \ref{sec:heuristic-arguments}\ref{subsec:heuristic-adding-one-qubit}, we study the act of adding a single qubit to a Schur-transformed state via the optimal cloning map.

\vspace{0.5\baselineskip}

In the Schur-transformed state $\rho(N_C,\lambda,\hat{n})$, the coefficients of the Dicke states $\ket{D^{(N_C)}_w}\bra{D^{(N_C)}_w}_{\hat{n}}$, in decreasing order of Hamming weight, look as follows (up to scaling to ensure that the state has unit trace):
\begin{align}
    N_C &\mapsto c_1^{N_C} \\
    N_C-1 &\mapsto c_1^{N_C-1}c_0 \\
    N_C-2 &\mapsto c_1^{N_C-2}c_0^2 \\
    &\vdots \\
    1 &\mapsto c_1c_0^{N_C-1} \\
    0 &\mapsto c_0^{N_C}.
\end{align}
As a general formula, we can write
\begin{equation}
    N_C-\delta \mapsto c_1^{N_C-\delta}c_0^\delta.
\end{equation}
The essential point is that the coefficients form a geometric sequence with common ratio $\frac{c_0}{c_1} = \frac{1-\lambda}{1+\lambda}$.

\vspace{0.5\baselineskip}

\begin{figure}
    \includegraphics[scale=0.5]{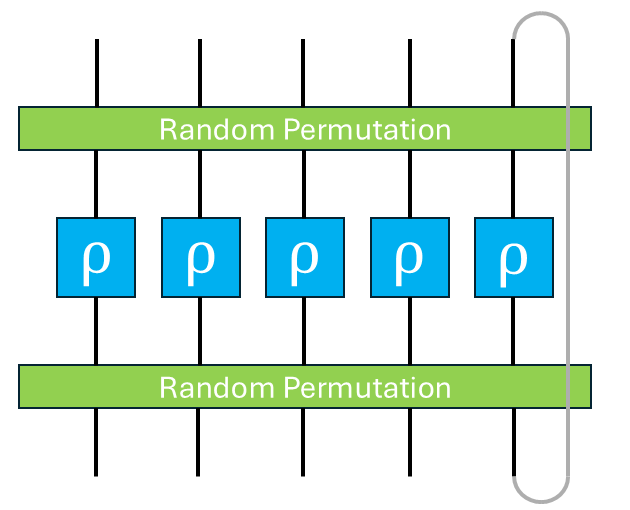}
    \caption{The result of discarding a single qubit from a Schur-transformed state. Suppose you have a bit string with $\tilde{\delta}$ zeros and $N_C-1-\tilde{\delta}$ ones. Each copy of $\rho$ multiplies the coefficient by $c_1$ if it receives a $1$ and by $c_0$ if it receives a $0$. Furthermore, the partial trace wire can either contribute an extra $c_0$ factor (if a $0$ is mapped to the last position), an extra $c_1$ factor (if a $1$ is mapped to the last position), or an extra $c_1 + c_0 = 1$ factor (if the last position is mapped to itself). As a result, up to proportionality, the coefficient of this bit string is multiplied by $(N_C-\tilde{\delta})c_1^{N_C-\tilde{\delta}}c_0^\delta + (\tilde{\delta}+1)c_1^{N_C-1-\tilde{\delta}}c_0^{\tilde{\delta}+1}$.}
    \label{fig:single-qubit-discarding}
\end{figure}

Now watch what happens when we discard a single qubit. Without loss of generality, suppose we keep qubits $1$ through $N_C-1$ and discard qubit $N_C$. Refer to the tensor network diagram shown in Figure \ref{fig:single-qubit-discarding}. Suppose you have a bit string with $\tilde{\delta}$ zeros and $N_C-1-\tilde{\delta}$ ones. Each copy of $\rho$ multiplies the coefficient by $c_1 = \frac{1+\lambda}{2}$ if it receives a $1$ and by $c_0 = \frac{1-\lambda}{2}$ if it receives a $0$.

\vspace{0.5\baselineskip}

Furthermore, each position from $1$ to $N_C$ gets mapped to the final position in a $1/N_C$ fraction of the permutations. Therefore, we need to do casework on which position gets mapped to position $N_C$ and thus gets looped around by the partial trace wire:
\begin{itemize}
    \item With probability $\tilde{\delta}/N_C$, a $0$ gets mapped to this position. Since this $0$ passes through two copies of $\rho$, it contributes a multiplier of $c_0^2$ instead of the usual $c_0$.
    \item With probability $(N_C-1-\tilde{\delta})/N_C$, a $1$ gets mapped to this position. Since this $1$ passes through two copies of $\rho$, it contributes a multiplier of $c_1^2$ instead of the usual $c_0$.
    \item With probability $1/N_C$, this position is mapped to itself. In this case, this wire produces a multiplier of $\text{Tr}[\rho] = c_1 + c_0 = 1$.
\end{itemize}
As a result, up to proportionality, the coefficient of this bit string is multiplied by
\begin{align}
    & (N_C-1-\tilde{\delta})c_1^{N_C-2-\tilde{\delta}}c_0^\delta(c_1^2) + (\tilde{\delta})c_1^{N_C-1-\tilde{\delta}}c_0^{\tilde{\delta}-1}(c_0^2) + (1)c_1^{N_C-1-\tilde{\delta}}c_0^{\tilde{\delta}}(c_0+c_1) \\
    = \,\, & (N_C-\tilde{\delta})c_1^{N_C-\tilde{\delta}}c_0^\delta + (\tilde{\delta}+1)c_1^{N_C-1-\tilde{\delta}}c_0^{\tilde{\delta}+1} \\
    = \,\, & c_1^{N_C-1-\tilde{\delta}}c_0^\delta\left[(N_C-\tilde{\delta})c_1 + (\tilde{\delta}+1)c_0\right].
\end{align}
Therefore, as a general formula, we can write
\begin{equation}
    N_C-1-\tilde{\delta} \mapsto c_1^{N_C-1-\tilde{\delta}}c_0^{\tilde{\delta}}\left[(N_C-\tilde{\delta})c_1 + (\tilde{\delta}+1)c_0\right].
\end{equation}
Therefore, just to write everything out more explicitly, we observe that the coefficients of the Dicke states, in decreasing order of Hamming weight, look as follows (again, up to scaling):
\begin{align}
    N_C-1 &\mapsto c_1^{N_C-1}\left[N_Cc_1 + c_0\right] \\
    N_C-2 &\mapsto c_1^{N_C-2}c_0\left[(N_C-1)c_1 + 2c_0\right] \\
    N_C-3 &\mapsto c_1^{N_C-3}c_0^2\left[(N_C-2)c_1 + 3c_0\right] \\
    & \vdots \\
    1 &\mapsto c_1c_0^{N_C-2}\left[2c_1 + (N_C-1)c_0\right] \\
    0 &\mapsto c_0^{N_C-1}\left[c_1 + Nc_0\right].
\end{align}
As a result, the ratio between successive coefficients now takes the form
\begin{align}
    \frac{\text{coeff}\left[N_C-1-(\tilde{\delta}+1)\right]}{\text{coeff}\left[N_C-1-\tilde{\delta}\right]} &= \frac{c_1^{N_C-2-\tilde{\delta}}c_0^{\tilde{\delta}+1}}{c_1^{N_C-1-\tilde{\delta}}c_0^{\tilde{\delta}}}\frac{(N_C-1-\tilde{\delta})c_1 + (\tilde{\delta}+2)c_0}{(N_C-\tilde{\delta})c_1 + (\tilde{\delta}+1)c_0} \\
    &= \frac{1-\lambda}{1+\lambda}\frac{(1+\lambda)N_C - (2\lambda\tilde{\delta}+3\lambda-1)}{(1+\lambda)N_C - (2\lambda\tilde{\delta}+\lambda-1)} \\
    &= \frac{1-\lambda}{1+\lambda}\left[1 - \frac{2\lambda}{1+\lambda}\frac{1}{N_C} + O\left(\frac{\tilde{\delta}}{N_C^2}\right)\right].
\end{align}
The upshot is that the decay ratio from one coefficient to the next becomes ever so slightly smaller, meaning that the remaining qubits become ever so slightly concentrated.

\vspace{0.5\baselineskip}

It is not an accident that this approximation is only good for small values of $\tilde{\delta}$, as we saw in the formal proof of Lemma \ref{lem:discarding-map-schur-transformed-state} and in the numerics in Appendix \ref{sec:numerical-analysis}. Fortunately, because the coefficients are exponentially decaying, almost all of the weight is concentrated in these small values of $\tilde{\delta}$, so the error between this state and a true Schur-transformed state is very small. Just as we observed in the formal proof, so too can we observe in the heuristic argument above that these coefficients are negligible for $\tilde{\delta} = \Omega(N_C^\varepsilon)$ for any $\varepsilon > 0$.

\vspace{0.5\baselineskip}

Since the ratio between consecutive coefficients should be $(1-\tilde{\lambda})/(1+\tilde{\lambda})$, where $\tilde{\lambda}$ is the new (ever so slightly higher) purity level, we can approximately solve for a function of $\tilde{\lambda}$ as follows:
\begin{align}
    \frac{1-\tilde{\lambda}}{1+\tilde{\lambda}} &= \frac{1-\lambda}{1+\lambda}\left[1 - \frac{2\lambda}{1+\lambda}\frac{1}{N_C} + O\left(N_C^{-2+\varepsilon}\right)\right] \\
    \implies \frac{1+\tilde{\lambda}}{1-\tilde{\lambda}} &= \frac{1+\lambda}{1-\lambda}\left[1 + \frac{2\lambda}{1+\lambda}\frac{1}{N_C} + O\left(N_C^{-2+\varepsilon}\right)\right] \\
    \therefore \frac{\tilde{\lambda}}{1-\tilde{\lambda}} &= \frac{1}{2}\left(1 + \frac{1+\tilde{\lambda}}{1-\tilde{\lambda}}\right) \\
    &= \frac{1}{2}\left[1 + \frac{1+\lambda}{1-\lambda} + \frac{2\lambda}{1-\lambda}\frac{1}{N_C} + O\left(N_C^{-2+\varepsilon}\right)\right] \\
    &= \frac{\lambda}{1-\lambda}\left[1 + \frac{1}{N_C} + O\left(N_C^{-2+\varepsilon}\right)\right].
\end{align}
This may not seem remarkable, but notice what happens when we compute $\text{RLD}_{\text{max}}$ associated to a Schur-transformed state with this purity level. Recall that $\text{RLD}_{\text{max}}(\lambda) = \frac{\lambda^2}{1-\lambda}$, but the Schur-transformed state has roughly $\lambda$ times as many qubits as the original i.i.d. collection. Therefore, we need to multiply the number of qubits in the Schur-transformed state, which is $(N_C-1)$, by $\frac{\tilde{\lambda}}{1-\tilde{\lambda}}$. The result is as follows:
\begin{align}
    \text{RLD}_{\text{max}}[\text{new}] &\approx (N_C-1)\frac{\tilde{\lambda}}{1-\tilde{\lambda}} \\
    &= (N_C-1)\frac{\lambda}{1-\lambda}\left[1 + \frac{1}{N_C} + O\left(N_C^{-2+\varepsilon}\right)\right] \\
    &= N_C\frac{\lambda}{1-\lambda}\left[1 + O\left(N_C^{-2+\varepsilon}\right)\right] \\
    &= \text{RLD}_{\text{max}}[\text{old}] + O\left(N_C^{-1+\varepsilon}\right).
\end{align}
Therefore, the amount by which these qubits become concentrated exactly matches what we would expect if $\text{RLD}_{\text{max}}$ is approximately conserved. In particular, $\text{RLD}_{\text{max}}$ reduces by an $o(1)$ amount, so even when this single-qubit discarding is repeated $\Theta(N)$ times, the total $\text{RLD}_{\text{max}}$ reduces by an $o(N)$ amount, meaning that it is conserved to the leading order. This explains intuitively why the discarding map yields the maximum concentration rate, which is governed by the monotonicity of $\text{RLD}_{\text{max}}$.

\appsubsec{Adding One Qubit at a Time}
{subsec:heuristic-adding-one-qubit}

In this subsection, we study the act of adding a single qubit to a Schur-transformed state via the optimal cloning map. This will closely mirror the discussion from Appendix \ref{sec:heuristic-arguments}\ref{subsec:heuristic-discarding-one-qubit}, where we studied the act of discarding a single qubit from a Schur-transformed state.

\vspace{0.5\baselineskip}

As a reminder, in the Schur-transformed state, the coefficients of the Dicke states look as follows (up to scaling to ensure that the state has unit trace):
\begin{equation}
    N_C-\delta \mapsto c_1^{N_C-\delta}c_0^\delta.
\end{equation}
The essential point is that the coefficients form a geometric sequence with common ratio $\frac{c_0}{c_1} = \frac{1-\lambda}{1+\lambda}$.

\vspace{0.5\baselineskip}

\begin{figure}
    \includegraphics[scale=0.5]{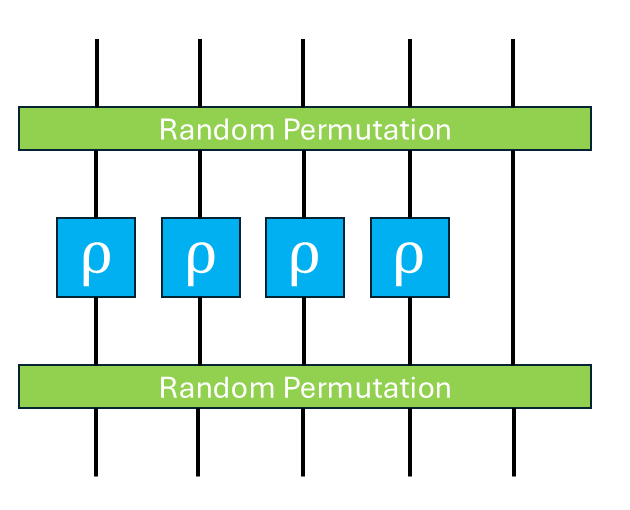}
    \caption{The result of adding a single qubit via the optimal cloning map. Suppose you have a bit string with $\tilde{\delta}$ zeros and $N_C+1-\tilde{\delta}$ ones. Each copy of $\rho$ multiplies the coefficient by $c_1$ if it receives a $1$ and by $c_0$ if it receives a $0$. Furthermore, with probability $\tilde{\delta}/(N_C+1)$, the identity wire receives a $0$, and with probability $(N_C+1-\tilde{\delta})/(N_C+1)$, the identity wire receives a $1$. Of course, the identity wire does not alter the coefficient regardless of whether it receives a $0$ or a $1$. As a result, up to proportionality, the coefficient of this bit string is multiplied by $(\tilde{\delta})c_1^{N_C-\tilde{\delta}+1}c_0^{\tilde{\delta}-1} + (N_C+1-\tilde{\delta})c_1^{N_C-\tilde{\delta}}c_0^{\tilde{\delta}}$.}
    \label{fig:single-qubit-optimal-cloning}
\end{figure}

Now watch what happens when we add a single qubit via the optimal cloning map. Without loss of generality, suppose we add a new identity wire in position $(N_C+1)$. Refer to the tensor network diagram shown in Figure \ref{fig:single-qubit-optimal-cloning}. Suppose you have a bit string with $\tilde{\delta}$ zeros and $N_C+1-\tilde{\delta}$ ones. Each copy of $\rho$ multiplies the coefficient by $c_1 = \frac{1+\lambda}{2}$ if it receives a $1$ and by $c_0 = \frac{1-\lambda}{2}$ if it receives a $0$.

\vspace{0.5\baselineskip}

Furthermore, each position from $1$ to $(N_C+1)$ gets mapped to the final position in a $1/(N_C+1)$ fraction of the permutations. Therefore, with probability $\tilde{\delta}/(N_C+1)$, the identity wire receives a $0$, and with probability $(N_C+1-\tilde{\delta})/(N_C+1)$, the identity wire receives a $1$. Of course, the identity wire does not alter the coefficient regardless of whether it receives a $0$ or a $1$.

\vspace{0.5\baselineskip}

As a result, up to proportionality, the coefficient of this bit string is multiplied by
\begin{align}
    & (\tilde{\delta})c_1^{N_C-\tilde{\delta}+1}c_0^{\tilde{\delta}-1} + (N_C+1-\tilde{\delta})c_1^{N_C-\tilde{\delta}}c_0^{\tilde{\delta}} \\
    = \,\, & c_1^{N_C-\tilde{\delta}}c_0^{\tilde{\delta}-1}\left[\tilde{\delta}c_1 + (N_C+1-\tilde{\delta})c_0\right] \\
    = \,\, & c_1^{N_C+1-\tilde{\delta}}c_0^{\tilde{\delta}}\left(\frac{N_C+1-\tilde{\delta}}{c_1} + \frac{\tilde{\delta}}{c_0}\right).
\end{align}
Therefore, as a general formula, we can write
\begin{equation}
    N_C+1-\tilde{\delta} \mapsto c_1^{N_C+1-\tilde{\delta}}c_0^{\tilde{\delta}}\left(\frac{N_C+1-\tilde{\delta}}{c_1} + \frac{\tilde{\delta}}{c_0}\right).
\end{equation}
Therefore, just to write everything out more explicitly, we observe that the coefficients of the Dicke states, in decreasing order of Hamming weight, look as follows (again, up to scaling):
\begin{align}
    N_C+1 &\mapsto (N_C+1)c_1^{N_C} = c_1^{N_C+1}c_0^0\left(\frac{N_C+1}{c_1} + \frac{0}{c_0}\right) \\
    N_C &\mapsto N_Cc_1^{N_C-1}c_0 + c_1^{N_C} = c_1^{N_C}c_0^1\left(\frac{N_C}{c_1} + \frac{1}{c_0}\right) \\
    N_C-1 &\mapsto (N_C-1)c_1^{N_C-2}c_0^2 + 2c_1^{N_C-1}c_0 = c_1^{N_C-1}c_0^2\left(\frac{N_C-1}{c_1} + \frac{2}{c_0}\right) \\
    & \vdots \\
    1 &\mapsto c_0^{N_C} + N_Cc_1c_0^{N_C-1} = c_1^1c_0^{N_C}\left(\frac{1}{c_1} + \frac{N_C}{c_0}\right) \\
    0 &\mapsto (N_C+1)c_0^{N_C} = c_1^0c_0^{N_C+1}\left(\frac{0}{c_1} + \frac{N_C+1}{c_0}\right).
\end{align}
As a result, the ratio between successive coefficients now takes the form
\begin{align}
    \frac{\text{coeff}\left[N_C+1-(\tilde{\delta}+1)\right]}{\text{coeff}\left[N_C+1-\tilde{\delta}\right]} &= \frac{c_1^{N_C-\tilde{\delta}}c_0^{\tilde{\delta}+1}}{c_1^{N_C+1-\tilde{\delta}}c_0^{\tilde{\delta}}}\frac{\frac{N_C-\tilde{\delta}}{c_1} + \frac{\tilde{\delta}+1}{c_0}}{\frac{N_C+1-\tilde{\delta}}{c_1} + \frac{\tilde{\delta}}{c_0}} \\
    &= \frac{1-\lambda}{1+\lambda}\frac{(1-\lambda)N_C + (2\lambda\tilde{\delta}+1+\lambda)}{(1-\lambda)N_C + (2\lambda\tilde{\delta}+1-\lambda)} \\
    &= \frac{1-\lambda}{1+\lambda}\left[1 + \frac{2\lambda}{1-\lambda}\frac{1}{N_C} + O\left(\frac{\tilde{\delta}}{N_C^2}\right)\right].
\end{align}
The upshot is that the decay ratio from one coefficient to the next becomes ever so slightly larger, meaning that the remaining qubits become ever so slightly diluted.

\vspace{0.5\baselineskip}

Once again, it is not an accident that this approximation is only good for small values of $\tilde{\delta}$, as we saw in the formal proof of Lemma \ref{lem:optimal-cloning-map-schur-transformed-state} and in the numerics in Appendix \ref{sec:numerical-analysis}. Fortunately, because the coefficients are exponentially decaying, almost all of the weight is concentrated in these small values of $\tilde{\delta}$, so the error between this state and a true Schur-transformed state is very small. Just as we observed in the formal proof, so too can we observe in the heuristic argument above that these coefficients are negligible for $\tilde{\delta} = \Omega(N_C^\varepsilon)$ for any $\varepsilon > 0$.

\vspace{0.5\baselineskip}

Since the ratio between consecutive coefficients should be $(1-\tilde{\lambda})/(1+\tilde{\lambda})$, where $\tilde{\lambda}$ is the new (ever so slightly lower) purity level, we can approximately solve for a function of $\tilde{\lambda}$ as follows:
\begin{align}
    \frac{1-\tilde{\lambda}}{1+\tilde{\lambda}} &= \frac{1-\lambda}{1+\lambda}\left[1 + \frac{2\lambda}{1-\lambda}\frac{1}{N_C} + O\left(N_C^{-2+\varepsilon}\right)\right] \\
    \therefore \frac{\tilde{\lambda}}{1+\tilde{\lambda}} &= \frac{1}{2}\left(1 - \frac{1-\tilde{\lambda}}{1+\tilde{\lambda}}\right) \\
    &= \frac{1}{2}\left[1 - \frac{1-\lambda}{1+\lambda} - \frac{2\lambda}{1+\lambda}\frac{1}{N_C} + O\left(N_C^{-2+\varepsilon}\right)\right] \\
    &= \frac{\lambda}{1+\lambda}\left[1 - \frac{1}{N_C} + O\left(N_C^{-2+\varepsilon}\right)\right].
\end{align}
In this case, notice what happens when we compute $\text{RLD}_{\text{min}}$ associated to a Schur-transformed state with this purity level. Recall that $\text{RLD}_{\text{min}}(\lambda) = \frac{\lambda^2}{1+\lambda}$, but the Schur-transformed state has roughly $\lambda$ times as many qubits as the original i.i.d. collection. Therefore, we need to multiply the number of qubits in the Schur-transformed state, which is $(N_C+1)$, by $\frac{\tilde{\lambda}}{1+\tilde{\lambda}}$. The result is as follows:
\begin{align}
    \text{RLD}_{\text{min}}[\text{new}] &\approx (N_C+1)\frac{\tilde{\lambda}}{1+\tilde{\lambda}} \\
    &= (N_C+1)\frac{\lambda}{1+\lambda}\left[1 - \frac{1}{N_C} + O\left(N_C^{-2+\varepsilon}\right)\right] \\
    &= N_C\frac{\lambda}{1+\lambda}\left[1 + O\left(N_C^{-2+\varepsilon}\right)\right] \\
    &= \text{RLD}_{\text{min}}[\text{old}] + O\left(N_C^{-1+\varepsilon}\right).
\end{align}
Therefore, the amount by which these qubits become diluted exactly matches what we would expect if $\text{RLD}_{\text{min}}$ is approximately conserved. In particular, $\text{RLD}_{\text{min}}$ reduces by an $o(1)$ amount, so even when this single-qubit optimal cloning is repeated $\Theta(N)$ times, the total $\text{RLD}_{\text{min}}$ reduces by an $o(N)$ amount, meaning that it is conserved to the leading order. This explains intuitively why the optimal cloning map yields the maximum dilution rate, which is governed by the monotonicity of $\text{RLD}_{\text{min}}$.

\newpage

\appsec{Numerical Analysis of the Concentration and Dilution Protocols}
{sec:numerical-analysis}

In this appendix, we perform numerical simulations of our various qubit linear-rate conversion procedures to corroborate the analytical results shown in Appendices \ref{sec:schur-transformed-state-conversion} and \ref{sec:unified-presentation}. We will show how the channels $\mE_{\text{discard}}$, $\mE_{\text{clone}}$, $\mE_{\text{MP}}$ successfully transform a Schur-transformed state into approximately a new Schur-transformed state.  We will additionally show how combining each of these channels with the Schur sampling and inverse Schur sampling steps achieves i.i.d. linear-rate conversion with vanishing trace distance. 

\vspace{0.5\baselineskip}

Note that we do not need to treat $\mE_{\text{WW}}$ (and the conversion task it helps achieve, namely, wrong-way conversion) separately. This is because it will yield the same result as $\mE_{\text{MP}}$ (and its associated task, namely, measure-and-prepare conversion), but with the Hamming weights inverted, which corresponds to the direction $\hat{n}$ being reversed to $-\hat{n}$. As a result, we do not study $\mE_{\text{WW}}$ and wrong-way conversion additionally in this appendix, beyond the results we will show for $\mE_{\text{MP}}$ and measure-and-prepare conversion.

\vspace{0.5\baselineskip}

This appendix is organized as follows:
\begin{itemize}
    \item In Appendix \ref{sec:numerical-analysis}\ref{subsec:numerical-analysis-schur-transformed-state-conversion-geometric-sequence}, we show how the channels $\mE_{\text{discard}}$, $\mE_{\text{clone}}$, $\mE_{\text{MP}}$ approximately maintain the property of the Dicke state coefficients being in geometric sequence, at least in the dominant Hamming weights. Recall from Definition \ref{def:Schur-transformed-states} that this ``geometric sequence'' property is one of the defining features of a Schur-transformed state.
    \item In Appendix \ref{sec:numerical-analysis}\ref{subsec:numerical-analysis-schur-transformed-state-conversion-error-analysis}, we analyze the trend of the trace distance of the Schur-transformed state conversion as a function of $N$, under the simplifying assumption that $\tilde{N}_C/N_C$ matches $R_C$ \textit{exactly}. We see the trace distance decay roughly as $N_C^{-1}$, and we also observe some modular arithmetic artifacts when $\tilde{N}_C/N_C$ cannot match $R_C$ exactly. (For example, when $R_C = 1/2$, we will see two distinct trends for even and odd $N_C$ values.) Beyond the analytical proofs of the lemmas in Appendix \ref{sec:schur-transformed-state-conversion}, this provides convincing evidence that our intermediate steps work as intended.
    \item In Appendix \ref{sec:numerical-analysis}\ref{subsec:numerical-analysis-iid-state-conversion-error-analysis}, we numerically simulate the full three-step conversion procedures. This requires much larger values of $N$, and we now see the trace distance decay roughly as $N^{-1/2}$. Beyond the analytical proofs of the lemmas in Appendix \ref{sec:unified-presentation}, this provides convincing evidence that our full qubit linear-rate conversion procedures work as intended.
\end{itemize}

\appsubsec{Schur-Transformed State Conversion: Preserving the ``Geometric Sequence'' Property}
{subsec:numerical-analysis-schur-transformed-state-conversion-geometric-sequence}

Recall from Definition \ref{def:Schur-transformed-states} that, for a state on the symmetric subspace of some number of qubits, the defining properties that make it a Schur-transformed state are that the state is diagonal in the basis of Dicke states with respect to some direction, and that the coefficients of the Dicke states are in geometric sequence.

\vspace{0.5\baselineskip}

In this first subsection, we will show examples of how each of the maps $\mE_{\text{discard}}$, $\mE_{\text{clone}}$, $\mE_{\text{MP}}$ to a Schur-transformed state roughly maintains the geometric sequence property while adjusting the common ratio, which equates to changing the purity level. This is the central intuition behind all of our qubit conversion protocols, as formalized in Appendix \ref{sec:schur-transformed-state-conversion} and explained heuristically in Appendix \ref{sec:heuristic-arguments}.

\vspace{0.5\baselineskip}

We present two examples for each of the three relevant channels on the symmetric subspace:
\begin{itemize}
    \item \textbf{Discarding map:} $\lambda_{\text{in}} = 1/2$, $\lambda_{\text{out}} = 2/3$, $(N_C,\tilde{N}_C) = (20,10), (100,50)$ (refer to Figure \ref{fig:dicke-state-coeff-schur-transformed-state-conversion-discarding-map})
    \item \textbf{Optimal cloning map:} $\lambda_{\text{in}} = 2/3$, $\lambda_{\text{out}} = 1/2$, $(N_C,\tilde{N}_C) = (10,12), (50,60)$ (refer to Figure \ref{fig:dicke-state-coeff-schur-transformed-state-conversion-optimal-cloning-map})
    \item \textbf{Optimal measure-and-prepare channel:} $\lambda_{\text{in}} = 1/2$, $\lambda_{\text{out}} = 2/3$, $(N_C,\tilde{N}_C) = (60,10), (300,50)$ (refer to Figure \ref{fig:dicke-state-coeff-schur-transformed-state-conversion-optimal-mp-channel})
\end{itemize}

\begin{figure}
    \includegraphics[scale=0.5]{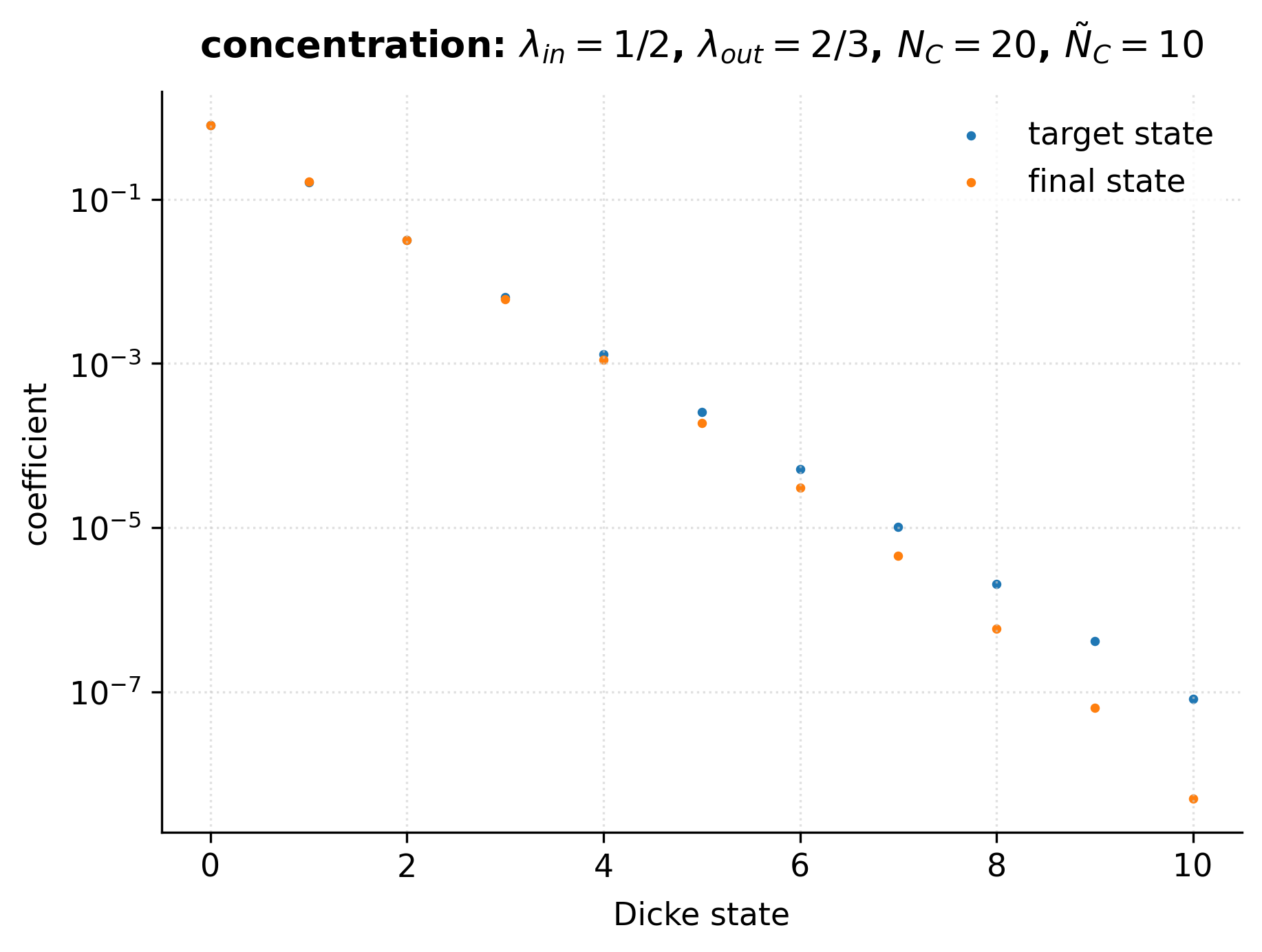}
    \includegraphics[scale=0.5]{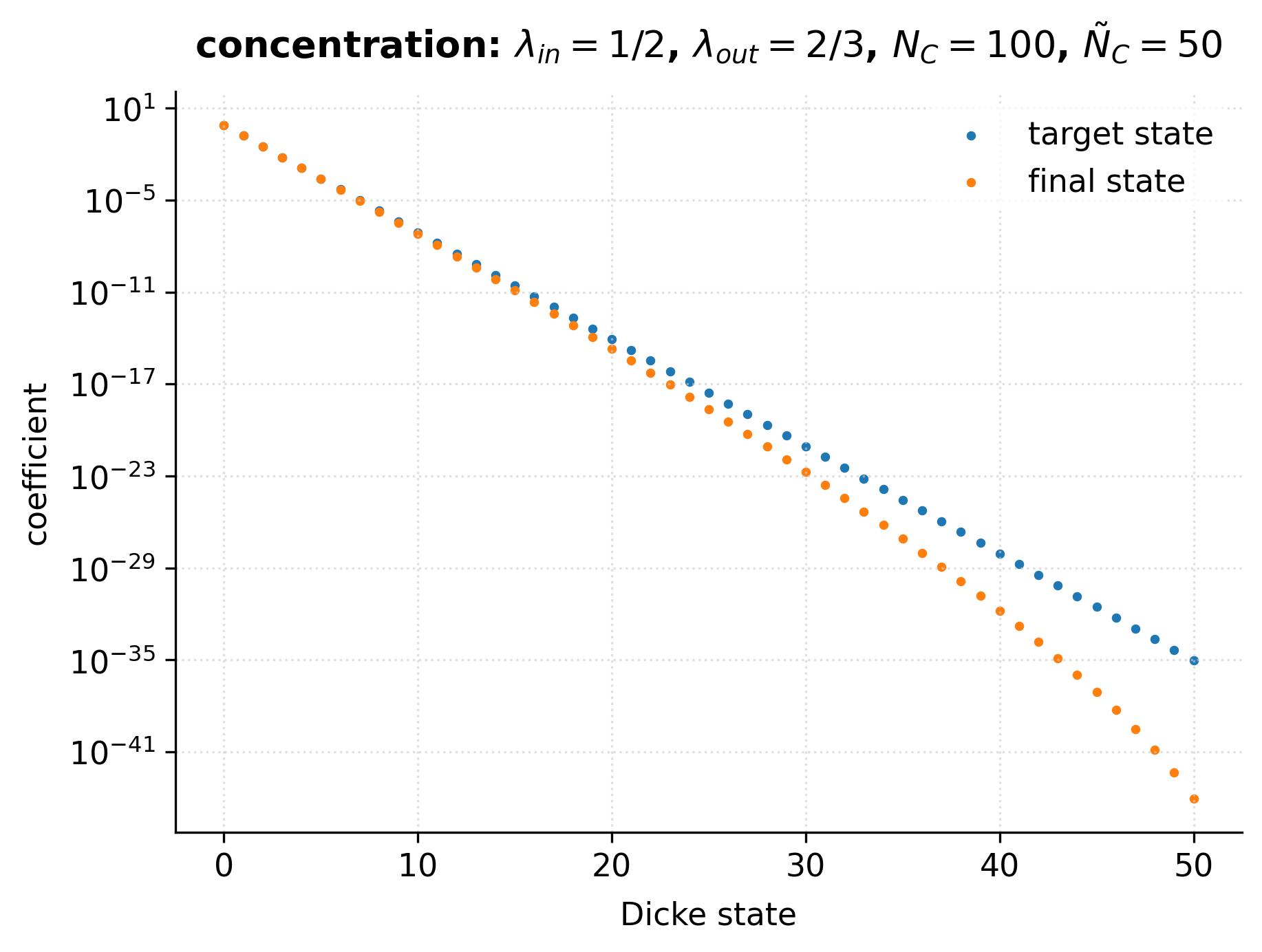}
    \caption{(LEFT) A comparison between the Dicke state coefficients of the output state $\mE_{\text{discard}}\left(\rho_C(20,1/2,\hat{n})\right)$ and the target state $\rho_C(10,2/3,\hat{n})$. (RIGHT) A comparison between the Dicke state coefficients of the output state $\mE_{\text{discard}}\left(\rho_C(100,1/2,\hat{n})\right)$ and the target state $\rho_C(50,2/3,\hat{n})$.}
    \label{fig:dicke-state-coeff-schur-transformed-state-conversion-discarding-map}
\end{figure}

\begin{figure}
    \includegraphics[scale=0.5]{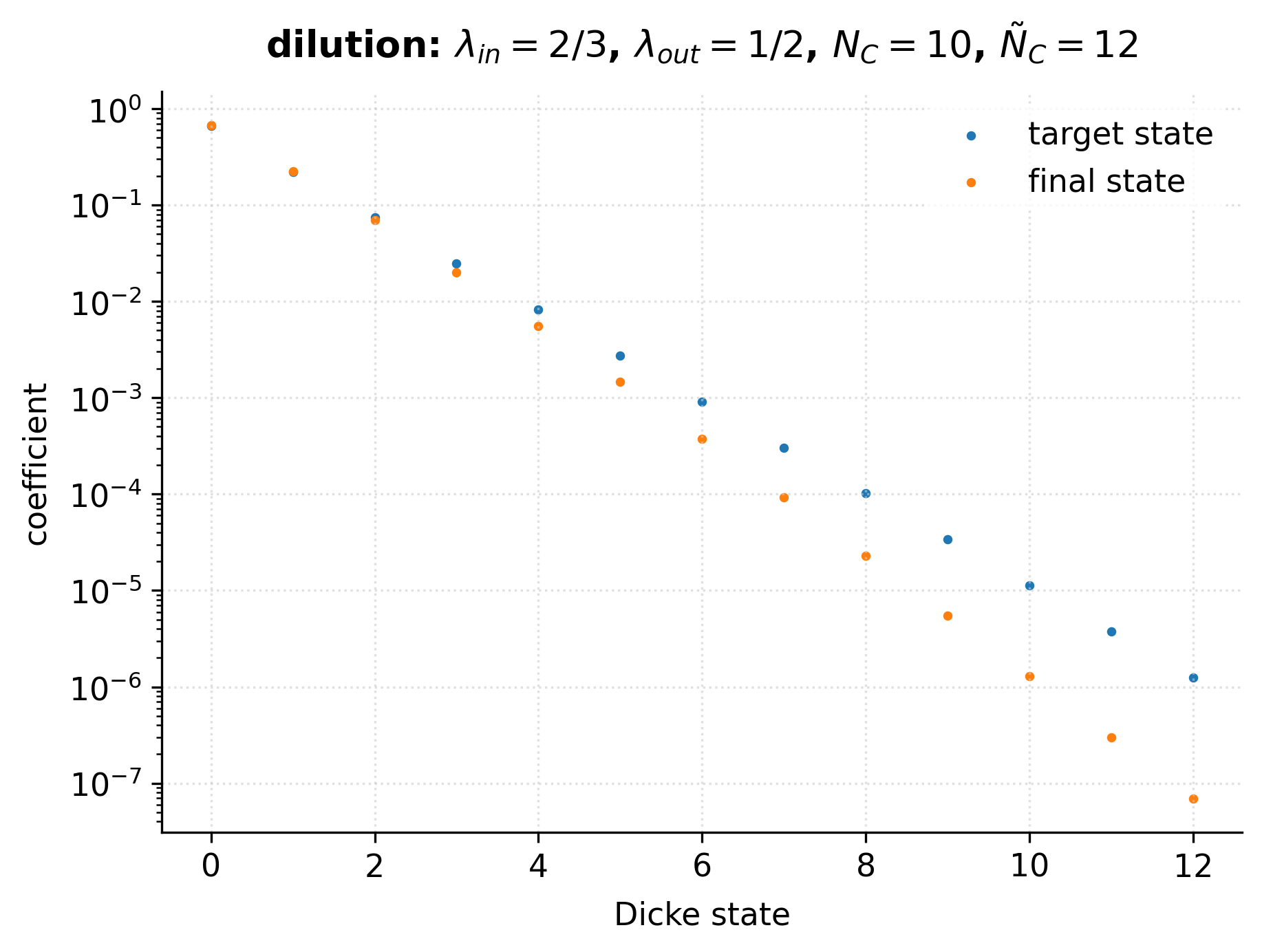}
    \includegraphics[scale=0.5]{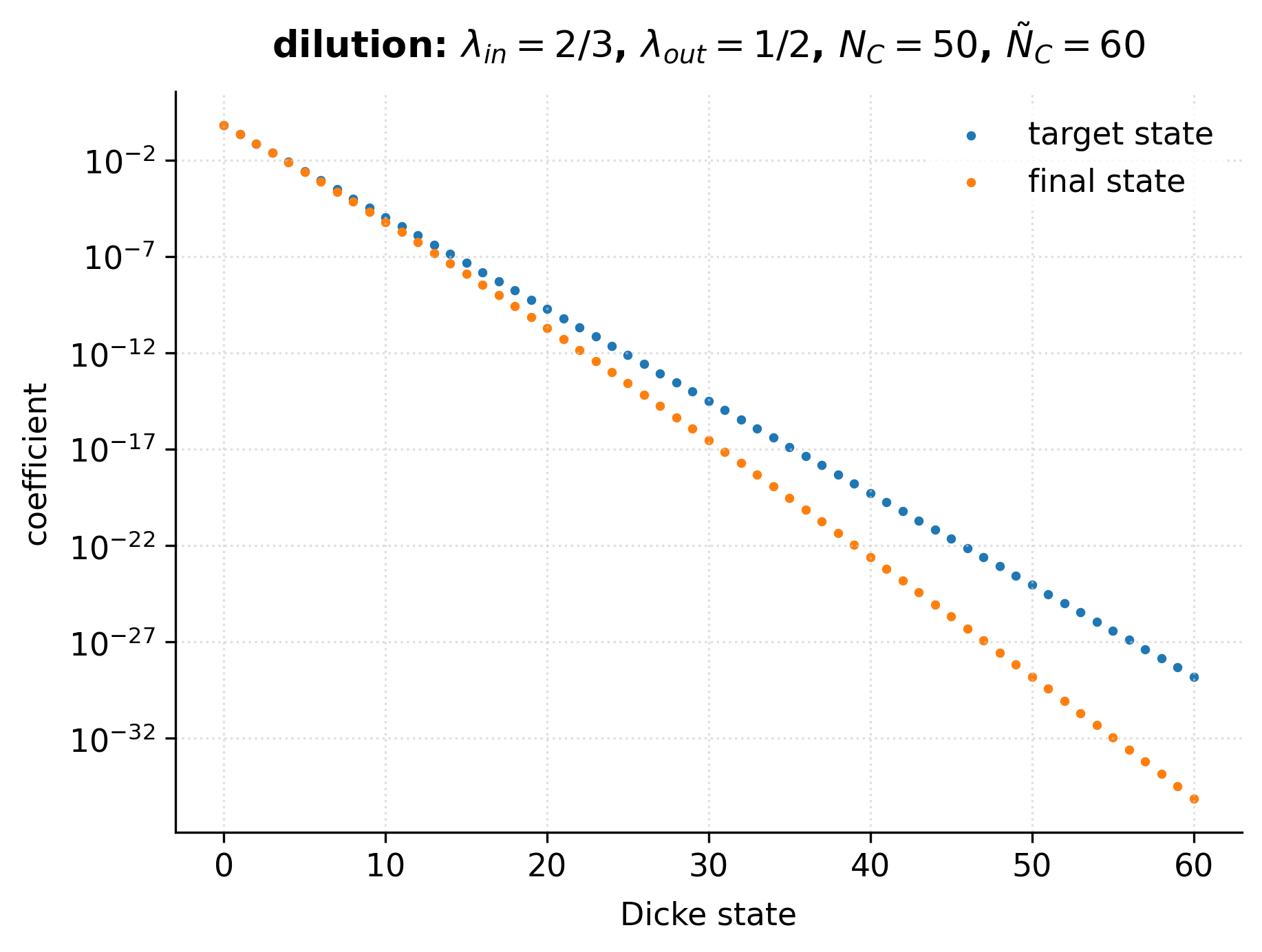}
    \caption{(LEFT) A comparison between the Dicke state coefficients of the output state $\mE_{\text{clone}}\left(\rho_C(10,2/3,\hat{n})\right)$ and the target state $\rho_C(12,1/2,\hat{n})$. (RIGHT) A comparison between the Dicke state coefficients of the output state $\mE_{\text{clone}}\left(\rho_C(50,2/3,\hat{n})\right)$ and the target state $\rho_C(60,1/2,\hat{n})$.}
    \label{fig:dicke-state-coeff-schur-transformed-state-conversion-optimal-cloning-map}
\end{figure}

\begin{figure}
    \includegraphics[scale=0.5]{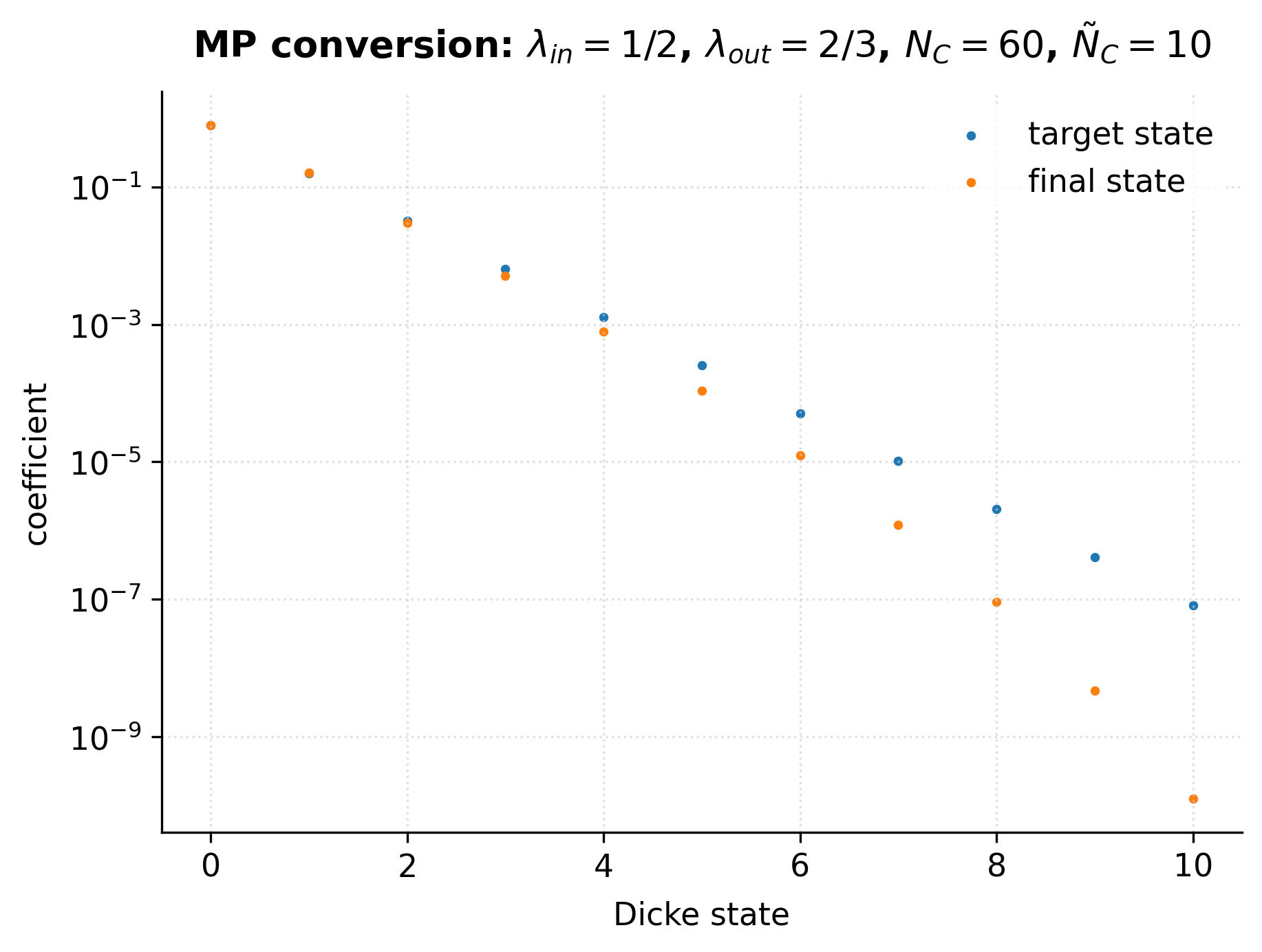}
    \includegraphics[scale=0.5]{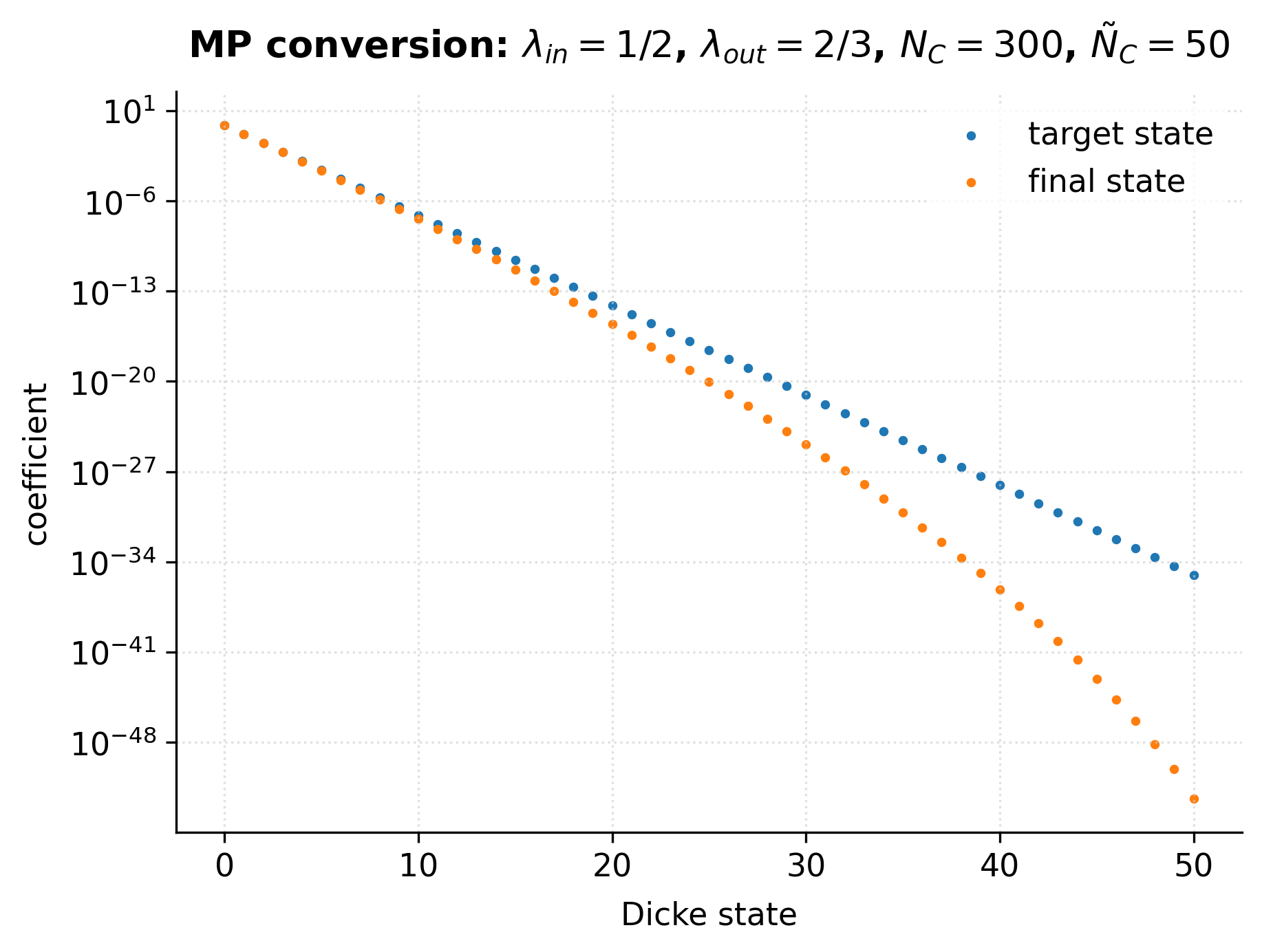}
    \caption{(LEFT) A comparison between the Dicke state coefficients of the output state $\mE_{\text{MP}}\left(\rho_C(60,1/2,\hat{n})\right)$ and the target state $\rho_C(10,2/3,\hat{n})$. (RIGHT) A comparison between the Dicke state coefficients of the output state $\mE_{\text{MP}}\left(\rho_C(300,1/2,\hat{n})\right)$ and the target state $\rho_C(50,2/3,\hat{n})$.}
    \label{fig:dicke-state-coeff-schur-transformed-state-conversion-optimal-mp-channel}
\end{figure}

For all the plots in this subsection, we use the convention that the horizontal axis is indexed by $\tilde{\delta}\coloneqq\tilde{N}_C-\tilde{w}$, so that the largest Hamming weights are for small values of $\tilde{\delta}$. Furthermore, we plot the Dicke state coefficients on a logarithmic scale. Therefore, the points for a Schur-transformed state (such as the target state) will lie on a straight line. A steeper slope corresponds to a higher purity level, whereas a gentler slope corresponds to a lower purity level.

\vspace{0.5\baselineskip}

In all three cases, we see that the Dicke state coefficients of the output state closely approximate those of the target state for small values of $\delta$, which are of course the dominant ones. The approximation becomes worse and worse for large values of $\delta$, but these coefficients are so tiny that they are negligible.

\vspace{0.5\baselineskip}

These plots help explain why the formal proofs in Appendix \ref{sec:schur-transformed-state-conversion} and the heuristic arguments in Appendix \ref{sec:heuristic-arguments} are structured the way they are. In particular, both of them separate the lowest values $\tilde{\delta}$ values (i.e., the highest $\tilde{w}$ values) from all the rest, because only for the lowest $\tilde{\delta}$ values is the geometric sequence approximation actually accurate.

\vspace{0.5\baselineskip}

For most of the $\tilde{\delta}$ values, the approximation is actually terrible. More precisely, the approximation is terrible in a \emph{relative} sense, i.e., there is a very large factor difference between the target state and output state coefficients. The logarithmic scale for all the plots makes this clear, since the vertical distance between two points represents the factor difference between their values. But because these coefficients are so tiny, their contribution to the trace distance is also tiny.

\vspace{0.5\baselineskip}

However, these plots alone should not yet convince you that we are successfully achieving approximate conversion between Schur-transformed states. After all, because the vertical axis spans many orders of magnitude, even points that look extremely nearby visually may secretly be very far apart. In the next subsection, we show that the trace distance indeed decays as the input and output qubit count grow.

\appsubsec{Schur-Transformed State Conversion: Error Analysis}
{subsec:numerical-analysis-schur-transformed-state-conversion-error-analysis}

In this subsection, we will analyze how the trace distance in Schur-transformed state conversion decays with the number of qubits. In particular, we will start with a Schur-transformed state $\rho_C(N_C,\lambda_{\text{in}},\hat{n})$, apply the map $\mE_{\text{discard}}$, $\mE_{\text{clone}}$, or $\mE_{\text{MP}}$ to it, and compute the trace distance of the resulting output state with the target Schur-transformed state $\rho_C(\tilde{N}_C,\lambda_{\text{out}},\hat{n})$. We will repeat this for many pairs $(N_C,\tilde{N}_C)$ and perform regression to see how the trace distance decays as $N_C$ grows.

\vspace{0.5\baselineskip}

For simplicity, throughout this subsection, we will choose $\tilde{N}_C = \lfloor N_CR_C\rfloor$, meaning that $N_C$ and $\tilde{N}_C$ are related \emph{exactly} by the factor $R_C$ (up to truncation error, since $R_CN_C$ may not be an integer). We do this for a few reasons:
\begin{itemize}
    \item We can numerically demonstrate successful conversion of Schur-transformed states in its own right, independently of its application to i.i.d. qubit conversion.
    \item We can corroborate the lemmas in Appendix \ref{sec:schur-transformed-state-conversion} even in the $p\to 0$ regime. In particular, we will see how the trace distance appears to decay roughly as $N_C^{-1}$, which is what we expect based on those lemmas (since $\varepsilon > 0$ can be made arbitrarily small).
    \item We can highlight some interesting artifacts of truncation error and modular arithmetic. In particular, if the Schur-transformed conversion rate is a rational number but not an integer (i.e., $R_C = a/b$ for $a,b\in\Nbb$ and $b > 1$), then different values of $N_C$ will be split into different trend lines based on remainders modulo $b$.
\end{itemize}

\vspace{0.5\baselineskip}

We present one example of each of the three relevant channels on the symmetric subspace:
\begin{itemize}
    \item \textbf{Discarding map:} $\lambda_{\text{in}} = 1/2$, $\lambda_{\text{out}} = 2/3$, $R_C = 1/2$ (refer to Figure \ref{fig:tr-dist-trend-schur-transformed-state-conversion-discarding-map})
    \item \textbf{Optimal cloning map:} $\lambda_{\text{in}} = 2/3$, $\lambda_{\text{out}} = 1/2$, $R_C = 6/5$ (refer to Figure \ref{fig:tr-dist-trend-schur-transformed-state-conversion-optimal-cloning-map})
    \item \textbf{Optimal measure-and-prepare channel:} $\lambda_{\text{in}} = 1/2$, $\lambda_{\text{out}} = 2/3$, $R_C = 1/6$ (refer to Figure \ref{fig:tr-dist-trend-schur-transformed-state-conversion-optimal-mp-channel})
\end{itemize}

\begin{figure}
    \includegraphics[scale=0.5]{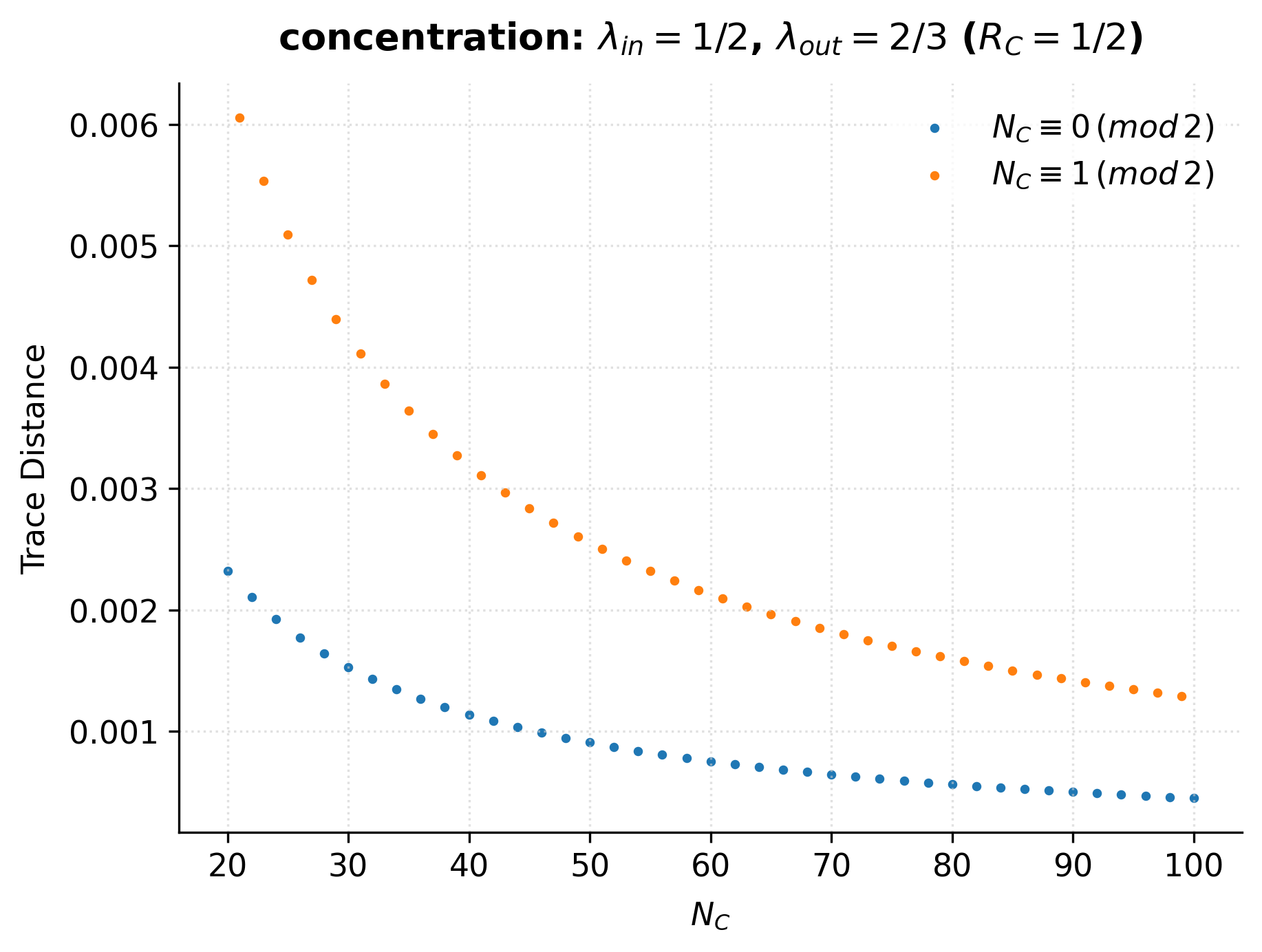}
    \includegraphics[scale=0.5]{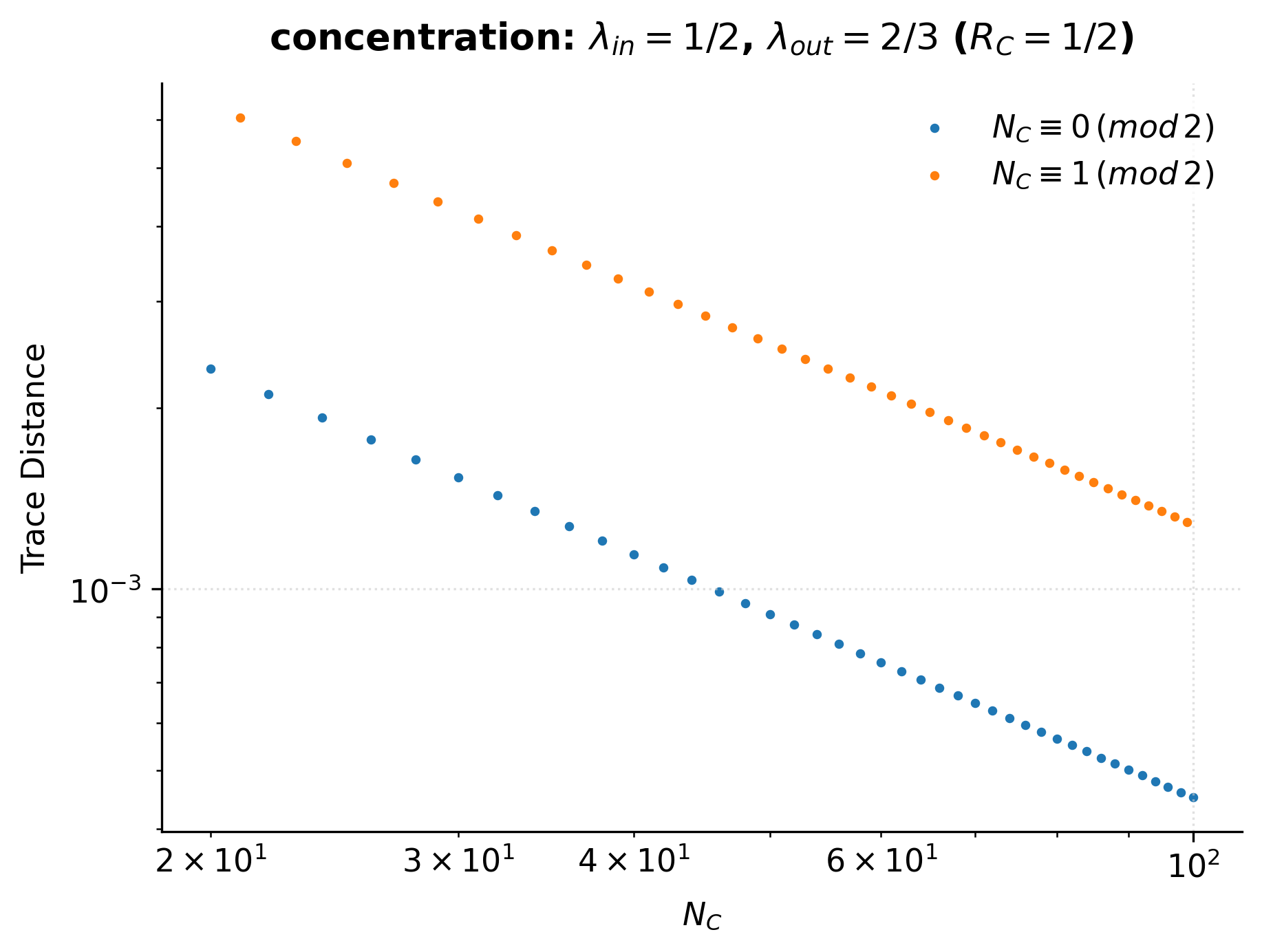}
    \caption{(LEFT) Trace distance as a function of input qubit count for Schur-transformed state conversion using the discarding map. In particular, the trace distance is computed between $\mE_{\text{discard}}\left(\rho_C(N_C,1/2,\hat{n})\right)$ and $\rho_C(\tilde{N}_C,2/3,\hat{n})$. Since $R_C = 1/2$, there are two distinct trend lines for $N_C\equiv 0\pmod{2}$ and $N_C\equiv 1\pmod{2}$. (RIGHT) The same plot, but with logarithmic scales for both axes. When we run linear regression on the transformed data (which reveals a power law relation) for each trend line, we obtain:}
    \begin{equation}
        N_C\equiv 0\pmod{2} \implies d_{\text{Tr}} = 0.04852 \times N_C^{-1.0162} \quad (R^2 = 0.9999881)
    \end{equation}
    \begin{equation}
        N_C\equiv 1\pmod{2} \implies d_{\text{Tr}} = 0.1262 \times N_C^{-0.9971} \quad (R^2 = 0.9999995).
    \end{equation}
    \label{fig:tr-dist-trend-schur-transformed-state-conversion-discarding-map}
\end{figure}

\begin{figure}
    \includegraphics[scale=0.5]{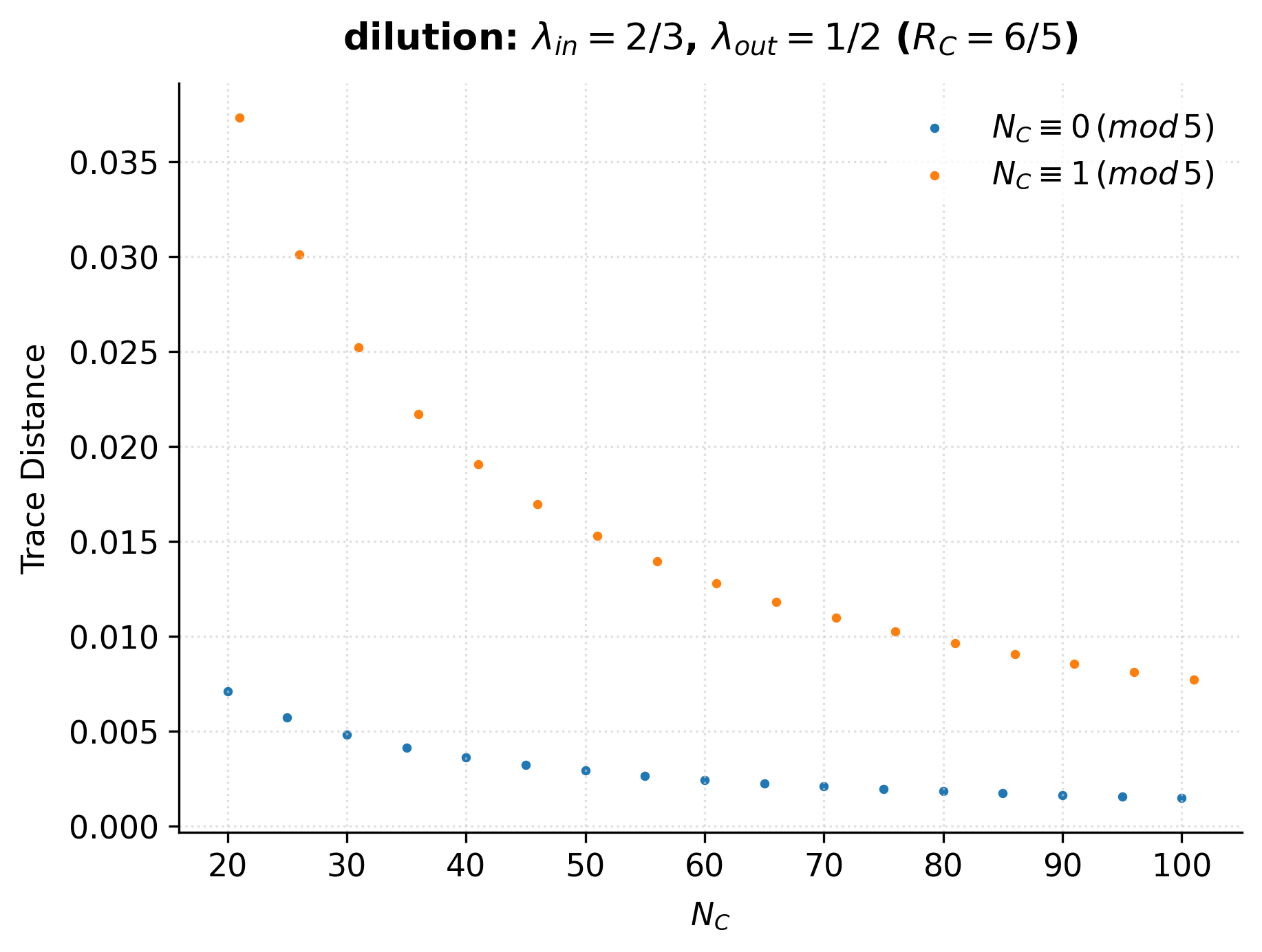}
    \includegraphics[scale=0.5]{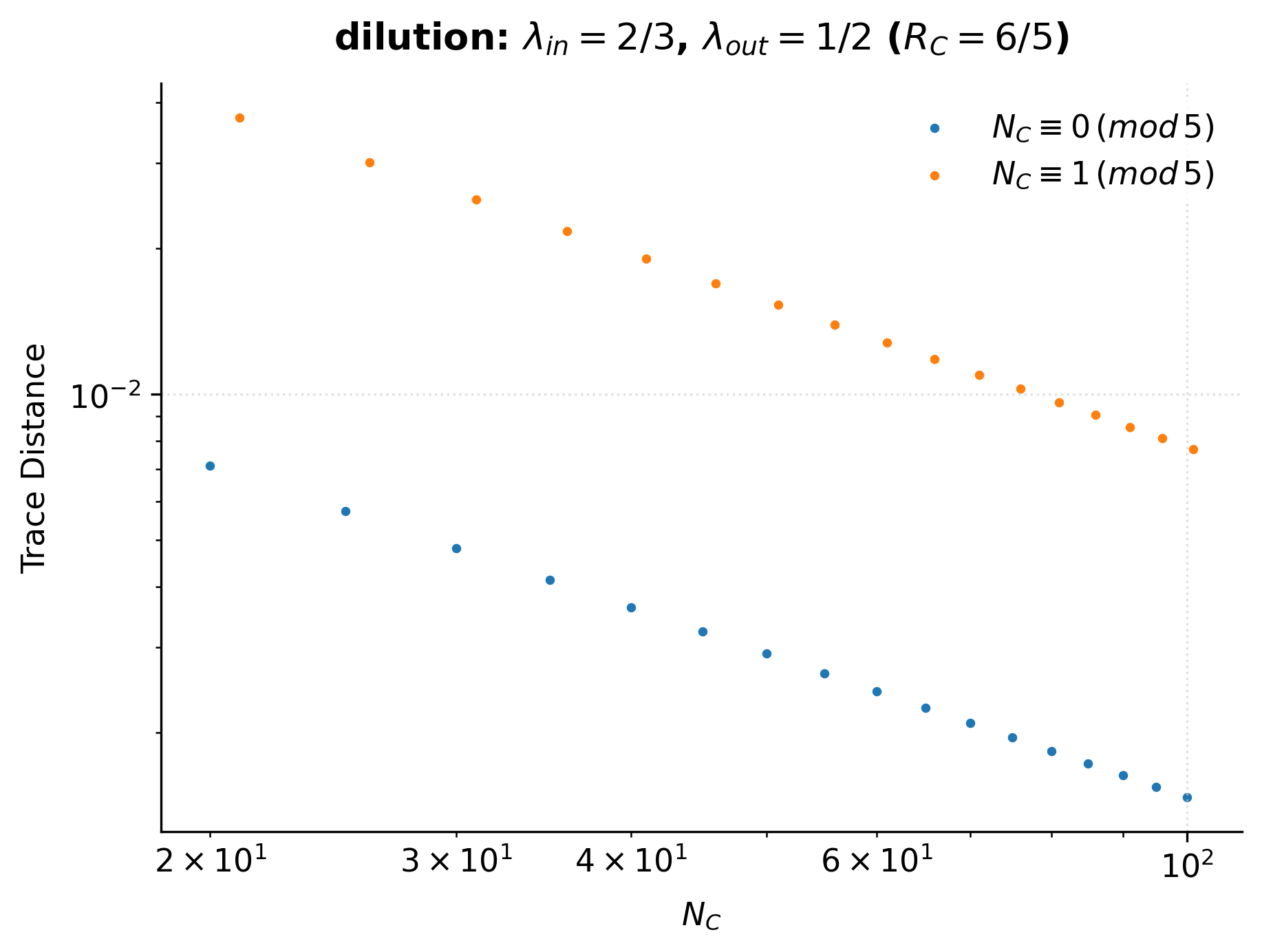}
    \caption{(LEFT) Trace distance as a function of input qubit count for Schur-transformed state conversion using the optimal cloning map. In particular, the trace distance is computed between $\mE_{\text{clone}}\left(\rho_C(N_C,2/3,\hat{n})\right)$ and $\rho_C(\tilde{N}_C,1/2,\hat{n})$. Since $R_C = 6/5$, there are two distinct trend lines for $N_C\equiv 0\pmod{5}$ and $N_C\equiv 1\pmod{5}$. (RIGHT) The same plot, but with logarithmic scales for both axes. When we run linear regression on the transformed data (which reveals a power law relation) for each trend line, we obtain:}
    \begin{equation}
        N_C\equiv 0\pmod{5} \implies d_{\text{Tr}} = 0.1352 \times N_C^{-0.9814} \quad (R^2 = 0.9999817)
    \end{equation}
    \begin{equation}
        N_C\equiv 1\pmod{5} \implies d_{\text{Tr}} = 0.7921 \times N_C^{-1.0037} \quad (R^2 = 0.9999993).
    \end{equation}
    \label{fig:tr-dist-trend-schur-transformed-state-conversion-optimal-cloning-map}
\end{figure}

\begin{figure}
    \includegraphics[scale=0.5]{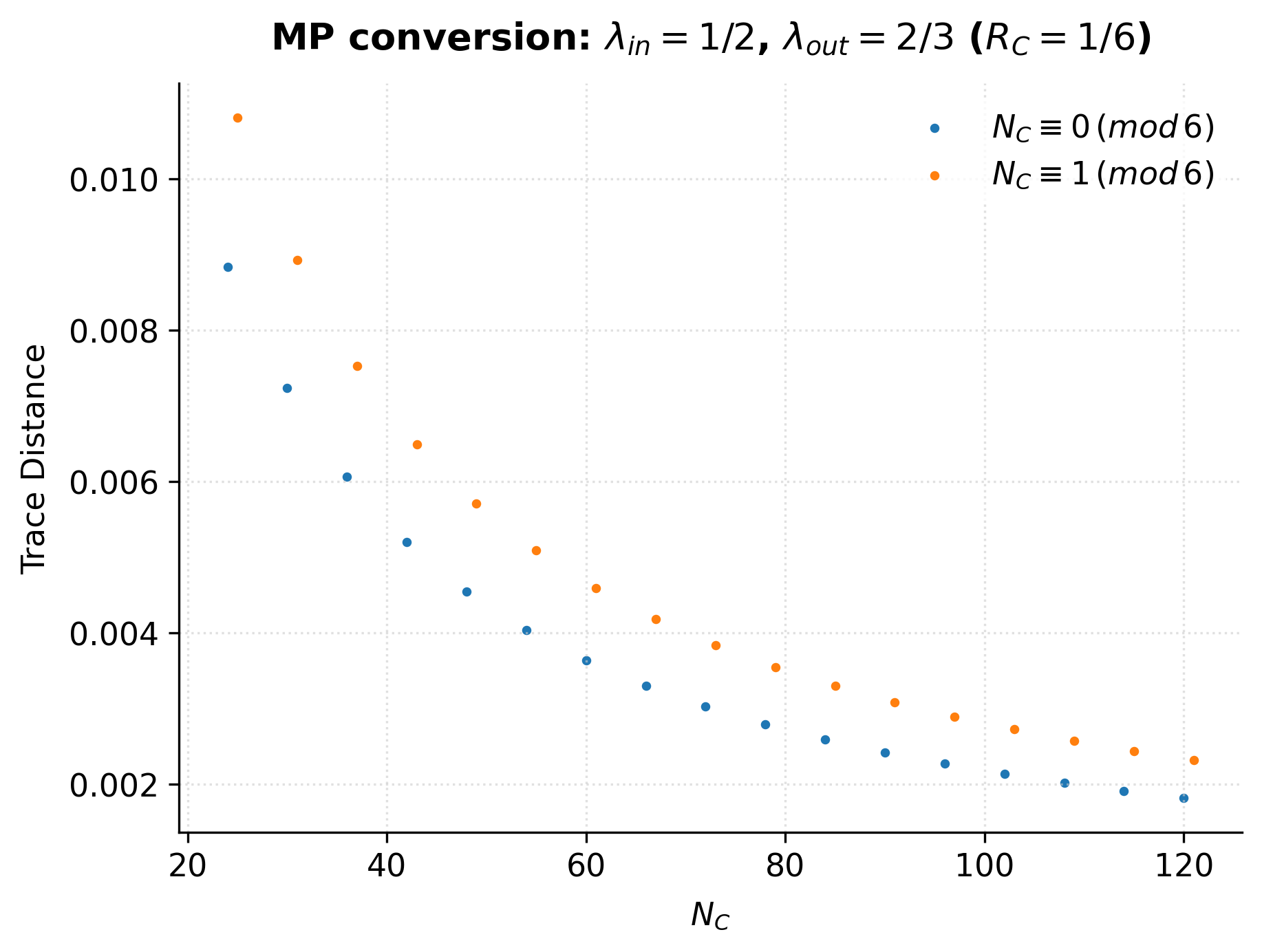}
    \includegraphics[scale=0.5]{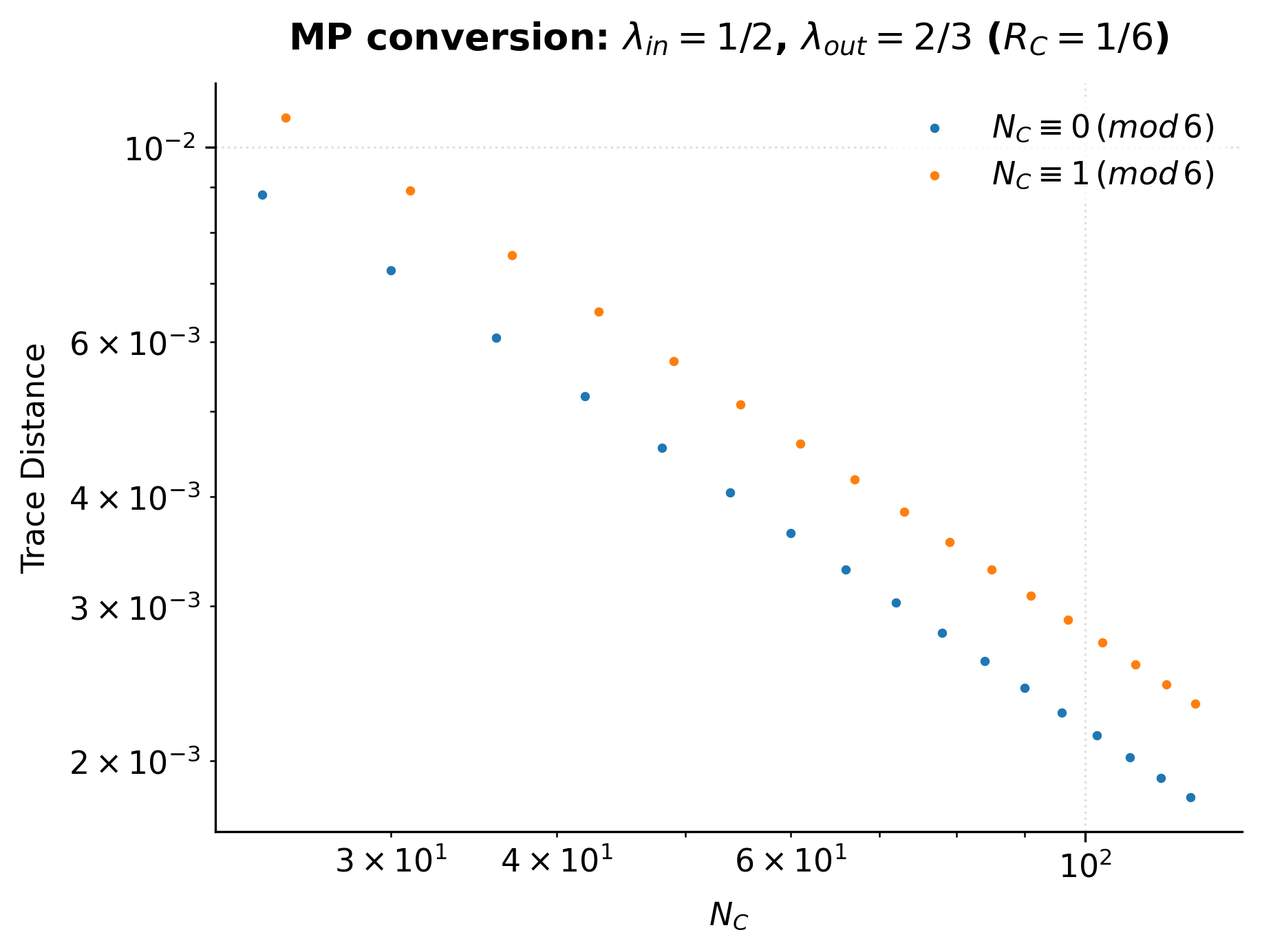}
    \caption{(LEFT) Trace distance as a function of input qubit count for Schur-transformed state conversion using the optimal measure-and-prepare channel. In particular, the trace distance is computed between $\mE_{\text{MP}}\left(\rho_C(N_C,1/2,\hat{n})\right)$ and $\rho_C(\tilde{N}_C,2/3,\hat{n})$. Since $R_C = 1/6$, there are two distinct trend lines for $N_C\equiv 0\pmod{6}$ and $N_C\equiv 1\pmod{6}$. (RIGHT) The same plot, but with logarithmic scales for both axes. When we run linear regression on the transformed data for each trend line (which reveals a power law relation), we obtain:}
    \begin{equation}
        N_C\equiv 0\pmod{6} \implies d_{\text{Tr}} = 0.2110 \times N_C^{-0.9927} \quad (R^2 = 0.9998532)
    \end{equation}
    \begin{equation}
        N_C\equiv 1\pmod{6} \implies d_{\text{Tr}} = 0.2622 \times N_C^{-0.9849} \quad (R^2 = 0.9998538).
    \end{equation}
    \label{fig:tr-dist-trend-schur-transformed-state-conversion-optimal-mp-channel}
\end{figure}

In all three cases, there are distinct trend lines based on modular arithmetic. This occurs because $R_CN_C$ need not be an integer, so $N_C = \lfloor R_CN_C\rfloor$ introduces some truncation error. As a result, in each plot, we color-code the data based on modular arithmetic.

\vspace{0.5\baselineskip}

After splitting the data based on modular arithmetic, we study each trend line by performing linear regression of $\ln(d_{\text{Tr}})$ as a function of $\ln(N_C)$, which corresponds to a power-law relation, with the slope corresponding to the exponent.

\vspace{0.5\baselineskip}

In all three cases, each trend line has a slope very close to $-1$ when graphed on a log-log plot. This means that the trace distance decays roughly as $N_C^{-1}$. Furthermore, the quality of each trend line is very impressive, with $R^2$ values consistently exceeding $0.9998$. This lends strong numerical support to the analytical results from Appendix \ref{sec:schur-transformed-state-conversion}.

\vspace{0.5\baselineskip}

In particular, if $N_C$ and $\tilde{N}_C$ are allowed to deviate from their ``ideal values'' ($\lambda_{\text{in}}N$ and $\lambda_{\text{out}}RN$, respectively) by $O(N^p)$, we expect to see $O(N^{p+\varepsilon-1})$ error for arbitrarily small $\varepsilon$. In this case, we can assume that $N_C$ and $\tilde{N}_C$ only deviate from their ideal values by at most $1$. Hence, we can set $p\to 0$, so it makes perfect sense that we see the error scale as $N_C^{-1}$.

\vspace{0.5\baselineskip}

Before moving on, let us say more about the modular arithmetic artifacts. For example, in Figure \ref{fig:tr-dist-trend-schur-transformed-state-conversion-optimal-cloning-map}, where $R_C = 6/5$, we see two separate trend lines for $N_C\equiv 0\pmod{5}$ and $N_C\equiv 1\pmod{5}$, but if we plotted every single value of $N_C$ in this range, we would see $5$ trend lines, one for each possible remainder modulo $5$.

\vspace{0.5\baselineskip}

Across all three procedures, the trend line in the case of no round-off error (i.e., where $R_CN_C$ is an integer) has the smaller trace distance. This suggests that, even at this level of granularity, the closer $\tilde{N}_C/N_C$ is to the ideal ratio $R_C$, the more accurately you achieve conversion between Schur-transformed states.

\vspace{0.5\baselineskip}

However, regardless of modular arithmetic, all the trend lines still show a decay in the trace distance going roughly as $N^{-1}$. This is reassuring, since all of these cases are still contained within the $p\to 0$ case of the lemmas in Appendix \ref{sec:schur-transformed-state-conversion}, regardless of round-off error.

\vspace{0.5\baselineskip}

Of course, in the actual implementation of our i.i.d. conversion protocols, $N_C$ and $\tilde{N}_C$ have some random variation coming from their respective Schur sampling distributions. We will introduce this variation in the final subsection of this appendix.

\appsubsec{IID State Conversion: Error Analysis}
{subsec:numerical-analysis-iid-state-conversion-error-analysis}

We will now numerically simulate the full i.i.d. state conversion procedures, as described in Appendix \ref{sec:unified-presentation}\ref{subsec:three-step-procedure-unified-presentation}.

\vspace{0.5\baselineskip}

We must now incorporate the Schur sampling outcome distributions for the input and target i.i.d. qubit collections, which are given by $p(N,N_C,\lambda_{\text{in}})$ and $p(\tilde{N},\tilde{N}_C,\lambda_{\text{out}})$, respectively. This can be done in two ways:
\begin{itemize}
    \item \textbf{Monte Carlo simulation:} Randomly sample values $N_C$ and $\tilde{N}_C$ from the input and output Schur sampling distributions, compute the action of the channel $\mE_{(\cdot)}[N_C\to\tilde{N}_C]\left(\rho_C(N_C,\lambda_{\text{in}},\hat{n})\right)$, and then compute the resulting trace distance with the target state $\rho_C(\tilde{N}_C,\lambda_{\text{out}},\hat{n})$. You can then average this trace distance over the different samples.
    \item \textbf{Exact calculation:} Compute the probabilities $p(N,N_C,\lambda_{\text{in}})$ and $p(\tilde{N},\tilde{N}_C,\lambda_{\text{out}})$, compute the trace distance resulting from the channel $\mE_{(\cdot)}[N_C\to\tilde{N}_C]$ for \textit{every} possible value of $N_C$ and $\tilde{N}_C$, and then compute the resulting average trace distance. This ensures that there is no error coming from random variation, but it is much more computationally expensive.
\end{itemize}

Regardless of which method you use, the variation of Schur sampling outcomes means that much larger values of $N$ are needed to reliably see the asymptotic behavior of the trace distance. Furthermore, the quality of the regression will become more imperfect (although it will still be very convincing).

\vspace{0.5\baselineskip}

Technically, the trace distance we compute is not quite the trace distance between the output of the whole protocol and the target i.i.d. state. Rather, it is the trace distance between the ensemble of output states on the symmetric subspace and the corresponding ensemble of target Schur-transformed states on the symmetric subspace. This is actually the trace distance \emph{before} step \textbf{(iii)}, inverse Schur sampling. Fortunately, since step \textbf{(iii)} cannot increase the trace distance, the value you compute using this procedure is an \emph{upper bound} on the trace distance.

\vspace{0.5\baselineskip}

In this appendix, we use Monte Carlo simulation with $1000$ samples per data point. In particular, we proceed as follows:
\begin{itemize}
    \item We start with pairs of qubit counts $(N,\tilde{N})$ such that $\tilde{N}=RN$.
    \item On each trial, we randomly sample $N_C$ and $\tilde{N}_C$ from their respective Schur sampling outcome distributions.
    \item For each trial, we compute the trace distance between the output state $\mE_{(\cdot)}[N_C\to\tilde{N}_C]\left(\rho_C(N_C,\lambda_{\text{in}},\hat{n})\right)$ and the target Schur-transformed state $\rho_C(\tilde{N}_C,\lambda_{\text{out}},\hat{n})$.
    \item Finally, we average over all the samples to produce an empirical estimate of the pre-step-\textbf{(iii)} trace distance.
    \item We also  compute the standard error (standard deviation, divided by the square root of the number of samples) of all the trace distance values.
    \item We plot each data point centered at the mean, with an error bar whose half-width is the standard error.
\end{itemize}

\vspace{0.5\baselineskip}

We present one example of each of the three linear-rate conversion protocols:
\begin{itemize}
    \item \textbf{Concentration:} $\lambda_{\text{in}} = 1/2$, $\lambda_{\text{out}} = 2/3$, $R^{\text{conc}}(\lambda_{\text{in}}\to\lambda_{\text{out}}) = 3/8$ (refer to Figure \ref{fig:tr-dist-trend-iid-state-conversion-conc})
    \item \textbf{Dilution:} $\lambda_{\text{in}} = 2/3$, $\lambda_{\text{out}} = 1/2$, $R^{\text{dilut}}(\lambda_{\text{in}}\to\lambda_{\text{out}}) = 8/5$ (refer to Figure \ref{fig:tr-dist-trend-iid-state-conversion-dilut})
    \item \textbf{Measure-and-prepare conversion:} $\lambda_{\text{in}} = 1/2$, $\lambda_{\text{out}} = 2/3$, $R^{\text{MP}}(\lambda_{\text{in}}\to\lambda_{\text{out}}) = 1/8$ (refer to Figure \ref{fig:tr-dist-trend-iid-state-conversion-mp})
\end{itemize}

\begin{figure}
    \includegraphics[scale=0.5]{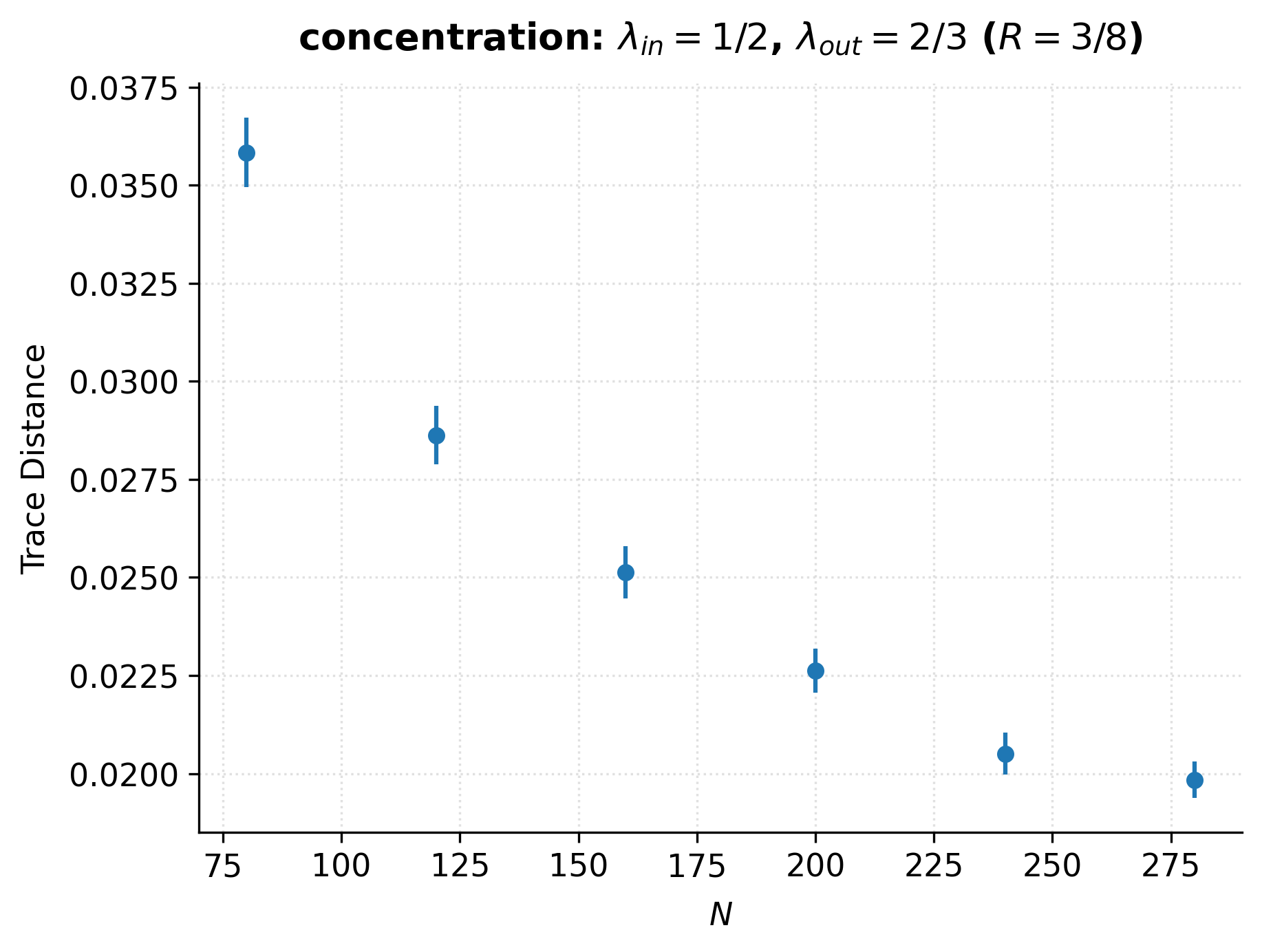}
    \includegraphics[scale=0.5]{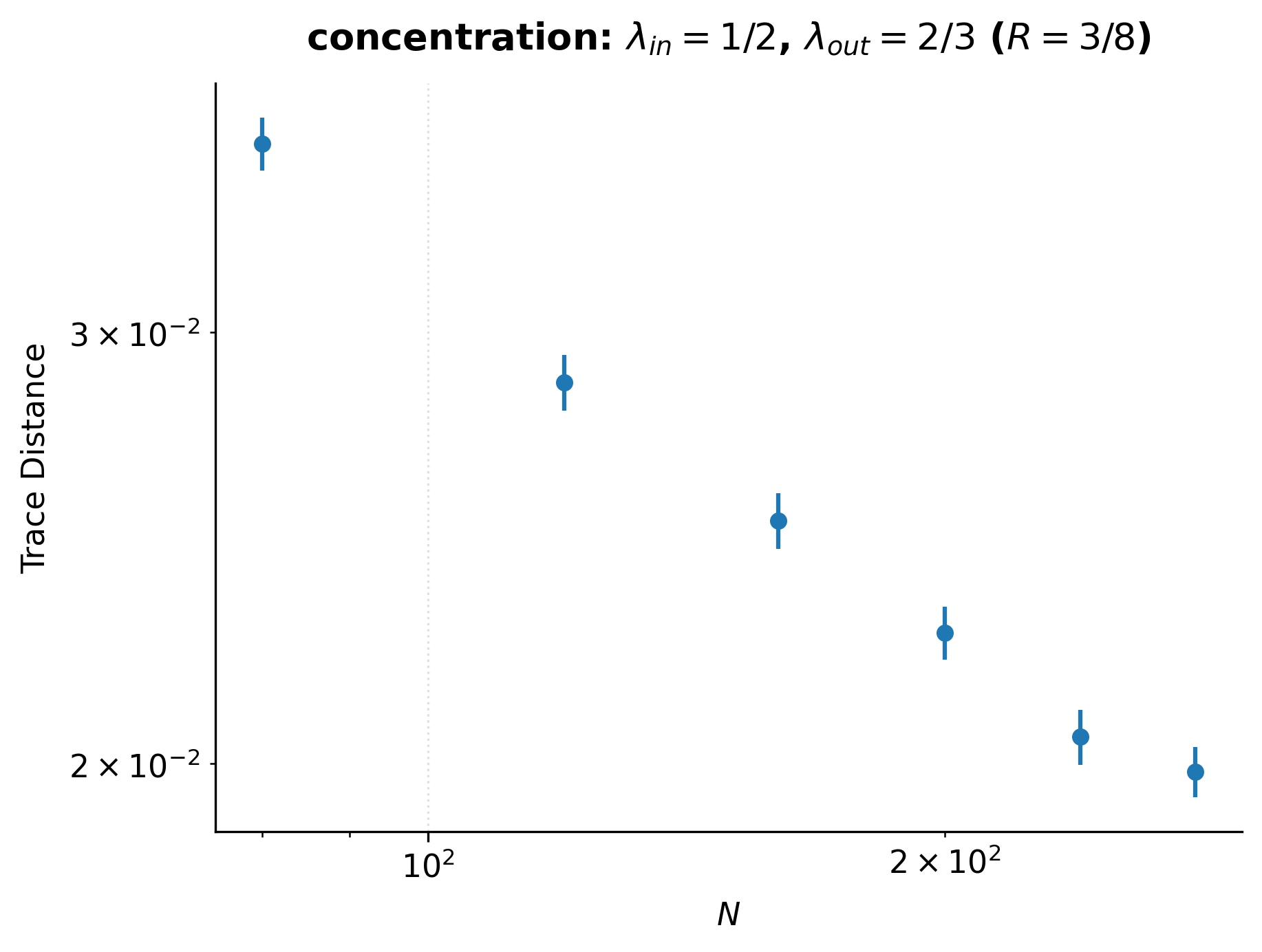}
    \caption{(LEFT) Trace distance as a function of input qubit count for concentration of i.i.d. qubits. In particular, we take pairs of qubit counts $(N,\tilde{N})$ such that $\tilde{N} = RN$, and we use Monte Carlo simulation to estimate an upper bound on the trace distance between the target state $\rho(2/3,\hat{n})^{\otimes\tilde{N}}$ and the result when the concentration protocol is applied to $\rho(1/2,\hat{n})^{\otimes N}$. (RIGHT) The same plot, but with logarithmic scales for both axes. When we run linear regression on the transformed data (which reveals a power law relation), we obtain $d_{\text{Tr}} = 0.2885 \times N^{-0.4796}$.}
    \label{fig:tr-dist-trend-iid-state-conversion-conc}
\end{figure}

\begin{figure}
    \includegraphics[scale=0.5]{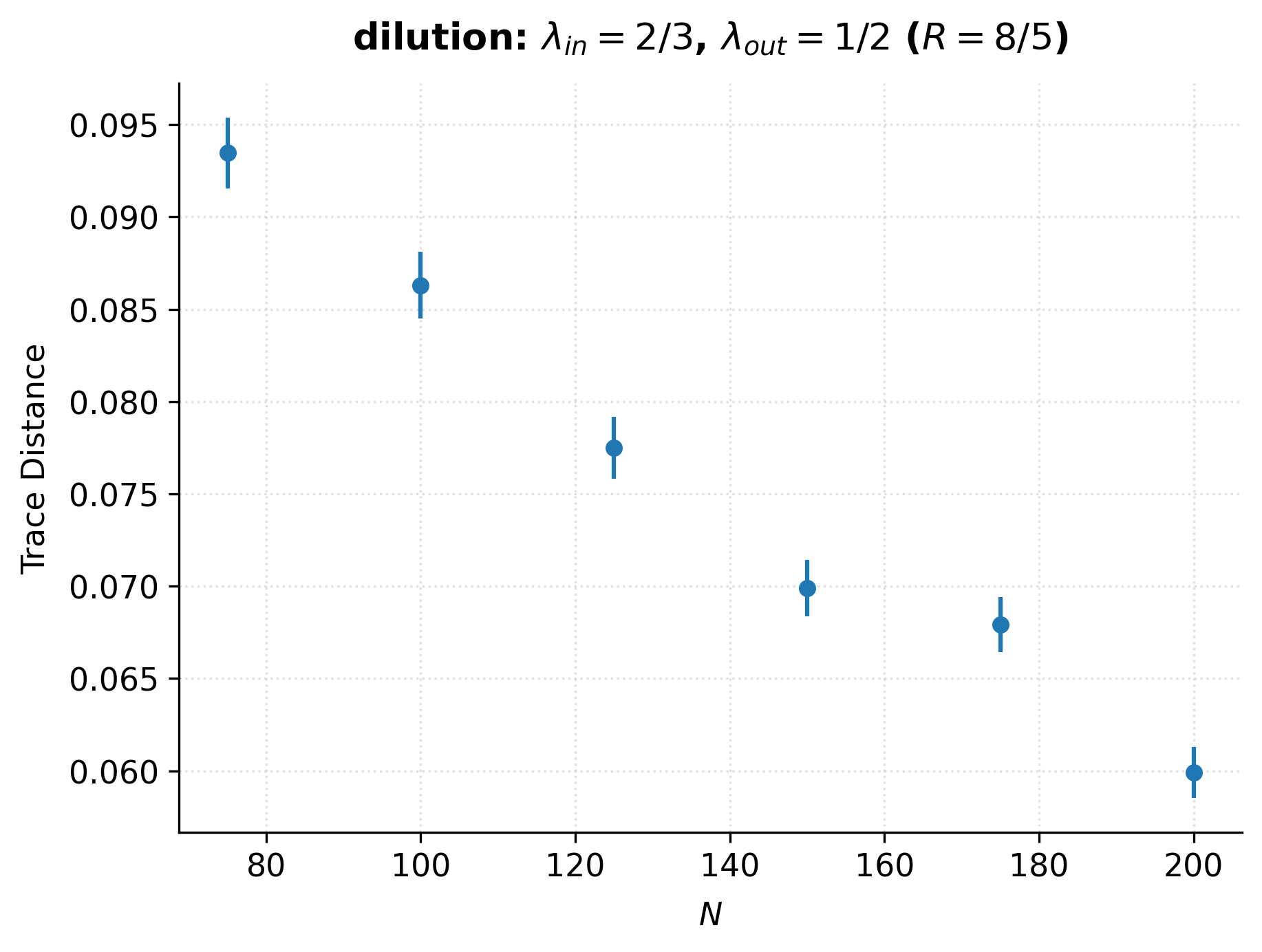}
    \includegraphics[scale=0.5]{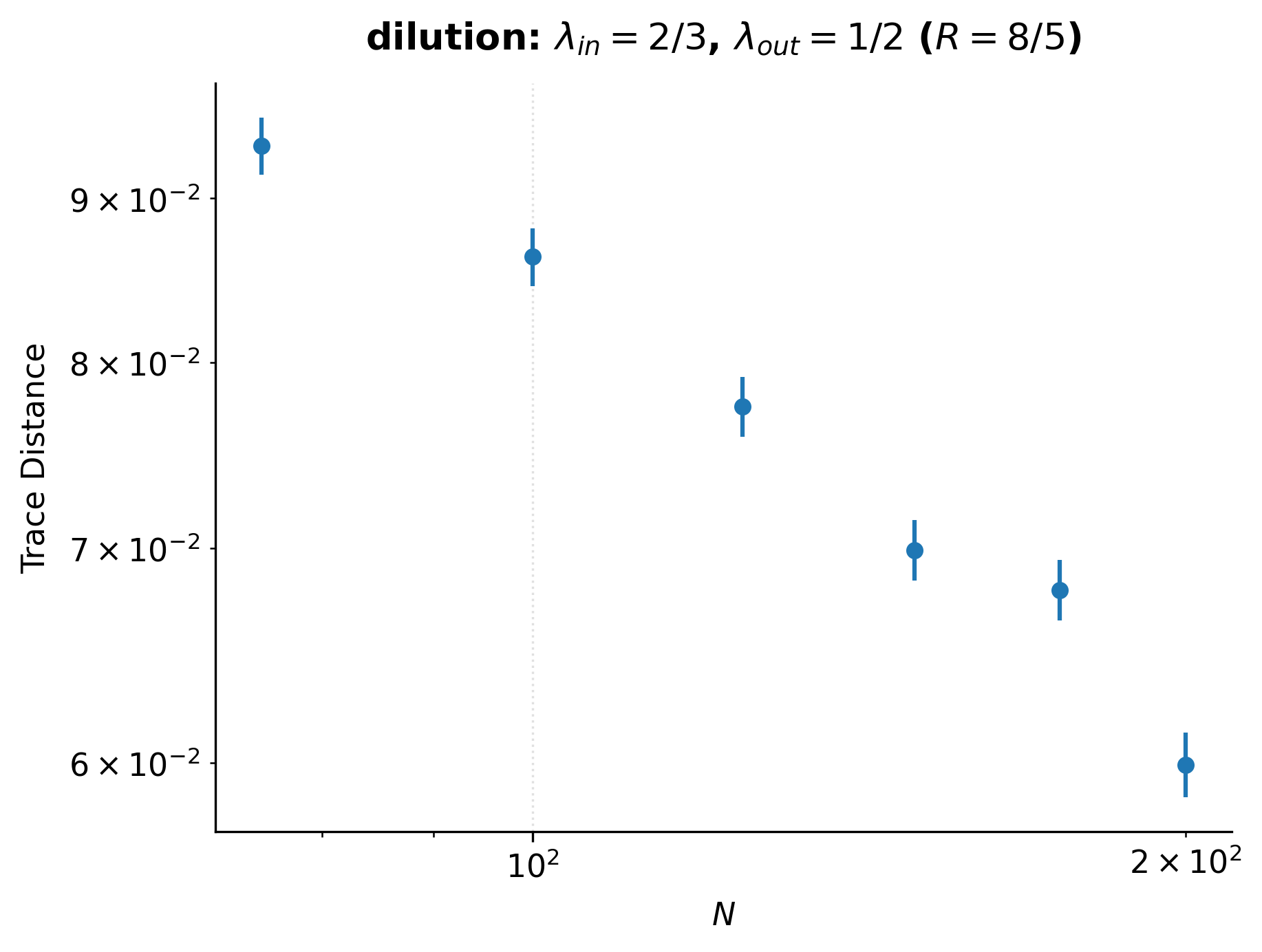}
    \caption{(LEFT) Trace distance as a function of input qubit count for dilution of i.i.d. qubits. In particular, we take pairs of qubit counts $(N,\tilde{N})$ such that $\tilde{N} = RN$, and we use Monte Carlo simulation to estimate an upper bound on the trace distance between the target state $\rho(1/2,\hat{n})^{\otimes\tilde{N}}$ and the result when the concentration protocol is applied to $\rho(2/3,\hat{n})^{\otimes N}$. (RIGHT) The same plot, but with logarithmic scales for both axes. When we run linear regression on the transformed data (which reveals a power law relation), we obtain $d_{\text{Tr}} = 0.6479 \times N^{-0.4427}$.}
    \label{fig:tr-dist-trend-iid-state-conversion-dilut}
\end{figure}

\begin{figure}
    \includegraphics[scale=0.5]{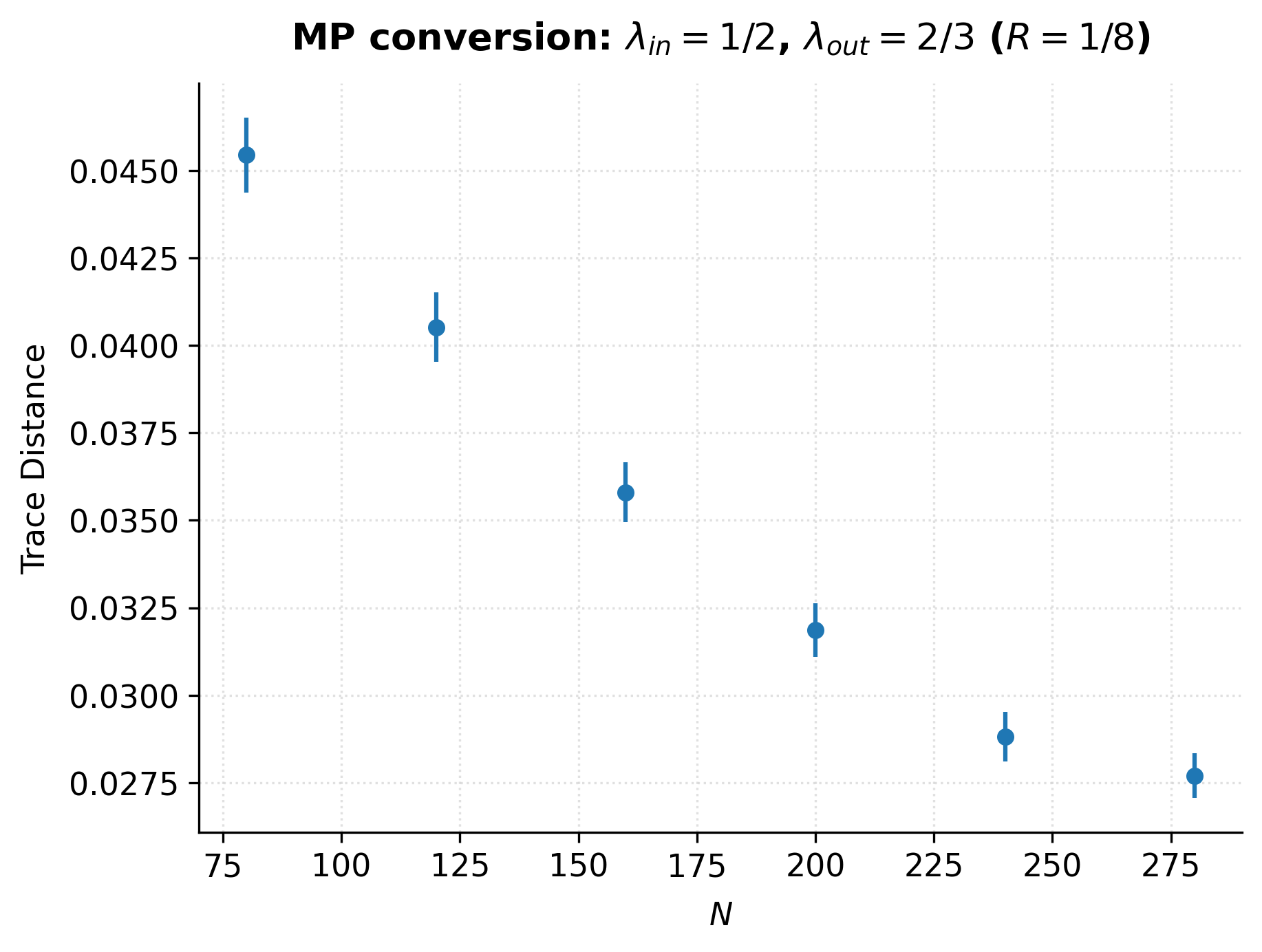}
    \includegraphics[scale=0.5]{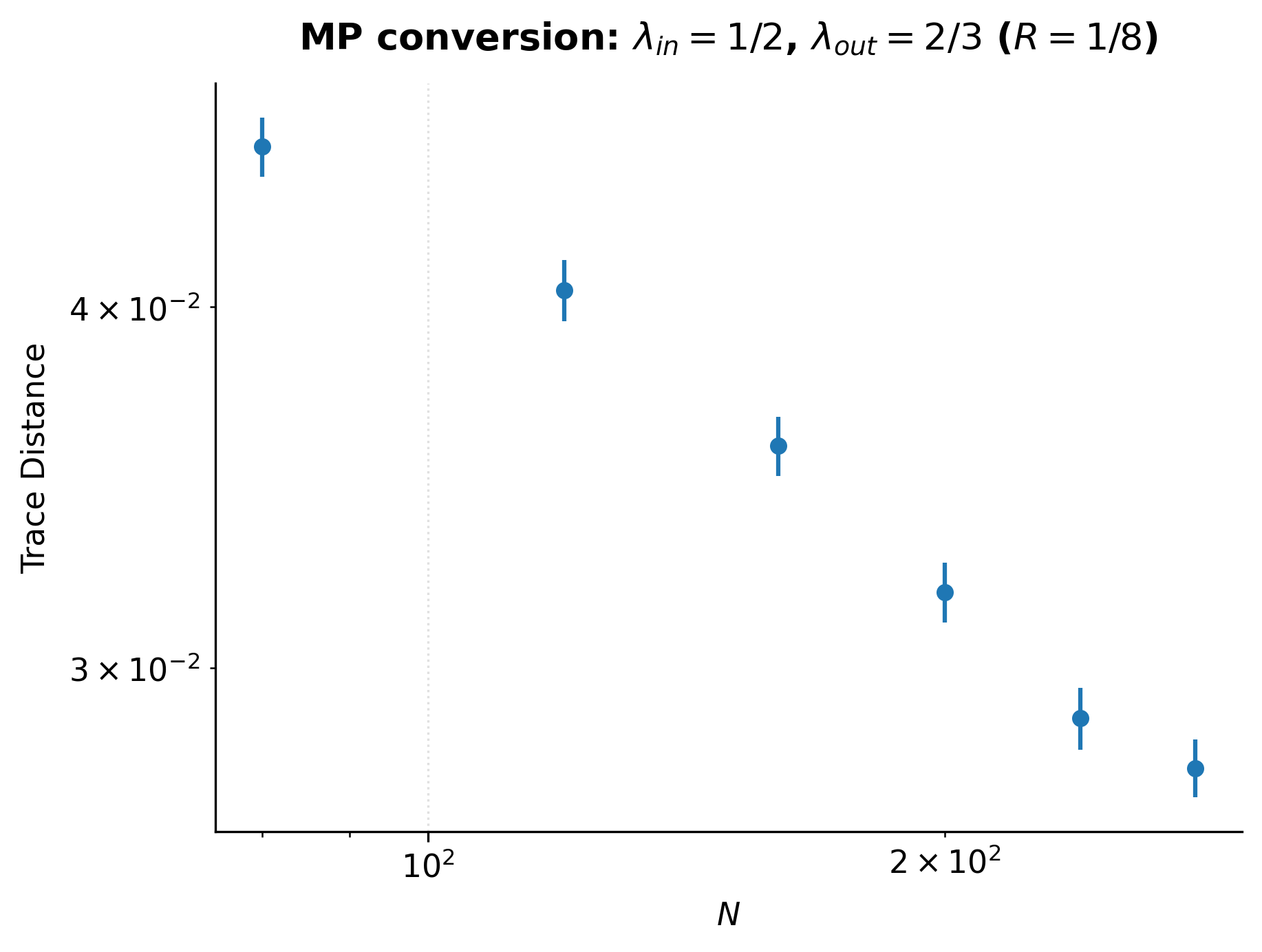}
    \caption{(LEFT) Trace distance as a function of input qubit count for measure-and-prepare conversion of i.i.d. qubits. In particular, we take pairs of qubit counts $(N,\tilde{N})$ such that $\tilde{N} = RN$, and we use Monte Carlo simulation to estimate an upper bound on the trace distance between the target state $\rho(2/3,\hat{n})^{\otimes\tilde{N}}$ and the result when the measure-and-prepare protocol is applied to $\rho(1/2,\hat{n})^{\otimes N}$. (RIGHT) The same plot, but with logarithmic scales for both axes. When we run linear regression on the transformed data (which reveals a power law relation), we obtain $d_{\text{Tr}} = 0.2888 \times N^{-0.4160}$.}
    \label{fig:tr-dist-trend-iid-state-conversion-mp}
\end{figure}

We can see a consistent pattern across all three protocols. Once we incorporate Schur sampling and inverse Schur sampling to execute the full protocol, the decay of trace distance now occurs much more slowly. In particular, the trace distance decay roughly as $N^{-1/2}$, rather than roughly as $N^{-1}$.

\vspace{0.5\baselineskip}

This phenomenon is nicely explained by the natural variation of Schur sampling outcomes. In particular, $N_C$ and $\tilde{N}_C$ are no longer related by \textit{exactly} the Schur-transformed conversion rate, but rather only \textit{approximately}. In particular, both $N_C$ and $\tilde{N}_C$ now fluctuate by $\Theta(\sqrt{N})$, so we must now plug $p=\frac{1}{2}+\varepsilon'$. As a result, we can only guarantee $O\left(N^{\varepsilon-1/2}\right)$ error, and this is borne out in the numerical simulations.

\newpage

\appsec{Monotonicity of the Complexified RLD Fisher Information Matrix}
{sec:monotonicity-complexified-rld}

In this appendix, we prove that the complexified RLD Fisher information matrix for a family of quantum states is decreasing under quantum channels. This result is a special case of the classification by Petz of monotone inner products \cite{petz1996}, which we discuss further in Appendix \ref{sec:qfi-metrics}. In particular, our focus in this appendix is this theorem:

\begin{theorem}[Monotonicity of RLD matrix under CPTP maps \cite{petz1996}]
\label{thm:rld-monotone-cptp}
The complexified RLD Fisher information matrix is decreasing under completely positive trace-preserving maps. More precisely, if $\rho(\vec{x})$ is a parametrized family of quantum states, and $\mE$ is a CPTP map, then
\begin{equation}
    \text{RLD}[\mE(\rho)](\vec{x}) \le \text{RLD}[\rho](\vec{x}),
\end{equation}
where the inequality should be understood as a matrix inequality for Hermitian matrices.
\end{theorem}

Theorem \ref{thm:rld-monotone-cptp} is an immediate consequence of a slightly more general result proven by Petz \cite{petz1996}:

\begin{theorem}[Monotonicity of RLD matrix under $2$-positive TP maps \cite{petz1996}]
\label{thm:rld-monotone-2-positive}
Let $\mM(\mH)$ denote the space of all complex matrices on a Hilbert space $\mH$. Also, let $\mM_{\ge 0}(\mH)$ and $\mM_{>0}(\mH)$ respectively denote the space of all positive semidefinite (PSD) matrices on $\mH$ and positive definite (PD) matrices on $\mH$. Suppose that $D\in\mM_{>0}(\mH_{\text{in}})$ and $A\in\mM(\mH_{\text{in}})$. Also suppose that $T:\mM(\mH_{\text{in}})\rightarrow\mM(\mH_{\text{out}})$ is a $2$-positive trace-preserving map such that $T(D)\in\mM_{>0}(\mH_{\text{out}})$. Then
\begin{equation}
    K_{T(D)}(T(A),T(A)) \le K_D(A,A),
\end{equation}
where $K_M$ is an inner product defined by
\begin{equation}
    K_M(P,Q) = \text{Tr}\left[M^{-1}P^\dagger Q\right].
\end{equation}
\end{theorem}

As a reminder, a linear map $T:\mM(\mH_{\text{in}})\rightarrow\mM(\mH_{\text{out}})$ is said to be $\mathbf{k}$\textbf{-positive} if, for an arbitrary positive semidefinite (PSD) matrix $\rho_{AB}\in\mM_{\ge 0}(\mH_{\text{in}}\otimes\mH_{\text{aux}})$ with $\text{dim}(\mH_{\text{aux}})=k$,
\begin{equation}
    \left(T_A\otimes\text{id}_B\right)(\rho_{AB})\in\mM_{\ge 0}(\mH_{\text{out}}\otimes\mH_{\text{aux}}).
\end{equation}
In other words, $T$ must always map valid quantum states to valid quantum states, even if the input of $T$ is entangled with an auxiliary system of Hilbert space dimension $k$. A linear map is \textbf{positive} if it is $1$-positive, and it is \textbf{completely positive} if it is $k$-positive for all positive integers $k$. The characterization of completely positive linear maps by Choi immediately implies that $\text{dim}(\mH_{\text{in}})$-positivity is sufficient for complete positivity \cite{choi1975}.

\vspace{0.5\baselineskip}

For the reader's benefit, we restate Petz's proof of Theorem \ref{thm:rld-monotone-2-positive} here, with a bit more explanation than what Petz provided:

\begin{proof}[Proof of Theorem \ref{thm:rld-monotone-2-positive}]
First, notice that the block matrix
\begin{equation}
    \begin{bmatrix}
        D & A^\dagger \\
        A & AD^{-1}A^\dagger
    \end{bmatrix}
\end{equation}
is PSD. To see this, notice that
\begin{align}
    \begin{bmatrix}
        u^\dagger & v^\dagger
    \end{bmatrix}\begin{bmatrix}
        D & A^\dagger \\
        A & AD^{-1}A^\dagger
    \end{bmatrix}\begin{bmatrix}
        u \\ v
    \end{bmatrix} &= u^\dagger Du + u^\dagger A^\dagger v + v^\dagger Au + v^\dagger AD^{-1}A^\dagger v \\
    &= \left\Vert D^{1/2}u + D^{-1/2}A^\dagger v\right\Vert_2^2 \ge 0.
\end{align}
Since $T$ is $2$-positive, the application of $T\otimes\text{id}$ to this block matrix, which is
\begin{equation}
    \begin{bmatrix}
        T(D) & T(A)^\dagger \\
        T(A) & T(AD^{-1}A^\dagger)
    \end{bmatrix},
\end{equation}
must also be PSD. (We also used $T(A^\dagger) = T(A)^\dagger$, since $T$ is positive and hence Hermiticity-preserving.)

\vspace{0.5\baselineskip}

In general, if the block matrix
\begin{equation}
    \begin{bmatrix}
        P & Q^\dagger \\
        Q & R
    \end{bmatrix}
\end{equation}
is PSD, then it is necessary to have $QP^{-1}Q^\dagger\le R$. To see this, suppose that the latter property does not hold. Then there exists a unit eigenvector $v$ of $QP^{-1}Q^\dagger - R$ with eigenvalue $\eta > 0$, from which it follows that
\begin{equation}
    \begin{bmatrix}
        v^\dagger QP^{-1} & -v^\dagger
    \end{bmatrix}\begin{bmatrix}
        P & Q^\dagger \\
        Q & R
    \end{bmatrix}\begin{bmatrix}
        P^{-1}Q^\dagger v \\ -v
    \end{bmatrix} = v^\dagger(-QP^{-1}Q^\dagger + R)v = -\eta < 0.
\end{equation}
The above statements together imply that
\begin{equation}
    T(A)T(D)^{-1}T(A)^\dagger \le T(AD^{-1}A^\dagger).
\end{equation}
Hence, the trace of the left side must be less than or equal to the trace of the right side. Finally, using the cyclic property of trace and the fact that $T$ is trace-preserving, we obtain
\begin{equation}
    \text{Tr}\left[T(D)^{-1}T(A)^\dagger T(A)\right] \le \text{Tr}\left[D^{-1}A^\dagger A\right],
\end{equation}
exactly as desired.
\end{proof}

Armed with Theorem \ref{thm:rld-monotone-2-positive}, we can make quick work of Theorem \ref{thm:rld-monotone-cptp}:

\begin{proof}[Proof of Theorem \ref{thm:rld-monotone-cptp}]
Since $\mE$ is completely positive, it is also $2$-positive. Let $\mu\in\Cbb^p$ be an arbitrary direction in the ($p$-dimensional) parameter space. Apply Theorem \ref{thm:rld-monotone-2-positive} to $\mE$ with the choices $D = \rho(\vec{x})$ and $A = \partial_{\mu^*}\rho(\vec{x}) = (\partial_{\mu}\rho(\vec{x}))^\dagger$. The result is
\begin{equation}
    \text{Tr}\left[\mE(\rho(\vec{x}))^{-1}\mE(\partial_\mu\rho(\vec{x}))\mE(\partial_\mu\rho(\vec{x}))^\dagger\right] \le \text{Tr}\left[\rho(\vec{x})^{-1}(\partial_\mu\rho(\vec{x}))(\partial_\mu\rho(\vec{x}))^\dagger\right].
\end{equation}
Using the cyclic property of trace and the fact that $\mE$ is linear and Hermiticity-preserving, we can rewrite this as
\begin{equation}
    \text{Tr}\left[\partial_{\mu^*}(\mE(\rho(\vec{x})))\mE(\rho(\vec{x}))^{-1}\partial_\mu(\mE(\rho(\vec{x})))\right] \le \text{Tr}\left[\partial_{\mu^*}\rho(\vec{x})\rho(\vec{x})^{-1}\partial_\mu\rho(\vec{x})\right].
\end{equation}
In other words, for an arbitrary vector $\mu\in\Cbb^p$,
\begin{equation}
    \mu^\dagger\left[\text{RLD}[\mE(\rho)](\vec{x})\right]\mu \le \mu^\dagger\left[\text{RLD}[\rho](\vec{x})\right]\mu,
\end{equation}
which is equivalent to saying that
\begin{equation}
    \text{RLD}[\mE(\rho)](\vec{x}) \le \text{RLD}[\rho](\vec{x}),
\end{equation}
exactly as desired.
\end{proof}

Is $2$-positivity essential for this result? The answer is a resounding YES! To show this, we construct a counterexample for $1$-positivity. Consider the family of states
\begin{equation}
    \sigma(a,b) = \frac{\Ibb + \lambda\left(aX + bY + \sqrt{1-a^2-b^2}Z\right)}{2}
\end{equation}
in the vicinity of $(0,0)$, and consider the map $T$ that reflects every point in the Bloch sphere across the origin, i.e.,
\begin{equation}
    T\left(c_i\Ibb + c_xX + c_yY + c_zZ\right) = c_i\Ibb - c_xX - c_yY - c_zZ.
\end{equation}
The map $T$ is positive, because in the Bloch sphere picture, the space of density operators is the closed unit ball, which is preserved under reflection about the origin. However, it is not $2$-positive; one way to see this is to apply $T\otimes\text{id}$ to the singlet state $\ket{\Psi^-}\bra{\Psi^-}$:
\begin{align}
    \ket{\Psi^-}\bra{\Psi^-} &= \frac{\Ibb - \text{Swap}}{2} = \frac{\Ibb - XX - YY - ZZ}{4} \\
    \implies \left(T\otimes\text{id}\right)\left(\ket{\Psi^-}\bra{\Psi^-}\right) &= \frac{\Ibb + XX + YY + ZZ}{4} = \frac{\text{Swap}}{2}.
\end{align}
Since $\text{Swap}$ is not PSD, $T$ is not $2$-positive. Sure enough, the RLD matrices for the input and output families at $(0,0)$ are
\begin{align}
    \text{RLD}[\sigma](0,0) &= \frac{\lambda^2}{1-\lambda^2}\begin{bmatrix}
        1 & +\lambda i \\
        -\lambda i & 1
    \end{bmatrix} \\
    \text{RLD}[T(\sigma)](0,0) &= \frac{\lambda^2}{1-\lambda^2}\begin{bmatrix}
        1 & -\lambda i \\
        +\lambda i & 1
    \end{bmatrix}.
\end{align}
Notice that $\text{RLD}[T(\sigma)](0,0)$ is NOT less than or equal to $\text{RLD}[\sigma](0,0)$, because the eigenvalues $\frac{\lambda^2}{1\pm\lambda}$ stay the same, but the eigenvectors switch places. Hence RLD Fisher information is not decreasing under the map $T$. As a result, complexified RLD offers a valuable insight into the distinction between positivity and complete positivity (or even more specifically, the distinction between $1$-positivity and $2$-positivity).

\vspace{0.5\baselineskip}

In fact, we can achieve the same effect by using everyone's favorite positive-but-not-completely-positive map, namely \textbf{transpose}. In general, if $D$ is a real positive definite matrix (for example, if we write everything in the eigenbasis of $D$), then transposition results in the RLD Fisher information matrix being transposed as well, which we can see as follows:
\begin{equation}
    K_{D^\intercal}(A^\intercal,B^\intercal) = \text{Tr}\left[(D^\intercal)^{-1}A^\intercal B^\intercal\right] = \text{Tr}\left[BAD^{-1}\right] = \text{Tr}[D^{-1}BA] = K_D(B,A).
\end{equation}
(Since the RLD matrix is Hermitian, the RLD matrix being transposed is the same as its being complex conjugated.) If the RLD matrix is real, then transposition leaves the RLD matrix unaffected, but if it is complex, then the transposition violates monotonicity, since a complex Hermitian matrix is not greater than or equal to its transpose unless it is a real symmetric matrix.

\vspace{0.5\baselineskip}

The vast majority of distance and divergence measures used in quantum information theory cannot ``tell the difference'' between a density matrix $\rho$ and its transpose (equivalently, complex conjugate) $\rho^\intercal = \rho^*$, in the sense that $D(\rho \,||\, \sigma) = D(\rho^* \,||\, \sigma^*)$. As a result, they do not provide any information-theoretic obstruction for transforming a state into its transpose. In contrast, RLD Fisher information does provide such an obstruction, which may be a fact of independent interest to quantum information theorists.

\newpage

\appsec{Single-Shot Qubit Distillation}
{sec:single-shot-qubit-distillation}

In this appendix, we use complexified RLD Fisher information to provide an information-theoretic justification for the performance of optimal single-shot qubit distillation, as derived by Cirac, Ekert, and Macchiavello nearly three decades ago \cite{Cirac1999}.

\vspace{0.5\baselineskip}

The first step is to compute the complexified RLD associated with the family of single-qubit states at a fixed purity level $\lambda$, which are related by the defining representation of $\mathrm{SU}(2)$. We will compute the RLD Fisher information in the vicinity of $\rho(\lambda,\hat{z}) = (\Ibb + \lambda Z)/2$, the state pointing in the $+z$-direction.

\vspace{0.5\baselineskip}

As a reminder, the complexified RLD Fisher information matrix is defined entrywise via the formula
\begin{equation}
    \text{RLD}_{\mu\nu}[\rho](\vec{x}) = \text{Tr}\left[(\partial_\mu\rho(\vec{x}))\rho(\vec{x})^{-1}(\partial_\nu\rho(\vec{x}))\right].
\end{equation}
This means that we need to choose a parametrization of our input and output state families. Fortunately, we can choose any parametrization we like, so long as it matches for the input and output families. This is because the RLD Fisher information matrix should be understood as a bilinear form, or even more generally, as a $(0,2)$-tensor. If the coordinates change via a transformation with Jacobian matrix $S$, then the RLD Fisher information matrix changes via the congruence transformation $M \mapsto S^\intercal MS$. Since both the input and output RLD Fisher information matrices undergo the same congruence transformation, the application of monotonicity yields exactly the same result.

\vspace{0.5\baselineskip}

One way to parametrize the family of states in the neighborhood of $\rho(\lambda,\hat{z})$ is as follows:
\begin{equation}
    \sigma(a,b) = \frac{\Ibb + \lambda\left(aX + bY + \sqrt{1-a^2-b^2}Z\right)}{2}.
\end{equation}
We first compute the partial derivatives of this state and also the inverse of the state at $(0,0)$:
\begin{align}
    \partial_a\sigma(0,0) &= \frac{\lambda}{2}X \quad\quad\quad\quad\quad \partial_b\sigma(0,0) = \frac{\lambda}{2}Y \\
    \sigma(0,0)^{-1} &= \left(\frac{\Ibb + \lambda Z}{2}\right)^{-1} = \frac{2}{1-\lambda^2}(\Ibb - \lambda Z).
\end{align}
We can now compute the complexified RLD Fisher information matrix at $(0,0)$:
\begin{equation}
    \text{RLD}^{\Cbb}[\sigma](0,0) = \frac{\lambda^2}{1-\lambda^2}\begin{bmatrix}
        1 & +\lambda i \\
        -\lambda i & 1
    \end{bmatrix}.
\end{equation}
Interestingly, the RLD matrix always has the same eigenvectors regardless of $\lambda$. Furthermore, it has two distinct positive eigenvalues, which we call $\text{RLD}_{\text{max}}$ and $\text{RLD}_{\text{min}}$ for convenience:
\begin{equation}
    \text{RLD}_{\text{max}}(\lambda) \coloneqq \frac{\lambda^2}{1-\lambda} \quad\quad\quad \text{RLD}_{\text{min}}(\lambda) \coloneqq \frac{\lambda^2}{1+\lambda}.
\end{equation}

\vspace{0.5\baselineskip}

First, we exploit the fact that we are free to impose $\mathrm{SU}(2)$ covariance on the protocol. In other words, if we rotate all our input qubits by some $U\in\mathrm{SU}(2)$, then the output qubit must be rotated by that same $U$. If the input qubits have Bloch vector $\lambda\hat{n}$, then for any $\alpha\in\Rbb$, the unitary $U = \exp(i\alpha\hat{n}\cdot\vec{\sigma})$ fixes the input qubits, so it must also fix the output qubit. Therefore, the output qubit must point in a direction parallel to the input qubits. In other words, the output qubit must have Bloch vector $\tilde{\lambda}\hat{n}$ for some output purity level $-1\le\tilde{\lambda}\le 1$. (If $\tilde{\lambda} > 0$, then the output qubit points in the \textit{same} direction, whereas if $\tilde{\lambda} < 0$, then the output qubit points in the \textit{opposite} direction.)

\vspace{0.5\baselineskip}

This fact implies that the RLD Fisher information assumes exactly the same functional form for an input qubit and the output qubit, except for different purity levels $\lambda$ and $\tilde{\lambda}$. Furthermore, we use the fact that RLD Fisher information is additive under tensor product, meaning that for $N$ i.i.d. qubits, we simply multiply the RLD Fisher information by $N$. We conclude that the input and output RLD Fisher information matrices are
\begin{align}
    \text{RLD}_{\text{in}} &= N\cdot\frac{\lambda^2}{1-\lambda^2}\begin{bmatrix}
        1 & +i\lambda \\ -i\lambda & 1
    \end{bmatrix} \\
    \text{RLD}_{\text{out}} &= \frac{\tilde{\lambda}^2}{1-\tilde{\lambda}^2}\begin{bmatrix}
        1 & +i\tilde{\lambda} \\ -i\tilde{\lambda} & 1
    \end{bmatrix}.
\end{align}
Therefore, the monotonicity of RLD Fisher information implies that $\text{RLD}_{\text{out}} \le \text{RLD}_{\text{in}}$ as a matrix inequality. Since both matrices have the same eigenvectors, this matrix inequality is equivalent to scalar inequalities on the corresponding eigenvalues:
\begin{align}
    \text{RLD}_{\text{max}}(\tilde{\lambda}) \le N\cdot\text{RLD}_{\text{max}}(\lambda) \\
    \text{RLD}_{\text{min}}(\tilde{\lambda}) \le N\cdot\text{RLD}_{\text{min}}(\lambda).
\end{align}
Interestingly, we will find use cases for both of these!

\appsubsec{Right-Way Distillation}
{subsec:right-way-distillation}

Suppose that we purify $N$ qubits with purity level $\lambda$ into a single qubit with purity level $\tilde{\lambda}$, where we want $\tilde{\lambda}$ to be as close to $1$ as possible. This was the problem studied by Cirac et al. \cite{Cirac1999}, and we may call this task \textbf{right-way distillation}.

\vspace{0.5\baselineskip}

In this setting, we will focus on the $\text{RLD}_{\text{max}}$ inequality. Writing out this inequality more explicitly yields
\begin{align}
    \frac{\tilde{\lambda}^2}{1-\tilde{\lambda}} \le N\cdot\frac{\lambda^2}{1-\lambda}.
\end{align}
Rearranging this (and applying the obvious inequality $\tilde{\lambda} < 1$) yields
\begin{equation}
    \tilde{\lambda}^2 + \frac{\lambda^2N}{1-\lambda}\tilde{\lambda} - \frac{\lambda^2N}{1-\lambda} \le 0.
\end{equation}
This is a quadratic inequality in $\tilde{\lambda}$, which we can solve using the quadratic formula. In general, if $ax^2 + bx + c\le 0$ and $a > 0$, then $\frac{-b-\sqrt{b^2-4ac}}{2a}\le x\le \frac{-b+\sqrt{b^2-4ac}}{2a}$. The former inequality turns out to be trivially satisfied, but the latter inequality yields
\begin{equation}
    \tilde{\lambda} \le \frac{1}{2}\Bigg\{-\frac{\lambda^2}{1-\lambda}N + \sqrt{\frac{\lambda^2}{1-\lambda}N\left(\frac{\lambda^2}{1-\lambda}N + 4\right)}\Bigg\}.
\end{equation}
The square root can be approximated as
\begin{equation}
    \sqrt{\frac{\lambda^2}{1-\lambda}N\left(\frac{\lambda^2}{1-\lambda}N + 4\right)} = \frac{\lambda^2}{1-\lambda}N + 2 - \frac{2(1-\lambda)}{\lambda^2}\frac{1}{N} + O\left(N^{-2}\right).
\end{equation}
Applying this approximation to the above inequality yields
\begin{align}
    \tilde{\lambda} &\le \frac{1}{2}\Bigg\{-\frac{\lambda^2}{1-\lambda}N + \sqrt{\frac{\lambda^2}{1-\lambda}N\left(\frac{\lambda^2}{1-\lambda}N + 4\right)}\Bigg\} \\
    &= 1 - \frac{1-\lambda}{\lambda^2}\frac{1}{N} + O\left(N^{-2}\right).
\end{align}
Finally, we relate $\tilde{\lambda}$ to the output fidelity to obtain an upper bound on the output fidelity:
\begin{equation}
    \text{Fid} = \frac{1+\tilde{\lambda}}{2} \le 1 - \frac{1-\lambda}{2\lambda^2}\frac{1}{N} + O\left(N^{-2}\right).
\end{equation}
This matches the fidelity obtained by Cirac et al. \cite{Cirac1999}. In particular, they achieve this fidelity by first performing Schur sampling and discarding the singlet states, and then just keeping one of the remaining qubits and throwing away the rest. In particular, the final step is simply an application of the discarding map $\mE_{\text{discard}}[N_C\to 1]$ (see Appendix \ref{sec:four-important-channels}\ref{subsec:discarding-map-unified-presentation}) on the $N_C$ remaining qubits in the symmetric subspace.

\appsubsec{Wrong-Way Distillation}
{subsec:wrong-way-distillation}

Now suppose that we actually want to distill a qubit pointing in the \textit{opposite} direction as the input qubits. We may call this task \textbf{wrong-way distillation}. This is equivalent to saying that we want $\tilde{\lambda}$ to be as close to $-1$ as possible.

\vspace{0.5\baselineskip}

In this setting, we will instead use the $\text{RLD}_{\text{min}}$ inequality. Writing out this inequality more explicitly yields
\begin{align}
    \frac{\tilde{\lambda}^2}{1+\tilde{\lambda}} \le N\cdot\frac{\lambda^2}{1+\lambda}.
\end{align}
Rearranging this (and applying the obvious inequality $\tilde{\lambda} > -1$) yields
\begin{equation}
    \tilde{\lambda}^2 - \frac{\lambda^2N}{1+\lambda}\tilde{\lambda} - \frac{\lambda^2N}{1+\lambda} \le 0.
\end{equation}
This is a quadratic inequality in $\tilde{\lambda}$, which we can solve using the quadratic formula. In general, if $ax^2 + bx + c\le 0$ and $a > 0$, then $\frac{-b-\sqrt{b^2-4ac}}{2a}\le x\le \frac{-b+\sqrt{b^2-4ac}}{2a}$. The latter inequality turns out to be trivially satisfied, but the former inequality yields
\begin{equation}
    \tilde{\lambda} \ge \frac{1}{2}\Bigg\{\frac{\lambda^2}{1+\lambda}N - \sqrt{\frac{\lambda^2}{1+\lambda}N\left(\frac{\lambda^2}{1+\lambda}N + 4\right)}\Bigg\}.
\end{equation}
The square root can be approximated as
\begin{equation}
    \sqrt{\frac{\lambda^2}{1+\lambda}N\left(\frac{\lambda^2}{1+\lambda}N + 4\right)} = \frac{\lambda^2}{1+\lambda}N + 2 - \frac{2(1+\lambda)}{\lambda^2}\frac{1}{N} + O\left(N^{-2}\right).
\end{equation}
Applying this approximation to the above inequality yields
\begin{align}
    \tilde{\lambda} &\ge \frac{1}{2}\Bigg\{\frac{\lambda^2}{1+\lambda}N - \sqrt{\frac{\lambda^2}{1+\lambda}N\left(\frac{\lambda^2}{1+\lambda}N + 4\right)}\Bigg\} \\
    &= -1 + \frac{1+\lambda}{\lambda^2}\frac{1}{N} + O\left(N^{-2}\right).
\end{align}
Finally, we relate $\tilde{\lambda}$ to the output fidelity to obtain an upper bound on the output fidelity:
\begin{equation}
    \text{Fid} = \frac{1-\tilde{\lambda}}{2} \le 1 - \frac{1+\lambda}{2\lambda^2}\frac{1}{N} + O\left(N^{-2}\right).
\end{equation}
This matches the fidelity obtained by first performing Schur sampling and discarding the singlet states, followed by performing the optimal measure-and-prepare channel $\mE_{\text{MP}}[N_C\to 1]$ (see Appendix \ref{sec:four-important-channels}\ref{subsec:optimal-mp-channel-unified-presentation}) on the $N_C$ remaining qubits in the symmetric subspace.

\vspace{0.5\baselineskip}

We thus conclude that both complexified RLD eigenvalues find novel operational interpretations even in the simpler problem of single-shot qubit distillation:
\begin{itemize}
    \item On one hand, the monotonicity of $\text{RLD}_{\text{max}}$ sets the tightest leading-order bound on the performance of right-way distillation, which makes sense because $\text{RLD}_{\text{max}}$ diverges as the purity level approaches $+1$, corresponding to a pure qubit indicating the \emph{same} direction as the input qubits.
    \item On the other hand, the monotonicity of $\text{RLD}_{\text{min}}$ sets the tightest leading-order bound on the performance of wrong-way distillation, which makes sense because $\text{RLD}_{\text{min}}$ diverges as the purity level approaches $-1$, corresponding to a pure qubit indicating the \emph{opposite} direction as the input qubits.
\end{itemize}

\newpage

\appsec{Quantum Fisher Information Metrics}
{sec:qfi-metrics}

The central information-theoretic quantity underlying this entire work is the so-called ``complexified right logarithmic derivative (RLD) Fisher information''. In the interest of making the main body of this work concise, we have so far said relatively little about what the terms ``complexified'', ``right logarithmic derivative'', and ``Fisher information'' even mean in this setting. In this appendix, we answer these questions by putting complexified RLD Fisher information in its rightful context as a member of the broader family of quantum Fisher information (QFI) metrics.

\appsubsec{The Morozova-Chentsov-Petz Classification}
{subsec:mcp-classification}

In classical probability theory and statistics, Fisher information describes a metric on the space of probability distributions. The significance of Fisher information arises from \textbf{Chentsov's theorem}, which states that Fisher information is the unique (up to a global scale factor) Riemannian metric on probability distributions that is non-increasing under stochastic maps \cite{Chentsov1982}. To generalize this result to the quantum setting, one should replace ``probability distributions'' with ``quantum states'' (positive semidefinite, unit-trace operators) and replace ``stochastic maps'' with ``quantum channels'' (completely positive trace-preserving linear maps). Hence, one should seek to classify all Riemannian metrics on quantum states that are non-increasing under quantum channels. Morozova and Chentsov found a number of necessary conditions for such a metric \cite{morozova1991}, and Petz later completed the characterization of these metrics \cite{petz1996}. Their combined result is as follows:

\begin{theorem}[Morozova-Chentsov-Petz (MCP) theorem \cite{morozova1991,petz1996}]
\label{thm:morozova-chentsov-petz}
A Riemannian metric on density operators is monotone (that is, non-increasing under quantum channels) if and only if, in the neighborhood of a diagonalized density matrix $\rho = \sum_{i}p_i\ket{i}\bra{i}$, it takes the form (up to scaling)
\begin{equation}
    ds^2 = \sum_{i}\frac{d\rho_{ii}^2}{p_i} + \sum_{i\neq j}\frac{\abs{d\rho_{ij}}^2}{m_f(p_i,p_j)},
\end{equation}
where $m_f(p_i,p_j) = p_jf(p_i/p_j)$ and $f:[0,\infty)\rightarrow[0,\infty)$ is a \textbf{Morozova-Chentsov (MC) function}, meaning that it satisfies the following three conditions:
\begin{itemize}
    \item normalized: $f(1) = 1$
    \item self-inverse: $f(t) = tf(1/t)$
    \item operator monotone: for any two positive semidefinite matrices $A$ and $B$, $A\ge B$ implies $f(A)\ge f(B)$.
\end{itemize}
These monotone Riemannian metrics are commonly referred to as \textbf{quantum Fisher information (QFI) metrics}, or otherwise as \textbf{monotone Riemannian metrics}.
\end{theorem}

Petz proved this result by looking for monotone inner products. In particular, consider an inner product $K_D(A,B)$, where $A$ and $B$ are any two matrices, and the positive definite matrix $D$ determines the nature of the inner product itself. We say that $K_D$ is a \textbf{monotone inner product} if $K_{\mE(D)}(\mE(A),\mE(A))\le K_D(A,A)$ for any matrix $A$, any positive definite matrix $D$, and any quantum channel $\mE$. Notice that, if $K_D$ is a monotone inner product, then
\begin{equation}
    ds^2 = K_\rho(d\rho,d\rho)
\end{equation}
is a monotone Riemannian metric on density operators.

\vspace{0.5\baselineskip}

It turns out that the same matrix $D$ can define many different inner products; an especially natural collection of such inner products is described by ``symmetrized multiplication'':

\begin{definition}[$f$-symmetrized multiplication by $D$]
\label{def:function-symmetrized-multiplication}
Let $D$ be a positive definite matrix, and let $f:[0,\infty)\rightarrow[0,\infty)$ be a function. The $\mathbf{f}$\textbf{-symmetrized multiplication by }$\mathbf{D}$, denoted $E^f_D$, is the superoperator given by
\begin{equation}
    E^f_D \coloneqq R_Df\left(L_DR_D^{-1}\right),
\end{equation}
where $L_D$ and $R_D$ denote left- and right-multiplication by $D$, respectively.
\end{definition}

In particular, the function $f$ determines some ``compromise'' between left- and right-multiplication. For example, $f(t)=1$ yields $E^f_D = R_D$ (right-multiplication), while $f(t)=t$ yields $E^f_D = L_D$ (left-multiplication). Using Definition \ref{def:function-symmetrized-multiplication}, we can now write down Petz's result for monotone inner products:

\begin{theorem}
[Classification of monotone inner products \cite{petz1996}]
\label{thm:petz-monotone-metric-classification}
An inner product on matrices is non-increasing under quantum channels if and only if it takes the form
\begin{equation}
    K^f_D(A,B) = \text{Tr}\left[(E^f_D)^{-1}(A^\dagger)B\right],
\end{equation}
where $f:[0,\infty)\rightarrow[0,\infty)$ is an operator monotone, meaning that, for any two positive semidefinite matrices $A$ and $B$, $A\ge B$ implies $f(A)\ge f(B)$.
\end{theorem}

\appsubsec{Real vs. Complexified}
{subsec:real-vs-complexified}

So what is the difference between Theorem \ref{thm:morozova-chentsov-petz} and Theorem \ref{thm:petz-monotone-metric-classification}? Notice that the self-inverse condition $f(t) = tf(1/t)$ in Theorem \ref{thm:morozova-chentsov-petz} is nowhere to be found in Theorem \ref{thm:petz-monotone-metric-classification}. The matrices $A$ and $B$ in Theorem \ref{thm:petz-monotone-metric-classification} play the role of $d\rho$ in Theorem \ref{thm:morozova-chentsov-petz}, in the sense that
\begin{equation}
    ds^2 = K^f_\rho(d\rho,d\rho).
\end{equation}
However, while $A$ and $B$ are allowed to be any matrices, $d\rho$ is required to be Hermitian. As a result, $d\rho_{ji} = d\rho_{ij}^*$, so if $f$ is not self-inverse, then plugging $f(t)$ into the formula for $ds^2$ will yield the same result as if we had used the harmonic mean of $f(t)$ and $tf(1/t)$, i.e.,
\begin{equation}
    f_{\text{SI}}(t) = \frac{2tf(t)f(1/t)}{f(t) + tf(1/t)}.
\end{equation}
So in fact, the self-inverse condition in Theorem \ref{thm:morozova-chentsov-petz} is actually not needed for $f$ to produce a \textit{valid} metric on the space of density operators. It is just needed to ensure that each choice of $f$ produces a \textit{distinct} metric on the space of density operators. For example, $f(t)=1$ and $f(t)=t$ are both operator monotones, but they produce the same self-inverse operator monotone $f_{\text{SI}}(t)=\frac{2t}{1+t}$.

\vspace{0.5\baselineskip}

In general, $K^f_D$ is a Hermitian (but not necessarily real) inner product, while $K^{f_{\text{SI}}}_D$ is its real part, which provides its more obvious interpretation as a distance metric on the space of density operators. However, if $f$ is not self-inverse, then the monotonicity of $K^f_D$ is strictly stronger than the monotonicity of $K^{f_{\text{SI}}}_D$. For example, it is well known that transposition is a positive map, but not a completely positive map, and hence it is not a quantum channel. However, the vast majority of distance and divergence measures on density operators that people have studied do not make any distinction between a density operator $\rho$ and its transpose (equivalently, complex conjugate) $\rho^\intercal = \rho^*$. As a result, they do not provide any information-theoretic explanation for why transposition is not a physical operation. However, $K^f_{D^*}(A^*,B^*) = K^f_D(A,B)^*$, and the complex conjugate of a Hermitian operator (such as the Gram matrix defined by the inner product $K^f_D$) is not less than or equal to the original operator unless that operator is real symmetric. As a result, transposition actually violates the monotonicity of $K^f_D$ for any non-self-inverse $f$.

\vspace{0.5\baselineskip}

For convenience, we will call a QFI metric \textbf{real} if its associated operator monotone $f$ satisfies the self-inverse condition $f(t) = tf(1/t)$ and \textbf{complexified} otherwise. A real QFI metric associated with a self-inverse operator monotone $f_{\text{SI}}$ may have many different complexified ``variants'', which are associated with all possible operator monotones $f$ such that $f_{\text{SI}}(t)$ is the harmonic mean of $f(t)$ and $tf(1/t)$. When we want to think solely about measuring distances on the manifold of density operators, we are free to restrict our attention to real QFI metrics as described in Theorem \ref{thm:morozova-chentsov-petz}. However, when we want to compute more general inner products, the complexified QFI metrics provide information not captured by the real QFI metrics.

\appsubsec{How Does RLD Fit In?}
{subsec:rld-in-context}

Now that we have discussed QFI metrics more generally and explained what the terms ``real'' and ``complexified'' mean in this context, we can explain both the real and complexified RLD metrics and how they fit into the broader picture of QFI metrics. We begin by restating the definition of complexified RLD and introducing the real RLD alongside it:

\begin{definition}[real and complexified RLD Fisher information matrices]
\label{def:real-complexified-RLD}
For a family of density operators $\rho(\vec{x})$ with parameter $\vec{x}\in\Rbb^k$, the \textbf{complexified RLD Fisher information matrix} $\text{RLD}^{\Cbb}[\rho](\vec{x})$ is defined entrywise as
\begin{equation}
    \text{RLD}^{\Cbb}_{\mu\nu}[\rho](\vec{x}) = \text{Tr}\left[(\partial_\mu\rho(\vec{x}))\rho(\vec{x})^{-1}(\partial_\nu\rho(\vec{x}))\right]
\end{equation}
for indices $1\le\mu,\nu\le k$. In contrast, the \textbf{real RLD Fisher information matrix} $\text{RLD}^{\Rbb}[\rho](\vec{x})$ is the real part (equivalently, symmetric part) of the complexified matrix, meaning that it has entrywise formula
\begin{align}
    \text{RLD}^{\Rbb}_{\mu\nu}[\rho](\vec{x}) &= \frac{\text{RLD}^{\Cbb}_{\mu\nu}[\rho](\vec{x}) + \text{RLD}^{\Cbb}_{\nu\mu}[\rho](\vec{x})}{2} \\
    &= \frac{1}{2}\text{Tr}\left[\rho(\vec{x})^{-1}\{\partial_\mu\rho(\vec{x}),\partial_\nu\rho(\vec{x})\}\right].
\end{align}
\end{definition}

Using the functions $f(t)=1$ and $f_{\text{SI}}(t) = \frac{2t}{1+t}$, respectively, the complexified and real RLD Fisher information metrics can be rewritten as
\begin{align}
    \text{RLD}^{\Cbb}_{\mu\nu}[\rho](\vec{x}) &= K^f_{\rho(\vec{x})}\left(\partial_\mu\rho(\vec{x}),\partial_\nu\rho(\vec{x})\right) \\
    \text{RLD}^{\Rbb}_{\mu\nu}[\rho](\vec{x}) &= K^{f_{\text{SI}}}_{\rho(\vec{x})}\left(\partial_\mu\rho(\vec{x}),\partial_\nu\rho(\vec{x})\right).
\end{align}

As an aside, the reason for the name ``right logarithmic derivative'' comes from the fact that $f(t)=1$ corresponds to right-multiplication in Definition \ref{def:function-symmetrized-multiplication}. The logarithmic derivative of a number-valued function is $\partial_\mu\ln f(\vec{x}) = \frac{\partial_\mu f(\vec{x})}{f(\vec{x})}$. But for a matrix-valued function $\rho(\vec{x})$, ``division'' by $\rho(\vec{x})$ is ambiguous if it does not commute with $\partial_\mu\rho(\vec{x})$. One possible choice is the inverse of right-multiplication, hence the name ``right logarithmic derivative'' (RLD). One could analogously define the ``left logarithmic derivative'' (LLD) and the LLD Fisher information matrix by using $f(t)=t$, but the LLD matrix would just be the transpose of the RLD matrix.

\vspace{0.5\baselineskip}

As a metric tensor on the space of density operators, $\text{RLD}^{\Rbb}$ and $\text{RLD}^{\Cbb}$ are identical, for precisely the reason discussed in Appendix \ref{sec:qfi-metrics}\ref{subsec:real-vs-complexified}. The MC function $f_{\text{SI}}(t) = \frac{2t}{1+t}$ turns out to be the smallest MC function, meaning that the RLD Fisher information metric is the largest of the QFI metrics in Theorem \ref{thm:morozova-chentsov-petz}. As such, RLD Fisher information plays a privileged role in quantum information theory. For example, real RLD Fisher information has been used to prove a bound on the performance of coherence distillation \cite{marvian2020}, a task in the resource theory of asymmetry where the relevant group is $\{e^{-itH}\,|\,t\in\Rbb\}$, i.e., the one-parameter group of time evolution operators emerging from a constant Hamiltonian. This bound was recently shown to be asymptotically tight in at least two notable cases: thermal coherent states of a quantum harmonic oscillator \cite{yadavalli2025} and arbitrary qubit states \cite{kazi2025}.

\vspace{0.5\baselineskip}

However, the complexified RLD matrix conveys additional information that is ignored by the real RLD matrix. To demonstrate this, we return to the example of single-qubit states at a fixed purity level, which we introduced in the main body:

\begin{example}[qubit states at a fixed purity level, continued]
\label{ex:qubit-states-fixed-purity-level-continued}
Consider the family of qubit states with a fixed purity level $\lambda$. One way to parametrize part of this family is as follows:
\begin{equation}
    \sigma(a,b) = \frac{\Ibb + \lambda\left(aX + bY + \sqrt{1-a^2-b^2}Z\right)}{2}.
\end{equation}
We first compute the partial derivatives of this state and also the inverse of the state at $(0,0)$:
\begin{align}
    \partial_a\sigma(0,0) &= \frac{\lambda}{2}X \\
    \partial_b\sigma(0,0) &= \frac{\lambda}{2}Y \\
    \sigma(0,0)^{-1} &= \left(\frac{\Ibb + \lambda Z}{2}\right)^{-1} = \frac{2}{1-\lambda^2}(\Ibb - \lambda Z).
\end{align}
We can now compute the complexified RLD Fisher information matrix at $(0,0)$:
\begin{equation}
    \text{RLD}^{\Cbb}[\sigma](0,0) = \frac{\lambda^2}{1-\lambda^2}\begin{bmatrix}
        1 & +\lambda i \\
        -\lambda i & 1
    \end{bmatrix}.
\end{equation}
This matrix always has the same eigenvectors regardless of $\lambda$. Furthermore, it has two distinct eigenvalues, which we call $\text{RLD}_{\text{max}}$ and $\text{RLD}_{\text{min}}$ for convenience:
\begin{equation}
    \text{RLD}_{\text{max}}(\lambda) \coloneqq \frac{\lambda^2}{1-\lambda}, \quad \text{RLD}_{\text{min}}(\lambda) \coloneqq \frac{\lambda^2}{1+\lambda}.
\end{equation}
\end{example}

As a result, if we were to convert $N$ i.i.d. qubits at purity level $\lambda_{\text{in}}$ to $RN$ qubits at purity level $\lambda_{\text{out}}$, then the monotonicity of complexified RLD Fisher information naturally suggests two upper bounds, one for each eigenvalue:
\begin{equation}
    R \le \frac{\lambda_{\text{in}}^2}{1-\lambda_{\text{in}}} \Big/ \frac{\lambda_{\text{out}}^2}{1-\lambda_{\text{out}}}, \quad R \le \frac{\lambda_{\text{in}}^2}{1+\lambda_{\text{in}}} \Big/ \frac{\lambda_{\text{out}}^2}{1+\lambda_{\text{out}}}.
\end{equation}
To prove these upper bounds rigorously, some extra technical details are required to handle the allowance of nonzero but vanishing trace distance (which we defer to Appendix \ref{sec:converse-bound-rld-sensitivity}), but the essential reason for these upper bounds is already shown here.

\vspace{0.5\baselineskip}

To emphasize the importance of the complex entries in $\text{RLD}^{\Cbb}[\rho](0,0)$, consider the real RLD matrix
\begin{equation}
    \text{RLD}^{\Rbb}[\sigma](0,0) = \frac{\lambda^2}{1-\lambda^2}\begin{bmatrix}
        1 & 0 \\
        0 & 1
    \end{bmatrix}.
\end{equation}
Solely considering the monotonicity of real RLD, if we were to convert $N$ i.i.d. qubits at purity level $\lambda_{\text{in}}$ to $RN$ qubits at purity level $\lambda_{\text{out}}$, we would naturally expect to see
\begin{equation}
    R \le \frac{\lambda_{\text{in}}^2}{1-\lambda_{\text{in}}^2} \Big/ \frac{\lambda_{\text{out}}^2}{1-\lambda_{\text{out}}^2}.
\end{equation}
This is still a correct bound, and it is certainly nontrivial, but it is strictly weaker than the bound imposed by complexified RLD. For example, if $\lambda_{\text{in}} = 1/2$ and $\lambda_{\text{out}} = 2/3$, then the monotonicity of \emph{real} RLD Fisher information suggests the upper bound $R\le\frac{5}{12}$, but the monotonicity of \emph{complexified} RLD Fisher information (in particular, focusing on the $\text{RLD}_{\text{max}}$ eigenvalue) suggests the stronger bound $R\le\frac{3}{8}$.

\appsubsec{QFI-Derived $\text{SU}(2)$ Asymmetry Measures}
{subsec:qfi-derived-su2-asymmetry-measures}

Since all QFI metrics are non-increasing under quantum channels (after all, that is what makes it appropriate to call them ``QFI metrics''), we can actually construct a measure of $\mathrm{SU}(2)$ asymmetry out of \textit{any} QFI metric. So what makes complexified RLD, and its associated eigenvalues $\text{RLD}_{\text{max}}$ and $\text{RLD}_{\text{min}}$, so special?

\vspace{0.5\baselineskip}

To answer this question, we will construct a measure of $\mathrm{SU}(2)$ asymmetry for a single-qubit state based on an \emph{arbitrary} QFI metric. We will see that, although both eigenvalues of each QFI metric imply upper bounds on the concentration rate, it is specifically the RLD eigenvalues $\text{RLD}_{\text{max}}$ and $\text{RLD}_{\text{min}}$ that provide the \textit{smallest} (hence \textit{strongest}) upper bound on the concentration and dilution rate, respectively.

\vspace{0.5\baselineskip}

Recall that the QFI matrix associated with a specific operator monotone $f$ has the following entrywise formula:
\begin{equation}
    \mI(\theta)_{\mu\nu} = \sum_{j,k}\frac{\left(\partial_\mu\rho_{jk}\right)\left(\partial_\nu\rho_{jk}\right)^*}{p_jf(p_k/p_j)}.
\end{equation}
If $f$ satisfies the self-inverse condition $f(t) = tf(1/t)$, then $p_jf(p_k/p_j) = p_kf(p_j/p_k)$, so the terms with $j\neq k$ will come in complex conjugate pairs, making the quantity $\mI(\theta)_{\mu\nu}$ necessarily real. This yields the real QFI metrics as described in Theorem \ref{thm:morozova-chentsov-petz}. However, if $f$ does not satisfy the self-inverse condition, then $\mI(\theta)_{\mu\nu}$ need not be real for $\mu\neq\nu$, and this yields the complexified metrics.

\vspace{0.5\baselineskip}

Let us apply this formula to the family of states that we focus on throughout this work, namely, the family of qubit states with a fixed purity level $\lambda$. One way to parametrize part of this family is as follows:
\begin{equation}
    \rho(a,b) = \frac{\Ibb + \lambda\left(aX + bY + \sqrt{1-a^2-b^2}Z\right)}{2}.
\end{equation}
At the point $(0,0)$, we have the state
\begin{equation}
    \rho(0,0) = \frac{\Ibb + \lambda Z}{2} = \begin{bmatrix}
        \frac{1+\lambda}{2} & 0 \\ 0 & \frac{1-\lambda}{2}
    \end{bmatrix},
\end{equation}
which has eigenvalues
\begin{equation}
    p_0 = \frac{1+\lambda}{2}, \quad p_1 = \frac{1-\lambda}{2}.
\end{equation}
We now compute the partial derivatives of this state at $(0,0)$. Since $\rho(0,0)$ is already diagonal in the computational basis, no change of basis is necessary:
\begin{equation}
    \partial_a\rho(0,0) = \frac{\lambda}{2}X = \frac{\lambda}{2}\begin{bmatrix}
        0 & 1 \\ 1 & 0
    \end{bmatrix}, \quad\quad \partial_b\rho(0,0) = \frac{\lambda}{2}Y = \frac{\lambda}{2}\begin{bmatrix}
        0 & -i \\ +i & 0
    \end{bmatrix}.
\end{equation}
We can now compute each entry of the QFI matrix:
\begin{align}
    \mI_{aa}(0,0) &= \left(\frac{\lambda}{2}\right)^2\left[\frac{(0)(0)^*}{p_0f(p_0/p_0)} + \frac{(1)(1)^*}{p_0f(p_1/p_0)} + \frac{(1)(1)^*}{p_1f(p_0/p_1)} + \frac{(0)(0)^*}{p_1f(p_1/p_1)}\right] \\
    &= \frac{\lambda^2}{2}\left[\frac{1}{(1+\lambda)f\left(\frac{1-\lambda}{1+\lambda}\right)} + \frac{1}{(1-\lambda)f\left(\frac{1+\lambda}{1-\lambda}\right)}\right]
\end{align}
\begin{align}
    \mI_{ab}(0,0) &= \left(\frac{\lambda}{2}\right)^2\left[\frac{(0)(0)^*}{p_0f(p_0/p_0)} + \frac{(1)(-i)^*}{p_0f(p_1/p_0)} + \frac{(1)(+i)^*}{p_1f(p_0/p_1)} + \frac{(0)(0)^*}{p_1f(p_1/p_1)}\right] \\
    &= \frac{\lambda^2}{2}i\left[\frac{1}{(1+\lambda)f\left(\frac{1-\lambda}{1+\lambda}\right)} - \frac{1}{(1-\lambda)f\left(\frac{1+\lambda}{1-\lambda}\right)}\right]
\end{align}
\begin{align}
    \mI_{ba}(0,0) &= \left(\frac{\lambda}{2}\right)^2\left[\frac{(0)(0)^*}{p_0f(p_0/p_0)} + \frac{(-i)(1)^*}{p_0f(p_1/p_0)} + \frac{(+i)(1)^*}{p_1f(p_0/p_1)} + \frac{(0)(0)^*}{p_1f(p_1/p_1)}\right] \\
    &= -\frac{\lambda^2}{2}i\left[\frac{1}{(1+\lambda)f\left(\frac{1-\lambda}{1+\lambda}\right)} - \frac{1}{(1-\lambda)f\left(\frac{1+\lambda}{1-\lambda}\right)}\right]
\end{align}
\begin{align}
    \mI_{bb}(0,0) &= \left(\frac{\lambda}{2}\right)^2\left[\frac{(0)(0)^*}{p_0f(p_0/p_0)} + \frac{(-i)(-i)^*}{p_0f(p_1/p_0)} + \frac{(+i)(+i)^*}{p_1f(p_0/p_1)} + \frac{(0)(0)^*}{p_1f(p_1/p_1)}\right] \\
    &= \frac{\lambda^2}{2}\left[\frac{1}{(1+\lambda)f\left(\frac{1-\lambda}{1+\lambda}\right)} + \frac{1}{(1-\lambda)f\left(\frac{1+\lambda}{1-\lambda}\right)}\right].
\end{align}
For convenience, we define the function $g:(-1,+1)\rightarrow\Rbb$ by
\begin{equation}
    g(u) \coloneqq f\left(\frac{1-u}{1+u}\right)^{-1}.
\end{equation}
First, the normalization condition $f(1) = 1$ is equivalent to $g(0) = 1$. Furthermore, the fact that $f$ is an operator monotone on $(0,\infty)$ is actually equivalent to the fact that $g$ is an operator monotone on $(-1,+1)$. This is because $u\mapsto\frac{1-u}{1+u}$ and $v\mapsto v^{-1}$ are operator order-reversing ($f$ is operator order-reversing if and only if $-f$ is operator monotone), and composing two operator order-reversing functions with an operator monotone function yields another operator monotone. Using this new operator monotone $g$, we can write the entries of the QFI matrix as
\begin{align}
    \mI_{aa}(0,0) = \mI_{bb}(0,0) &= \frac{\lambda^2}{2}\left[\frac{g(\lambda)}{1+\lambda} + \frac{g(-\lambda)}{1-\lambda}\right] \\
    \mI_{ab}(0,0) = \mI_{ba}(0,0)^* &= \frac{\lambda^2}{2}i\left[\frac{g(\lambda)}{1+\lambda} - \frac{g(-\lambda)}{1-\lambda}\right].
\end{align}
The QFI matrix thus has the following two eigenvalues:
\begin{equation}
    \mI_f^+(\lambda) = \frac{\lambda^2}{1+\lambda}g(\lambda), \quad\quad \mI_f^-(\lambda) = \frac{\lambda^2}{1-\lambda}g(-\lambda).
\end{equation}
When $f$ satisfies the self-inverse condition $f(t) = tf(1/t)$, then plugging in $t = \frac{1-u}{1+u}$ and rearranging yields $(1-u)g(u) = (1+u)g(-u)$, which in turn means that $\mI_f^+ = \mI_f^-$. Hence, any real QFI metric yields a multiple of the identity matrix, which makes sense given the symmetry of the family of states $\rho(a,b)$ with respect to rotation about the $z$-axis, which corresponds to a rotation of parameter space about $(0,0)$. In contrast, any non-self-inverse operator monotone $f$ yields two different eigenvalues. In particular, $f^\intercal(t) = tf(1/t)$ (which is also an operator monotone \cite{Simon2019}) will yield the transpose (equivalently, complex conjugate) QFI matrix, which means that the eigenvalues switch their labels: $\mI_{f^\intercal}^+ = \mI_f^-$ and $\mI_{f^\intercal}^- = \mI_f^+$.

\vspace{0.5\baselineskip}

One nice observation is that $\mI_f^+(\lambda) = \lambda^2 + O(\lambda^3)$ and $\mI_f^-(\lambda) = \lambda^2 + O(\lambda^3)$ for $\lambda\approx 0$. This shows that the two QFI eigenvalues become roughly equal in the limit of very dilute states, regardless of the choice of $f$. Intuitively, very dilute qubit states act effectively like classical direction indicators.

\appsubsec{Comparing Different QFI Metrics}
{subsec:comparing-different-qfi-metrics}

Now that we have a general formula for a QFI-derived $\text{SU}(2)$ asymmetry measure, let us highlight a few important examples:
\begin{itemize}
    \item Of course, we should look at complexified RLD, which corresponds to $f(t) = 1$.
    \item We can additionally look at the complexified \textbf{left logarithmic derivative (LLD)} metric, which corresponds to $f(t) = t$. The reason people do not study the LLD matrix separately is that it is simply the transpose (equivalently, complex conjugate) of the RLD matrix. In this context, it means that the $\mathrm{SU}(2)$ asymmetry measures switch places.
    \item Furthermore, we can take the average of RLD and LLD to obtain the real RLD metric, which corresponds to the MC function $f(t) = \frac{2t}{t+1}$.
    \item Finally, we will take the most famous of all QFI metrics, the \textbf{symmetric logarithmic derivative (SLD)} metric, which corresponds to the MC function $f(t) = \frac{t+1}{2}$.
\end{itemize}
These four QFI metrics and their associated $\mathrm{SU}(2)$ asymmetry measures are summarized in Table \ref{tab:operator-monotones-su2-asymmetry-measures}. In particular, notice that the quantities $\text{RLD}_{\text{max}}(\lambda)$ and $\text{RLD}_{\text{min}}(\lambda)$ that we have emphasized throughout this paper are exactly the eigenvalues coming from complexified RLD (or equivalently, complexified LLD).

\begin{table}
    \begin{tabular}{c|c|c|c|c}
        QFI Metric & Operator Monotone $f$ & Operator Monotone $g$ & Eigenvalue $\mI_f^+$ & Eigenvalue $\mI_f^-$ \\
        \hline
        Symmetric logarithmic derivative (SLD) & $f(t) = \frac{t+1}{2}$ & $g(u) = 1+u$ & $\lambda^2$ & $\lambda^2$ \\
        \hline
        Right logarithmic derivative (RLD) & $f(t) = 1$ & $g(u) = 1$ & $\frac{\lambda^2}{1+\lambda}$ & $\frac{\lambda^2}{1-\lambda}$ \\
        \hline
        Left logarithmic derivative (LLD) & $f(t) = t$ & $g(u) = \frac{1+u}{1-u}$ & $\frac{\lambda^2}{1-\lambda}$ & $\frac{\lambda^2}{1+\lambda}$ \\
        \hline
        Real RLD (average of RLD and LLD) & $f(t) = \frac{2t}{t+1}$ & $g(u) = \frac{1}{1-u}$ & $\frac{\lambda^2}{1-\lambda^2}$ & $\frac{\lambda^2}{1-\lambda^2}$
    \end{tabular}
    \caption{Different real and complexified QFI metrics, their associated operator monotones on $[0,\infty)$ and $(-1,+1)$, and their associated $\text{SU}(2)$ asymmetry measures for a single qubit with purity level $\lambda$. For the real QFI metrics, SLD and real RLD, the two eigenvalues coincide. In contrast, for the complexified QFI metrics, RLD and LLD, the two eigenvalues are distinct.}
    \label{tab:operator-monotones-su2-asymmetry-measures}
\end{table}

\vspace{0.5\baselineskip}

Any QFI-derived $\text{SU}(2)$ asymmetry measure yields a corresponding upper bound on the conversion rate. More precisely, for any operator monotone $f$, we can say that
\begin{equation}
    R\left(\lambda_{\text{in}}\rightarrow\lambda_{\text{out}}\right) \le \text{min}\Bigg\{\frac{\mI_f^+(\lambda_{\text{in}})}{\mI_f^+(\lambda_{\text{out}})}, \frac{\mI_f^-(\lambda_{\text{in}})}{\mI_f^-(\lambda_{\text{out}})}\Bigg\}.
\end{equation}
This comes from a combination of the monotonicity of each QFI metric, the additivity of each QFI metric under tensor product, and an additional technical argument that we cover in Appendix \ref{sec:converse-bound-rld-sensitivity}.

\vspace{0.5\baselineskip}

The question is now which $\text{SU}(2)$ asymmetry measure yields the \textit{smallest} (hence \textit{strongest}) upper bound on the conversion rate. To illustrate this point, we compare $\text{RLD}_{\text{max}}(\lambda)$ and $\text{RLD}_{\text{min}}(\lambda)$ against the eigenvalues of the real RLD matrix and the SLD matrix, which we assign more evocative names as follows:
\begin{equation}
    \text{RLD}(\lambda) \coloneqq \frac{\lambda^2}{1-\lambda^2}, \quad\quad \text{SLD}(\lambda) \coloneqq \lambda^2.
\end{equation}
For any purity level $0 < \lambda < 1$, these four $\text{SU}(2)$ asymmetry measures satisfy the ordering
\begin{equation}
    \text{RLD}_{\text{min}}(\lambda) < \text{SLD}(\lambda) < \text{RLD}(\lambda) < \text{RLD}_{\text{max}}(\lambda).
\end{equation}
Furthermore, if you look at the \textit{ratio} of each measure for two different purity levels $\lambda_{\text{big}} > \lambda_{\text{small}}$, then they satisfy the same ordering:
\begin{equation}
    \frac{\text{RLD}_{\text{min}}(\lambda_{\text{big}})}{\text{RLD}_{\text{min}}(\lambda_{\text{small}})} < \frac{\text{SLD}(\lambda_{\text{big}})}{\text{SLD}(\lambda_{\text{small}})} < \frac{\text{RLD}(\lambda_{\text{big}})}{\text{RLD}(\lambda_{\text{small}})} < \frac{\text{RLD}_{\text{max}}(\lambda_{\text{big}})}{\text{RLD}_{\text{max}}(\lambda_{\text{small}})}.
\end{equation}
In Table \ref{tab:conversion-upper-bounds-qfi-resources}, we show these four different measures, along with the upper bounds they imply for a specific instance of concentration and dilution. We see that $\text{RLD}_{\text{max}}(\lambda)$ sets the strongest upper bound on concentration, and $\text{RLD}_{\text{min}}(\lambda)$ sets the strongest upper bound on dilution.

\begin{table}
    \begin{tabular}{c|c|c|c|c}
        QFI Asymmetry Measure & $\lambda=1/2$ Value & $\lambda=2/3$ Value & Concentration Bound & Dilution Bound \\
        \hline
        $\text{RLD}_{\text{max}}(\lambda) = \frac{\lambda^2}{1-\lambda}$ & $1/2$ & $4/3$ & $\mathbf{R\le 3/8 = 0.375}$ & $R\le 8/3 = 2.\overline{6}$ \\
        \hline
        $\text{RLD}(\lambda) = \frac{\lambda^2}{1-\lambda^2}$ & $1/3$ & $4/5$ & $R\le 5/12 = 0.41\overline{6}$ & $R\le 12/5 = 2.4$ \\
        \hline
        $\text{SLD}(\lambda) = \lambda^2$ & $1/4$ & $4/9$ & $R\le 9/16 = 0.5625$ & $R\le 16/9 = 1.\overline{7}$ \\
        \hline
        $\text{RLD}_{\text{min}}(\lambda) = \frac{\lambda^2}{1+\lambda}$ & $1/6$ & $4/15$ & $R\le 5/8 = 0.625$ & $\mathbf{R\le 8/5 = 1.6}$
    \end{tabular}
    \caption{Four different measures of $\mathrm{SU}(2)$ asymmetry for single-qubit states that can be derived from QFI metrics. We present their general formulas, their values for qubits at purity levels $\lambda=1/2$ and $\lambda=2/3$, and the resulting upper bounds on conversion rate that they imply in the concentration case $\frac{1}{2}\rightarrow\frac{2}{3}$ and the dilution case $\frac{2}{3}\rightarrow\frac{1}{2}$. Notice that the concentration upper bound becomes stronger as you go up the table, while the dilution upper bound becomes stronger as you go down the table.}
    \label{tab:conversion-upper-bounds-qfi-resources}
\end{table}

\appsubsec{Resource Wastage and Irreversibility}
{subsec:resource-wastage-irreversibility}

To conclude our discussion of the multitude of QFI metrics, we observe a very important consequence of the fact that the maximum concentration and dilution rates are governed by different measures of $\mathrm{SU}(2)$ asymmetry. In particular, qubit concentration maintains the total $\text{RLD}_{\text{max}}$ to leading order, but it substantially reduces all other QFI-derived $\mathrm{SU}(2)$ asymmetry measures. Analogously, qubit dilution maintains the total $\text{RLD}_{\text{min}}$ to leading order, but it substantially reduces all other QFI-derived $\mathrm{SU}(2)$ asymmetry measures. As a result, both concentration and dilution are asymptotically lossy, and therefore \textbf{irreversible}.

\vspace{0.5\baselineskip}

To illustrate this point, suppose we start with a large collection of i.i.d. qubits at some purity level $\lambda_{\text{small}}$. We now concentrate them to a higher purity level $\lambda_{\text{big}}$, and then dilute them back to the original purity level. We will find that we end up with substantially fewer qubits than we started with. This is because the maximum concentration rate $R^{\text{conc}}(\lambda_{\text{small}}\to\lambda_{\text{big}})$ and the maximum dilution rate $R^{\text{dilut}}(\lambda_{\text{big}}\to\lambda_{\text{small}})$ are NOT reciprocals of each other. Instead, they actually multiply to
\begin{equation}
    R^{\text{conc}}(\lambda_{\text{small}}\to\lambda_{\text{big}})R^{\text{dilut}}(\lambda_{\text{big}}\to\lambda_{\text{small}}) = \frac{1+\lambda_{\text{small}}}{1-\lambda_{\text{small}}} \Bigg/ \frac{1+\lambda_{\text{big}}}{1-\lambda_{\text{big}}},
\end{equation}
which is clearly less than $1$ for any $0 < \lambda_\text{small} < \lambda_{\text{big}} < 1$.

\vspace{0.5\baselineskip}

In Table \ref{tab:concentration-dilution-resource-wastage}, we show a specific example of this phenomenon. We start with $6000$ i.i.d. qubits at purity level $\lambda_{\text{small}} = 1/2$. We then concentrate them to $2250$ (approximately) i.i.d. qubits at purity level $\lambda_{\text{big}} = 2/3$, at the maximum achievable concentration rate $R^{\text{conc}}(1/2\to 2/3) = 3/8$. After that, we dilute these qubits to $3600$ (approximately) i.i.d. qubits at purity level $\lambda_{\text{small}} = 1/2$, at the maximum achievable dilution rate $R^{\text{dilut}}(2/3\to 1/2) = 8/5$. Throughout this process, we can track the values of various $\mathrm{SU}(2)$ asymmetry measures. Notice that only $\text{RLD}_{\text{max}}$ is (approximately) conserved in the concentration step, whereas only $\text{RLD}_{\text{min}}$ is (approximately) conserved in the dilution step.

\begin{table}
    \begin{tabular}{|c|c|c|c|c|c|}
    \hline
    $N$ & $6000$ & $\searrow$ & $2250$ & $\nearrow$ & $3600$ \\
    \hline
    $\lambda$ & $1/2$ & $\nearrow$ & $2/3$ & $\searrow$ & $1/2$ \\
    \hline
    $\text{RLD}_{\text{max}}$ & $3000$ & $\textcolor{blue}{\boldsymbol{\rightarrow}}$ & $3000$ & $\textcolor{red}{\boldsymbol{\searrow}}$ & $1800$ \\
    \hline
    $\text{RLD}$ & $2000$ & $\textcolor{red}{\boldsymbol{\searrow}}$ & $1800$ & $\textcolor{red}{\boldsymbol{\searrow}}$ & $1200$ \\
    \hline
    $\text{SLD}$ & $1500$ & $\textcolor{red}{\boldsymbol{\searrow}}$ & $1000$ & $\textcolor{red}{\boldsymbol{\searrow}}$ & $900$ \\
    \hline
    $\text{RLD}_{\text{min}}$ & $1000$ & $\textcolor{red}{\boldsymbol{\searrow}}$ & $600$ & $\textcolor{blue}{\boldsymbol{\rightarrow}}$ & $600$ \\
    \hline
    & & $R = 3/8$ & & $R = 8/5$ & \\
    \hline
    \end{tabular}
    \caption{A demonstration of the irreversibility of linear-rate concentration and dilution as a result of wastage of the $\mathrm{SU}(2)$ asymmetry measures. In the first step, $6000$ qubits with purity level $1/2$ are concentrated into $2250$ qubits with purity level $2/3$ at the maximum achievable rate $R^{\text{conc}}(1/2\rightarrow 2/3) = 3/8$. In the second step, those qubits are then diluted back into qubits with purity level $1/2$, but since the maximum achievable rate is $R^{\text{dilut}}(2/3\rightarrow 1/2) = 8/5$, we only end up with $3600$ qubits, which is less than what we started with. Notice that concentration roughly conserves $\text{RLD}_{\text{max}}$ while substantially wasting all the others, whereas dilution roughly conserves $\text{RLD}_{\text{min}}$ while substantially wasting all the others. We use \textcolor{blue}{blue} to denote (approximate) resource \textcolor{blue}{conservation} and \textcolor{red}{red} to denote resource \textcolor{red}{wastage}.}
    \label{tab:concentration-dilution-resource-wastage}
\end{table}

\vspace{0.5\baselineskip}

In a chemical reaction, the \textbf{limiting reagent} is the input substance that constrains the amount of output substances that can be produced. The limiting reagent is the input substance whose amount is the smallest relative to the numbers that take part in a single unit of the reaction. (For example, consider the reaction $2\text{H}_2 + \text{O}_2 \rightarrow 2\text{H}_2\text{O}$, where two molecules of hydrogen gas react with one molecule of oxygen gas to produce two molecules of water. If the number of hydrogen molecules is less than double the number of oxygen molecules, then hydrogen is the limiting reagent, whereas if it is more than double, then oxygen is the limiting reagent.) When the reaction is carried out to completion, the limiting reagent is completely consumed, meaning that all of it has been successfully used to produce the output substances of the reaction, but all other input substances will have some amount left over, meaning that they are ``wasted'' to some extent. The same logic applies to qubit linear-rate conversion: $\text{RLD}_{\text{max}}$ is the ``limiting reagent'' of concentration, while $\text{RLD}_{\text{min}}$ is the ``limiting reagent'' of dilution.

\newpage

\appsec{Upper Bounding the Achievable Conversion Rate Using RLD Fisher Information}
{sec:converse-bound-rld-sensitivity}

In this appendix, we use the monotonicity of RLD Fisher information to prove the desired upper bound on the linear conversion rate for i.i.d. qubits. In particular, we will prove the following theorems:

\begin{theorem}[Upper bound on qubit linear conversion rates]
\label{thm:qubit-linear-conversion-rate-upper-bound}
The maximum achievable conversion rates satisfy the following upper bounds:
\begin{align}
    R^{\text{conc}}(\lambda_{\text{in}}\to\lambda_{\text{out}}) &\le \frac{\text{RLD}_{\text{max}}(\lambda_{\text{in}})}{\text{RLD}_{\text{max}}(\lambda_{\text{out}})} \\
    R^{\text{dilut}}(\lambda_{\text{in}}\to\lambda_{\text{out}}) &\le \frac{\text{RLD}_{\text{min}}(\lambda_{\text{in}})}{\text{RLD}_{\text{min}}(\lambda_{\text{out}})} \\
    R^{\text{MP}}(\lambda_{\text{in}}\to\lambda_{\text{out}}) &\le \frac{\text{RLD}_{\text{min}}(\lambda_{\text{in}})}{\text{RLD}_{\text{max}}(\lambda_{\text{out}})} \\
    R^{\text{WW}}(\lambda_{\text{in}}\to\lambda_{\text{out}}) &\le \frac{\text{RLD}_{\text{min}}(\lambda_{\text{in}})}{\text{RLD}_{\text{max}}(\lambda_{\text{out}})}.
\end{align}
\end{theorem}

Combining Theorem \ref{thm:qubit-linear-conversion-rate-upper-bound} with the procedures described in Appendix \ref{sec:unified-presentation} that achieve the above rates for the four tasks of interest, we conclude that these are indeed the maximum achievable rates.

\vspace{0.5\baselineskip}

This appendix is organized as follows:
\begin{itemize}
    \item In Appendix \ref{sec:converse-bound-rld-sensitivity}\ref{subsec:conversion-rate-upper-bounds-intuitive}, we demonstrate how the upper bounds in Theorem \ref{thm:qubit-linear-conversion-rate-upper-bound} can be understood from the monotonicity of (complexified) RLD Fisher information. In particular, if we were required to perform the conversion \textit{exactly}, we could stop here.
    \item In Appendix \ref{sec:converse-bound-rld-sensitivity}\ref{subsec:output-state-restrictions}, we use covariant conditions to heavily restrict the space of output states that we have to consider. In particular, we will show that we can assume without generality that:
    \begin{itemize}
        \item the output state is permutation-invariant;
        \item the output state commutes with the total spin operator $S(\lfloor RN\rfloor,\hat{n})\coloneqq\sum_{j=1}^{\lfloor RN\rfloor}\hat{n}\cdot\vec{\sigma}$.
    \end{itemize}
    These two facts will greatly assist in the computation of both the RLD Fisher information of the output state and the trace distance between the output and target states.
    \item In Appendix \ref{sec:converse-bound-rld-sensitivity}\ref{subsec:rld-sensitivity}, we rigorously prove Theorem \ref{thm:qubit-linear-conversion-rate-upper-bound} by addressing the fact that our conversion problem permits nonzero trace distance, so long as it vanishes in the $N\to\infty$ limit. In particular, we show that, if our output states $\mE_N\left(\rho(\lambda_{\text{in}},\hat{n})^{\otimes N}\right)$ have vanishing trace distance from the target states $\rho(\lambda_{\text{out}},\hat{n})^{\otimes\lfloor RN\rfloor}$, then their RLD Fisher information cannot be too much less than that of the target states, and thus the conversion rate upper bounds persist.
\end{itemize}

\appsubsec{Intuitive Upper Bounds on Conversion Rate}
{subsec:conversion-rate-upper-bounds-intuitive}

In this appendix, we show how each upper bound in Theorem \ref{thm:qubit-linear-conversion-rate-upper-bound} can be readily understood from the monotonicity of complexified RLD Fisher information.

\vspace{0.5\baselineskip}

For the sake of this appendix, we will assume that our objective is to transform $\rho(\lambda_{\text{in}},\hat{n})^{\otimes N}$ into $\rho(\lambda_{\text{out}},\hat{n})^{\otimes\lfloor RN\rfloor}$ \emph{exactly}. The allowance of nonzero but vanishing trace distance will be addressed in Appendix \ref{sec:converse-bound-rld-sensitivity}\ref{subsec:rld-sensitivity}.

\vspace{0.5\baselineskip}

As we showed in Appendix \ref{sec:single-shot-qubit-distillation}, if we consider the family of single-qubit states $\rho(\lambda,\hat{n})$ with fixed purity level $\lambda$ but variable direction $\hat{n}$, then the complexified RLD Fisher information can be written as
\begin{equation}
    \text{RLD} = \frac{\lambda^2}{1-\lambda^2}\begin{bmatrix}
        1 & +i\lambda \\ -i\lambda & 1
    \end{bmatrix}.
\end{equation}
Also, observe that this formula applies even for $\lambda < 0$, where we extend our notation to negative purity levels via $\rho(-\lambda,\hat{n}) \coloneqq \rho(\lambda,-\hat{n})$. This will be important for wrong-way conversion.

\vspace{0.5\baselineskip}

For the sake of clarity, this is the RLD Fisher information matrix when the family of qubits is parametrized as
\begin{equation}
    \sigma(a,b) = \frac{\Ibb + \lambda\left(aX + bY + \sqrt{1-a^2-b^2}Z\right)}{2},
\end{equation}
and then the RLD Fisher information is evaluated at the point $(a,b)=(0,0)$. However, the choice of parametrization is arbitrary, since if we locally reparametrize the family of states $\rho(\vec{x})$ in the vicinity of $\vec{x} = \vec{0}$ via  the transformation $\vec{x}' = S\vec{x} + o(\abs{\vec{x}})$, then the RLD Fisher information at $\vec{x} = \vec{0}$ would simply change via the congruence transformation $M\mapsto S^\intercal MS$, and none of the consequences of RLD monotonicity would be altered, since $A\ge B$ is equivalent to $X^\dagger AX\ge X^\dagger BX$ for any invertible matrix $X$. From now on, we stick to this choice of parametrization, which gives us the single-qubit RLD Fisher information matrix shown above.

\vspace{0.5\baselineskip}

Because all QFI metrics are additive under tensor product, the RLD matrix of an i.i.d. qubit state is simply
\begin{equation}
    \text{RLD}\left(\rho(\lambda,\hat{n})^{\otimes N}\right) = N\frac{\lambda^2}{1-\lambda^2}\begin{bmatrix}
        1 & +i\lambda \\ -i\lambda & 1
    \end{bmatrix}.
\end{equation}
Regardless of the values of $N$ and $\lambda$, the above matrix has eigenvectors and eigenvalues as follows:
\begin{align}
    u_+ = \begin{bmatrix}
        1 \\ -i
    \end{bmatrix} \quad\quad & \text{RLD}_+\left(\rho(\lambda,\hat{n})^{\otimes N}\right) = \frac{\lambda^2}{1-\lambda}N = N\cdot\text{RLD}_{\text{max}}(\lambda) \\
    u_- = \begin{bmatrix}
        1 \\ +i
    \end{bmatrix} \quad\quad & \text{RLD}_-\left(\rho(\lambda,\hat{n})^{\otimes N}\right) = \frac{\lambda^2}{1+\lambda}N = N\cdot\text{RLD}_{\text{min}}(\lambda).
\end{align}
With all the preliminaries in place, we are ready to apply RLD monotonicity.

\vspace{0.5\baselineskip}

We begin with \textbf{concentration} and \textbf{dilution}. The monotonicity of RLD Fisher information implies that
\begin{equation}
    \text{RLD}\left(\rho(\lambda_{\text{in}},\hat{n})^{\otimes N}\right) \ge \text{RLD}\left(\rho(\lambda_{\text{out}},\hat{n})^{\otimes\lfloor RN\rfloor}\right).
\end{equation}
Since the eigenvectors of both sides are the same (namely, $u_+$ and $u_-$ as defined above), the matrix inequality is equivalent to an inequality on the corresponding eigenvalues:
\begin{align}
    N\cdot\text{RLD}_{\text{max}}(\lambda_{\text{in}}) &\ge \lfloor RN\rfloor\cdot\text{RLD}_{\text{max}}(\lambda_{\text{out}}) \\
    N\cdot\text{RLD}_{\text{min}}(\lambda_{\text{in}}) &\ge \lfloor RN\rfloor\cdot\text{RLD}_{\text{min}}(\lambda_{\text{out}}).
\end{align}
Using the fact that $\lfloor RN\rfloor\ge RN-1$, we conclude that
\begin{align}
    N\cdot\text{RLD}_{\text{max}}(\lambda_{\text{in}}) &\ge \lfloor RN\rfloor\cdot\text{RLD}_{\text{max}}(\lambda_{\text{out}}) \\
    \implies \lfloor RN\rfloor &\le N\frac{\text{RLD}_{\text{max}}(\lambda_{\text{in}})}{\text{RLD}_{\text{max}}(\lambda_{\text{out}})} \\
    \implies RN-1 &\le N\frac{\text{RLD}_{\text{max}}(\lambda_{\text{in}})}{\text{RLD}_{\text{max}}(\lambda_{\text{out}})} \\
    \implies R &\le \frac{\text{RLD}_{\text{max}}(\lambda_{\text{in}})}{\text{RLD}_{\text{max}}(\lambda_{\text{out}})} + \frac{1}{N}.
\end{align}
In the limit as $N\to\infty$, we obtain
\begin{equation}
    R \le \frac{\text{RLD}_{\text{max}}(\lambda_{\text{in}})}{\text{RLD}_{\text{max}}(\lambda_{\text{out}})}.
\end{equation}
By applying exactly the same reasoning with the other eigenvalue, we also obtain
\begin{equation}
    R \le \frac{\text{RLD}_{\text{min}}(\lambda_{\text{in}})}{\text{RLD}_{\text{min}}(\lambda_{\text{out}})}.
\end{equation}
Now, we simply observe which of these two bounds is lower, and thus stronger. Notice that
\begin{equation}
    \frac{\text{RLD}_{\text{max}}(\lambda_{\text{in}})}{\text{RLD}_{\text{max}}(\lambda_{\text{out}})} \Bigg/ \frac{\text{RLD}_{\text{min}}(\lambda_{\text{in}})}{\text{RLD}_{\text{min}}(\lambda_{\text{out}})} = \frac{1+\lambda_{\text{in}}}{1-\lambda_{\text{in}}} \Bigg/ \frac{1+\lambda_{\text{out}}}{1-\lambda_{\text{out}}}.
\end{equation}
Since the function $f(\lambda) = \frac{1+\lambda}{1-\lambda}$ is increasing on the interval $-1 < \lambda < 1$, the above quantity $f(\lambda_{\text{in}})/f(\lambda_{\text{out}})$ is less than $1$ as long as $\lambda_{\text{out}} > \lambda_{\text{in}}$ (concentration) and greater than $1$ as long as $\lambda_{\text{out}} < \lambda_{\text{in}}$ (dilution). Therefore, the upper bound with the $\text{RLD}_{\text{max}}$ values is stronger for concentration, whereas the upper bound with the $\text{RLD}_{\text{min}}$ values is stronger for dilution. We conclude that
\begin{align}
    R^{\text{conc}}(\lambda_{\text{in}}\to\lambda_{\text{out}}) &\le \frac{\text{RLD}_{\text{max}}(\lambda_{\text{in}})}{\text{RLD}_{\text{max}}(\lambda_{\text{out}})} \\
    R^{\text{dilut}}(\lambda_{\text{in}}\to\lambda_{\text{out}}) &\le \frac{\text{RLD}_{\text{min}}(\lambda_{\text{in}})}{\text{RLD}_{\text{min}}(\lambda_{\text{out}})},
\end{align}
exactly as desired.

\vspace{0.5\baselineskip}

We now proceed to \textbf{wrong-way conversion}. This is the same as saying that the target state is $\rho(\hat{n},-\lambda_{\text{out}})^{\otimes\lfloor RN\rfloor}$ (where we still have $\lambda_{\text{out}} > 0$), so we can repeat the above reasoning, but replacing $\lambda_{\text{out}}$ with $-\lambda_{\text{out}}$. In particular, the monotonicity of RLD Fisher information implies that
\begin{equation}
    \text{RLD}\left(\rho(\lambda_{\text{in}},\hat{n})^{\otimes N}\right) \ge \text{RLD}\left(\rho(-\lambda_{\text{out}},\hat{n})^{\otimes\lfloor RN\rfloor}\right).
\end{equation}
Comparing corresponding eigenvectors tells us that
\begin{align}
    N\cdot\text{RLD}_{\text{max}}(\lambda_{\text{in}}) &\ge \lfloor RN\rfloor\cdot\text{RLD}_{\text{max}}(-\lambda_{\text{out}}) \\
    N\cdot\text{RLD}_{\text{min}}(\lambda_{\text{in}}) &\ge \lfloor RN\rfloor\cdot\text{RLD}_{\text{min}}(-\lambda_{\text{out}}).
\end{align}
By the same reasoning as above, we conclude that
\begin{equation}
    R \le \frac{\text{RLD}_{\text{max}}(\lambda_{\text{in}})}{\text{RLD}_{\text{max}}(-\lambda_{\text{out}})}, \quad\quad R \le \frac{\text{RLD}_{\text{min}}(\lambda_{\text{in}})}{\text{RLD}_{\text{min}}(-\lambda_{\text{out}})}.
\end{equation}
But now notice that the formulas for $\text{RLD}_{\text{max}}$ and $\text{RLD}_{\text{min}}$ satisfy
\begin{equation}
    \text{RLD}_{\text{min}}(+\lambda) = \text{RLD}_{\text{max}}(-\lambda) \quad \forall -1 < \lambda < +1.
\end{equation}
In other words, the two eigenvalues switch places when you flip the sign of the purity level! Therefore, the above bounds can be rewritten as
\begin{equation}
    R \le \frac{\text{RLD}_{\text{max}}(\lambda_{\text{in}})}{\text{RLD}_{\text{min}}(\lambda_{\text{out}})}, \quad\quad R \le \frac{\text{RLD}_{\text{min}}(\lambda_{\text{in}})}{\text{RLD}_{\text{max}}(\lambda_{\text{out}})}.
\end{equation}
But now, there is no doubt as to which of these upper bounds is stronger. Clearly the latter upper bound is stronger, since the numerator is smaller, and the denominator is bigger. We conclude that
\begin{equation}
    R^{\text{WW}}(\lambda_{\text{in}},\lambda_{\text{out}}) \le \frac{\text{RLD}_{\text{min}}(\lambda_{\text{in}})}{\text{RLD}_{\text{max}}(\lambda_{\text{out}})},
\end{equation}
exactly as desired.

\vspace{0.5\baselineskip}

As a brief aside, notice that the wrong-way conversion result can be thought of as simply an extension of the dilution result to the setting where the output purity level is negative, since $R^{\text{WW}}(\lambda_{\text{in}}\to\lambda_{\text{out}}) = R^{\text{dilut}}(\lambda_{\text{in}}\to -\lambda_{\text{out}})$.

\vspace{0.5\baselineskip}

Finally, we tackle \textbf{measure-and-prepare conversion}. This will be the most challenging of the four, but still not too challenging. Suppose that we perform measure-and-prepare conversion as follows:
\begin{equation}
    \rho(\lambda_{\text{in}},\hat{n})^{\otimes N} \xrightarrow{\text{measure}} X(\lambda_{\text{in}},\lambda_{\text{out}},\hat{n}) \xrightarrow{\text{prepare}} \rho(\lambda_{\text{out}},\hat{n})^{\otimes\lfloor RN\rfloor}.
\end{equation}
The details of the intermediate classical random variable $X(\lambda_{\text{in}},\lambda_{\text{out}},\hat{n})$ are not important. We only need the following facts:
\begin{itemize}
    \item The monotonicity of RLD Fisher information applies to the ``measure'' and ``prepare'' steps, since a classical random variable such as $X(\lambda_{\text{in}},\lambda_{\text{out}},\hat{n})$ can always be thought of as a diagonal density matrix.
    \item Because $X(\lambda_{\text{in}},\lambda_{\text{out}},\hat{n})$ is a classical random variable, all of its QFI matrices collapse to a single classical Fisher information matrix. In particular, this is a \emph{real symmetric} matrix, because a classical Fisher information matrix is always a real symmetric matrix.
\end{itemize}
First, the monotonicity of RLD Fisher information on the ``measure'' step tells us that
\begin{equation}
    \left(u_-\right)^\dagger\text{RLD}\left(\rho(\lambda_{\text{in}},\hat{n})^{\otimes N}\right)\left(u_-\right) \ge \left(u_-\right)^\dagger\text{RLD}\left(X(\lambda_{\text{in}},\lambda_{\text{out}},\hat{n})\right)\left(u_-\right).
\end{equation}
Second, the fact that $\text{RLD}\left(X(\lambda_{\text{in}},\lambda_{\text{out}},\hat{n})\right)$ is a real matrix tells us that
\begin{equation}
    \left(u_-\right)^\dagger\text{RLD}\left(X(\lambda_{\text{in}},\lambda_{\text{out}},\hat{n})\right)\left(u_-\right) = \left(u_+\right)^\dagger\text{RLD}\left(X(\lambda_{\text{in}},\lambda_{\text{out}},\hat{n})\right)\left(u_+\right).
\end{equation}
For clarity, this is because $u_+$ and $u_-$ are complex conjugates, along with the fact that, for any complex vector $v$ and any real symmetric matrix $A$,
\begin{align}
    v^\dagger Av &= \left(v^\dagger Av\right)^* \\
    &= v^\intercal A^*v^* \\
    &= (v^*)^\dagger A(v^*).
\end{align}
Third, the monotonicity of RLD Fisher information on the ``prepare'' step tells us that
\begin{equation}
    \left(u_+\right)^\dagger\text{RLD}\left(X(\lambda_{\text{in}},\lambda_{\text{out}},\hat{n})\right)\left(u_+\right) \ge \left(u_+\right)^\dagger\text{RLD}\left(\rho(\lambda_{\text{out}},\hat{n})^{\otimes\lfloor RN\rfloor}\right)\left(u_+\right).
\end{equation}
Combining these three facts yields
\begin{equation}
    \left(u_-\right)^\dagger\text{RLD}\left(\rho(\lambda_{\text{in}},\hat{n})^{\otimes N}\right)\left(u_-\right) \ge \left(u_+\right)^\dagger\text{RLD}\left(\rho(\lambda_{\text{out}},\hat{n})^{\otimes\lfloor RN\rfloor}\right)\left(u_+\right).
\end{equation}
But in the above inequality, we know that $u_+$ and $u_-$ are eigenvectors of the relevant RLD matrices. We thus obtain
\begin{equation}
    N\cdot\text{RLD}_{\text{min}}(\lambda_{\text{in}}) \ge \lfloor RN\rfloor\cdot\text{RLD}_{\text{max}}(\lambda_{\text{out}}).
\end{equation}
By the same reasoning as above, we conclude that
\begin{equation}
    R \le \frac{\text{RLD}_{\text{min}}(\lambda_{\text{in}})}{\text{RLD}_{\text{max}}(\lambda_{\text{out}})},
\end{equation}
exactly as desired.

\vspace{0.5\baselineskip}

As a brief aside, notice that the measure-and-prepare result can be thought of intuitively as applying the monotonicity of each individual eigenvalue $\text{RLD}_{\text{max}}$ and $\text{RLD}_{\text{min}}$, but with the fact that $X$ is a classical random variable producing the extra stipulation that $\text{RLD}_{\text{max}}(X) = \text{RLD}_{\text{min}}(X)$. However, the above line of reasoning allows us to avoid any restriction on the nature of $X$. (In particular, we do not need to prove that $\text{RLD}(X)$ actually has the same eigenvectors $u_+$ and $u_-$.)

\vspace{0.5\baselineskip}

As a reminder, all the results above assume that we perform the linear-rate conversion \emph{exactly}, meaning that the trace distance between the output state and the target state is exactly zero. In the following appendices, we relax this requirement so that the trace distance can be nonzero, so long as it vanishes in the $N\to\infty$ limit.

\appsubsec{Using Covariance To Restrict the Output State}
{subsec:output-state-restrictions}

To rigorously prove our upper bounds on conversion rate, we will first heavily restrict the space of output states we need to consider. We do this by making the following observations:
\begin{itemize}
    \item First, since our protocols must work regardless of the direction $\hat {n}$, we may assume that the protocol $\mE_N$ is $\mathrm{SU}(2)$-covariant.
    \item Second, since our target state is permutation-invariant, we may assume that the output state is also permutation-invariant. In particular, applying a random permutation to our output state can only decrease the trace distance to the target state.
\end{itemize}
More formally, we can prove the following statement:

\begin{lemma}[Restrictions on output state]
\label{lem:restrictions-on-output-state}
Suppose that there exists a sequence of $\mathrm{SU}(2)$-covariant protocols $\mE_N$ (which may still explicitly depend on $N$, $\lambda_{\text{in}}$, and $\lambda_{\text{out}}$) such that, for all directions $\hat{n}$, the trace distance between $\sigma_N\coloneqq\mE_N\left(\rho(\lambda_{\text{in}},\hat{n})^{\otimes N}\right)$ and $\rho(\lambda_{\text{out}},\hat{n})^{\otimes\lfloor RN\rfloor}$ vanishes in the $N\to\infty$ limit. Then there also exists a sequence of protocols $\mE_N$ such that the trace distance between $\tilde{\sigma}_N\coloneqq\tilde{\mE}_N\left(\rho(\lambda_{\text{in}},\hat{n})^{\otimes N}\right)$ and $\rho(\lambda_{\text{out}},\hat{n})^{\otimes\lfloor RN\rfloor}$ vanishes in the $N\to\infty$ limit, but where each output state $\tilde{\sigma}_N$ additionally satisfies the following:
\begin{itemize}
    \item $\tilde{\sigma}_N$ is invariant under arbitrary permutations on the $\lfloor RN\rfloor$ qubits;
    \item $\tilde{\sigma}_N$ commutes with the total spin operator $S(\lfloor RN\rfloor,\hat{n})\coloneqq\sum_{j=1}^{\lfloor RN\rfloor}(\hat{n}\cdot\vec{\sigma})_j$.
    \end{itemize}
\end{lemma}

\begin{proof}[Proof of Lemma \ref{lem:restrictions-on-output-state}]
As the final step of our protocol, we can always apply a random permutation to our output state, which will automatically make it permutation-invariant. Furthermore, since a random permutation does not affect our target state $\rho(\lambda_{\text{out}},\hat{n})^{\otimes\lfloor RN\rfloor}$, and since trace distance is always decreasing under quantum channels, the trace distance between our output state and target state can only get better, not worse. Hence we may assume without loss of generality that our output state $\sigma_N\coloneqq\mE_N\left(\rho(\lambda_{\text{in}},\hat{n})^{\otimes N}\right)$ is permutation-invariant.

\vspace{0.5\baselineskip}

Next, the $\text{SU}(2)$ covariance condition implies that, for any $U\in\text{SU}(2)$,
\begin{equation}
    \mE_N\left(U^{\otimes N}\rho(\lambda_{\text{in}},\hat{n})^{\otimes N}\left(U^\dagger\right)^{\otimes N}\right) = U^{\otimes\lfloor RN\rfloor}\mE_N\left(\rho(\lambda_{\text{in}},\hat{n})^{\otimes N}\right)\left(U^\dagger\right)^{\otimes\lfloor RN\rfloor}.
\end{equation}
Plugging in $U = e^{-it(\vec{\sigma}\cdot\hat{n})}$ for arbitrary $t\in\Rbb$ yields
\begin{equation}
    \mE_N\left(e^{-itS(N,\hat{n})}\rho(\lambda_{\text{in}},\hat{n})^{\otimes N}e^{+itS(N,\hat{n})}\right) = e^{-itS(\lfloor RN\rfloor,\hat{n})}\mE_N\left(\rho(\lambda_{\text{in}},\hat{n})^{\otimes N}\right)e^{+itS(\lfloor RN\rfloor,\hat{n})}.
\end{equation}
Since $\rho(\lambda_{\text{in}},\hat{n})^{\otimes N}$ commutes with $S(N,\hat{n})$ (in fact, $\rho(\lambda_{\text{in}},\hat{n})^{\otimes N}$ can be written as a degree-$N$ polynomial in $S(N,\hat{n})$), we obtain
\begin{equation}
    \mE_N\left(\rho(\lambda_{\text{in}},\hat{n})^{\otimes N}\right) = e^{-itS(\lfloor RN\rfloor,\hat{n})}\mE_N\left(\rho(\lambda_{\text{in}},\hat{n})^{\otimes N}\right)e^{+itS(\lfloor RN\rfloor,\hat{n})}.
\end{equation}
We conclude that $\sigma_N\coloneqq\mE_N\left(\rho(\lambda_{\text{in}},\hat{n})^{\otimes N}\right)$ commutes with $e^{-itS(\lfloor RN\rfloor,\hat{n})}$ for all $t\in\Rbb$, which implies that $\sigma_N$ commutes with $S(\lfloor RN\rfloor,\hat{n})$.
\end{proof}

Lemma \ref{lem:restrictions-on-output-state} states that, without loss of generality, we may assume that our output state $\sigma(\lfloor RN\rfloor,\lambda_{\text{in}},\lambda_{\text{out}},\hat{n}) \coloneqq \mE_N\left(\rho(\lambda_{\text{in}},\hat{n})^{\otimes N}\right)$ is permutation-invariant and commutes with the total spin operator in the $\hat{n}$ direction, defined by $S(\lfloor RN\rfloor,\hat{n}) = \sum_{j=1}^{\lfloor RN\rfloor}(\vec{\sigma}\cdot\hat{n})_j$. This will greatly simplify our computation of both the RLD Fisher information of the output state and the trace distance between the output and target states, which will finally allow us to rigorously prove the upper bounds on conversion rate shown in Theorem \ref{thm:qubit-linear-conversion-rate-upper-bound}.

\vspace{0.5\baselineskip}

In particular, if a state in the $N$-qubit symmetric subspace commutes with $S(N,\hat{n})$, then it must be diagonal in the basis of Dicke states $\ket{D^{(N)}_w}\bra{D^{(N)}_w}_{\hat{n}}$. Fortunately, for such states, it is easy to compute the RLD Fisher information eigenvalues with respect to simultaneous qubit rotations. We do so with the following lemma:

\begin{lemma}[RLD Fisher information for a state diagonal in the Dicke state basis]
Consider the following two-parameter family of states:
\begin{align}
    \sigma(a,b) &= \sum_{w=0}^{N}A(w)\ket{D^{(N)}_w}\bra{D^{(N)}_w}_{\hat{n}} \\
    \hat{n} &= \langle a, b, \sqrt{1-a^2-b^2}\rangle.
\end{align}
Then the RLD Fisher information matrix at $(a,b) = (0,0)$ evaluates to
\begin{equation}
    \text{RLD}[\sigma](0,0) = \begin{bmatrix}
        \text{RLD}_{aa} & \text{RLD}_{ab} \\
        \text{RLD}_{ba} & \text{RLD}_{bb}
    \end{bmatrix},
\end{equation}
where the four entries are given by
\begin{align}
    \text{RLD}_{aa} = \text{RLD}_{bb} &= \sum_{w=1}^{N}w(N-w+1)\frac{[A(w)-A(w-1)]^2[A(w-1)+A(w)]}{2A(w)A(w-1)} \\
    \text{RLD}_{ab} = -\text{RLD}_{ba} &= i\sum_{w=1}^{N}w(N-w+1)\frac{[A(w)-A(w-1)]^3}{2A(w)A(w-1)}.
\end{align}
In particular, the eigenvectors and eigenvalues are always given by $u_+$ and $u_-$, with the respective eigenvalues being
\begin{align}
    \text{RLD}_+ &= \frac{1}{2}\sum_{w=1}^{N}w(N-w+1)\frac{[A(w)-A(w-1)]^2}{A(w-1)} \\
    \text{RLD}_- &= \frac{1}{2}\sum_{w=1}^{N}w(N-w+1)\frac{[A(w)-A(w-1)]^2}{A(w)}.
\end{align}
\label{lem:rld-fisher-info-state-diagonal-dicke-state}
\end{lemma}

Notice that the only difference between the $\text{RLD}_+$ and $\text{RLD}_-$ eigenvalues is that the coefficient in the denominator corresponds to Hamming weight $w-1$ for $\text{RLD}_+$, but $w$ for $\text{RLD}_-$.

\begin{proof}
The Lie group of simultaneous qubit rotations $\{U^{\otimes N} \,|\, U\in\mathrm{SU}(2)\}$ is generated by the Lie algebra of total spin Hamiltonians $\{S(N,\vec{v}) \,|\, \vec{v}\in\Rbb^3\}$. Increasing $a$ by an infinitesimal amount from $0$ corresponds to rotating \emph{counterclockwise} about the $y$-axis as viewed from the $+y$ direction, whereas increasing $b$ by an infinitesimal amount from $0$ corresponds to rotating \emph{clockwise} about the $x$-axis as viewed from the $+x$ direction. This implies that
\begin{align}
    \partial_a\sigma(0,0) &= -\frac{i}{2}[S(N,\hat{y}),\sigma(0,0)] \\
    \partial_b\sigma(0,0) &= +\frac{i}{2}[S(N,\hat{x}),\sigma(0,0)].
\end{align}
Fortunately, the relevant total spin Hamiltonians have very nice standard forms in the Dicke state basis:
\begin{align}
    S(N,\hat{x}) &= \sum_{w=1}^{N}\sqrt{w(N-w+1)}\left[\ket{D^{(N)}_w}\bra{D^{(N)}_{w-1}} + \ket{D^{(N)}_{w-1}}\bra{D^{(N)}_w}\right] \\
    S(N,\hat{y}) &= i\sum_{w=1}^{N}\sqrt{w(N-w+1)}\left[\ket{D^{(N)}_w}\bra{D^{(N)}_{w-1}} - \ket{D^{(N)}_{w-1}}\bra{D^{(N)}_w}\right].
\end{align}
Plugging these forms into the above commutator formulas yields
\begin{align}
    \partial_a\sigma(0,0) &= \frac{1}{2}\sum_{w=1}^{N}\sqrt{w(N-w+1)}\left[A(w) - A(w-1)\right]\left[\ket{D^{(N)}_w}\bra{D^{(N)}_{w-1}} + \ket{D^{(N)}_{w-1}}\bra{D^{(N)}_w}\right] \\
    \partial_b\sigma(0,0) &= \frac{i}{2}\sum_{w=1}^{N}\sqrt{w(N-w+1)}\left[A(w) - A(w-1)\right]\left[\ket{D^{(N)}_w}\bra{D^{(N)}_{w-1}} - \ket{D^{(N)}_{w-1}}\bra{D^{(N)}_w}\right].
\end{align}
Of course, we must also compute
\begin{equation}
    \sigma(0,0)^{-1} = \sum_{w=0}^{N}A(w)^{-1}\ket{D^{(N)}_w}\bra{D^{(N)}_w}.
\end{equation}
We can now plug everything into the entrywise formula for the RLD Fisher information matrix:
\begin{equation}
    \text{RLD}_{\mu\nu} = \text{Tr}\left[(\partial_\mu\rho)(\rho^{-1})(\partial_\nu\rho)\right].
\end{equation}
From this point onward, computing the four matrix entries and subsequently determining the eigenvalues and eigenvectors are both routine calculations.
\end{proof}

For additional intuition, it is worthwhile to comment on a few special values of the purity level $\lambda$:
\begin{itemize}
    \item On one extreme, when $\lambda=1$, we have $A(w) = 1$ for $w=N$ and $A(w) = 0$ otherwise. This causes the $w=N$ term of the above $\text{RLD}_+$ formula to diverge, meaning that the state has infinite $\text{RLD}_+$, as expected. However, the state still has finite $\text{RLD}_-$.
    \item On the other extreme, when $\lambda=-1$, we have $A(w) = 1$ for $w=0$ and $A(w) = 0$ otherwise. This causes the $w=1$ term of the above $\text{RLD}_-$ formula to diverge, meaning that the state has infinite $\text{RLD}_-$, as expected. However, the state still has finite $\text{RLD}_+$.
    \item In the middle, when $\lambda=0$, we have $A(w) = \left(N+1\right)^{-1}$ for all $0\le w\le N$. Therefore both summations become zero, meaning that the state has zero RLD Fisher information, exactly as expected.
\end{itemize}

\appsubsec{``Sensitivity'' of RLD Fisher Information Near an IID Qubit State}
{subsec:rld-sensitivity}

If we were required to convert $\rho(\lambda_{\text{in}},\hat{n})^{\otimes N}$ into $\rho(\lambda_{\text{out}},\hat{n})^{\otimes\lfloor RN\rfloor}$ exactly, then the monotonicity of RLD Fisher information as proved by Petz would already be enough to upper bound the rate $R$, as we showed in Appendix \ref{sec:converse-bound-rld-sensitivity}\ref{subsec:conversion-rate-upper-bounds-intuitive}.

\vspace{0.5\baselineskip}

However, our allowance of vanishing but still nonzero trace distance between our output state $\mE_N\left(\rho(\lambda_{\text{in}},\hat{n})^{\otimes N}\right)$ and our target state $\rho(\lambda_{\text{out}},\hat{n})^{\otimes\lfloor RN\rfloor}$ poses a danger. What if, by introducing a cleverly chosen deviation, we could decrease the RLD Fisher information of the family of output states substantially, and thereby ruin the ability of RLD monotonicity to upper bound the rate?

\vspace{0.5\baselineskip}

It is now finally time for us to dismiss this concern, thus establishing that the bound on the qubit conversion rate holds, even while allowing a small but vanishing trace distance. We will show that, if a state is ``sufficiently close'' to an i.i.d. state, then its RLD Fisher information cannot be ``too much less'' than that of the i.i.d. state. More precisely, we establish the following result:

\begin{lemma}[RLD Fisher information lower bound for vanishing trace distance]
\label{lem:approx-conversion-rld-lower-bound}
Suppose there exists a sequence of $M$-qubit states $\sigma(M,\lambda,\hat{n})$ that are permutation-invariant and commute with the total spin operator $S(M,\hat{n}) = \sum_{j=1}^{M}(\vec{\sigma}\cdot\hat{n})_j$, and such that $\lim_{M\rightarrow\infty}d_{\text{Tr}}\left(\sigma_M,\rho(\lambda,\hat{n})^{\otimes M}\right)=0$. Then
\begin{align}
    \text{RLD}_+(\sigma(M,\lambda,\hat{n})) &\ge M\cdot\text{RLD}_{\text{max}}(\lambda) - o(M) \\
    \text{RLD}_-(\sigma(M,\lambda,\hat{n})) &\ge M\cdot\text{RLD}_{\text{min}}(\lambda) - o(M).
\end{align}
\end{lemma}

Proving Lemma \ref{lem:approx-conversion-rld-lower-bound} turns out to be somewhat tricky. One tempting route would be to try to show that the decrease in RLD Fisher information is at most $O(M)$ times the trace distance. However, this turns out to be false! Although the various QFI metrics are $O(M)$ for $M$-qubit tensor product states (and hence also for separable states by convexity), they can in general be $\Theta(M^2)$ for highly entangled states. As a result, introducing a deviation in the direction away from one of these highly entangled states can actually yield a decrease in RLD Fisher information that is $\Theta(M^2)$ times the trace distance.

\vspace{0.5\baselineskip}

Previous works in the resource theory of asymmetry that establish information-theoretic upper bounds on linear conversion rates have also contended with precisely this issue \cite{yamaguchi2026,marvian2022}, and they tackle it using bespoke methods in their respective settings of interest. In this work, we do the same, presenting a proof that only applies to our very specific problem. The task of rigorously proving a more general information-theoretic upper bound on conversion rates in a general RTA thus remains an open problem for future investigation.

\vspace{0.5\baselineskip}

We first explain why Lemma \ref{lem:approx-conversion-rld-lower-bound}, along with the monotonicity of RLD Fisher information, suffices to prove Theorem \ref{thm:qubit-linear-conversion-rate-upper-bound}:

\begin{proof}[Proof of Theorem \ref{thm:qubit-linear-conversion-rate-upper-bound}]
Suppose we denote the $\lfloor RN\rfloor$-qubit output state as follows:
\begin{equation}
    \sigma_N \coloneqq \sigma(\lfloor RN\rfloor,\lambda_{\text{out}},\hat{n}) \coloneqq \mE_N\left(\rho(\lambda_{\text{in}},\hat{n})^{\otimes N}\right).
\end{equation}
Since $\mE_N$ can depend on $N$, $\lambda_{\text{in}}$, and $\lambda_{\text{out}}$, the output state $\sigma_N$ can depend on $N$, $\lambda_{\text{in}}$, $\lambda_{\text{out}}$, and $\hat{n}$. However, for convenience, we suppress these dependences in the notation (other than the dependence on $N$).

\vspace{0.5\baselineskip}

First, we apply Lemma \ref{lem:restrictions-on-output-state}. This means that we can assume without loss of generality that $\sigma_N$ is permutation-invariant and commutes with the total spin operator $S(\lfloor RN\rfloor,\hat{n})$. Therefore, we have all the conditions needed to apply Lemma \ref{lem:approx-conversion-rld-lower-bound}.

\vspace{0.5\baselineskip}

First, for \textbf{concentration} and \textbf{dilution}, Lemma \ref{lem:approx-conversion-rld-lower-bound} tells us that
\begin{align}
    \text{RLD}_+(\sigma_N) &\ge \lfloor RN\rfloor\cdot\text{RLD}_{\text{max}}(\lambda_{\text{out}}) - o(N) \\
    \text{RLD}_-(\sigma_N) &\ge \lfloor RN\rfloor\cdot\text{RLD}_{\text{min}}(\lambda_{\text{out}}) - o(N).
\end{align}
Therefore, when we apply the monotonicity of RLD Fisher information and compare corresponding eigenvalues (just as we did in Appendix \ref{subsec:conversion-rate-upper-bounds-intuitive}), we obtain
\begin{align}
    N\cdot\text{RLD}_{\text{max}}(\lambda_{\text{in}}) \ge \text{RLD}_+(\sigma_N) &\ge \lfloor RN\rfloor\cdot\text{RLD}_{\text{max}}(\lambda_{\text{out}}) - o(N) \\
    N\cdot\text{RLD}_{\text{min}}(\lambda_{\text{in}}) \ge \text{RLD}_-(\sigma_N) &\ge \lfloor RN\rfloor\cdot\text{RLD}_{\text{min}}(\lambda_{\text{out}}) - o(N).
\end{align}
Rearranging these inequalities to solve for $R$, exactly as we did in Appendix \ref{subsec:conversion-rate-upper-bounds-intuitive}, yields
\begin{align}
    R &\le \frac{\text{RLD}_{\text{max}}(\lambda_{\text{in}})}{\text{RLD}_{\text{max}}(\lambda_{\text{out}})} + \frac{1}{N} + \frac{o(N)}{N} \\
    R &\le \frac{\text{RLD}_{\text{min}}(\lambda_{\text{in}})}{\text{RLD}_{\text{min}}(\lambda_{\text{out}})} + \frac{1}{N} + \frac{o(N)}{N}.
\end{align}
Taking the limit as $N\to\infty$, we obtain the same upper bounds as we did in Appendix \ref{subsec:conversion-rate-upper-bounds-intuitive}.

\vspace{0.5\baselineskip}

Next, for \textbf{wrong-way conversion}, we apply the same reasoning as above, except we replace $\lambda_{\text{out}}$ with $-\lambda_{\text{out}}$. Therefore, Lemma \ref{lem:approx-conversion-rld-lower-bound} and the monotonicity of RLD Fisher information yield
\begin{align}
    N\cdot\text{RLD}_{\text{max}}(\lambda_{\text{in}}) &\ge \lfloor RN\rfloor\cdot\text{RLD}_{\text{max}}(-\lambda_{\text{out}}) - o(N) \\
    &= \lfloor RN\rfloor\cdot\text{RLD}_{\text{min}}(\lambda_{\text{out}}) - o(N) \\
    N\cdot\text{RLD}_{\text{min}}(\lambda_{\text{in}}) &\ge \lfloor RN\rfloor\cdot\text{RLD}_{\text{min}}(-\lambda_{\text{out}}) - o(N) \\
    &= \lfloor RN\rfloor\cdot\text{RLD}_{\text{max}}(\lambda_{\text{out}}) - o(N).
\end{align}
Rearranging the latter inequality (which yields the stronger upper bound) to solve for $R$ yields
\begin{align}
    R &\le \frac{\text{RLD}_{\text{min}}(\lambda_{\text{in}})}{\text{RLD}_{\text{max}}(\lambda_{\text{out}})} + \frac{1}{N} + \frac{o(N)}{N}.
\end{align}
Taking the limit as $N\to\infty$, we obtain the same upper bound as we did in Appendix \ref{subsec:conversion-rate-upper-bounds-intuitive}.

\vspace{0.5\baselineskip}

Finally, for \textbf{measure-and-prepare conversion}, we apply the same reasoning as above, except we also introduce the intermediate random variable $X$ (just as for $\sigma_N$, we suppress the notation for the dependence of $X$ on the various relevant quantities). Note that we do not assume anything about the eigenvectors of $\text{RLD}(X)$. However, we can still define the $\text{RLD}_+$ and $\text{RLD}_-$ values for $X$ in a natural way as follows:
\begin{equation}
    \text{RLD}_\pm(X) \coloneqq (u_\pm)^\dagger\text{RLD}(X)(u_\pm).
\end{equation}
In particular, this definition is consistent with the definition of $\text{RLD}_+$ and $\text{RLD}_-$ for families of states whose RLD Fisher information eigenvalues are necessarily $u_+$ and $u_-$. Furthermore, the monotonicity of RLD Fisher information is a matrix inequality, meaning that it enforces monotonicity on both $\text{RLD}_+$ and $\text{RLD}_-$.

\vspace{0.5\baselineskip}

Since $X$ is a classical random variable, $\text{RLD}(X)$ collapses to a real Fisher information matrix, which means that
\begin{equation}
    \text{RLD}_+(X) = (u_+)^\dagger\text{RLD}(X)(u_+) = (u_-)^\dagger\text{RLD}(X)(u_-) = \text{RLD}_-(X).
\end{equation}
Therefore, Lemma \ref{lem:approx-conversion-rld-lower-bound} and the monotonicity of RLD Fisher information yield
\begin{align}
    N\cdot\text{RLD}_{\text{min}}(\lambda_{\text{in}}) &= \text{RLD}_-\left(\rho(\lambda_{\text{in}},\hat{n})^{\otimes N}\right) \\
    &\ge \text{RLD}_-(X) \\
    &= \text{RLD}_+(X) \\
    &\ge \text{RLD}_+(\sigma_N) \\
    &\ge \lfloor RN\rfloor\cdot\text{RLD}_{\text{max}}(\lambda_{\text{out}}) - o(N).
\end{align}
Rearranging this inequality to solve for $R$ yields
\begin{align}
    R &\le \frac{\text{RLD}_{\text{min}}(\lambda_{\text{in}})}{\text{RLD}_{\text{max}}(\lambda_{\text{out}})} + \frac{1}{N} + \frac{o(N)}{N}.
\end{align}
Taking the limit as $N\to\infty$, we obtain the same upper bound as we did in Appendix \ref{subsec:conversion-rate-upper-bounds-intuitive}.
\end{proof}

Now that we have established that Lemma \ref{lem:approx-conversion-rld-lower-bound} suffices for us to rigorously upper bound the conversion rates and thus prove Theorem \ref{thm:qubit-linear-conversion-rate-upper-bound}, we proceed to prove it. The key ingredient is the following lemma:

\begin{lemma}[more specific lower bound on RLD for allowed output state]
Suppose that $\sigma(M,\lambda,\hat{n})$ is a permutation-invariant state on $M$ qubits that commutes with the total spin operator $S(M,\hat{n}) = \sum_{j=1}^{M}(\vec{\sigma}\cdot\hat{n})_j$. Furthermore, let $d_{\text{Tr}}$ denote the trace distance between $\sigma(M,\lambda,\hat{n})$ and the i.i.d. state $\rho(\lambda,\hat{n})^{\otimes M}$. Then
\begin{align}
    \text{RLD}_+\left(\rho(\lambda,\hat{n})^{\otimes M}\right) - \text{RLD}_+\left(\sigma(M,\lambda,\hat{n})\right) &\le \left[2C_{1,+}\cdot d_{\text{Tr}}^{1/2} + \frac{1}{2}C_{2,+}\cdot\left(\frac{c_0}{c_1}\right)^{d_{\text{Tr}}^{-1/2}}d_{\text{Tr}}^{-1/2}\right]M \\
    \text{RLD}_-\left(\rho(\lambda,\hat{n})^{\otimes M}\right) - \text{RLD}_-\left(\sigma(M,\lambda,\hat{n})\right) &\le \left[2C_{1,-}\cdot d_{\text{Tr}}^{1/2} + \frac{1}{2}C_{2,-}\cdot\left(\frac{c_0}{c_1}\right)^{d_{\text{Tr}}^{-1/2}}d_{\text{Tr}}^{-1/2}\right]M,
\end{align}
where we have defined the constants
\begin{equation}
    C_{1,+} = \frac{4\lambda}{(1-\lambda)^2}, \quad C_{1,-} = \frac{4\lambda(1+2\lambda)}{(1+\lambda)^2}, \quad C_{2,+} = \frac{2(1+\lambda)^2}{(1-\lambda)^2}, \quad C_{2,-} = \frac{2(1+\lambda)}{1-\lambda}.
\end{equation}
\label{lem:rld-lower-bound-approx-iid-state-specific}
\end{lemma}

Lemma \ref{lem:rld-lower-bound-approx-iid-state-specific} says that, as long as a state is sufficiently close in trace distance to an i.i.d. state and also has the form specified by Lemma \ref{lem:restrictions-on-output-state}, its RLD Fisher information cannot be too much less than that of the i.i.d. state.

\vspace{0.5\baselineskip}

In particular, as $d_{\text{Tr}}\to 0$, it is clear that the quantity in brackets goes to $0$:
\begin{itemize}
    \item The first term goes to $0$ because $d_{\text{Tr}}\to 0$ implies $d_{\text{Tr}}^{1/2}\to 0$.
    \item The second term goes to $0$ because $y\to\infty$ implies $yr^y\to 0$ for any $0 < r < 1$. In this case, we have $y = d_{\text{Tr}}^{-1/2}$ and $r = \frac{c_0}{c_1}$.
\end{itemize}
Hence Lemma \ref{lem:approx-conversion-rld-lower-bound} is an immediate corollary of Lemma \ref{lem:rld-lower-bound-approx-iid-state-specific}. Now that we have established that Lemma \ref{lem:rld-lower-bound-approx-iid-state-specific} suffices for us to rigorously upper bound the conversion rates and thus prove Theorem \ref{thm:qubit-linear-conversion-rate-upper-bound}, we proceed to prove it:

\begin{proof}[Proof of Lemma \ref{lem:rld-lower-bound-approx-iid-state-specific}]
Since $\sigma(M,\lambda,\hat{n})$ is assumed to be permutation-invariant, it has a block-diagonal structure implied by Schur-Weyl duality:
\begin{equation}
    \sigma(M,\lambda,\hat{n}) = \bigoplus_{M_C}q(M,M_C,\lambda)\tau(M,M_C,\hat{n},\lambda)\otimes\frac{\Ibb}{d(M,M_C)},
\end{equation}
where $q(M,M_C,\lambda)$ is some probability distribution over the values of $M_C$ (all integers from $0$ to $M$ that have the same parity as $M$), and $\tau(M,M_C,\lambda,\hat{n})$ is a state on the symmetric subspace of $M_C$ qubits. In other words, if we applied Schur sampling to $\sigma(M,\lambda,\hat{n})$, we would obtain $\tau(M,M_C,\lambda,\hat{n})$ with probability $q(M,M_C,\lambda)$, for each possible value of $M_C$.

\vspace{0.5\baselineskip}

The second part of Lemma \ref{lem:restrictions-on-output-state} tells us that $\sigma(M,\lambda,\hat{n})$ commutes with $S(M,\hat{n})$. Applying this to the form above, we deduce that $\tau(M,M_C,\lambda,\hat{n})$ commutes with $S(M_C,\hat{n})$ for every possible value of $M_C$. This means that $\tau(M,M_C,\lambda,\hat{n})$ is a classical mixture of Dicke states with respect to $\hat{n}$:
\begin{equation}
    \tau(M,M_C,\lambda,\hat{n}) = \sum_{w=0}^{M_C}B(M_C,w)\ket{D^{(M_C)}_w}\bra{D^{(M_C)}_w}_{\hat{n}},
\end{equation}
where $B(M_C,w)$ is a probability distribution over Hamming weights $0\le w\le M_C$ for each $M_C$.

\vspace{0.5\baselineskip}

As we discuss in Appendix \ref{sec:schur-sampling-commentary}, the i.i.d. state $\rho(\lambda,\hat{n})^{\otimes M}$ also has a block-diagonal form owing to its permutation invariance:
\begin{equation}
    \rho(\lambda,\hat{n})^{\otimes M} = \bigoplus_{M_C}p(M,M_C,\lambda)\rho_C(M_C,\lambda,\hat{n})\otimes\frac{\Ibb}{d(M,M_C)}.
\end{equation}
In particular, the Schur-transformed state $\rho_C(M_C,\lambda,\hat{n})$ is known to have the form
\begin{equation}
    \rho_C(M_C,\lambda,\hat{n}) = \frac{c_1 - c_0}{c_1^{M_C+1} - c_0^{M_C+1}}\sum_{w=0}^{M_C}c_1^wc_0^{M_C-w}\ket{D^{(M_C)}_w}\bra{D^{(M_C)}_w}_{\hat{n}}.
\end{equation}
For convenience, define
\begin{equation}
    A(M_C,w) \coloneqq \frac{c_1 - c_0}{c_1^{M_C+1} - c_0^{M_C+1}}c_1^wc_0^{M_C-w}.
\end{equation}
Then $A(M_C,w)$ is a probability distribution over Hamming weights $0\le w\le M_C$ for each $M_C$ (similarly to the $B(M_C,w)$ values). As we do in previous appendices, define $\delta\coloneqq M_C - w$. Then we can instead write
\begin{equation}
    A(M_C,w) = \frac{c_1^{M_C}\left(c_1 - c_0\right)}{c_1^{M_C+1} - c_0^{M_C+1}}\left(\frac{c_0}{c_1}\right)^\delta.
\end{equation}
This formulation makes it clear that these coefficients form a geometric series that decays as the Hamming weight $w$ decreases (or equivalently, as $\delta$ increases).

\vspace{0.5\baselineskip}

Since $\sigma(M,\lambda,\hat{n})$ and $\rho(\lambda,\hat{n})^{\otimes M}$ are simultaneously diagonalizable, their trace distance collapses to a total variation distance between classical probability distributions:
\begin{equation}
    d_{\text{Tr}}\left(\sigma(M,\lambda,\hat{n}), \rho(\lambda,\hat{n})^{\otimes M}\right) = \frac{1}{2}\sum_{M_C}\sum_{w=0}^{M_C}\abs{p(M,M_C,\lambda)A(M_C,w) - q(M,M_C,\lambda)B(M_C,w)}.
\end{equation}

\vspace{0.5\baselineskip}

Furthermore, by Lemma \ref{lem:rld-fisher-info-state-diagonal-dicke-state}, the RLD Fisher information eigenvalues corresponding to $u_+$ (which we call $\text{RLD}_+$) carried by $\sigma(M,\lambda,\hat{n})$ and $\rho(\lambda,\hat{n})^{\otimes M}$ have the following formulas:
\begin{align}
    \text{RLD}_+\left(\rho(\lambda,\hat{n})^{\otimes M}\right) &= \sum_{M_C}p(M,M_C,\lambda)\left[\frac{1}{2}\sum_{w=1}^{M_C}w(M_C-w+1)\frac{\left[A(M_C,w) - A(M_C,w-1)\right]^2}{A(M_C,w-1)}\right] \\
    \text{RLD}_+\left(\sigma(M,\lambda,\hat{n})\right) &= \sum_{M_C}q(M,M_C,\lambda)\left[\frac{1}{2}\sum_{w=1}^{M_C}w(M_C-w+1)\frac{\left[B(M_C,w) - B(M_C,w-1)\right]^2}{B(M_C,w-1)}\right].
\end{align}
Similarly, the RLD Fisher information eigenvalues corresponding to $u_-$ (which we call $\text{RLD}_-$) carried by $\sigma(M,\lambda,\hat{n})$ and $\rho(\lambda,\hat{n})^{\otimes M}$ have the following formulas:
\begin{align}
    \text{RLD}_-\left(\rho(\lambda,\hat{n})^{\otimes M}\right) &= \sum_{M_C}p(M,M_C,\lambda)\left[\frac{1}{2}\sum_{w=1}^{M_C}w(M_C-w+1)\frac{\left[A(M_C,w) - A(M_C,w-1)\right]^2}{A(M_C,w)}\right] \\
    \text{RLD}_-\left(\sigma(M,\lambda,\hat{n})\right) &= \sum_{M_C}q(M,M_C,\lambda)\left[\frac{1}{2}\sum_{w=1}^{M_C}w(M_C-w+1)\frac{\left[B(M_C,w) - B(M_C,w-1)\right]^2}{B(M_C,w)}\right].
\end{align}

\vspace{0.5\baselineskip}

As a brief aside, if we wanted to compute the eigenvalue of some other quantum Fisher information (QFI) metric, corresponding to Morozova-Chentsov (MC) function $f$, then the denominators would instead be $m_f(A(M_C,w-1),A(M_C,w))$ and $m_f(B(M_C,w-1),B(M_C,w))$, where $m_f(x,y)\coloneqq xf(y/x)$. The special cases above correspond to $f(t)=1$ and $f(t)=t$ for $\text{RLD}_+$ and $\text{RLD}_-$, respectively. See Appendix \ref{sec:qfi-metrics} for a more detailed discussion of the full family of QFI metrics.

\vspace{0.5\baselineskip}

It will be convenient to define the following probability distributions:
\begin{align}
    A_C(M_C,w) &\coloneqq p(M,M_C,\lambda)A(M_C,w) \\
    B_C(M_C,w) &\coloneqq q(M,M_C,\lambda)B(M_C,w).
\end{align}
This is simply a matter of adjusting each Dicke state coefficient by the probability of the appropriate Schur sampling outcome. Using these two distributions, we can rewrite the trace distance as
\begin{equation}
    d_{\text{Tr}}\left(\sigma(M,\lambda,\hat{n}), \rho(\lambda,\hat{n})^{\otimes M}\right) = \frac{1}{2}\sum_{M_C}\sum_{w=0}^{M_C}\abs{A_C(M_C,w) - B_C(M_C,w)}.
\end{equation}
Furthermore, we can rewrite the RLD Fisher information eigenvalues as
\begin{align}
    \text{RLD}_+\left(\rho(\lambda,\hat{n})^{\otimes M}\right) &= \frac{1}{2}\sum_{M_C}\sum_{w=1}^{M_C}w(M_C-w+1)\frac{\left[A_C(M_C,w) - A_C(M_C,w-1)\right]^2}{A_C(M_C,w-1)} \\
    \text{RLD}_+\left(\sigma(M,\lambda,\hat{n})\right) &= \frac{1}{2}\sum_{M_C}\sum_{w=1}^{M_C}w(M_C-w+1)\frac{\left[B_C(M_C,w) - B_C(M_C,w-1)\right]^2}{B_C(M_C,w-1)} \\
    \text{RLD}_-\left(\rho(\lambda,\hat{n})^{\otimes M}\right) &= \frac{1}{2}\sum_{M_C}\sum_{w=1}^{M_C}w(M_C-w+1)\frac{\left[A_C(M_C,w) - A_C(M_C,w-1)\right]^2}{A_C(M_C,w)} \\
    \text{RLD}_-\left(\sigma(M,\lambda,\hat{n})\right) &= \frac{1}{2}\sum_{M_C}\sum_{w=1}^{M_C}w(M_C-w+1)\frac{\left[B_C(M_C,w) - B_C(M_C,w-1)\right]^2}{B_C(M_C,w)}.
\end{align}

\vspace{0.5\baselineskip}

The coefficient $w(M_C-w+1)$ can generally be $\Theta(M^2)$, which is why small deviations can generally create changes in RLD Fisher information that are $\Theta(M^2)$ times the trace distance. However, in our case of interest, where the output state is i.i.d., the coefficients decay like a geometric series and thus are negligible except for Hamming weights $w$ very close to $M_C$. This is what ultimately allows us to circumvent this difficulty.

\vspace{0.5\baselineskip}

Let $\delta_{\text{max}} = \lfloor d_{\text{Tr}}^{-1/2}\rfloor$. We will split the summation over Hamming weights $w$ into the terms where $0\le\delta\le\delta_{\text{max}}-1$ and the terms where $\delta_{\text{max}}\le\delta\le M_C-1$.

\vspace{0.5\baselineskip}

For the terms where $0\le\delta\le\delta_{\text{max}}-1$, the coefficient $w(M_C-w+1)$ will be $O(\delta_{\text{max}}M_C)$, rather than the generic $O(M_C^2)$. This allows us to nicely upper bound the deficiency in RLD Fisher information of the output state relative to the target state. We will do so using the following niche technical lemma:

\begin{lemma}
Define the following quantities:
\begin{align}
    \varepsilon(M_C,w) &\coloneqq \abs{A_C(w) - B_C(M_C,w)} \\
    \eta(M_C,w) &\coloneqq \frac{\varepsilon(M_C,w)}{A_C(M_C,w)} = \frac{\abs{A_C(M_C,w) - B_C(M_C,w)}}{A_C(M_C,w)}.
\end{align}
Then
\begin{align}
    \frac{\left[A_C(M_C,w) - A_C(M_C,w-1)\right]^2}{A_C(M_C,w-1)} - \frac{\left[B_C(M_C,w) - B_C(M_C,w-1)\right]^2}{B_C(M_C,w-1)} &\le C_{1,+}\cdot\left[\varepsilon(M_C,w-1) + \varepsilon(M_C,w)\right] \\
    \frac{\left[A_CM_C,w) - A_C(M_C,w-1)\right]^2}{A_C(M_C,w)} - \frac{\left[B_C(M_C,w) - B_C(M_C,w-1)\right]^2}{B_C(M_C,w)} &\le C_{1,-}\cdot\left[\varepsilon(M_C,w-1) + \varepsilon(M_C,w)\right],
\label{eq:low-delta-rld-diff-upper-bound}
\end{align}
where the constant factors $C_{1,+}$ and $C_{1,-}$ match those defined in Lemma \ref{lem:rld-lower-bound-approx-iid-state-specific}.
\label{lem:low-delta-rld-diff-upper-bound}
\end{lemma}

To avoid cluttering the flow of the main proof, we defer the proof of Lemma \ref{lem:low-delta-rld-diff-upper-bound} to Appendix \ref{sec:converse-bound-rld-sensitivity}\ref{subsec:proofs-niche-lemmas-converse-bound}.

\vspace{0.5\baselineskip}

For the terms where $0\le\delta\le\delta_{\text{max}}-1$, we apply Lemma \ref{lem:low-delta-rld-diff-upper-bound} to upper bound their contribution to the difference in $\text{RLD}_+$ as follows:
\begin{align}
    & \quad\,\, \frac{1}{2}\sum_{M_C}\sum_{w=M_C-\delta_{\text{max}}+1}^{M_C}w(M_C-w+1) \\
    & \quad\quad \Bigg\{\frac{\left[A_C(M_C,w) - A_C(M_C,w-1)\right]^2}{A_C(M_C,w-1)} - \frac{\left[B_C(M_C,w) - B_C(M_C,w-1)\right]^2}{B_C(M_C,w-1)}\Bigg\} \\
    &\le \frac{1}{2}\sum_{M_C}\delta_{\text{max}}M_C\sum_{w=M_C-\delta_{\text{max}}+1}^{M_C} \\
    & \quad\quad \Bigg\{\frac{\left[A_C(M_C,w) - A_C(M_C,w-1)\right]^2}{A_C(M_C,w-1)} - \frac{\left[B_C(M_C,w) - B_C(M_C,w-1)\right]^2}{B_C(M_C,w-1)}\Bigg\} \\
    &\stackrel{(1)}{\le} \frac{1}{2}\delta_{\text{max}}M\sum_{M_C}\sum_{w=M_C-\delta_{\text{max}}+1}^{M_C}C_1\cdot\left[\varepsilon(M_C,w-1) + \varepsilon(M_C,w)\right] \\
    &= \frac{1}{2}\delta_{\text{max}}M\cdot C_{1,+}\left[\sum_{M_C}\sum_{w=M_C-\delta_{\text{max}}+1}^{M_C}\left[\varepsilon(M_C,w-1) + \varepsilon(M_C,w)\right]\right] \\
    &\stackrel{(2)}{\le} \frac{1}{2}\delta_{\text{max}}M\cdot C_{1,+}\left(4d_{\text{Tr}}\right) \\
    &\le \frac{1}{2}d_{\text{Tr}}^{-1/2}M\cdot C_{1,+}\left(4d_{\text{Tr}}\right) \\
    &= \boxed{2C_{1,+}\cdot d_{\text{Tr}}^{1/2}M}.
\end{align}
For clarity, inequality (1) comes from applying the first statement in Lemma \ref{lem:low-delta-rld-diff-upper-bound}, and inequality (2) comes from the fact that
\begin{equation}
    \sum_{M_C}\sum_{w=M_C-\delta_{\text{max}}+1}^{M_C}\left[\varepsilon(M_C,w-1) + \varepsilon(M_C,w)\right] \le 4d_{\text{Tr}}.
\end{equation}
Running through this same argument, but applying the second statement in Lemma \ref{lem:low-delta-rld-diff-upper-bound}, we can obtain the analogous result
\begin{align}
    & \quad\,\, \frac{1}{2}\sum_{M_C}\sum_{w=M_C-\delta_{\text{max}}+1}^{M_C}w(M_C-w+1) \\
    & \quad\quad \Bigg\{\frac{\left[A_C(M_C,w) - A_C(M_C,w-1)\right]^2}{A_C(M_C,w)} - \frac{\left[B_C(M_C,w) - B_C(M_C,w-1)\right]^2}{B_C(M_C,w)}\Bigg\} \\
    &\le \boxed{2C_{1,-}\cdot d_{\text{Tr}}^{1/2}M}.
\end{align}

\vspace{0.5\baselineskip}

For the terms where $\delta_{\text{max}}\le\delta\le M_C-1$, the  contribution to the RLD Fisher information carried by the target state is itself very small, due to the exponential decay of the Dicke state coefficients. Therefore, the deficiency in RLD Fisher information of the output state relative to the target state coming from these terms can be upper bounded simply by the RLD Fisher information carried by the target state:
\begin{align}
    & \quad\,\, \frac{1}{2}\sum_{M_C}\sum_{w=1}^{M_C-\delta_{\text{max}}}w(M_C-w+1)\Bigg\{\frac{\left[A_C(M_C,w) - A_C(M_C,w-1)\right]^2}{A_C(M_C,w-1)} - \frac{\left[B_C(M_C,w) - B_C(M_C,w-1)\right]^2}{B_C(M_C,w-1)}\Bigg\} \\
    &\le \frac{1}{2}\sum_{M_C}\sum_{w=1}^{M_C-\delta_{\text{max}}}w(M_C-w+1)\frac{\left[A_C(M_C,w) - A_C(M_C,w-1)\right]^2}{A_C(M_C,w-1)}.
\end{align}
Due to the form of $A(M_C,w)$, we can additionally write
\begin{align}
    \frac{\left[A(M_C,w) - A(M_C,w-1)\right]^2}{A(M_C,w-1)} &= \frac{A(M_C,w)^2\left(1-\frac{c_0}{c_1}\right)^2}{A(M_C,w)\frac{c_0}{c_1}} \\
    &= \frac{(c_1-c_0)^2}{c_1c_0}A(M_C,w) \\
    &= \frac{(c_1-c_0)^2}{c_1c_0}\cdot\frac{c_1-c_0}{c_1^{M_C+1}-c_0^{M_C+1}}c_1^wc_0^{M_C-w} \\
    &= \frac{(c_1-c_0)^2}{c_1c_0}\cdot\frac{c_1^{M_C}(c_1-c_0)}{c_1^{M_C+1}-c_0^{M_C+1}}\left(\frac{c_0}{c_1}\right)^{M_C-w}.
\label{eq:Aw-RLD-frac-eval}
\end{align}
We now introduce yet another niche lemma:

\begin{lemma}
Define the quantity $C_2 = \frac{(1+\lambda)^2}{2\lambda^2}$. Then
\begin{equation}
    \sum_{\delta=\delta_{\text{max}}}^{M_C-1}(\delta+1)(M_C-\delta)\left(\frac{c_0}{c_1}\right)^{\delta} \le C_2\cdot\left(\frac{c_0}{c_1}\right)^{\delta_{\text{max}}}\delta_{\text{max}}M_C.
\label{eq:high-delta-rld-diff-upper-bound}
\end{equation}
\label{lem:high-delta-rld-diff-upper-bound}
\end{lemma}

Once again, to avoid cluttering the flow of the main proof, we defer the proof of Lemma \ref{lem:high-delta-rld-diff-upper-bound} to Appendix \ref{sec:converse-bound-rld-sensitivity}\ref{subsec:proofs-niche-lemmas-converse-bound}.

\vspace{0.5\baselineskip}

We can apply Lemma \ref{lem:high-delta-rld-diff-upper-bound} as follows:
\begin{align}
    & \quad\,\, \sum_{w=1}^{M_C-\delta_{\text{max}}}w(M_C-w+1)\frac{\left[A(M_C,w) - A(M_C,w-1)\right]^2}{A(M_C,w-1)} \\
    &= \sum_{w=1}^{M_C-\delta_{\text{max}}}w(M_C-w+1)\frac{(c_1-c_0)^2}{c_1c_0}\cdot\frac{c_1^{M_C}(c_1-c_0)}{c_1^{M_C+1}-c_0^{M_C+1}}\left(\frac{c_0}{c_1}\right)^{M_C-w} & (\text{by Eq. }\ref{eq:Aw-RLD-frac-eval}) \\
    &= \frac{(c_1-c_0)^2}{c_1c_0}\frac{c_1^{M_C}(c_1-c_0)}{c_1^{M_C+1} - c_0^{M_C+1}}\sum_{\delta=\delta_{\text{max}}}^{M_C-1}(\delta+1)(M_C-\delta)\left(\frac{c_0}{c_1}\right)^{\delta} & (\delta\coloneqq M_C-w) \\
    &\le \frac{(c_1-c_0)^2}{c_1c_0}\sum_{\delta=\delta_{\text{max}}}^{M_C-1}(\delta+1)(M_C-\delta)\left(\frac{c_0}{c_1}\right)^{\delta} \\
    &\le \frac{(c_1-c_0)^2}{c_1c_0}C_2\cdot\left(\frac{c_0}{c_1}\right)^{\delta_{\text{max}}}\delta_{\text{max}}M_C & (\text{by Lem. }\ref{lem:high-delta-rld-diff-upper-bound}).
\end{align}
Now, using the fact that $d_{\text{Tr}}^{-1/2}-1\le\delta_{\text{max}}\le d_{\text{Tr}}^{-1/2}$, we can say that
\begin{align}
    \frac{(c_1-c_0)^2}{c_1c_0}C_2\cdot\left(\frac{c_0}{c_1}\right)^{\delta_{\text{max}}}\delta_{\text{max}}M_C &\le \frac{(c_1-c_0)^2}{c_1c_0}C_2\cdot\left(\frac{c_0}{c_1}\right)^{d_{\text{Tr}}^{-1/2}-1}d_{\text{Tr}}^{-1/2}M_C \\
    &= \frac{(c_1-c_0)^2}{c_0^2}C_2\cdot\left(\frac{c_0}{c_1}\right)^{d_{\text{Tr}}^{-1/2}}d_{\text{Tr}}^{-1/2}M_C \\
    &= \frac{4\lambda^2}{(1-\lambda)^2}C_2\cdot\left(\frac{c_0}{c_1}\right)^{d_{\text{Tr}}^{-1/2}}d_{\text{Tr}}^{-1/2}M_C \\
    &= C_{2,+}\cdot\left(\frac{c_0}{c_1}\right)^{d_{\text{Tr}}^{-1/2}}d_{\text{Tr}}^{-1/2}M_C,
\end{align}
where the last equality comes from the fact that $C_{2,+} = \frac{4\lambda^2}{(1-\lambda)^2}C_2$.

\vspace{0.5\baselineskip}

Combining the previous two inequalities yields
\begin{equation}
    \sum_{w=1}^{M_C-\delta_{\text{max}}}w(M_C-w+1)\frac{\left[A(M_C,w) - A(M_C,w-1)\right]^2}{A(M_C,w-1)} \le C_{2,+}\cdot\left(\frac{c_0}{c_1}\right)^{d_{\text{Tr}}^{-1/2}}d_{\text{Tr}}^{-1/2}M_C.
\end{equation}
Finally, applying this inequality across the possible values of $M_C$ yields
\begin{align}
    & \quad\,\, \frac{1}{2}\sum_{M_C}\sum_{w=1}^{M_C-\delta_{\text{max}}}w(M_C-w+1) \frac{\left[A_C(M_C,w) - A_C(M_C,w-1)\right]^2}{A_C(M_C,w-1)} \\
    &= \frac{1}{2}\sum_{M_C}p(M,M_C,\lambda)\left[\sum_{w=1}^{M_C-\delta_{\text{max}}}w(M_C-w+1) \frac{\left[A(M_C,w) - A(M_C,w-1)\right]^2}{A(M_C,w-1)}\right] \\
    &\le \frac{1}{2}\sum_{M_C}p(M,M_C,\lambda)\left[C_{2,+}\cdot\left(\frac{c_0}{c_1}\right)^{d_{\text{Tr}}^{-1/2}}d_{\text{Tr}}^{-1/2}M\right] \\
    &= \boxed{\frac{1}{2}C_{2,+}\cdot\left(\frac{c_0}{c_1}\right)^{d_{\text{Tr}}^{-1/2}}d_{\text{Tr}}^{-1/2}M}.
\end{align}
If we replace the denominator $A_C(M_C,w-1)$ with $A_C(M_C,w)$, the only change is a constant factor, since $A_C(M_C,w-1) = \frac{1-\lambda}{1+\lambda}A_C(M_C,w)$. Therefore, we can produce the analogous result:
\begin{align}
    & \quad\,\, \frac{1}{2}\sum_{M_C}\sum_{w=1}^{M_C-\delta_{\text{max}}}w(M_C-w+1) \frac{\left[A_C(M_C,w) - A_C(M_C,w-1)\right]^2}{A_C(M_C,w)} \\
    &\le \frac{1-\lambda}{1+\lambda}\left[\frac{1}{2}C_{2,+}\cdot\left(\frac{c_0}{c_1}\right)^{d_{\text{Tr}}^{-1/2}}d_{\text{Tr}}^{-1/2}M\right] \\
    &= \boxed{\frac{1}{2}C_{2,-}\cdot\left(\frac{c_0}{c_1}\right)^{d_{\text{Tr}}^{-1/2}}d_{\text{Tr}}^{-1/2}M},
\end{align}
where the last equality comes from the fact that $C_{2,-} = \frac{1-\lambda}{1+\lambda}C_{2,+}$.

\vspace{0.5\baselineskip}

We now combine the inequality coming from small $\delta$ values and the inequality coming from large $\delta$ values. We thus upper bound the deficiency of $\text{RLD}_+$ as follows:
\begin{align}
    & \quad\,\,\,\, \text{RLD}_+\left(\rho(\lambda,\hat{n})^{\otimes M}\right) - \text{RLD}_+\left(\sigma(M,\lambda,\hat{n})\right) \\
    &= \frac{1}{2}\sum_{M_C}\sum_{w=1}^{M_C}w(M_C-w+1) \\
    & \quad\quad \Bigg\{\frac{\left[A_C(M_C,w) - A_C(M_C,w-1)\right]^2}{A_C(M_C,w-1)} - \frac{\left[B_C(M_C,w) - B_C(M_C,w-1)\right]^2}{B_C(M_C,w-1)}\Bigg\} \\
    &= \frac{1}{2}\sum_{M_C}\left[\sum_{w=1}^{M_C-\delta_{\text{max}}}w(M_C-w+1)\Big\{\cdots\Big\} + \sum_{w=M_C-\delta_{\text{max}}+1}^{M_C}w(M_C-w+1)\Big\{\cdots\Big\}\right] \\
    &\le \left[2C_{1,+}\cdot d_{\text{Tr}}^{1/2} + \frac{1}{2}C_{2,+}\cdot\left(\frac{c_0}{c_1}\right)^{d_{\text{Tr}}^{-1/2}}d_{\text{Tr}}^{-1/2}\right]M.
\end{align}
We can carry out exactly the same reasoning as above using $\text{RLD}_-$ instead, with the only difference being that the $A_C(M_C,w-1)$ and $B_C(M_C,w-1)$ in the denominators are replaced with $A_C(M_C,w)$ and $B_C(M_C,w)$, respectively. The only change is to replace the constants $C_{1,+}$ and $C_{2,+}$ with $C_{1,-}$ and $C_{2,-}$, respectively.
\end{proof}

Although this proof gets the job done for our qubit conversion problems of interest, it relies on a number of tricks specific to the families of states we study in this paper, such as the expansion of a Schur-transformed state in the basis of Dicke states. The first part of Lemma \ref{lem:restrictions-on-output-state} is broadly applicable to any conversion problem where the target state is permutation-invariant (including but not limited to linear-rate conversion problems, where the target state is i.i.d.). However, many quantum state conversion problems may not necessarily have a nice analogue of the second part of Lemma \ref{lem:restrictions-on-output-state}. As a result, we suspect that upper bounding conversion rates for other families of states may require a meaningfully different approach from the one taken here.

\appsubsec{Proofs of Niche Lemmas}
{subsec:proofs-niche-lemmas-converse-bound}

As part of the proof of Lemma \ref{lem:rld-lower-bound-approx-iid-state-specific}, we introduced two additional niche lemmas whose sole purpose is to help us upper bound the deficiency in RLD Fisher information coming from low and high values of $\delta\coloneqq M_C-w$. We include their proofs here so that they do not excessively clutter the flow of the main argument.

\begin{proof}[Proof of Lemma \ref{lem:low-delta-rld-diff-upper-bound}]
First, recall that the coefficients $A_C(M_C,w)$ are in geometric sequence with respect to $w$:
\begin{equation}
    A_C(M_C,w) = p(M,M_C,\lambda)\frac{c_1 - c_0}{c_1^{M_C+1} - c_0^{M_C+1}}c_1^wc_0^{M_C-w} \implies \frac{A_C(M_C,w-1)}{A_C(M_C,w)} = \frac{c_0}{c_1}.
\end{equation}
Therefore, the first term can be simplified as
\begin{equation}
    \frac{\left[A_C(M_C,w) - A_C(M_C,w-1)\right]^2}{A_C(M_C,w-1)} = A_C(M_C,w)\frac{\left(1-\frac{c_0}{c_1}\right)^2}{\frac{c_0}{c_1}} = A_C(M_C,w)\frac{\lambda^2}{c_1c_0}.
\end{equation}
Then we can write
\begin{align}
    \frac{1}{B_C(M_C,w-1)} &\ge \frac{1}{A_C(M_C,w-1)\left[1+\eta(M_C,w-1)\right]} \\
    &\ge \frac{1}{A_C(M_C,w-1)}\left[1-\eta(M_C,w-1)\right] \\
    &= \frac{1}{A_C(M_C,w)}\frac{c_1}{c_0}\left[1-\eta(M_C,w-1)\right] \\
    &\ge \frac{1}{A_C(M_C,w)}\frac{c_1}{c_0}\Bigg\{1-\frac{c_1}{c_0}\left[\frac{c_0}{c_1}\eta(M_C,w-1) + \eta(M_C,w)\right]\Bigg\}.
\end{align}
The last step may seem odd, but the utility of the quantity in square brackets comes from the fact that
\begin{equation}
    \frac{c_0}{c_1}\eta(M_C,w-1) + \eta(M_C,w) = \frac{\varepsilon(M_C,w-1) + \varepsilon(M_C,w)}{A_C(M_C,w)}.
\label{eq:eta-eps-relation}
\end{equation}
In particular, the quantities $\eta(M_C,w-1)$ and $\eta(M_C,w)$ are summands in the trace distance formula, so it is helpful to have a lower bound involving their sum.

\vspace{0.5\baselineskip}

Just to ensure that our lower bound is actually nonnegative, we will write this as
\begin{equation}
    \frac{1}{B_C(M_C,w-1)} \ge \frac{1}{A_C(M_C,w)}\frac{c_1}{c_0}(1-\min\{x,1\}), \quad x \coloneqq \frac{c_1}{c_0}\left[\frac{c_0}{c_1}\eta(M_C,w-1) + \eta(M_C,w)\right].
\label{eq:BC(w-1)-reciprocal-lower-bound}
\end{equation}
We can also write
\begin{align}
    & \quad\,\,\,\, \abs{B_C(M_C,w) - B_C(M_C,w-1)} \\
    &\ge \abs{A_C(M_C,w) - A_C(M_C,w-1)} - \left[\varepsilon(M_C,w-1) + \varepsilon(M_C,w)\right] \\
    &= \abs{A_C(M_C,w) - A_C(M_C,w-1)} - \left[\eta(M_C,w-1)A_C(M_Cw-1) + \eta(M_C,w)A_C(M_C,w)\right] \\
    &= A_C(M_C,w)\left[1 - \frac{c_0}{c_1} - \frac{c_0}{c_1}\eta(M_C,w-1) - \eta(M_C,w)\right] \\
    &\ge A_C(M_C,w)\frac{\lambda}{c_1}\Bigg\{1 - \frac{c_1}{\lambda}\left[\frac{c_0}{c_1}\eta(M_C,w-1) + \eta(M_C,w)\right]\Bigg\}.
\end{align}
Once again, just to ensure that our lower bound is actually nonnegative, we will write this as
\begin{equation}
    \abs{B_C(M_C,w) - B_C(M_C,w-1)} \ge A_C(M_C,w)\frac{\lambda}{c_1}(1-\min\{y,1\}), \quad y \coloneqq \frac{c_1}{\lambda}\left[\frac{c_0}{c_1}\eta(M_C,w-1) + \eta(M_C,w)\right].
\label{eq:BC-diff-lower-bound}
\end{equation}
Combining Equations \ref{eq:BC(w-1)-reciprocal-lower-bound} and \ref{eq:BC-diff-lower-bound} (and color-coding them \textbf{\textcolor{blue}{blue}} and \textbf{\textcolor{red}{red}}, respectively, for readability) yields
\begin{align}
    & \quad\,\, \frac{\textcolor{red}{\left[B_C(M_C,w) - B_C(M_C,w-1)\right]^2}}{\textcolor{blue}{B_C(M_C,w-1)}} \\
    &\ge \textcolor{red}{\left[A_C(M_C,w)\frac{\lambda}{c_1}(1-\min\{y,1\})\right]^2}\textcolor{blue}{\frac{1}{A_C(M_C,w)}\frac{c_1}{c_0}(1-\min\{x,1\})} \\
    &= A_C(M_C,w)\frac{\lambda^2}{c_1c_0}(1-\min\{x,1\})(1-\min\{y,1\})^2.
\end{align}
Now we take advantage of the fact that
\begin{equation}
    \prod_{j}(1-\min\{x_j,1\}) \ge 1 - \sum_{j}\min\{x_j,1\} \ge 1 - \sum_{j}x_j.
\end{equation}
Applying this fact with $x_1 = x$ and $x_2 = x_3 = y$ yields
\begin{align}
    & \quad\,\, \frac{\left[B(M_C,w) - B(M_C,w-1)\right]^2}{B(M_C,w-1)} \\
    &\ge A_C(M_C,w)\frac{\lambda^2}{c_1c_0}\left[1 - (x+2y)\right] \\
    &= A_C(M_C,w)\frac{\lambda^2}{c_1c_0}\Bigg\{1 - \left(\frac{c_1}{c_0} + 2\frac{c_1}{\lambda}\right)\left[\frac{c_0}{c_1}\eta(M_C,w-1) + \eta(M_C,w)\right]\Bigg\} \\
    &\stackrel{(1)}{=} A_C(M_C,w)\frac{\lambda^2}{c_1c_0} - \frac{\lambda^2}{c_1c_0}\left(\frac{c_1}{c_0} + 2\frac{c_1}{\lambda}\right)\left[\varepsilon(M_C,w-1) + \varepsilon(M_C,w)\right] \\
    &= A_C(M_C,w)\frac{\lambda^2}{c_1c_0} - \frac{4\lambda}{(1-\lambda)^2}\left[\varepsilon(M_C,w-1) + \varepsilon(M_C,w)\right].
\end{align}
For clarity, equality (1) comes from Equation \ref{eq:eta-eps-relation}. We thus conclude that
\begin{align}
    & \quad\,\, \frac{\left[A_C(M_C,w) - A_C(M_C,w-1)\right]^2}{A_C(M_C,w-1)} - \frac{\left[B_C(M_C,w) - B_C(M_C,w-1)\right]^2}{B_C(M_C,w-1)} \\
    &= A_C(M_C,w)\frac{\lambda^2}{c_1c_0} - \frac{\left[B_C(M_C,w) - B_C(M_C,w-1)\right]^2}{B_C(M_C,w-1)} \\
    &\le \frac{4\lambda}{(1-\lambda)^2}\left[\varepsilon(M_C,w-1) + \varepsilon(M_C,w)\right] \\
    &= C_{1,+}\cdot\left[\varepsilon(M_C,w-1) + \varepsilon(M_C,w)\right].
\end{align}

\vspace{0.5\baselineskip}

Now we need to construct the analogous lower bound, but where the denominators are $A_C(M_C,w)$ and $B_C(M_C,w)$, rather than $A_C(M_C,w-1)$ and $B_C(M_C,w-1)$. We can readily see that
\begin{equation}
    \frac{\left[A_C(M_C,w) - A_C(M_C,w-1)\right]^2}{A_C(M_C,w)} = A_C(M_C,w)\left(1-\frac{c_0}{c_1}\right)^2 = A_C(M_C,w)\frac{\lambda^2}{c_1^2}.
\end{equation}
We can also construct a lower bound on $B_C(M_C,w)^{-1}$, analogous to the previously constructed lower bound on $B_C(M_C,w-1)^{-1}$:
\begin{align}
    \frac{1}{B_C(M_C,w)} &\ge \frac{1}{A_C(M_C,w)\left[1+\eta(M_C,w)\right]} \\
    &\ge \frac{1}{A_C(M_C,w)}\left[1-\eta(M_C,w)\right] \\
    &\ge \frac{1}{A_C(M_C,w)}\Bigg\{1-\left[\frac{c_0}{c_1}\eta(M_C,w-1) + \eta(M_C,w)\right]\Bigg\}.
\end{align}
Just to ensure that our lower bound is actually nonnegative, we will write this as
\begin{equation}
    \frac{1}{B_C(M_C,w)}\ge \frac{1}{A_C(M_C,w)}(1-\min\{x',1\}), x' \coloneqq \frac{c_0}{c_1}\eta(M_C,w-1) + \eta(M_C,w).
\label{eq:BC(w)-reciprocal-lower-bound}
\end{equation}
Combining Equation \ref{eq:BC(w)-reciprocal-lower-bound} with \ref{eq:BC-diff-lower-bound} from earlier (and color-coding them \textbf{\textcolor{blue}{blue}} and \textbf{\textcolor{red}{red}}, respectively, for readability) yields
\begin{align}
    & \quad\,\, \frac{\textcolor{red}{\left[B_C(M_C,w) - B_C(M_C,w-1)\right]^2}}{\textcolor{blue}{B_C(M_C,w)}} \\
    &\ge \textcolor{red}{\left[A_C(M_C,w)\frac{\lambda}{c_1}(1-\min\{y,1\})\right]^2}\textcolor{blue}{\frac{1}{A_C(M_C,w)}(1-\min\{x',1\})} \\
    &= A(w)\frac{\lambda^2}{c_1^2}(1-\min\{x',1\})(1-\min\{y,1\})^2.
\end{align}
We now again take advantage of the fact that
\begin{equation}
    \prod_{j}(1-\min\{x_j,1\}) \ge 1 - \sum_{j}\min\{x_j,1\} \ge 1 - \sum_{j}x_j.
\end{equation}
Applying this fact with $x_1 = x'$ and $x_2 = x_3 = y$ yields
\begin{align}
    & \quad\,\,\,\, \frac{\left[B_C(M_C,w) - B_C(M_C,w-1)\right]^2}{B_C(M_C,w)} \\
    &\ge A_C(M_C,w)\frac{\lambda^2}{c_1^2}\left[1 - (x'+2y)\right] \\
    &\ge A_C(M_C,w)\frac{\lambda^2}{c_1^2}\Bigg\{1 - \left(1 + 2\frac{c_1}{\lambda}\right)\left[\frac{c_0}{c_1}\eta(M_C,w-1) + \eta(M_C,w)\right]\Bigg\} \\
    &\stackrel{(1)}{=} A_C(M_C,w)\frac{\lambda^2}{c_1^2} - \frac{\lambda^2}{c_1^2}\left(1 + 2\frac{c_1}{\lambda}\right)\left[\varepsilon(M_C,w-1) + \varepsilon(M_C,w)\right] \\
    &= A(M_C,w)\frac{\lambda^2}{c_1^2} - \frac{4\lambda(1+2\lambda)}{(1+\lambda)^2}\left[\varepsilon(M_C,w-1) + \varepsilon(M_C,w)\right].
\end{align}
For clarity, equality (1) comes from Equation \ref{eq:eta-eps-relation}. We thus conclude that
\begin{align}
    & \quad\,\, \frac{\left[A_C(M_C,w) - A_C(M_C,w-1)\right]^2}{A_C(M_C,w)} - \frac{\left[B_C(M_C,w) - B_C(M_C,w-1)\right]^2}{B_C(M_C,w)} \\
    &= A_C(M_C,w)\frac{\lambda^2}{c_1^2} - \frac{\left[B_C(M_C,w) - B_C(M_C,w-1)\right]^2}{B_C(M_C,w-1)} \\
    &\le \frac{4\lambda(1+2\lambda)}{(1+\lambda)^2}\left[\varepsilon(M_C,w-1) + \varepsilon(M_C,w)\right] \\
    &= C_{1,-}\cdot\left[\varepsilon(M_C,w-1) + \varepsilon(M_C,w)\right].
\end{align}
\end{proof}

\begin{proof}[Proof of Lemma \ref{lem:high-delta-rld-diff-upper-bound}]
First, observe that
\begin{align}
    \sum_{\delta=\delta_{\text{max}}}^{M_C-1}(\delta+1)(M_C-\delta)\left(\frac{c_0}{c_1}\right)^\delta &\le M_C\sum_{\delta=\delta_{\text{max}}}^{M_C-1}(\delta+1)\left(\frac{c_0}{c_1}\right)^\delta \\
    &\le M_C\sum_{\delta=\delta_{\text{max}}}^{\infty}(\delta+1)\left(\frac{c_0}{c_1}\right)^\delta \\
    &= \left(\frac{c_0}{c_1}\right)^{\delta_{\text{max}}}M_C\sum_{n=0}^{\infty}(n+\delta_{\text{max}}+1)\left(\frac{c_0}{c_1}\right)^n.
\end{align}
We now use the following formula, which can be obtained via standard manipulations of the infinite geometric series:
\begin{equation}
    \sum_{n=0}^{\infty}(n+A)r^n = (A-1)\frac{1}{1-r} + \frac{1}{(1-r)^2}.
\end{equation}
In particular, if $A\ge 1$ and $0 < r < 1$, then we can upper bound the above quantity as
\begin{equation}
    \sum_{n=0}^{\infty}(n+A)r^n = (A-1)\frac{1}{1-r} + \frac{1}{(1-r)^2} \le \frac{A-1}{(1-r)^2} + \frac{1}{(1-r)^2} = \frac{A}{(1-r)^2}.
\end{equation}
Applying this inequality to the above summation with $A = \delta_{\text{max}}+1$ and $r = \frac{c_0}{c_1}$ yields
\begin{align}
    \sum_{n=0}^{\infty}(n+\delta_{\text{max}}+1)\left(\frac{c_0}{c_1}\right)^n &\le \frac{\delta_{\text{max}}+1}{\left(1-\frac{c_0}{c_1}\right)^2} \\
    &\le \frac{c_1^2}{(c_1-c_0)^2}(2\delta_{\text{max}}) \\
    &= \frac{(1+\lambda)^2}{2\lambda^2}\delta_{\text{max}}.
\end{align}
Putting this all together yields
\begin{align}
    \sum_{\delta=\delta_{\text{max}}}^{M_C-1}(\delta+1)(M_C-\delta)\left(\frac{c_0}{c_1}\right)^\delta &\le \frac{(1+\lambda)^2}{2\lambda^2}\left(\frac{c_0}{c_1}\right)^{\delta_{\text{max}}}\delta_{\text{max}}M_C \\
    &= C_2\cdot\left(\frac{c_0}{c_1}\right)^{\delta_{\text{max}}}\delta_{\text{max}}M_C.
\end{align}
\end{proof}

\newpage

\appsec{Friendly Implementations of Concentration and Dilution}
{sec:friendly-implementations}

In this appendix, we show an elegant implementation of the Werner optimal cloning map for qubits \cite{Werner1998} that we believe is not common knowledge. This implementation makes it more obvious why the optimal cloning map has the mathematical properties that it does. Furthermore, thanks to a recent work by Brahmachari et al. that implements qubit Schur sampling using two-qubit SWAP tests \cite{Brahmachari2025}, we can show that our concentration and dilution procedures both have relatively friendly implementations, in particular avoiding the need for entangling operations that act on many qubits simultaneously. See \cite{harrow2013church} for a useful review of the properties of the symmetric subspace that will be used in what follows.

\vspace{0.5\baselineskip}

Here is a summary of our method. Suppose we wish to perform $\mE_{\text{clone}}[N\to N+K]$ for some positive integer $K$.
\begin{itemize}
    \item We first introduce $K$ singlets.
    \item We then perform random two-qubit SWAP tests on the $N$ input qubits and one qubit from each singlet pair.
    \item If, after sufficiently many SWAP tests on those $(N+K)$ qubits, we have not observed any singlet outcome, the state is projected approximately onto the totally symmetric subspace of $M$ qubits, and we have obtained the desired output of the optimal cloner $\mE_{\text{clone}}[N\to N+K]$, with an error that is exponentially small in the number of SWAP tests. In particular, at this point, we can discard the $K$ qubits (one from each singlet pair) that did not participate in the SWAP tests, and we are done.
    \item Otherwise, as soon as we observe a singlet, we bring back all the $K$ qubits from the singlet pairs that we set aside and continue performing random SWAP tests on all $(N+2K)$ qubits until we hunt down all $K$ singlets. We then repeat the entire process.
\end{itemize}

\appsubsec{Implementation of the Optimal Cloning Map}
{subsec:optimal-cloning-implementation}

As a reminder, the optimal cloning map from $N$ qudits to $M$ qudits takes the form
\begin{equation}
    \mE_{\text{clone}}[N\rightarrow M](\rho_N) = \frac{\binom{N+d-1}{N}}{\binom{M+d-1}{M}}\Pi^M_{\text{sym}}\left(\rho_N\otimes\Ibb_d^{\otimes(M-N)}\right)\Pi^M_{\text{sym}},
\end{equation}
where $\rho_N$ is a state in the symmetric subspace on $N$ qudits, $d$ is the dimension of the Hilbert space for each qudit, and $\Pi^M_{\text{sym}}$ is the projector to the symmetric subspace on $M$ qudits \cite{Werner1998}. In other words, you take the tensor product with $(M-N)$ identity operators, project to the symmetric subspace on $M$ qudits, and then rescale by the appropriate constant factor to normalize the state.

\vspace{0.5\baselineskip}

If we specialize to qubits ($d=2$), then this formula simplifies to
\begin{equation}
    \mE_{\text{clone}}[N\rightarrow M](\rho_N) = \frac{N+1}{M+1}\Pi^M_{\text{sym}}\left(\rho_N\otimes\Ibb_2^{\otimes(M-N)}\right)\Pi^M_{\text{sym}}.
\end{equation}
Therefore, if you introduce $(M-N)$ maximally mixed qubits, perform a total angular momentum measurement, and obtain the highest possible outcome ($j=M/2$), then the result you obtain will indeed be the result of the optimal cloning map from $N$ qubits to $M$ qubits.

\vspace{0.5\baselineskip}

As a result, a tempting first pass at implementing the optimal cloning map is to perform this procedure and hope that you get the highest possible total angular momentum. The obvious problem with this method is that it does not always succeed! In fact, it succeeds with probability $\frac{M+1}{(N+1)2^{M-N}}$. For example, if $M=N+1$, then you get $j=(N+1)/2$ with probability $\frac{N+2}{2N+2}$, but you get $j=(N-1)/2$ with probability $\frac{N}{2N+2}$. If we obtain the desired outcome $j=(N+1)/2$, then the optimal cloning map has been successfully performed, but if we instead obtain $j=(N-1)/2$, then we have actually lost some amount of information, and our cloning can no longer be optimal.

\vspace{0.5\baselineskip}

So how can we ensure that we always get the desired outcome, or otherwise ensure that we do not lose information in the case of the undesired outcome? The trick is to use a singlet state! Instead of running the total angular momentum measurement with the original $N$ qubits and an isolated maximally mixed state, run it with the original $N$ qubits and one qubit from a singlet state $\ket{\Psi^-}\coloneqq\frac{\ket{01}-\ket{10}}{\sqrt{2}}$. Since the reduced state of one qubit from the singlet is maximally mixed, the outcome probabilities are exactly the same as before. Following the successful outcome $j=(N+1)/2$, one can discard the other qubit from the singlet, and the optimal cloning map has been successfully performed. However, following the failure outcome $j=(N-1)/2$, some amount of the useful information will be stored in the other qubit from the singlet. As a result, as we will show later, one can actually restore the original $N$-qubit state and try again.

\vspace{0.5\baselineskip}

Let us write this out more explicitly. Suppose that we have the $N$-qubit state $\ket{0}\bra{0}^{\otimes N}$, and we want to apply the optimal cloning map to increase the number of qubits by $1$. Applying Lemma \ref{lem:optimal-cloning-map-dicke-state-unified-presentation} with $w=0$ and $M=N+1$, we see that
\begin{equation}
    \mE_{\text{clone}}[N\rightarrow N+1]\left(\ket{0}\bra{0}^{\otimes N}\right) = \frac{N+1}{N+2}\ket{D^{(N+1)}_0}\bra{D^{(N+1)}_0} + \frac{1}{N+2}\ket{D^{(N+1)}_1}\bra{D^{(N+1)}_1}.
\end{equation}
In fact, it suffices to check that our implementation of the optimal cloning map maps this one state to its intended target. This is because it is known that the operators $\{\ket{\psi}\bra{\psi}^{\otimes N} \,|\, \ket{\psi}\in\Cbb^d\}$ span the space of Hermitian operators on the $N$-qudit symmetric subspace \cite{harrow2013church}. As a result, an $\mathrm{SU}(2)$-covariant channel on the $N$-qubit symmetric subspace is fully determined by its action on $\ket{0}\bra{0}^{\otimes N}$. The covariance condition allows you to compute the action on arbitrary $\ket{\psi}\bra{\psi}^{\otimes N}$, and then linearity allows you to compute the action on any Hermitian operator.

\vspace{0.5\baselineskip}

For convenience, define the following two states:
\begin{align}
    \ket{E} &\coloneqq \frac{1}{\sqrt{N(N+1)}}\left[\sum_{i=0}^{N-1}\ket{0^i10^{N-i}} - N\ket{0^N1}\right] \\
    \ket{F} &\coloneqq \frac{1}{\sqrt{2}}\left[\ket{0^N1} - \ket{0^{N-1}10}\right] = \ket{0}^{\otimes(N-1)}\ket{\Psi^-}.
\end{align}
Notice that $\ket{E}$ and $\ket{F}$ are superpositions of computational basis states with Hamming weight $1$, which means that they live fully within the $m=-(N-1)/2$ subspace. Furthermore, both of them are orthogonal to $\ket{D^{(N+1)}_1}$, which is the only state in the $m=-(N-1)/2$ subspace that is fully symmetric and thus lives in the $j=(N+1)/2$ subspace. We conclude that $\ket{E}$ and $\ket{F}$ both live fully within the $j=(N-1)/2$ subspace. These facts imply that, when we write $\ket{E}\bra{E}$ and $\ket{F}\bra{F}$ in the basis implied by Schur-Weyl duality and trace out the multiplicity subsystem, they both have the same reduced state. As a result, a theorem by Marvian and Spekkens guarantees that $\ket{E}$ can be turned into $\ket{F}$ via an $\mathrm{SU}(2)$-covariant unitary (see Theorem 1 of \cite{marvian2013}).

\vspace{0.5\baselineskip}

Now suppose that we begin with the state $\ket{0}^{\otimes N} = \ket{D^{(N)}_0}$. We append a singlet to obtain $\ket{D^{(N)}_0}\otimes\ket{\Psi^-}$, and we perform a total angular momentum measurement on the first $(N+1)$ qubits. To understand the possible outcomes, we rewrite this $(N+2)$-qubit density operator in the total angular momentum basis on the first $(N+1)$ qubits:
\begin{align}
    & \quad\,\,\,\, \ket{D^{(N)}_0}\bra{D^{(N)}_0}\otimes\ket{\Psi^-}\bra{\Psi^-} \\
    &= \frac{1}{2}\left(\ket{0^N01} - \ket{0^N10}\right)\left(\bra{0^N01} - \bra{0^N10}\right) \\
    &= \frac{1}{2}\left[\ket{D^{(N+1)}_0}\ket{1} - \frac{1}{\sqrt{N+1}}\ket{D^{(N+1)}_1}\ket{0} + \sqrt{\frac{N}{N+1}}\ket{E}\ket{0}\right] \\
    & \quad\quad\quad \left[\bra{D^{(N+1)}_0}\bra{1} - \frac{1}{\sqrt{N+1}}\bra{D^{(N+1)}_1}\bra{0} + \sqrt{\frac{N}{N+1}}\bra{E}\bra{0}\right].
\end{align}
In particular, notice that the first two terms in each set of square brackets live in the $j=(N+1)/2$ sector on the first $(N+1)$ qubits, while the third term lives in the $j=(N-1)/2$ sector.

\vspace{0.5\baselineskip}

Now we condition on the outcome of the total angular momentum measurement. On one hand, if we obtain the $j=(N+1)/2$ outcome, which is the outcome we want, then the subsequent state is
\begin{equation}
    \frac{N+1}{N+2}\left[\ket{D^{(N+1)}_0}\ket{1} - \frac{1}{\sqrt{N+1}}\ket{D^{(N+1)}_1}\ket{0}\right]\left[\bra{D^{(N+1)}_0}\bra{1} - \frac{1}{\sqrt{N+1}}\bra{D^{(N+1)}_1}\bra{0}\right].
\end{equation}
Therefore, if we discard the final qubit, the reduced state of the first $(N+1)$ qubits is
\begin{equation}
    \frac{N+1}{N+2}\ket{D^{(N+1)}_0}\bra{D^{(N+1)}_0} + \frac{1}{N+2}\ket{D^{(N+1)}_1}\bra{D^{(N+1)}_1},
\end{equation}
which is exactly the desired outcome of the optimal cloning map.

\vspace{0.5\baselineskip}

On the other hand, if we obtain the undesired $j=(N-1)/2$ outcome, then the subsequent state is $\ket{E}\ket{0}$. As discussed earlier, there exists an $\mathrm{SU}(2)$-covariant unitary that maps $\ket{E}$ to $\ket{F}$, so we can transform this state to $\ket{F}\ket{0} = \ket{0}^{\otimes(N-1)}\ket{\Psi^-}\ket{0}$. But this is just what we started with, namely, $N$ independent qubits in the $\ket{0}$ state and a singlet! Hence, we can just re-order the qubits to have the last two qubits be the singlet, and we can try again.

\vspace{0.5\baselineskip}

As a result, we have successfully described an implementation of the optimal cloning map $\mE_{\text{clone}}[N\to N+1]$. But in fact, this is all we need! The optimal cloning map satisfies a natural composition property: for all nonnegative integers $N\le P\le M$,
\begin{equation}
    \mE_{\text{clone}}[N\to M] = \mE_{\text{clone}}[P\to M] \circ \mE_{\text{clone}}[N\to P].
\end{equation}
We present and prove this fact in Appendix \ref{sec:math-tidbits} as Lemma \ref{lem:optimal-cloning-map-composition}. In particular, repeatedly applying this fact means that
\begin{equation}
    \mE_{\text{clone}}[N\to M] = \mE_{\text{clone}}[M-1\to M] \circ \cdots \circ \mE_{\text{clone}}[N+1\to N+2] \circ \mE_{\text{clone}}[N\to N+1].
\end{equation}
Therefore, to implement $\mE_{\text{clone}}[N\to M]$, you can simply repeat the above procedure exactly $(M-N)$ times.

\vspace{0.5\baselineskip}

For example, suppose that you want to increase the qubit count from $N$ to $N+K$ for some positive integer $K$. Then each single-qubit increase requires you to prepare a fresh singlet state, so you will need to prepare $K$ singlets in total. In addition, when increasing from $(N+q)$ qubits to $(N+q+1)$ qubits for each $0\le q\le k-1$, you have success probability $\frac{N+q+2}{2N+2q+2}$, so you need $\frac{2N+2q+2}{N+q+2}$ attempts on average. Therefore, the total number of angular momentum measurements you need on average is
\begin{equation}
    \sum_{q=0}^{K-1}\frac{2N+2q+2}{N+q+2}.
\end{equation}
Crucially, since the success probability is always greater than $1/2$, you need fewer than $2$ attempts on average for each single-qubit increase, so you need fewer than $2K$ angular momentum measurements on average. If $N$ is large, then each success probability becomes very close to $1/2$, so the average number of angular momentum measurements becomes very close to $2K$.

\vspace{0.5\baselineskip}

The implementation above is certainly more explicit than the abstract description originally offered by Werner \cite{Werner1998}. However, two of the steps are still concerning:
\begin{itemize}
    \item How difficult is the total angular momentum measurement in practice?
    \item How feasible it is to transform $\ket{E}$ into $\ket{F}$? Recall that this part of the implementation was non-constructive: we did not explicitly write out the $\mathrm{SU}(2)$-covariant unitary that maps $\ket{E}$ to $\ket{F}$. 
\end{itemize}
We will now make both of these steps practical by leveraging a recent work by Brahmachari et al. that explains how to perform a total angular momentum measurement using two-qubit SWAP tests \cite{Brahmachari2025}.

\appsubsec{``Universality'' of the Two-Qubit SWAP Test}
{subsec:swap-test-universality}

We will now explain how our optimal concentration and dilution procedures can be implemented solely using the two-qubit SWAP test and a few other basic operations. This will include addressing the questions posed above regarding the practicality of our implementation of the optimal cloning map.

\vspace{0.5\baselineskip}

As a reminder, a SWAP test is a projective measurement on the $2$-qudit Hilbert space defined by the $\frac{d(d+1)}{2}$-dimensional symmetric subspace and the $\frac{d(d-1)}{2}$-dimensional antisymmetric subspace. (In fact, the SWAP test can be understood as the simplest nontrivial example of Schur sampling. The symmetric subspace corresponds to the Young diagram with the two boxes laid horizontally, while the antisymmetric subspace corresponds to the Young diagram with the two boxes stacked vertically.) For the special case of qubits ($d=2$), the $3$-dimensional symmetric (triplet) subspace is spanned by $\ket{00}$, $\ket{\Psi^+} = \frac{\ket{01}+\ket{10}}{\sqrt{2}}$, and $\ket{11}$, while the $1$-dimensional antisymmetric (singlet) subspace is spanned by $\ket{\Psi^-} = \frac{\ket{01}-\ket{10}}{\sqrt{2}}$. It is common to say that a SWAP test ``succeeds'' if it yields the symmetric outcome and ``fails'' if it yields the antisymmetric outcome. (Despite the connotations suggested by these terms, whether one would ``prefer'' a SWAP test to succeed or fail is highly context-dependent.)

\vspace{0.5\baselineskip}

Brahmachari et al. recently showed that two-qubit SWAP tests and random permutations can be used to realize unitary Schur sampling \cite{Brahmachari2025}. The essential idea is to repeatedly perform SWAP tests on randomly chosen pairs of qubits. If a SWAP test yields the symmetric (triplet) outcome, then one can simply do nothing, but if a SWAP test yields the antisymmetric (singlet) outcome, then one can set aside these two qubits and not touch them again, since they are now known to be in a singlet state $\ket{\Psi^-}$. Brahmachari et al. showed that, after $\sim 2N\ln(N/\varepsilon)$ SWAP tests, the remaining qubits will live in the symmetric subspace (meaning that all singlet states have been extracted) with $\ge 1-\varepsilon$ probability \cite{Brahmachari2025}.

\vspace{0.5\baselineskip}

This means that our original application of Schur sampling, namely, as the first of the three big steps for both our concentration and dilution procedures, can be carried out solely using two-qubit SWAP tests and potentially the discarding operation. (We are free to discard the singlets, but we can also just save them for later.)

\vspace{0.5\baselineskip}

However, the angular momentum measurement that we use in our implementation of the optimal cloning map, which we described in Appendix \ref{sec:friendly-implementations}\ref{subsec:optimal-cloning-implementation}, is also just an example of Schur sampling! As a result, it can also be performed using the procedure developed by Brahmachari et al. \cite{Brahmachari2025}.

\vspace{0.5\baselineskip}

Finally, what about the conversion from $\ket{E}$ to $\ket{F}$ in our implementation of the optimal cloning map from Appendix \ref{sec:friendly-implementations}\ref{subsec:optimal-cloning-implementation}? If you implement Schur sampling abstractly, you may generally need to perform an extra $\mathrm{SU}(2)$-covariant unitary to isolate the singlet states so that you can set them aside. (In fact, Cirac et al. also have this non-constructive step in their work on qubit distillation, as they do not bother to explain how one would find the unitary that isolates the singlet states \cite{Cirac1999}.)

\vspace{0.5\baselineskip}

But amazingly, the implementation of Schur sampling by Brahmachari et al. \cite{Brahmachari2025} allows us to avoid this difficulty entirely! In particular, in their procedure, every time a SWAP test produces a singlet, you can set it aside and continue the procedure on the remaining qubits. This means that the singlet states are identified and isolated directly as part of the procedure. 

\vspace{0.5\baselineskip}

Therefore, our implementation of the optimal cloning map $\mE_{\text{clone}}[N\to N+1]$ looks as follows:
\begin{enumerate}
    \item Prepare a fresh singlet state and put it in positions $(N+1)$ and $(N+2)$ (the input qubits should occupy positions $1$ through $N$).
    \item Run a bunch of SWAP tests on randomly chosen pairs out of the first $(N+1)$ qubits \cite{Brahmachari2025}.
    \item If a SWAP test ever produces a singlet, permute the qubits so that the singlet is in the last two positions, and return to Step 2.
    \item If you have enough successful SWAP tests in a row, discard the last qubit (the one that was excluded from the SWAP tests) and return the remaining $(N+1)$ qubits.
\end{enumerate}
And of course, we can repeat this whole procedure $K$ times if we want to implement $\mE_{\text{clone}}[N\to N+K]$.

\vspace{0.5\baselineskip}

Once we combine this improved implementation of the optimal cloning map with the remaining steps, it is easy to see that our concentration and dilution procedures can be carried out using just four basic operations: \textbf{(1)} the two-qubit SWAP test; \textbf{(2)} randomly permuting the qubits; \textbf{(3)} discarding a qubit; \textbf{(4)} manufacturing a singlet.

\vspace{0.5\baselineskip}

All of these procedures can be carried out by interacting at most two qubits at a time. (Even randomly permuting the qubits can be carried out by performing $\sim N\log_2N$ transpositions on randomly selected pairs of qubits.) This makes our concentration and dilution procedures much friendlier for practical implementation.

\vspace{0.5\baselineskip}

If we desire, we can have primitive operation \textbf{(4)} be the manufacturing of a maximally mixed qubit, rather than the manufacturing of a singlet. In particular, running a SWAP test on two independent maximally mixed qubits produces a singlet state with probability $1/4$. Hence you would need $4$ tries on average to obtain a singlet state, which would consume $8$ maximally mixed qubits. Therefore, in this framework, the singlet cost of any procedure should be multiplied by $8$ to obtain the corresponding maximally mixed qubit cost.

\newpage

\appsec{Mathematical Tidbits}
{sec:math-tidbits}

In Appendix \ref{sec:four-important-channels}, we presented four useful $\mathrm{SU}(2)$-covariant channels whose input and output Hilbert spaces are both the symmetric subspace on some number of qubits.

\vspace{0.5\baselineskip}

We always apply these maps to Schur-transformed states, which are classical mixtures of Dicke states with respect to a fixed direction, as described in Appendix \ref{sec:schur-sampling-commentary}. As a result, through the process of working on these qubit conversion problems, we came to appreciate that these channels have some very pleasing mathematical properties when applied to Dicke states.

\vspace{0.5\baselineskip}

This appendix is dedicated to exploring these mathematical properties. Although this appendix is mostly not necessary to understand our work on qubit linear-rate conversion, we believe the curious reader may find its contents interesting nonetheless. This appendix is organized as follows:
\begin{itemize}
    \item In Appendix \ref{sec:math-tidbits}\ref{subsec:optimal-cloning-map-composition}, we show that the optimal cloning map respects a natural composition property: applying it multiple times is equivalent to applying it just once to directly reach the final qubit count. This fact is helpful for our novel implementation of the optimal cloning map, which we presented in Appendix \ref{sec:friendly-implementations}.
    \item In Appendix \ref{sec:math-tidbits}\ref{subsec:optimal-cloning-map-polynomial}, we show how the action of the optimal cloning map on a Dicke state can be expressed very nicely using polynomials.
    \item In Appendix \ref{sec:math-tidbits}\ref{subsec:optimal-mp-channel-polynomial}, we show how the action of the optimal measure-and-prepare channel on a Dicke state can be expressed very nicely using polynomials.
    \item In Appendix \ref{sec:math-tidbits}\ref{subsec:extreme-dilution-extreme-concentration}, we show how the optimal measure-and-prepare channel can be approximately implemented by first massively increasing the qubit count via optimal cloning, and then massively decreasing the qubit count via discarding.
    \item In Appendix \ref{sec:math-tidbits}\ref{subsec:optimal-cloning-sjsm}, we show how the optimal cloning map does not affect the distribution over the unit sphere obtained from the standard joint symmetric measurement (SJSM).
    \item In Appendix \ref{sec:math-tidbits}\ref{subsec:connection-discarding-optimal-cloning}, we explain how the optimal cloning map can be understood as the best possible attempt to undo the discarding map.
\end{itemize}

\appsubsec{Composition Property of the Optimal Cloning Map}
{subsec:optimal-cloning-map-composition}

It is obvious that the discarding map $\mE_{\text{discard}}$ satisfies a natural composition property: if you throw away some qubits, and then throw away some more, you obtain the same result as if you had just thrown away both subsets at once. In particular, for any three nonnegative integers $N\ge P\ge M$,
\begin{equation}
    \mE_{\text{discard}}[N\to M] = \mE_{\text{discard}}[P\to M]\circ\mE_{\text{discard}}[N\to P].
\end{equation}
What is much less obvious is that the optimal cloning map also satisfies a natural composition property! In particular:

\begin{lemma}
For any three nonnegative integers $N\le P\le M$,
\begin{equation}
    \mE_{\text{clone}}[N\to M] = \mE_{\text{clone}}[P\to M]\circ\mE_{\text{clone}}[N\to P].
\end{equation}
\label{lem:optimal-cloning-map-composition}
\end{lemma}

An important consequence of Lemma \ref{lem:optimal-cloning-map-composition} is that as long as we understand $\mE_{\text{clone}}[N\rightarrow N+1]$ (the optimal cloning map that adds just one qubit) well enough, we can understand any other optimal cloning map as the repeated application of this map.

\vspace{0.5\baselineskip}

In particular, if we find a good implementation for $\mE_{\text{clone}}[N\to N+1]$ (the optimal cloning map to add just one qubit), we can simply repeat it $(M-N)$ times to perform $\mE_{\text{clone}}[N\to M]$. This is exactly what we do in Appendix \ref{sec:friendly-implementations}.

\vspace{0.5\baselineskip}

Before we present the formal proof, we will mention two intuitive ways to understand this fact:
\begin{itemize}
    \item The optimal cloning map $\mE_{\text{clone}}[N\to M]$, modulo the global constant factor, has a very nice tensor network representation, as shown in Figure \ref{fig:symmetric-subspace-state-discarding-map-optimal-cloning-map-general} in Appendix \ref{sec:heuristic-arguments}\ref{subsec:tensor-network-symmetric-subspace}. In particular, you add $(M-N)$ identity wires, and then you sandwich the whole $M$-qudit system between two independent permutations chosen uniformly at random (more precisely, you take the average over all such diagrams). If you perform $\mE_{\text{clone}}[N\to P]$, followed immediately by $\mE_{\text{clone}}[P\to M]$ (as shown in Figure \ref{fig:SUB-optimal_cloning_map_general}), you can absorb the inner random permutations (acting on the first $P$ qudits) into the outer random permutations (acting on all $M$ qudits). It then becomes clear that you obtain the same result as if you had just applied $\mE_{\text{clone}}[N\to M]$ (as shown in Figure \ref{fig:SUB-optimal_cloning_map_general_permutations_absorbed}).
    \item As we discuss immediately after the proof of Lemma \ref{lem:optimal-cloning-map-dicke-state-unified-presentation}, when you apply the optimal cloning map to a Dicke state, the Hamming weight difference $\tilde{w}-w$ follows a \textbf{negative hypergeometric distribution}. It turns out that sampling from the negative hypergeometric distribution can be understood as a \textbf{P\'{o}lya urn model}. In particular, suppose that you have an urn with $(w+1)$ red balls and $(N-w+1)$ black balls. Now suppose that you repeat the following process $(M-N)$ times where you randomly draw a ball and then introduce an extra ball of that same color. If we define $\tilde{w}$ such that the urn ends with $(\tilde{w}+1)$ red balls and $(M-\tilde{w}+1)$ black balls, then it turns out that $\tilde{w}-w$ follows the exact negative hypergeometric distribution we see in the optimal cloning map. However, this process clearly respects composition, in the sense that it just repeats the ``add one ball'' procedure $(M-N)$ times.
\end{itemize}

\begin{proof}[Proof of Lemma \ref{lem:optimal-cloning-map-composition}]
Recall that the $N$-qubit symmetric subspace is spanned by operators of the form $\ket{\psi}\bra{\psi}^{\otimes N}$ for $\ket{\psi}\in\Cbb^2$ \cite{harrow2013church}. Therefore, it suffices to prove the lemma for input states of that form. Furthermore, since all the relevant channels are $\mathrm{SU}(2)$-covariant, it suffices to prove the statement for an input state of the form $\ket{\psi}\bra{\psi}^{\otimes N}$ for a single $\ket{\psi}$ value. In other words, it suffices to prove the lemma for the input state $\ket{D^{(N)}_0}\bra{D^{(N)}_0}$. We will actually prove it for a slightly more general input state of the form $\ket{D^{(N)}_w}\bra{D^{(N)}_w}$ for an arbitrary $0\le w\le N$. Therefore, we will prove that
\begin{equation}
    \mE_{\text{clone}}[N\to M]\left(\ket{D^{(N)}_w}\bra{D^{(N)}_w}\right) = \mE_{\text{clone}}[P\to M]\left(\mE_{\text{clone}}[N\to P]\left(\ket{D^{(N)}_w}\bra{D^{(N)}_w}\right)\right).
\end{equation}

\vspace{0.5\baselineskip}

We begin by applying the first optimal cloning map to the Dicke state (refer to Lemma \ref{lem:optimal-cloning-map-dicke-state-unified-presentation}):
\begin{equation}
    \mE_{\text{clone}}[N\to P]\left(\ket{D^{(N)}_w}\bra{D^{(N)}_w}\right) = \binom{P+1}{N+1}^{-1}\sum_{w'=0}^{P}\binom{w'}{w}\binom{P-w'}{N-w}\ket{D^{(P)}_{w'}}\bra{D^{(P)}_{w'}}.
\end{equation}
We now apply the second optimal cloning map to each Dicke state in this result (once again applying Lemma \ref{lem:optimal-cloning-map-dicke-state-unified-presentation}):
\begin{align}
    \mE_{\text{clone}}[P\to M]\left(\ket{D^{(P)}_{w'}}\bra{D^{(P)}_{w'}}\right) = \binom{M+1}{P+1}^{-1}\sum_{\tilde{w}=0}^{P}\binom{\tilde{w}}{w'}\binom{M-\tilde{w}}{P-w'}\ket{D^{(M)}_{\tilde{w}}}\bra{D^{(M)}_{\tilde{w}}}.
\end{align}
Applying these two formulas in sequence yields
\begin{equation}
    \mE_{\text{clone}}[P\to M]\left(\mE_{\text{clone}}[N\to P]\left(\ket{D^{(N)}_w}\bra{D^{(N)}_w}\right)\right) = \sum_{\tilde{w}=0}^{M}D(w,\tilde{w})\ket{D^{(M)}_{\tilde{w}}}\bra{D^{(M)}_{\tilde{w}}},
\end{equation}
where the coefficient of each output Dicke state is
\begin{equation}
    D(w,\tilde{w}) = \binom{M+1}{P+1}^{-1}\binom{P+1}{N+1}^{-1}\sum_{w'=0}^{P}\binom{\tilde{w}}{w'}\binom{M-\tilde{w}}{P-w'}\binom{w'}{w}\binom{P-w'}{N-w}.
\end{equation}

\vspace{0.5\baselineskip}

Now we need to show that this matches the coefficient that would result from $\mE_{\text{clone}}[N\to M]$, which is (again invoking Lemma \ref{lem:optimal-cloning-map-dicke-state-unified-presentation})
\begin{equation}
    \tilde{D}(w,\tilde{w}) = \binom{M+1}{N+1}^{-1}\binom{\tilde{w}}{w}\binom{M-\tilde{w}}{N-w}.
\end{equation}
Fortunately, this is straightforward. We can rearrange the binomial coefficients in $D(w,\tilde{w})$ as follows:
\begin{align}
    \binom{M+1}{P+1}\binom{P+1}{N+1} &= \binom{M+1}{N+1}\binom{M-N}{P-N} \\
    \binom{\tilde{w}}{w'}\binom{w'}{w} &= \binom{\tilde{w}}{w}\binom{\tilde{w}-w}{w'-w} \\
    \binom{M-\tilde{w}}{P-w'}\binom{P-w'}{N-w} &= \binom{M-\tilde{w}}{N-w}\binom{(M-\tilde{w})-(N-w)}{(P-w')-(N-w)}.
\end{align}
Therefore, we can rewrite $D(w,\tilde{w})$ as follows:
\begin{align}
    D(w,\tilde{w}) &= \binom{M+1}{N+1}^{-1}\binom{M-N}{M-P}^{-1}\sum_{w'=0}^{P}\binom{\tilde{w}}{w}\binom{\tilde{w}-w}{\tilde{w}-w'}\binom{M-\tilde{w}}{N-w}\binom{(M-\tilde{w})-(N-w)}{(M-\tilde{w})-(P-w')} \\
    &= \binom{M+1}{N+1}^{-1}\binom{\tilde{w}}{w}\binom{M-\tilde{w}}{N-w}\Bigg\{\binom{M-N}{P-N}^{-1}\sum_{w'=0}^{P}\binom{\tilde{w}-w}{w'-w}\binom{(M-\tilde{w})-(N-w)}{(P-w')-(N-w)}\Bigg\} \\
    &= \tilde{D}(\tilde{w})\Bigg\{\binom{M-N}{P-N}^{-1}\sum_{w'=0}^{P}\binom{\tilde{w}-w}{w'-w}\binom{(M-\tilde{w})-(N-w)}{(P-w')-(N-w)}\Bigg\}.
\end{align}
At this point, we can apply the \textbf{Vandermonde identity}, which says that
\begin{equation}
    \binom{m+n}{r} = \sum_{k=0}^{r}\binom{m}{k}\binom{n}{r-k}.
\end{equation}
As a brief aside, this is easy to prove by bijection. The left side counts the number of ways to choose $r$ objects out of a collection of $(m+n)$ objects. The right side does the same, but conditioned on the number $k$ of objects chosen out of the first $m$ objects.

\vspace{0.5\baselineskip}

We now apply the Vandermonde identity with
\begin{equation}
    m\mapsto\tilde{w}-w, \quad n\mapsto (M-\tilde{w})-(N-w), \quad r\mapsto P-N, \quad k\mapsto w'-w.
\end{equation}
We thus obtain
\begin{align}
    \binom{M-N}{P-N} &= \sum_{w'=w}^{w+(P-N)}\binom{\tilde{w}-w}{w'-w}\binom{(M-N)-(\tilde{w}-w)}{(P-N)-(w'-w)} \\
    &= \sum_{w'=0}^{P}\binom{\tilde{w}-w}{w'-w}\binom{(M-\tilde{w})-(N-w)}{(P-w')-(N-w)}.
\end{align}
Hence, the quantity in curly braces above equals $1$, so we conclude that $D(w,\tilde{w}) = \tilde{D}(w,\tilde{w})$. Therefore, the composition of two cloning maps $\mE_{\text{clone}}[P\to M]\circ\mE_{\text{clone}}[N\to P]$ has the same action on a Dicke state as the the single cloning map $\mE_{\text{clone}}[N\to M]$, so the two channels are equal.
\end{proof}

\appsubsec{Polynomial for the Optimal Cloning Map}
{subsec:optimal-cloning-map-polynomial}

As a reminder, the optimal cloning map from $N$ qudits to $M$ qudits takes the form
\begin{equation}
    \mE_{\text{clone}}[N\rightarrow M](\rho_N) = \frac{\binom{N+d-1}{N}}{\binom{M+d-1}{M}}\Pi^M_{\text{sym}}\left(\rho_N\otimes\Ibb_d^{\otimes(M-N)}\right)\Pi^M_{\text{sym}},
\end{equation}
where $\rho_N$ is a state in the symmetric subspace on $N$ qudits, $d$ is the dimension of the Hilbert space, and $\Pi^M_{\text{sym}}$ is the projector to the symmetric subspace on $M$ qudits \cite{Werner1998}.

\vspace{0.5\baselineskip}

The optimal cloning map has the following crucial properties:
\begin{itemize}
    \item $\mathbf{SU(d)}$ \textbf{covariance:} For any $U\in\text{SU}(d)$, applying $U^{\otimes N}$ to the input and then applying $\mE_{\text{clone}}[N\rightarrow M]$ yields the same outcome as applying $\mE_{\text{clone}}[N\rightarrow M]$ and then applying $U^{\otimes M}$ to the output. Written as an equation,
    \begin{equation}
        \mE_{\text{clone}}[N\rightarrow M]\left(U^{\otimes N}\rho_N\left(U^\dagger\right)^{\otimes N}\right) = U^{\otimes M}\mE_{\text{clone}}[N\rightarrow M](\rho_N)\left(U^\dagger\right)^{\otimes M} \quad \forall \,\, U\in\text{SU}(d).
    \end{equation}
    \item \textbf{Composition:} Applying multiple optimal cloning maps in succession yields the same result as just applying one optimal cloning map. Written as an equation,
    \begin{equation}
        \mE_{\text{clone}}[N\rightarrow M] = \mE_{\text{clone}}[P\rightarrow M]\circ \mE_{\text{clone}}[N\rightarrow P] =  \quad \forall \,\, N\le P\le M.
    \end{equation}
    For qubits, we explicitly proved this result as Lemma \ref{lem:optimal-cloning-map-composition}, but it holds more generally for qudits.
\end{itemize}

\vspace{0.5\baselineskip}

From now on, we specialize to qubits ($d=2$), meaning that the formula for the optimal cloning map simplifies as follows:
\begin{equation}
    \mE_{\text{clone}}[N\rightarrow M](\rho_N) = \frac{N+1}{M+1}\Pi^M_{\text{sym}}\left(\rho_N\otimes\Ibb_2^{\otimes(M-N)}\right)\Pi^M_{\text{sym}}.
\end{equation}
We recall the result of Lemma \ref{lem:optimal-cloning-map-dicke-state-unified-presentation}, which shows the result of applying the optimal cloning map to a single Dicke state. We will set $M=N+k$, so that $k\ge 0$ represents the number of added qubits:
\begin{equation}
    \mE_{\text{clone}}[N\rightarrow N+k]\left(\ket{D^{(N)}_w}\bra{D^{(N)}_w}\right) = \binom{N+k+1}{N+1}^{-1}\sum_{\tilde{w}=0}^{N+k}\binom{\tilde{w}}{w}\binom{N+k-\tilde{w}}{N-w}\ket{D^{(N+k)}_{\tilde{w}}}\bra{D^{(N+k)}_{\tilde{w}}}.
\end{equation}
There are a few remarkable observations we can make:
\begin{itemize}
    \item First, the coefficient of the Dicke state is nonzero if and only if $w\le\tilde{w}\le w+k$. In general, each added qubit can only maintain the Hamming weight or increase it by one. Therefore, applying this operation $k$ times to get the optimal cloning map $\mE_{\text{clone}}[N\to N+k]$ can only increase the Hamming weight by some integer from $0$ to $k$.
    \item Second, the coefficients follow a degree-$N$ polynomial in the Hamming weight $\tilde{w}$.
    \item Third, this polynomial is the unique degree-$N$ polynomial whose zeros are the integers $0\le\tilde{w}\le w-1$ and $w+k+1\le\tilde{w}\le N+k$, and whose normalization is chosen so that the coefficients for $w\le\tilde{w}\le w+k$ add up to $1$. (In general, a degree-$D$ polynomial is determined by $(D+1)$ parameters. One common way to specify these parameters is to specify the leading coefficient and the $D$ roots, i.e., $p(x) = a(x-r_1)(x-r_2)\cdots(x-r_D)$.) In particular, the zeros are at the integers adjacent to the Hamming weights with nonzero coefficients, and there are $w$ zeros on the side below Hamming weight $w$ and $(N-w)$ zeros on the side above Hamming weight $w+k$.
\end{itemize}
To highlight a special case of these observations, suppose that we apply the optimal cloning map to the state $\ket{0}^{\otimes N} = \ket{D^{(N)}_0}$. Then we obtain
\begin{equation}
    \mE_{\text{clone}}[N\to N+k]\left(\ket{0}\bra{0}^{\otimes N}\right) = \binom{N+k+1}{N+1}^{-1}\sum_{\tilde{w}=0}^{k}\binom{N+k-\tilde{w}}{N}\ket{D^{(N+k)}_{\tilde{w}}}\bra{D^{(N+k)}_{\tilde{w}}}.
\end{equation}
In this special case, we see that the Hamming weights greater than $k$ do not appear at all. In other words, the ``wrongest'' Dicke states (with the largest Hamming weights) are completely absent from the output.

\vspace{0.5\baselineskip}

Generally speaking, the famous fact about the optimal cloning map is that the coefficient of $\tilde{w}=0$ is as large as possible in this setting. (The optimal cloning map is named as such precisely because it achieves the maximum possible fidelity with the target state $\ket{0}^{\otimes(N+k)}$, which precisely equals the coefficient of $\tilde{w}=0$ \cite{Werner1998}.) However, we find it especially interesting that the polynomial structure of the full profile of coefficients is illuminated by looking at the opposite end and noticing that the highest Hamming weights have zero contribution.

\vspace{0.5\baselineskip}

The fact that the Hamming weights $\tilde{w}\ge k+1$ have zero contribution even helped us come up with the implementation of the optimal cloning map in Appendix \ref{sec:friendly-implementations}\ref{subsec:optimal-cloning-implementation}. For example, when considering optimal cloning from $1$ qubit to $2$ qubits, some tempting ideas can be ruled out immediately because they will clearly map $\ket{0}\bra{0}$ to a state with nonzero coefficient for $\ket{11}\bra{11}$, whereas the intended output has zero coefficient for this state. (In particular, the intended output is $\frac{2}{3}\ket{00}\bra{00} + \frac{1}{3}\ket{\Psi^+}\bra{\Psi^+}$, where $\ket{\Psi^+} \coloneqq \frac{\ket{01}+\ket{10}}{\sqrt{2}} = \ket{D^{(2)}_1}$.)

\appsubsec{Polynomial for the Optimal Measure-and-Prepare Channel}
{subsec:optimal-mp-channel-polynomial}

In Lemma \ref{lem:optimal-mp-channel-dicke-state-unified-presentation}, we derived the formula for the optimal measure-and-prepare channel. The optimal measure-and-prepare channel $\mE_{\text{MP}}[N\to M]$ acts as follows on a single Dicke state:
\begin{equation}
    \mE_{\text{MP}}[N\to M]\left(\ket{D^{(N)}_w}\bra{D^{(N)}_w}\right) = \binom{N+M+1}{N+1}^{-1}\sum_{\tilde{w}=0}^{M}\binom{w+\tilde{w}}{w}\binom{(N-w)+(M-\tilde{w})}{N-w}\ket{D^{(M)}_{\tilde{w}}}\bra{D^{(M)}_{\tilde{w}}}.
\end{equation}
Like the discarding and optimal cloning maps, this channel is $\mathrm{SU}(2)$-covariant. However, unlike the discarding and optimal cloning maps, it does not behave nicely with respect to composition.

\vspace{0.5\baselineskip}

We can make some interesting observations that are analogous to those we made in Appendix \ref{subsec:optimal-cloning-map-polynomial}:
\begin{itemize}
    \item First, the coefficient of the Dicke state is nonzero for all $0\le\tilde{w}\le M$. This is in stark contrast to the discarding and optimal cloning maps, which only yield nonzero coefficients in a restricted range.
    \item Second, the coefficients follow a degree-$N$ polynomial in $\tilde{w}$. This is similar to the optimal cloning map.
    \item Third, this polynomial is the unique degree-$N$ polynomial whose zeros are the integers $-w\le\tilde{w}\le -1$ and $M+1\le\tilde{w}\le M+N-w$, and whose normalization is chosen so that the coefficients for $0\le\tilde{w}\le M$ add up to $1$. In particular, the zeros are at the integers that are immediately ``out of bounds'' from the allowed Hamming weights, and there are $w$ zeros on the side below Hamming weight $0$ and $(N-w)$ zeros on the side above Hamming weight $M$. This is similar to the optimal cloning map, except for how the zeros are located.
\end{itemize}
To highlight a special case of these observations, suppose that we apply the optimal measure-and-prepare channel to the state $\ket{0}^{\otimes N} = \ket{D^{(N)}_0}$. Then we obtain
\begin{equation}
    \mE_{\text{MP}}[N\to M]\left(\ket{0}\bra{0}^{\otimes N}\right) = \binom{N+M+1}{N+1}^{-1}\sum_{\tilde{w}=0}^{M}\binom{N+M-\tilde{w}}{N}\ket{D^{(M)}_{\tilde{w}}}\bra{D^{(M)}_{\tilde{w}}}.
\end{equation}
In this case, the zeros of $\binom{N+M-\tilde{w}}{N}$ as a degree-$N$ polynomial in $\tilde{w}$ are exactly the integers $N+1\le\tilde{w}\le N+M$.

\appsubsec{Extreme Dilution Plus Extreme Concentration Equals Measure-and-Prepare Conversion}
{subsec:extreme-dilution-extreme-concentration}

One interesting observation that can be made from the maximum achievable conversion rates is that
\begin{equation}
    R^{\text{MP}}(\lambda_{\text{in}}\to\lambda_{\text{out}}) = \lim_{\lambda_{\text{mid}}\to 0^+} R^{\text{dilut}}(\lambda_{\text{in}}\to\lambda_{\text{mid}})R^{\text{conc}}(\lambda_{\text{mid}}\to\lambda_{\text{out}}).
\end{equation}
This suggests that we could perhaps interpret the optimal measure-and-prepare protocol as extreme dilution, followed by extreme concentration. This would in turn suggest that we could perhaps interpret $\mE_{\text{MP}}$ as extreme cloning, followed by extreme discarding. This would provide an alternate (approximate) implementation for the optimal measure-and-prepare channel, beyond the usual way that uses the SJSM.

\vspace{0.5\baselineskip}

Here, we will make this connection more explicit. In particular, we will prove the following:

\begin{lemma}[extreme cloning, followed by extreme discarding, yields optimal measure-and-prepare]
For any fixed integers $N,M\ge 0$,
\begin{equation}
    \mE_{\text{MP}}[N\to M] = \lim_{P\to\infty}\mE_{\text{discard}}[P\to M]\circ\mE_{\text{clone}}[N\to P].
\end{equation}
\label{lem:optimal-mp-equals-extreme-cloning-then-extreme-discarding}
\end{lemma}

In other words, if we apply the optimal cloning map from the initial $N$ qubits to an extremely large number of qubits $P$ (in particular, much larger than $N$ and $M$), and then apply the discarding map to come back down to $M$ qubits, the result should approximately match the optimal measure-and-prepare channel from $N$ qubits to $M$ qubits.

\begin{proof}
For the same reason as we discussed at the beginning of the proof of Lemma \ref{lem:optimal-cloning-map-composition}, it suffices to prove the lemma for the input state $\ket{D^{(N)}_0}\bra{D^{(N)}_0}$. Once again, we will prove it for a slightly more general input state of the form $\ket{D^{(N)}_w}\bra{D^{(N)}_w}$ for an arbitrary $0\le w\le N$. Therefore, we will prove that
\begin{equation}
    \mE_{\text{MP}}[N\to M]\left(\ket{D^{(N)}_w}\bra{D^{(N)}_w}\right) = \lim_{P\to\infty}\mE_{\text{discard}}[P\to M]\left(\mE_{\text{clone}}[N\to P]\left(\ket{D^{(N)}_w}\bra{D^{(N)}_w}\right)\right).
\end{equation}

\vspace{0.5\baselineskip}

We begin by applying the optimal cloning map to a Dicke state (refer to Lemma \ref{lem:optimal-cloning-map-dicke-state-unified-presentation}):
\begin{equation}
    \mE_{\text{clone}}[N\to P]\left(\ket{D^{(N)}_w}\bra{D^{(N)}_w}\right) = \binom{P+1}{N+1}^{-1}\sum_{w'=0}^{P}\binom{w'}{w}\binom{P-w'}{N-w}\ket{D^{(P)}_{w'}}\bra{D^{(P)}_{w'}}.
\end{equation}
We now apply the discarding map to each resulting Dicke state (refer to Lemma \ref{lem:discarding-map-dicke-state-unified-presentation}):
\begin{equation}
    \mE_{\text{discard}}[P\to M]\left(\ket{D^{(P)}_{w'}}\bra{D^{(P)}_{w'}}\right) = \binom{P}{M}^{-1}\sum_{\tilde{w}=0}^{M}\binom{w'}{\tilde{w}}\binom{P-w'}{M-\tilde{w}}\ket{D^{(M)}_{\tilde{w}}}\bra{D^{(M)}_{\tilde{w}}}.
\end{equation}
Applying these two formulas in sequence yields
\begin{equation}
    \mE_{\text{discard}}[P\to M]\left(\mE_{\text{clone}}[N\to P]\left(\ket{D^{(N)}_w}\bra{D^{(N)}_w}\right)\right) = \sum_{\tilde{w}=0}^{M}D(w,\tilde{w})\ket{D^{(M)}_{\tilde{w}}}\bra{D^{(M)}_{\tilde{w}}},
\end{equation}
where the coefficient of each output Dicke state is
\begin{equation}
    D(w,\tilde{w}) = \binom{P}{M}^{-1}\binom{P+1}{N+1}^{-1}\sum_{w'=0}^{P}\binom{w'}{\tilde{w}}\binom{P-w'}{M-\tilde{w}}\binom{w'}{w}\binom{P-w'}{N-w}.
\end{equation}

\vspace{0.5\baselineskip}

Now we need to approximate the quantity $D(w,\tilde{w})$ in the regime where $P >> N,M$. Our goal is to obtain the expression that would match the action of $\mE_{\text{MP}}[N\to M]$, which is (refer to Lemma \ref{lem:optimal-mp-channel-dicke-state-unified-presentation}
\begin{align}
    \tilde{D}(w,\tilde{w}) &= \binom{N+M+1}{N+1}^{-1}\binom{w+\tilde{w}}{w}\binom{(N-w)+(M-\tilde{w})}{N-w} \\
    &= \frac{N+1}{N+M+1}\binom{N+M}{N}^{-1}\binom{w+\tilde{w}}{w}\binom{(N-w)+(M-\tilde{w})}{N-w} \\
    &= \frac{N+1}{N+M+1}\binom{N+M}{w+\tilde{w}}^{-1}\binom{N}{w}\binom{M}{\tilde{w}}.
\end{align}

\vspace{0.5\baselineskip}

We first approximate the relevant binomial coefficients. In addition to assuming that $P >> N,M$, we also assume that $w' >> w,\tilde{w}$, since that will be true for the vast majority of terms in the summation:
\begin{align}
    \binom{P}{M} &\approx \frac{P^M}{M!} & \binom{P}{N} &\approx \frac{P^N}{N!} \\
    \binom{P-w'}{N-w} &\approx \frac{(P-w')^{N-w}}{(N-w)!} & \binom{P-w'}{M-\tilde{w}} &\approx \frac{(P-w')^{M-\tilde{w}}}{(M-\tilde{w})!} \\
    \binom{w'}{w} &\approx \frac{(w')^w}{w!} & \binom{w'}{\tilde{w}} &\approx \frac{(w')^{\tilde{w}}}{\tilde{w}!}.
\end{align}
Plugging these approximations into the $D(w,\tilde{w})$ formula yields
\begin{align}
    D(w,\tilde{w}) &\approx \frac{N+1}{P+1}\frac{M!}{P^M}\frac{N!}{P^N}\sum_{w'=0}^{P}\frac{(w')^{\tilde{w}}}{\tilde{w}!}\frac{(P-w')^{M-\tilde{w}}}{(M-\tilde{w})!}\frac{(w')^w}{w!}\frac{(P-w')^{N-w}}{(N-w)!} \\
    &= \frac{N+1}{P+1}\frac{1}{P^{N+M}}\binom{N}{w}\binom{M}{\tilde{w}}\sum_{w'=0}^{P}(w')^{w+\tilde{w}}(P-w')^{(N+M)-(w+\tilde{w})} \\
    &= \frac{N+1}{P+1}\binom{N}{w}\binom{M}{\tilde{w}}\sum_{w'=0}^{P}\left(\frac{w'}{P}\right)^{w+\tilde{w}}\left(1-\frac{w'}{P}\right)^{(M+N)-(w+\tilde{w})}.
\end{align}
Now we approximate the sum as an integral, as follows:
\begin{align}
    &\quad\,\,\,\, \frac{1}{P+1}\sum_{w'=0}^{P}\left(\frac{w'}{P}\right)^{w+\tilde{w}}\left(1-\frac{w'}{P}\right)^{(N+M)-(w+\tilde{w})} \\
    &\approx \int_{0}^{1}x^{w+\tilde{w}}(1-x)^{(N+M)-(w+\tilde{w})} \\
    &= \frac{1}{N+M+1}\binom{N+M}{w+\tilde{w}}^{-1}.
\end{align}

\vspace{0.5\baselineskip}

To evaluate the above integral, we invoked the formula
\begin{equation}
    \int_{0}^{1}x^a(1-x)^b\,dx = \frac{1}{a+b+1}\binom{a+b}{a}^{-1} = \frac{a!\,b!}{(a+b+1)!}.
\end{equation}
As a brief digression, this integral has a very elegant combinatorial interpretation. Suppose you choose $(a+b+1)$ points $x_0,\cdots,x_{a+b}$ independently and uniformly at random on the interval $[0,1]$. What is the probability that $x_0 > x_j$ for all $1\le j\le a$ but $x_0 < x_j$ for all $a+1\le j\le a+b$?
\begin{itemize}
    \item For the first method of evaluation, suppose that $x_0 = x$. Then $x_j < x_0$ with probability $x$, and $x_j > x_0$ with probability $1-x$, so the desired property is satisfied with conditional probability $x^a(1-x)^b$. Since $x$ is sampled uniformly over $[0,1]$, the overall probability is the integral on the left.
    \item For the second method of evaluation, notice that all $(a+b+1)!$ orderings of the points are equally likely. There are $a!$ ways to order $x_j$ for $1\le j\le a$, and there are $b!$ ways to order $x_j$ for $a+1\le j\le a+b$, so there are $a!\,b!$ good orderings. Therefore, the desired property is satisfied with probability given by the fraction on the right.
\end{itemize}
Since we have computed the same probability in two different ways, the two expressions must be equal. By the way, although the combinatorial interpretation would no longer apply, this integral formula is actually valid for all $a,b\in\Cbb$ with $\text{Re}[a],\text{Re}[b] > -1$, except each factorial $z!$ needs to be replaced with a gamma function $\Gamma(z+1)$.

\vspace{0.5\baselineskip}

Finally, we can complete the calculation as follows:
\begin{align}
    D(w,\tilde{w}) &\approx (N+1)\binom{N}{w}\binom{M}{\tilde{w}}\Bigg\{\frac{1}{P+1}\sum_{w'=0}^{P}\left(\frac{w'}{P}\right)^{w+\tilde{w}}\left(1-\frac{w'}{P}\right)^{(M+N)-(w+\tilde{w})}\Bigg\} \\
    &\approx (N+1)\binom{N}{w}\binom{M}{\tilde{w}}\Bigg\{\frac{1}{N+M+1}\binom{N+M}{w+\tilde{w}}^{-1}\Bigg\} \\
    &= \frac{N+1}{N+M+1}\binom{N+M}{w+\tilde{w}}^{-1}\binom{N}{w}\binom{M}{\tilde{w}}.
\end{align}
We see that this approximation matches the quantity $\tilde{D}(w,\tilde{w})$ that we defined above. We conclude that, in the limit as $P >> N,M$, the optimal cloning map followed by the discarding map yields the optimal measure-and-prepare channel.
\end{proof}

The above proof addresses the setting where $N$ and $M$ are fixed. But what if $N$ and $M$ are asymptotic quantities, as they are for our qubit linear-rate conversion procedures? If so, it is natural to ask how $P$ should scale relative to $N$ and $M$. By analyzing the various approximations in the proof above, we can then establish a growth rate for $P$ that ensures that the error vanishes as $N,M\to\infty$.

\vspace{0.5\baselineskip}

In general, $\binom{A}{B} = \frac{A^B}{B!}\exp[O(B^2/A)]$. Therefore, we should definitely have $P\in\omega(\max\{M,N\}^2)$. This ensures that the approximations for $\binom{P}{N}$ and $\binom{P}{M}$ are accurate. Furthermore, this will automatically ensure that the other binomial coefficient approximations are accurate for the vast majority of values $0\le\tilde{w}\le P$. Some of these approximations will fail for $\tilde{w}$ extremely close to $0$ or $P$, but these terms are only a $O(\max\{M,N\}^2/P)$ fraction out of the entire sum. Finally, the approximation of the summation as an integral has $O(P^{-1})$ relative error, since it is a Riemann sum with $P+1$ terms. We conclude that all the approximations together have relative error on the order of $O(\max\{M,N\}^2/P)$.

\vspace{0.5\baselineskip}

However, the total variation distance between the target output state and the actual output state can be affected by up to $O(\max\{M,N\})$ contributions from different Hamming weights. Therefore, to ensure that the total variation distance goes to zero for any input state diagonal in the Dicke state basis (hence the output state is also diagonal in the Dicke state basis), it suffices to have
\begin{equation}
    \boxed{P \in \omega\left(\max\{N,M\}\right)^3}.
\end{equation}
We do not offer any guarantee that this is the smallest allowed asymptotic scaling of $P$, merely that it is indeed sufficient. In fact, we suspect that a substantially smaller asymptotic scaling behavior for $P$ will suffice, but we do not attempt to study this further.

\appsubsec{Optimal Cloning and the Standard Joint Symmetric Measurement}
{subsec:optimal-cloning-sjsm}

We now present a fun fact that relates the optimal cloning map to the standard joint symmetric measurement (SJSM) (refer to Definition \ref{def:standard-joint-symmetric-measurement}), which is the ``measure'' step of the measure-and-prepare channels $\mE_{\text{MP}}$ and $\mE_{\text{WW}}$. In particular, we will prove the following:

\begin{lemma}[optimal cloning does not affect SJSM outcome distribution]
For any integers $M\ge N\ge 0$, and any state $\rho_N$ in the $N$-qubit symmetric subspace, applying the $M$-qubit SJSM to $\mE_{\text{clone}}[N\to M](\rho_N)$ yields the same distribution over the unit sphere as applying the $N$-qubit SJSM to $\rho_N$.
\label{lem:optimal-cloning-sjsm}
\end{lemma}

Intuitively, this says that the optimal cloning map has no effect on the optimal procedure for direction estimation, which is to apply the SJSM. We will actually present two proofs of this fact. The first proof uses Dicke state manipulations of the type common throughout this paper, whereas the second proof cleverly applies our earlier demonstration of the optimal measure-and-prepare channel as extreme cloning followed by extreme discarding.

\begin{proof}[Proof using Dicke state manipulations]
We begin with the following three simplifications:
\begin{itemize}
    \item For the same reason as we discussed at the beginning of the proof of Lemma \ref{lem:optimal-cloning-map-composition}, it suffices to prove the lemma for the input state $\ket{D^{(N)}_0}\bra{D^{(N)}_0}$. Once again, we will prove it for a slightly more general input state of the form $\ket{D^{(N)}_w}\bra{D^{(N)}_w}$ for an arbitrary $0\le w\le N$.
    \item If the input is a Dicke state $\ket{D^{(N)}_w}\bra{D^{(N)}_w}$, then the distribution over the azimuthal angle $\phi\in[0,2\pi)$ will be uniform in either case. It thus suffices to check that the distributions over the polar angle $\theta\in[0,\pi]$ will match in the two cases.
    \item Due to the composition property of the optimal cloning map, it suffices to prove the lemma for the optimal cloning map that adds a single qubit, i.e., $M=N+1$.
\end{itemize}

In the proof of Lemma \ref{lem:optimal-mp-channel-dicke-state-unified-presentation}, we defined and evaluated the state
\begin{align}
    \sigma(N,\theta) &\coloneqq \frac{1}{2\pi}\int_{0}^{2\pi}\ket{\Psi(\theta,\phi)}\bra{\Psi(\theta,\phi)}^{\otimes N}\,d\phi \\
    &= 2^{-N}\sum_{w=0}^{N}\binom{N}{w}(1+\cos\theta)^{N-w}(1-\cos\theta)^w\ket{D^{(N)}_w}\bra{D^{(N)}_w}.
\end{align}
From this, we obtained the following quantity, which was crucial in the proof of Lemma \ref{lem:optimal-mp-channel-dicke-state-unified-presentation}:
\begin{equation}
    \text{Tr}\left[\sigma(N,\theta)\ket{D^{(N)}_w}\bra{D^{(N)}_w}\right] = 2^{-N}\binom{N}{w}(1+\cos\theta)^{N-w}(1-\cos\theta)^w.
\end{equation}
Therefore, if we directly apply the $N$-qubit SJSM, the resulting probability density over $\theta$ is
\begin{equation}
    p_1(\theta) = \frac{N+1}{2}\text{Tr}\left[\sigma(N,\theta)\ket{D^{(N)}_w}\bra{D^{(N)}_w}\right] = \frac{N+1}{2^{N+1}}\binom{N}{w}(1+\cos\theta)^{N-w}(1-\cos\theta)^w.
\end{equation}
In contrast, suppose that we first apply the optimal cloning map $\mE_{\text{clone}}[N\to N+1]$. Then by Lemma \ref{lem:optimal-cloning-map-dicke-state-unified-presentation} with $M=N+1$,
\begin{equation}
    \mE_{\text{clone}}[N\to N+1]\left(\ket{D^{(N)}_w}\bra{D^{(N)}_w}\right) = \frac{N-w+1}{N+2}\ket{D^{(N+1)}_w}\bra{D^{(N+1)}_w} + \frac{w+1}{N+2}\ket{D^{(N+1)}_{w+1}}\bra{D^{(N+1)}_{w+1}}.
\end{equation}
Therefore, if we apply the $(N+1)$-qubit SJSM to this new state, the resulting probability density over $\theta$ is
\begin{align}
    p_2(\theta) &= \frac{N+2}{2}\text{Tr}\left[\sigma(N+1,\theta)\mE_{\text{clone}}[N\to N+1]\left(\ket{D^{(N)}_w}\bra{D^{(N)}_w}\right)\right] \\
    &= \frac{N-w+1}{2}\text{Tr}\left[\sigma(N+1,\theta)\ket{D^{(N+1)}_w}\bra{D^{(N+1)}_w}\right] + \frac{w+1}{2}\text{Tr}\left[\sigma(N+1,\theta)\ket{D^{(N+1)}_{w+1}}\bra{D^{(N+1)}_{w+1}}\right] \\
    &= \frac{N-w+1}{2}\left[2^{-(N+1)}\binom{N+1}{w}(1+\cos\theta)^{(N+1)-w}(1-\cos\theta)^w\right] \\
    & \quad\,\, + \frac{w+1}{2}\left[2^{-(N+1)}\binom{N+1}{w+1}(1+\cos\theta)^{(N+1)-(w+1)}(1-\cos\theta)^{w+1}\right] \\
    &= 2^{-(N+2)}(1+\cos\theta)^{N-w}(1-\cos\theta)^w \\
    & \quad\,\, \left[(N-w+1)\binom{N+1}{w}(1+\cos\theta) + (w+1)\binom{N+1}{w+1}(1-\cos\theta)\right] \\
    &= 2^{-(N+2)}(1+\cos\theta)^{N-w}(1-\cos\theta)^w \\
    & \quad\,\, \left[(N+1)\binom{N}{w}(1+\cos\theta) + (N+1)\binom{N}{w}(1-\cos\theta)\right] \\
    &= \frac{N+1}{2^{N+1}}\binom{N}{w}(1+\cos\theta)^{N-w}(1-\cos\theta)^w.
\end{align}
Thus $p_2(\theta) = p_1(\theta)$, so we obtain the same distribution either way.
\end{proof}

\begin{proof}[Proof using extreme cloning and extreme discarding]
For any nonnegative integers $N,P,M$ such that $P\ge N$,
\begin{align}
    & \quad\,\, \mE_{\text{MP}}[P\to M]\circ\mE_{\text{clone}}[N\to P] \\
    &= \left(\lim_{Q\to\infty}\mE_{\text{discard}}[Q\to M]\circ\mE_{\text{clone}}[P\to Q]\right)\circ\mE_{\text{clone}}[N\to P] \\
    &= \lim_{Q\to\infty}\mE_{\text{discard}}[Q\to M]\circ\mE_{\text{clone}}[P\to Q]\circ\mE_{\text{clone}}[N\to P] \\
    &= \lim_{Q\to\infty}\mE_{\text{discard}}[Q\to M]\circ\mE_{\text{clone}}[N\to Q] \\
    &= \mE_{\text{MP}}[N\to M].
\end{align}
For clarity, the first and fourth equalities use Lemma \ref{lem:optimal-mp-equals-extreme-cloning-then-extreme-discarding}, and the third equality uses the composition property of the optimal cloning map.

\vspace{0.5\baselineskip}

Therefore, regardless of the final number of output qubits, the optimal cloning map can be absorbed into the optimal measure-and-prepare channel. Thus, if $p_1(\hat{n})$ is the probability density produced by the SJSM applied directly to $\rho_N$, whereas $p_2(\hat{n})$ is the probability density produced by the SJSM applied to $\mE_{\text{clone}}[N\to P](\rho_N)$, then
\begin{equation}
    \int_{\Sbb^2}p_1(\hat{n})\left(\ket{1}\bra{1}_{\hat{n}}\right)^{\otimes M}\,d\Omega = \int_{\Sbb^2}p_2(\hat{n})\left(\ket{1}\bra{1}_{\hat{n}}\right)^{\otimes M}\,d\Omega \quad \forall \, M\in\Zbb_{\ge 0}.
\end{equation}
This implies that the distributions $p_1(\hat{n})$ and $p_2(\hat{n})$ have identical moments. Since they are both continuous and bounded probability density functions over a compact space, this in turn implies that they must be equal.
\end{proof}

\appsubsec{Connection Between Discarding and Optimal Cloning}
{subsec:connection-discarding-optimal-cloning}

We conclude by highlighting a fascinating connection between the optimal cloning map and the discarding map. In particular, the discarding map that removes a single qubit acts as follows on a Dicke state:
\begin{equation}
    \mE_{\text{discard}}[N\to N-1]\left(\ket{D^{(N)}_w}\bra{D^{(N)}_w}\right) = \frac{w}{N}\ket{D^{(N-1)}_{w-1}}\bra{D^{(N-1)}_{w-1}} + \frac{N-w}{N}\ket{D^{(N-1)}_w}\bra{D^{(N-1)}_w}.
\end{equation}
This action is actually very intuitive. Remember that $\ket{D^{(N)}_w}$ is the uniform superposition of all computational basis states with $(N-w)$ zeros and $w$ ones. Now suppose you have $(N-w)$ zeros and $w$ ones, and you randomly choose one of them to discard. Then with probability $w/N$, you discard a $1$, leaving you with $(N-w)$ zeros and $(w-1)$ ones (this corresponds to the first term), and with probability $(N-w)/N$, you discard a $0$, leaving you with $(N-w-1)$ zeros and $w$ ones (this corresponds to the second term).

\vspace{0.5\baselineskip}

Furthermore, when we apply the optimal cloning map to add just one qubit to a Dicke state, we obtain the following formula:
\begin{equation}
    \mE_{\text{clone}}[N\rightarrow N+1]\left(\ket{D^{(N)}_w}\bra{D^{(N)}_w}\right) = \frac{N-w+1}{N+2}\ket{D^{(N+1)}_w}\bra{D^{(N+1)}_w} + \frac{w+1}{N+2}\ket{D^{(N+1)}_{w+1}}\bra{D^{(N+1)}_{w+1}}.
\end{equation}
It is worth remembering that, with each qubit added, the optimal cloning map only ever maintains the Hamming weight or increases it by $1$.

\vspace{0.5\baselineskip}

Now consider the following question. Suppose that you have a Dicke state $\ket{D^{(N)}_w}$, but you have no information about $w$. More formally, you have a uniform prior $\Pbb(w=x) = 1/(N+1)$ on the integers $0\le x\le N$. Now suppose that one qubit is lost, and then you measure the Hamming weight $w'$ of the $(N-1)$-qubit state. What can you say about the original Hamming weight $w$?

\vspace{0.5\baselineskip}

Based on the action of the discarding map, we have the following conditional probabilities:
\begin{equation}
    \Pbb(w'=w_0-1 \,|\, w=x) = \frac{x}{N}, \quad\quad \Pbb(w'=x \,|\, w=x) = \frac{N-x}{N}.
\end{equation}
Therefore, by Bayes' theorem, we can compute the posterior probabilities as follows:
\begin{align}
    \Pbb(w=y \,|\, w'=y) &= \frac{\Pbb(w'=y \,|\, w=y)\Pbb(w=y)}{\Pbb(w'=y)} \\
    &= \frac{\Pbb(w'=y \,|\, w=y)\Pbb(w=y)}{\Pbb(w'=y \,|\, w=y)\Pbb(w=y) + \Pbb(w'=y \,|\, w=y+1)\Pbb(w=y+1)} \\
    &= \frac{\frac{N-y}{N}\cdot\frac{1}{N+1}}{\frac{N-y}{N}\cdot\frac{1}{N+1} + \frac{y+1}{N}\cdot\frac{1}{N+1}} \\
    &= \frac{N-y}{(N-y) + (y+1)} \\
    &= \frac{N-y}{N+1}
\end{align}
\begin{align}
    \Pbb(w=y+1 \,|\, w'=y) &= \frac{\Pbb(w'=y \,|\, w=y+1)\Pbb(w=y+1)}{\Pbb(w'=y)} \\
    &= \frac{\Pbb(w'=y \,|\, w=y+1)\Pbb(w=y+1)}{\Pbb(w'=y \,|\, w=y+1)\Pbb(w=y+1) + \Pbb(w'=y \,|\, w=y)\Pbb(w=y)} \\
    &= \frac{\frac{y+1}{N}\cdot\frac{1}{N+1}}{\frac{y+1}{N}\cdot\frac{1}{N+1} + \frac{N-y}{N}\cdot\frac{1}{N+1}} \\
    &= \frac{y+1}{(y+1) + (N-y)} \\
    &= \frac{y+1}{N+1}.
\end{align}
Notice that the posterior probabilities corresponds precisely to the coefficients of the optimal cloning map $\mE_{\text{clone}}[N-1\rightarrow N]$ applied to the Dicke state $\ket{D^{(N-1)}_{w'}}$. As a result, we can understand the optimal cloning map as an application of Bayes' theorem to attempt to reconstruct a state in the symmetric subspace after a qubit is lost.

\vspace{0.5\baselineskip}

In fact, this connection via Bayes' theorem can be further understood as a consequence of the fact that the optimal cloning map is the \textbf{Petz recovery map} of the discarding map \cite{marvian2016clocks}. In particular, in the classical setting (restricting to diagonal density matrices and stochastic maps), the Petz recovery map reduces to an application of Bayes' theorem. This is exactly what we see in the above discussion regarding the discarding map and the optimal cloning map, since we concern ourselves with states that indicate a single direction (which are diagonal in the Dicke state basis defined with respect to that direction) and $\mathrm{SU}(2)$-covariant channels (which will maintain the preferred direction of such a state).

\vspace{0.5\baselineskip}

We will say a few more words about the Petz recovery map for the curious reader. In general, if you have a channel $\mE$ such that $\mE(\rho_0) = \sigma_0$, then the Petz recovery map corresponding to $\mE$ and $\rho_0$ is given by
\begin{equation}
    \mE_{\text{Petz}}(\cdot) = \rho_0^{1/2}\mE^\dagger\left(\sigma_0^{-1/2}(\cdot)\sigma_0^{-1/2}\right)\rho_0^{1/2},
\end{equation}
where $\mE^\dagger$ denotes the \textbf{adjoint channel} of $\mE$.

\vspace{0.5\baselineskip}

Before continuing, we offer a brief aside on the adjoint channel. If we study the \textbf{Schr\"{o}dinger picture}, where states evolve and observables remain fixed, we apply $\mE$ to states. But if we study the complementary \textbf{Heisenberg picture}, where states remain fixed and observables evolve, we apply $\mE^\dagger$ to observables. If $\mE$ has Kraus representation $\mE(\cdot) = \sum_{j}K_j(\cdot)K_j^\dagger$, then $\mE^\dagger$ has Kraus representation $\mE(\cdot) = \sum_{j}K_j^\dagger(\cdot)K_j$. The trace-preserving property of a quantum channel $\mE$ ensures that a state always stays properly normalized. This corresponds to the unital property of the adjoint channel $\mE^\dagger$, which ensures that the identity observable always stays fixed.

\vspace{0.5\baselineskip}

The Petz recovery map can be understood as an attempt to undo the action of the channel $\mE$ in the vicinity of the state $\rho_0$. (For the state $\rho_0$, it achieves this perfectly, since it is easy to verify that $\mE_{\text{Petz}}(\sigma_0) = \rho_0$.) As a result, the Petz recovery map finds numerous applications in quantum information theory. For example, for channels that conserve certain QFI metrics on a one-parameter family of states, the Petz recovery map can be used to perfectly undo the channel for any state on that one-parameter family \cite{Gao2023}. (However, the Petz recovery map does NOT always help you recover the original state. For example, the single-qubit depolarizing channel is actually its own Petz recovery map if the maximally mixed state is used as the basepoint. But clearly qubit depolarization cannot be ``fixed'', and certainly not by depolarizing further, which is what the Petz recovery map would do.)

\vspace{0.5\baselineskip}

In fact, in the classical setting (restricting to diagonal density matrices and stochastic maps), the Petz recovery map reduces to $\mE_{\text{Petz}}(\cdot) = \rho_0\mE^\dagger(\sigma_0(\cdot))$, which can be understood as an application of Bayes' theorem. More precisely, suppose a random variable $X$ with prior probability distribution $p^{(0)}$ is mapped to a random variable $Y$ with probability distribution $q^{(0)}$ via a fixed transition matrix $T$, so $q^{(0)} Tp^{(0)}$, or written entrywise,
\begin{equation}
    q^{(0)}_j = \sum_{k}T_{jk}p^{(0)}_k
\end{equation}
Then suppose that some observation of $Y$ tells us that $Y$ follows a new distribution $q^{(1)}$. (For example, perhaps the observation tells us $Y$ exactly, in which case $q^{(1)}$ would be fully concentrated on that single element in the sample space of $Y$.) Then the posterior probability distribution $p^{(1)}$ of $X$ has the entrywise formula
\begin{equation}
    p^{(1)}_k = p^{(0)}_k\sum_{j}\frac{q^{(1)}_j}{q^{(0)}_j}T_{jk}.
\end{equation}
Written using matrix operations, the posterior distribution has the formula
\begin{equation}
    p^{(1)} = p^{(0)}\odot T^\intercal\left(q^{(1)}\oslash q^{(0)}\right),
\end{equation}
where $\odot$ and $\oslash$ represent entrywise multiplication and division, respectively. Notice how this precisely matches the form of the Petz recovery map when restricted to classical probability distributions and stochastic maps.

\vspace{0.5\baselineskip}

Before observing anything about $Y$, meaning that $Y$ still has distribution $q^{(0)}$, the posterior $p^{(1)}$ of $X$ should clearly just match the prior $p^{(0)}$. (This corresponds to the fact that the Petz recovery map returns $\sigma_0$ perfectly back to $\rho_0$.) But after observing some additional information about $Y$, which will make $q^{(1)}\neq q^{(0)}$, the posterior $p^{(1)}$ of $X$ should update accordingly.

\end{document}